\documentclass[11pt]{article}

\usepackage{research5}
\usepackage{epsfig}
\usepackage{psfrag}
\usepackage{amsthm}
\usepackage{IEEEtrantools}
\usepackage{verbatim}

\theoremstyle{plain}
\newtheorem{thm}{Theorem}[section]
\newtheorem{prp}{Proposition}[section]
\newtheorem{cor}{Corollary}[section]
\newtheorem{lm}{Lemma}[section]
\theoremstyle{definition}
\newtheorem{dfn}{Definition}[section]

\newtheorem{rdc}{Reduction}[section]
\theoremstyle{remark}
\newtheorem{rmk}{Remark}[section]

\title{Sending a Bi-Variate Gaussian over a Gaussian MAC}

\author{Amos Lapidoth \and Stephan Tinguely}

\date{}

\begin{document}

\maketitle

\begin{abstract}
\renewcommand{\thefootnote}{}

We study the power versus distortion trade-off for the distributed
transmission of a memoryless bi-variate Gaussian source over a
two-to-one average-power limited Gaussian multiple-access channel. In
this problem, each of two separate transmitters observes a different
component of a memoryless bi-variate Gaussian source. The two
transmitters then describe their source component to a common receiver
via an average-power constrained Gaussian multiple-access
channel. From the output of the multiple-access channel, the receiver
wishes to reconstruct each source component with the least possible
expected squared-error distortion. Our interest is in characterizing
the distortion pairs that are simultaneously achievable on the two
source components.

We present sufficient conditions and necessary conditions for the
achievability of a distortion pair. These conditions are expressed as
a function of the channel signal-to-noise ratio (SNR) and of the
source correlation. In several cases the necessary conditions and
sufficient conditions are shown to agree. In particular, we show that
if the channel SNR is below a certain threshold, then an uncoded
transmission scheme is optimal. We also derive the precise
high-SNR asymptotics of an optimal scheme.



\footnote{The work of Stephan Tinguely was partially supported by the
  Swiss National Science Foundation under Grant 200021-111863/1. The
  results in this paper were presented in part at the 2006 IEEE
  Communications Theory Workshop, Dorado, Puerto Rico and at the 2006
  IEEE International Symposium on Information Theory, Seattle, USA.

  A.~Lapidoth and S.~Tinguely are with the Signal and Information
  Processing Laboratory (ISI), ETH Zurich, Switzerland (e-mail:
  lapidoth@isi.ee.ethz.ch; tinguely@isi.ee.ethz.ch).}
\setcounter{footnote}{0}
\end{abstract}

\section{Introduction}

We study the power versus distortion trade-off for the distributed
transmission of a memoryless bi-variate Gaussian source over a
two-to-one average-power limited Gaussian multiple-access channel. In
this problem, each of two separate transmitters observes a different
component of a memoryless bi-variate Gaussian source. The two
transmitters then describe their source component to a common receiver
via an average-power constrained Gaussian multiple-access
channel. From the output of the multiple-access channel, the receiver
wishes to reconstruct each source component with the least possible
expected squared-error distortion. Our interest is in characterizing
the distortion pairs that are simultaneously achievable on the two
source components.



We present sufficient conditions and necessary conditions for the
achievability of a distortion pair. These conditions are expressed as
a function of the channel signal-to-noise ratio (SNR) and of the
source correlation. In several cases the necessary conditions and
sufficient conditions are shown to agree, thus yielding a full
characterization of the achievable distortions. In particular, we show
that if the channel SNR is below a certain threshold (that we
compute), then an uncoded transmission scheme is optimal. We also
derive the precise high-SNR asymptotics of an optimal scheme. The
uncoded result is reminiscent of Goblick's result \cite{goblick65}
that for the transmission of a Gaussian source over an AWGN channel
the minimal squared-error distortion is achieved by uncoded
transmission. But in our setting uncoded transmission is only optimal
for some SNRs.



Our problem can be viewed as a lossy Gaussian version of the problem
addressed by Cover, El Gamal and Salehi \cite{cover_elgamal_salehi80}
(see also \cite{dueck81, kang_ulukus05}) in which a bi-variate
finite-alphabet source is to be transmitted losslessly over a
two-to-one multiple-access channel. Our problem is also related to the
quadratic Gaussian two-terminal source-coding problem \cite{oohama97,
  wagner-tavildar-vishwanath05} and to the quadratic Gaussian CEO
problem \cite{berger_vishwanathan97, oohama98}. In both of these
problems, correlated Gaussians are described distributedly to a
central receiver. But, in the quadratic Gaussian CEO problem the
interest is in reconstructing a single Gaussian random variable that
underlies the observations of the different transmitters, rather than
reconstructing each transmitter's observation itself. But more
importantly, the above two problems are source-coding problems whereas
ours is one of combined source-channel coding. We emphasize that, as
our results show, source-channel separation is suboptimal for our
setting.

The problem of transmitting correlated sources over multiple-access
channels has so far only been studied sparsely. One of the first
results is due to Cover, El Gamal and Salehi
\cite{cover_elgamal_salehi80} who presented sufficient conditions for
the lossless transmission of a finite-alphabet bi-variate source over
a multiple-access channel. Later, several variations of this problem
were considered. Salehi \cite{salehi95} studied a lossy versions of
the problem with a finite-alphabet source and arbitrary distortion
measures on each source component. For this problem he derived
sufficient conditions for the achievability of a distortion pair. More
recently, another variation where the two source components are binary
with Hamming distortion and where the multiple-access channel is
Gaussian was considered by Murugan, Gopala and El Gamal
\cite{murugan-gopala-elgamal04} who derived sufficient conditions for
the achievability of a distortion pair. Gastpar \cite{gastpar07}
considered a combined source-channel coding analog of the quadratic
Gaussian CEO problem. In this problem, distributed transmitters
observe independently corrupted versions of the same univariate
Gaussian source. These transmitters are connected to a central
receiver by means of a many-to-one Gaussian multiple-access
channel. The central receiver wishes to reconstruct the original
univariate source as accurately as possible. For this problem, Gastpar
showed that the minimal expected squared-error distortion is achieved
by an uncoded transmission scheme. The extension of our problem to the
case where perfect causal feedback from the receiver to each
transmitter is available is studied in \cite{lapidoth-tinguely09} (see
also \cite{lapidoth-tinguely07}).

\section{Problem Statement}\label{sec:problem-statement}


\subsection{Setup}\label{subsec:setup}



Our setup is illustrated in Figure \ref{fig:setup-mac}.
\begin{figure}[h]
  \centering
  \psfrag{s1}[cc][cc]{$S_{1,k}$}
  \psfrag{s2}[cc][cc]{$S_{2,k}$}
  \psfrag{x1}[cc][cc]{$X_{1,k}$}
  \psfrag{x2}[cc][cc]{$X_{2,k}$}
  \psfrag{src}[cc][cc]{Source}
  \psfrag{f1}[cc][cc]{$f_1^{(n)}(\cdot)$}
  \psfrag{f2}[cc][cc]{$f_2^{(n)}(\cdot)$}
  \psfrag{z}[cc][cc]{$Z_{k}$}
  \psfrag{y}[cc][cc]{$Y_{k}$}
  \psfrag{p1}[cc][cc]{$\phi_1^{(n)}(\cdot)$}
  \psfrag{p2}[cc][cc]{$\phi_2^{(n)}(\cdot)$}
  \psfrag{s1h}[cc][cc]{$\hat{S}_{1,k}$}
  \psfrag{s2h}[cc][cc]{$\hat{S}_{2,k}$}
  \epsfig{file=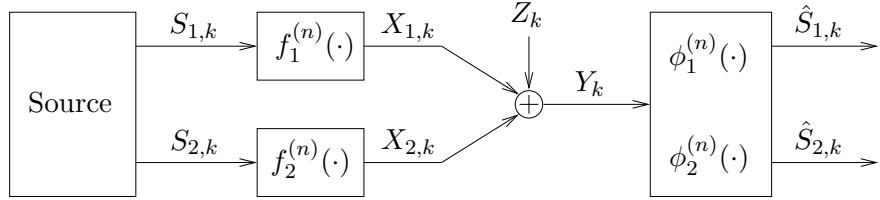, width=0.8\textwidth}
  \caption{Bi-variate Gaussian source with one-to-two Gaussian
    multiple-access channel.}
  \label{fig:setup-mac}
\end{figure}
A memoryless bi-variate Gaussian source is connected to a two-to-one
Gaussian multiple-access channel. Each transmitter observes one of the
source components and wishes to describe it to the common
receiver. The source symbols produced at time $k \in \Integers$ are
denoted by $(S_{1,k}, S_{2,k})$. The source output pairs $\{ (S_{1,k},
S_{2,k}) \}$ are independent identically distributed (IID) zero-mean
Gaussians of covariance matrix
\begin{equation}\label{eq:source-law}
\cov{S} = \left( \begin{array}{c c}
\sigma_1^2 & \rho \sigma_1 \sigma_2\\[3mm]
\rho \sigma_1 \sigma_2 & \sigma_2^2
\end{array} \right),
\end{equation}
where $\rho \in [-1,1]$ and where $0 < \sigma_i^2 < \infty$, $i \in \{
1,2 \}$. The sequence $\{ S_{1,k} \}$ of the first source component is
observed by Transmitter~1 and the sequence $\{ S_{2,k} \}$ of the
second source component is observed by Transmitter~2. The two source
components are to be described over the multiple-access channel to the
common receiver by means of the channel input sequences $\{ X_{1,k}
\}$ and $\{ X_{2,k} \}$, where $x_{1,k} \in \Reals$ and $x_{2,k} \in
\Reals$. The corresponding time-$k$ channel output is given by
\begin{equation}
Y_k = X_{1,k} + X_{2,k} + Z_k,
\end{equation}
where $Z_k$ is the time-$k$ additive noise term, and where $\{ Z_{k}\}$
are IID zero-mean variance-$N$ Gaussian random variables that are
independent of the source sequence.

For the transmission of the source $\{ S_{1,k}, S_{2,k}\}$, we
consider block encoding schemes and denote the block-length by $n$ and
the corresponding $n$-sequences in boldface, e.g.~${\bf S}_1 =
(S_{1,1}, S_{1,2}, \ldots , S_{1,n})$. Transmitter $i$ is modeled as a
function $f_i^{(n)} \colon \Reals^n \rightarrow \Reals^n$ which produces the
channel input sequence ${\bf X}_i$ based on the observed source
sequence ${\bf S}_i = (S_{i,1}, S_{i,2}, \ldots ,S_{i,n})$, i.e.
\begin{IEEEeqnarray}{rCl}\label{eq:encoders-nofb}
  {\bf X}_i & = & f_i^{(n)} \left( {\bf S}_i \right) \qquad i \in \{
  1,2 \}.
\end{IEEEeqnarray}
The channel input sequences are subjected to expected average power
constraints
\begin{IEEEeqnarray}{rCl}\label{eq:power-constraint}
  \frac{1}{n} \sum_{k=1}^n \E{X_{i,k}^2} & \leq & P_i \qquad i \in \{
  1,2 \},
\end{IEEEeqnarray}
for some given $P_i > 0$.

The decoder consists of two functions $\phi_{i}^{(n)} \colon
\Reals^{n} \rightarrow \Reals^{n}$, $i \in \{ 1,2 \}$, which perform
estimates $\hat{\bf S}_i$ of the respective source sequences ${\bf
  S}_i$, based on the observed channel output sequence ${\bf Y}$, i.e.
\begin{IEEEeqnarray}{rCl}\label{eq:reconstructors}
  \hat{\bf S}_i & = & \phi_{i}^{(n)} \left( {\bf Y} \right) \qquad i
  \in \{ 1,2 \}.
\end{IEEEeqnarray}

Our interest is in the pairs of expected squared-error distortions
that can be achieved simultaneously on the source-pair as the
blocklength $n$ tends to infinity. In view of this, we next define the
notion of achievability.

\subsection{Achievability of Distortion Pairs}\label{subsec:achievability}

\begin{dfn}\label{def:achv-dist}
  Given $\sigma_{1}, \sigma_{2} > 0$, $\rho \in [-1,1]$, $P_{1}, P_{2}
  > 0$, and $N > 0$ we say that the tuple $(D_{1}, D_{2},
  \sigma^{2}_{1}, \sigma_{2}^{2},$ $\rho, P_{1}, P_{2}, N )$ is
  \emph{achievable} if there exists a sequence of encoding functions
  $\{ f_{1}^{(n)}, f_{2}^{(n)} \}$ as in \eqref{eq:encoders-nofb},
  satisfying the average power constraints
  (\ref{eq:power-constraint}), and a sequence of reconstruction pairs
  $\{ \phi_{1}^{(n)}, \phi_{2}^{(n)} \}$ as in
  (\ref{eq:reconstructors}), such that the average distortions
  resulting from these encoding and reconstruction functions fulfill
  \begin{displaymath}
    \varlimsup_{n \rightarrow \infty} \frac{1}{n} \sum_{k=1}^n \E{
      \left( S_{i,k} - \hat{S}_{i,k} \right)^2} \leq D_{i}, \quad i=1,2,
  \end{displaymath}
  whenever
  \begin{displaymath}
    {\bf Y} = f_{1}^{(n)}({\bf S}_1) + f_{2}^{(n)}({\bf S}_2) + {\bf Z},
  \end{displaymath}
  and where $\{(S_{1,k},S_{2,k})\}$ are IID zero-mean bi-variate
  Gaussian vectors of covariance matrix $\cov{S}$ as in
  \eqref{eq:source-law} and $\{Z_{k}\}$ are IID zero-mean variance-$N$
  Gaussians that are independent of $\{(S_{1,k},S_{2,k})\}$.
\end{dfn}

The problem we address here is, for given $\sigma_1^{2}$,
$\sigma_2^{2}$, $\rho$, $P_{1}$, $P_{2}$, and $N$, to find the set of
pairs $(D_{1}, D_{2})$ such that $(D_{1}, D_{2}, \sigma_1^{2},
\sigma_2^{2}, \rho, P_{1}, P_{2}, N)$ is achievable. Sometimes, we
will refer to the set of all $(D_1,D_2)$ such that $(D_{1}, D_{2},
\sigma_1^{2}, \sigma_2^{2}, \rho, P_{1}, P_{2}, N)$ is achievable as
the distortion region associated to $( \sigma_1^{2}, \sigma_2^{2},
\rho, P_{1}, P_{2}, N )$. In that sense, we will often say, with
respect to some $( \sigma_{1}$, $\sigma_{2}$, $\rho$, $P_{1}$, $P_{2},
N )$, that the pair $(D_1,D_2)$ is achievable, instead of saying that
the tuple $(D_{1}, D_{2}, \sigma^{2}_{1}, \sigma_{2}^{2},$ $\rho,
P_{1}, P_{2}, N )$ is achievable.

\subsection{Normalization}

For the described problem we now show that, without loss in
generality, the source law given in \eqref{eq:source-law} can be
restricted to a simpler form. This restriction will ease the statement
of our results as well as their derivations.

\begin{rdc}\label{rdc:source-normalization}
  For the problem stated in Sections \ref{subsec:setup} and
  \ref{subsec:achievability}, there is no loss in generality in
  restricting the source law to satisfy
  \begin{IEEEeqnarray}{rCl}\label{eq:source-normalization}
    \sigma_1^2 = \sigma_2^2 = \sigma^2 & \hspace{12mm} \text{and}
    \hspace{12mm} & \rho \in [0,1].
  \end{IEEEeqnarray}
\end{rdc}

\begin{proof}
  The proof follows by noting that the described problem has certain
  symmetry properties with respect to the source law. We prove the
  reductions on the source variance and on the correlation coefficient
  separately.
  \renewcommand{\labelenumi}{\roman{enumi})}
  \begin{enumerate}
  \item The reduction to correlation coefficients $\rho \in [0,1]$ holds
    because the optimal distortion region depends on the correlation
    coefficient only via its absolute value $|\rho |$. That is, the
    tuple $(D_{1}, D_{2}, \sigma_{1}^{2}, \sigma_{2}^{2}, \rho,$ $P_{1},
    P_{2},N)$ is achievable if, and only if, the tuple $(D_{1}, D_{2},
    \sigma_{1}^{2}, \sigma_{2}^{2}, -\rho, P_{1}, P_{2},N)$ is achievable.
    To see this, note that if $\{ f_{1}^{(n)}, f_{2}^{(n)},
    \phi_{1}^{(n)}, \phi_{2}^{(n)} \}$ achieves the
    distortion $(D_{1}, D_{2})$ for the source of correlation
    coefficient $\rho$, then $\{ \tilde{f}_{1}^{(n)}, f_{2}^{(n)},
    \tilde{\phi}_{1}^{(n)}, \phi_{2}^{(n)} \}$, where
    \begin{equation*}
      \tilde{f}_{1}^{(n)}( {\bf S}_{1} ) = f_{1}^{(n)}( - {\bf S}_{1} )
      \qquad \text{and} \qquad 
      \tilde{\phi}_{1}^{(n)}( {\bf Y} ) = - \phi_{1}^{(n)}( {\bf Y} ) 
    \end{equation*}
    achieves  $(D_{1}, D_{2})$ on the source with correlation
    coefficient $-\rho$.
    
  \item The restriction to source variances satisfying $\sigma_1^2 =
    \sigma_2^2 = \sigma^2$ incurs no loss of generality because the
    distortion region scales linearly with the source variances. That
    is, the tuple $(D_{1}, D_{2}, \sigma_{1}^{2}, \sigma_{2}^{2},
    \rho, P_{1}, P_{2},N)$ is achievable if, and only if, for every
    $\alpha_{1}, \alpha_{2} \in \Reals^+$, the tuple
    $(\alpha_{1}D_{1}, \alpha_{2}D_{2}, \alpha_{1} \sigma_{1}^{2},
    \alpha_{2} \sigma_{2}^{2}, \rho, P_{1}, P_{2},N)$ is achievable.

    This can be seen as follows. If $\{ f_{1}^{(n)}, f_{2}^{(n)},
    \phi_{1}^{(n)}, \phi_{2}^{(n)} \}$ achieves $(D_{1}, D_{2},
    \sigma_{1}^{2}, \sigma_{2}^{2}, \rho,$ $P_{1}, P_{2},N)$, then the
    combination of the encoders
    \begin{equation*}
      \tilde{f}_{i}^{(n)}( {\bf S}_{i}) = f_{i}^{(n)}( {\bf
        S}_{i}/\sqrt{\alpha_{i}}), \hspace{10mm} i \in \{ 1,2 \},
    \end{equation*}
    with the reconstructors
    \begin{equation*}
      \tilde{\phi}_{i}^{(n)}( {\bf Y} ) = \sqrt{\alpha_{i}} \cdot
      \phi_{i}^{(n)}( {\bf Y} ), \hspace{10mm} i \in \{ 1,2 \},
    \end{equation*}
    achieves the tuple $(\alpha_1 D_{1}, \alpha_2 D_{2}, \alpha_{1}
    \sigma_{1}^{2}, \alpha_{2} \sigma_{2}^{2}, \rho, P_{1},
    P_{2},N)$. And by an analogous argument it follows that if
    $(\alpha_1 D_{1}, \alpha_2 D_{2}, \alpha_{1} \sigma_{1}^{2},
    \alpha_{2} \sigma_{2}^{2}, \rho, P_{1}, P_{2},N)$ is achievable,
    then also $(D_{1}, D_{2}, \sigma_{1}^{2}, \sigma_{2}^{2}, \rho,
    P_{1}, P_{2},N)$ is achievable. \hfill \qedhere
  \end{enumerate}
\end{proof}

In view of Reduction \ref{rdc:source-normalization} we assume for the
remainder that the source law additionally satisfies
\eqref{eq:source-normalization}.

\subsection{``Symmetric Version'' and a Convexity Property}

The ``symmetric version'' of our problem corresponds to the case where
the transmitters are subjected to the same power constraint, and where
we seek to achieve the same distortion on each source component. That
is, $P_{1} = P_{2} = P$, and we are interested in the minimal
distortion
\begin{IEEEeqnarray*}{rCl}
  D^{*}(\sigma^{2}, \rho, P, N) & \triangleq & \inf \{ D \colon (D, D,
  \sigma^{2}, \sigma^{2}, \rho, P, P, N) \text{ is achievable}\},
\end{IEEEeqnarray*}
that is simultaneously achievable on $\{ S_{1,k} \}$ and on $\{
S_{2,k} \}$. In this case, we define the SNR as $P/N$ and seek the
distortion $D^{*}(\sigma^{2}, \rho, P, N)$, for some fixed $\sigma^2$
and $\rho$, and as a function of the SNR.

We conclude this section with a convexity property of the achievable
distortions.

\begin{rmk}\label{rmk:convexity}
  If $(D_{1}, D_{2}, \sigma_{1}^{2}$, $\sigma_{2}^{2}, \rho, P_{1},
  P_{2}, N)$ and $(\tilde{D}_{1}, \tilde{D}_{2}, \sigma_{1}^{2},
  \sigma_{2}^{2}, \rho, \tilde{P}_{1}, \tilde{P}_{2}, N)$ are
  achievable, then
  \begin{equation*}
    \left( \lambda D_{1} + \bar{\lambda} \tilde{D}_{1}, \lambda D_{2} + 
      \bar{\lambda}\tilde{D}_{2},  \sigma_{1}^{2}, \sigma_{2}^{2}, \rho,
      \lambda P_{1} + \bar{\lambda} \tilde{P}_{1}, \lambda P_{2} + 
      \bar{\lambda}\tilde{P}_{2}, N \right),
  \end{equation*}
  is also achievable for every $\lambda \in [0,1]$, where
  $\bar{\lambda} = (1-\lambda)$.
\end{rmk}

\begin{proof}
  Follows by a time-sharing argument.
\end{proof}


\section{Preliminaries: Sending a Bi-Variate Gaussian over an AWGN
  Channel}\label{sec:pt2pt}

In this section we lay the ground for our main results. We study a
point-to-point analog of the multiple-access problem described in
Section \ref{subsec:setup}. More concretely, we consider the
transmission of a memoryless bi-variate Gaussian source, subject to
expected squared-error distortion on each source component, over the
additive white Gaussian noise (AWGN) channel. For this problem, we
characterize the power versus distortion trade-off and show that below
a certain SNR threshold, an uncoded transmission scheme is
optimal. This problem is simpler than our multiple-access problem
because here source-channel separation is optimal.

\subsection{Problem Statement}

The setup considered in this section is illustrated in Figure
\ref{fig:awgn-problem}.
\begin{figure}[h]
  \centering
  \psfrag{s1}[cc][cc]{$S_{1,k}$}
  \psfrag{s2}[cc][cc]{$S_{2,k}$}
  \psfrag{x}[cc][cc]{$X_{k}$}
  \psfrag{src}[cc][cc]{Source}
  \psfrag{f}[cc][cc]{$f^{(n)}(\cdot)$}
  \psfrag{z}[cc][cc]{$Z_{k}$}
  \psfrag{y}[cc][cc]{$Y_{k}$}
  \psfrag{p1}[cc][cc]{$\phi_1^{(n)}(\cdot)$}
  \psfrag{p2}[cc][cc]{$\phi_2^{(n)}(\cdot)$}
  \psfrag{s1h}[cc][cc]{$\hat{S}_{1,k}$}
  \psfrag{s2h}[cc][cc]{$\hat{S}_{2,k}$}
  \epsfig{file=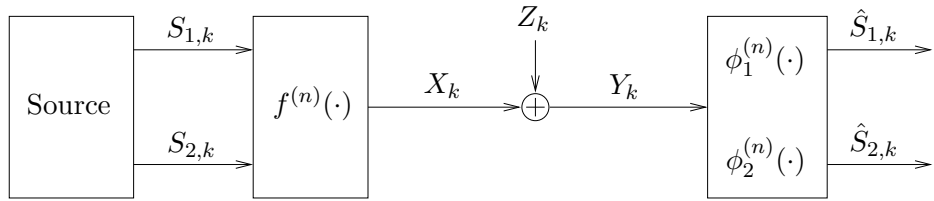, width=0.85\textwidth}
  \caption{Bi-variate Gaussian source with additive white Gaussian
    noise channel.}
  \label{fig:awgn-problem}
\end{figure}
The difference to the multiple-access problem of Section
\ref{subsec:setup} is that now the two source sequences ${\bf S}_1$
and ${\bf S}_2$ are observed and transmitted jointly by one single
transmitter rather than by two distributed transmitters. Thus, the
channel input sequence ${\bf X}$ is a function $f^{(n)} \colon
\Reals^n \times \Reals^n \rightarrow \Reals^n$ of the source sequences
$( {\bf S}_1, {\bf S}_2)$, i.e.
\begin{IEEEeqnarray}{rCl}\label{eq:encoder-awgn}
  {\bf X} & = & f^{(n)} \left( {\bf S}_1, {\bf S}_2 \right).
\end{IEEEeqnarray}
This channel input sequence is subjected to an average power
constraint
\begin{IEEEeqnarray}{rCl}\label{eq:power-awgn}
  \frac{1}{n} \sum_{k=1}^n \E{X_k^2} & \leq & P,
\end{IEEEeqnarray}
for some given $P > 0$.

The remainder of the problem statement is as in the multiple-access
problem. The source law is assumed to be given by
\eqref{eq:source-law} and to satisfy
\eqref{eq:source-normalization}. The reconstruction functions are as
defined in \eqref{eq:reconstructors}, and the achievability of
distortion pairs is defined analogously as in Section
\ref{def:achv-dist}. Our interest is in the set of achievable
distortion pairs $(D_{1}, D_{2})$.

\subsection{Rate-Distortion Function of a Bi-Variate Gaussian}

Denoting the rate-distortion function of the source $\{
(S_{1,k},S_{2,k}) \}$ by $R_{S_1,S_2}(D_1,D_2)$, the set of achievable
distortion pairs is given by all pairs $(D_1,D_2)$ satisfying
\begin{IEEEeqnarray}{rCl}\label{eq:awgn-nec+suff-cond}
  R_{S_1,S_2}(D_1,D_2) & \leq & \frac{1}{2} \log_2 \left( 1 +
    \frac{P}{N} \right).
\end{IEEEeqnarray}
We next compute the rate-distortion function $R_{S_1,S_2}(D_1,D_2)$.

\begin{thm}\label{thm:rd1d2-main}
  The rate-distortion function $R_{S_1,S_2}(D_1,D_2)$ is given by
  \begin{IEEEeqnarray}{rCl}\label{eq:RD1D2-solution}
    R_{S_1,S_2}(D_1,D_2) & = & \left\{ \begin{array}{l c}
        \frac{1}{2} \log_2^+ \left( \frac{\sigma^2}{D_\textnormal{min}}
        \right) & \text{if } (D_1,D_2) \in \mathscr{D}_1\\[5mm]
        \frac{1}{2} \log_2^+ \left( \frac{\sigma^4 (1-\rho^2)}{D_1 D_2} \right)
        & \text{if } (D_1,D_2) \in \mathscr{D}_2 \\[5mm]
        \frac{1}{2} \log_2^+ \left( \frac{\sigma^4 (1-\rho^2)}{D_1D_2 - \left(
              \rho \sigma^2 - \sqrt{(\sigma^2-D_1)(\sigma^2-D_2)} \right)^2}
        \right) & \text{if } (D_1,D_2) \in \mathscr{D}_3.
      \end{array} \right. \hspace{5mm}
  \end{IEEEeqnarray}
  where $\log_2^+(x) = \max \{ 0,\log_2(x) \}$, $D_{\textnormal{min}}
  = \min \left\{ D_1,D_2 \right\}$ and where the regions
  $\mathscr{D}_1$, $\mathscr{D}_2$ and $\mathscr{D}_3$ are given by
  \begin{IEEEeqnarray*}{rCl}
    \mathscr{D}_1 & = \Bigg\{ (D_1,D_2): \: &0 \leq D_1 \leq \sigma^2
    (1-\rho^2), \: \, D_2 \geq \sigma^2 (1-\rho^2) + \rho^2 D_1;\\[3mm]
    & & \hspace{7mm} \sigma^2(1-\rho^2) < D_1 \leq \sigma^2, \: \,
    D_2 \geq \sigma^2 (1-\rho^2) + \rho^2 D_1,\\
    & & \hspace{73mm} D_2 \leq \frac{D_1 - \sigma^2(1-\rho^2)}{\rho^2}
    \Bigg\},\\[7mm]
    \mathscr{D}_2 & = \Bigg\{ (D_1,D_2): \; &0 \leq D_1 \leq \sigma^2
    (1-\rho^2), \; 0 \leq D_2 < (\sigma^2(1-\rho^2) - D_1)
    \frac{\sigma^2}{\sigma^2-D_1} \Bigg\},\\[8mm]
    \mathscr{D}_3 & = \Bigg\{ (D_1,D_2): \: &0 \leq D_1 \leq
    \sigma^2(1-\rho^2),\\[-3mm]
    & & \hspace{12mm} (\sigma^2(1-\rho^2) - D_1)
    \frac{\sigma^2}{\sigma^2-D_1} \leq D_2 < \sigma^2 (1-\rho^2) +
    \rho^2 D_1;\\[4mm]
    & & \hspace{-2mm} \sigma^2(1-\rho^2) < D_1 \leq \sigma^2, \: \, \frac{D_1 -
      \sigma^2(1-\rho^2)}{\rho^2} < D_2 < \sigma^2 (1-\rho^2) + \rho^2
    D_1 \Bigg\}.
  \end{IEEEeqnarray*}
\end{thm}

\begin{proof}
See Appendix \ref{appx:prf-thm-RD1D2}.
\end{proof}

The result of Theorem \ref{thm:rd1d2-main} was also established
independently (and by a different proof) in \cite{xiao-luo05}. The
regions $\mathscr{D}_1$, $\mathscr{D}_2$, and $\mathscr{D}_3$ are
illustrated in Figure \ref{fig:regions}.
\begin{figure}[h]
 \centering
 \psfrag{d1}[cc][cc]{$D_1$}
 \psfrag{d2}[cc][cc]{$D_2$}
 \psfrag{D1}[cc][cc]{$\mathscr{D}_2$}
 \psfrag{D2}[cc][cc]{$\mathscr{D}_3$}
 \psfrag{D3}[cc][cc]{$\mathscr{D}_1$}
 \psfrag{s}[cc][cc]{$\sigma^2$}
 \psfrag{sr}[cc][cc]{$\sigma^2(1-\rho)$}
 \psfrag{sr2}[cc][cc]{$\sigma^2(1-\rho^2)$}
 \epsfig{file=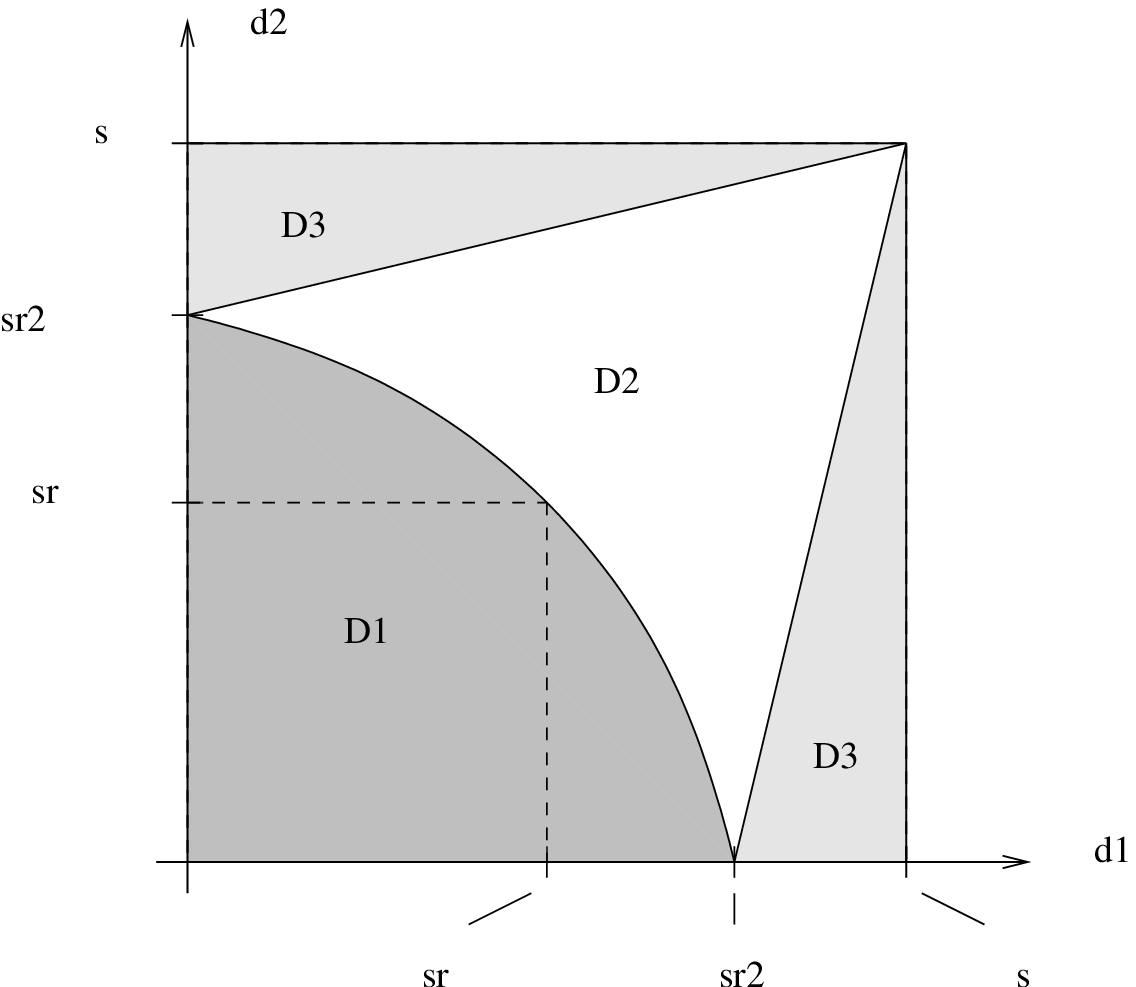, width=0.6\textwidth}
 \caption{The regions $\mathscr{D}_1$, $\mathscr{D}_2$,
   $\mathscr{D}_3$.}
 \label{fig:regions}
\end{figure}

\begin{rmk}
  Let $R_{S_1}(D_1)$ denote the rate-distortion function for the
  source component $\{ S_{1,k} \}$, i.e.,
  \begin{IEEEeqnarray*}{rCl}
    R_{S_1}(D_1) & = & \frac{1}{2} \log_2^+ \left(
      \frac{\sigma^2}{D_1} \right), 
  \end{IEEEeqnarray*}
  and let $R_{S_2 | S_1}(D_2)$ denote the rate-distortion function for
  $\{ S_{2,k} \}$ when $\{ S_{1,k} \}$ is given as side-information to
  both, the encoder and the decoder, i.e.,
  \begin{IEEEeqnarray*}{rCl}
    R_{S_2 | S_1} (D_2) & = & \frac{1}{2} \log_2^+ \left(
      \frac{\sigma^2 (1-\rho^2) }{D_2} \right).
  \end{IEEEeqnarray*}
  Then, for every $(D_1,D_2) \in \mathscr{D}_2$ the rate-distortion
  function $R_{S_1,S_2}(D_1,D_2)$ satisfies
  \begin{IEEEeqnarray*}{rCl}
    R_{S_1,S_2}(D_1,D_2) & = & \frac{1}{2} \log_2^+ \left(
      \frac{\sigma^4
        (1-\rho^2)}{D_1D_2} \right)\\[3mm]
    & \stackrel{a)}{=} & \frac{1}{2} \log_2^+ \left(
      \frac{\sigma^2}{D_1} \right) +
    \frac{1}{2} \log_2^+ \left( \frac{\sigma^2(1-\rho^2)}{D_2} \right)\\[3mm]
    & = & R_{S_1}(D_1) + R_{S_2|S_1}(D_2),
  \end{IEEEeqnarray*}
  where $a)$ holds since for $(D_1,D_2) \in \mathscr{D}_2$ we have
  $D_1, D_2 \geq 0$.
\end{rmk}


\subsection{Optimal Uncoded Scheme}\label{subsec:pt2pt-uncoded}

As an alternative to the separation-based approach, we now present an
uncoded scheme that, for all SNR below a certain threshold, is
optimal. The optimality of this uncoded scheme will be useful for
understanding a similar result in the multiple-access problem.

The uncoded scheme can be described as follows. At every time instant
$k$, the transmitter produces a channel input $X_k^{\textnormal{u}}$
of the form
\begin{IEEEeqnarray*}{rCl}
  X_k^{\textnormal{u}}(\alpha,\beta) & = & \sqrt{\frac{P}{\sigma^2
      (\alpha^2 + 2 \rho \alpha \beta + \beta^2)}} \left( \alpha
    S_{1,k} + \beta S_{2,k} \right) \qquad k \in \{ 1,2, \ldots ,n \},
\end{IEEEeqnarray*}
for some $\alpha, \beta \in \Reals$. From the resulting channel output
$Y_k$, the receiver makes a minimum mean squared-error (MMSE) estimate
$\hat{S}_{i,k}^{\textnormal{u}}$, $i \in \{ 1,2 \}$, of the source
sample $S_{i,k}$, i.e.,
\begin{IEEEeqnarray*}{rCl}
  \hat{S}_{i,k}^{\textnormal{u}} & = & \E{S_{i,k} | Y_k}, \qquad i \in
  \{ 1,2 \}.
\end{IEEEeqnarray*}
The corresponding expected distortions on $\{ S_{1,k} \}$ and on $\{
S_{2,k} \}$ are
\begin{IEEEeqnarray*}{rCl}
  \tilde{D}_1^{\textnormal{u}}(\alpha, \beta) & = & \sigma^2
  \frac{P^2 \beta^2 (1-\rho^2) + PN(\alpha^2 + 2 \rho \alpha \beta +
    \beta^2(2-\rho^2)) + N^2 (\alpha^2 + 2 \rho \alpha \beta +
    \beta^2)}{(P + N)^2 (\alpha^2 + 2 \rho \alpha \beta + \beta^2)},\\[3mm]
  \tilde{D}_2^{\textnormal{u}}(\alpha, \beta) & = & \sigma^2
  \frac{P^2 \alpha^2 (1-\rho^2) + PN(\beta^2 + 2 \rho \alpha \beta +
    \alpha^2(2-\rho^2)) + N^2 (\alpha^2 + 2 \rho \alpha \beta +
    \beta^2)}{(P + N)^2 (\alpha^2 + 2 \rho \alpha \beta + \beta^2)}.
\end{IEEEeqnarray*}
The optimality of this uncoded scheme below a certain SNR-threshold is
stated next.

\begin{prp}\label{prp:uncoded-AWGN-ch}
  Let $(D_1,D_2)$ be an achievable distortion pair for our
  point-to-point setting. If
  \begin{IEEEeqnarray}{rCl}\label{eq:threshold-pt2pt}
    \frac{P}{N} & \leq & \Gamma \left( D_1, \sigma^2, \rho \right),
  \end{IEEEeqnarray}
  where the threshold $\Gamma$ is given by
  \begin{IEEEeqnarray}{rCl}\label{eq:SNR-threshold}
    \Gamma \left( D_1, \sigma^2, \rho \right) & = &
    \left\{ \begin{array}{l l}
        \frac{\sigma^4(1-\rho^2) - 2 D_1 \sigma^2 (1-\rho^2)
          + D_1^2}{D_1 \left( \sigma^2 (1-\rho^2) - D_1\right)} &
        \text{if } 0 < D_1 < \sigma^2(1-\rho^2),\\[3mm]
        + \infty & \text{else,}
    \end{array} \right.
  \end{IEEEeqnarray}
  then there exist $\alpha^{\ast}, \beta^{\ast} \geq 0$ such that
  \begin{IEEEeqnarray*}{rCl}
    \tilde{D}_1^{\textnormal{u}} (\alpha^{\ast}, \beta^{\ast}) \leq D_1 \qquad
    & \text{and} & \qquad \tilde{D}_2^{\textnormal{u}} (\alpha^{\ast},
    \beta^{\ast}) \leq D_2.
  \end{IEEEeqnarray*}
\end{prp}

\begin{proof}
  See Appendix \ref{appx:prf-unc-AWGN}.
\end{proof}

In the symmetric case, Proposition \ref{prp:uncoded-AWGN-ch}
simplifies as follows.

\begin{cor}\label{cor:uncoded-AWGN-symmetric}
  Let $D > 0$ be such that $(D,D)$ is an achievable distortion pair
  for the point-to-point problem. If
  \begin{IEEEeqnarray}{rCl}\label{eq:threshold-AWGN-SNR-symmetric}
    \frac{P}{N} & \leq & \frac{2 \rho}{1-\rho},
  \end{IEEEeqnarray}
  then the pair $(D,D)$ is achieved by the uncoded scheme with
  time-$k$ channel input
  \begin{IEEEeqnarray*}{rCl}
    X_k^{\textnormal{u}}(\alpha,\alpha) & = &
    \sqrt{\frac{P}{2\sigma^2(1+\rho)}} \left( S_{1,k} + S_{2,k}
    \right) \qquad \text{for } k \in \{ 1,2, \ldots ,n \}.
  \end{IEEEeqnarray*}
\end{cor}

Corollary \ref{cor:uncoded-AWGN-symmetric} can also be verified
without relying on Proposition \ref{prp:uncoded-AWGN-ch}. This is
discussed in the following remark.

\begin{rmk}\label{rmk:prf-uncoded-awgn-symmetric}
  The distortions resulting from the uncoded scheme with any choice of
  $(\alpha,\beta)$ such that $\alpha = \beta$ are
  \begin{IEEEeqnarray*}{rCl}
    \tilde{D}_i^{\textnormal{u}}(\alpha,\alpha) & = & \sigma^2
    \frac{P(1-\rho) + 2N}{2 (P + N)} \qquad i \in \{ 1,2 \}.
  \end{IEEEeqnarray*}
  By evaluating the necessary and sufficient condition of
  \eqref{eq:awgn-nec+suff-cond} for the case where $D_1 = D_2 = D$, it
  follows that this is indeed the minimal achievable distortion for
  all $P/N$ satisfying \eqref{eq:threshold-AWGN-SNR-symmetric}.
\end{rmk}

This concludes our discussion of the point-to-point problem.


\section{Main Results}\label{sec:rslts-nofb}

\subsection{Necessary Condition for Achievability of
  $(D_1,D_2)$}\label{subsec:nec-cond}

\begin{thm}\label{thm:nofb-nec-cond}
  A necessary condition for the achievability of a distortion pair
  $(D_1,D_2)$ is that
  \begin{IEEEeqnarray}{rCl}\label{eq:noFB-nec-cond}
    R_{S_1,S_2}(D_1,D_2) \leq \frac{1}{2} \log_2 \left( 1 + \frac{P_1 +
        P_2 + 2\rho \sqrt{P_1P_2}}{N} \right).
  \end{IEEEeqnarray}
\end{thm}

\begin{proof}
See Appendix \ref{appx:prf-thm-nofb-nec-cond}.
\end{proof}

\begin{rmk}
  Theorem \ref{thm:nofb-nec-cond} can be extended to a wider class of
  sources and distortion measures. Indeed, if the source is any
  memoryless bi-variate source (not necessarily zero-mean Gaussian)
  and if the fidelity measures $d_{1}(s_{1}, \hat{s}_{1}),
  d_{2}(s_{2}, \hat{s}_{2}) \geq 0$ that are used to measure the
  distortion in reconstructing each of the source components are
  arbitrary, then the pair $(D_{1}, D_{2})$ is achievable with powers
  $P_{1}, P_{2}$ only if
  \begin{IEEEeqnarray}{rCl}\label{eq:biv-RD-problem-nonGauss}
    \inf_{\substack{P_{\widehat{S}_1, \widehat{S}_2 | S_1,S_2}:\\
        \E{d_1 (S_1,\hat{S}_1)} \leq D_1 \\
        \E{d_2 (S_2,\hat{S}_2)} \leq D_2}} I(S_1,S_2 ; \widehat{S}_1,
    \widehat{S}_2 ) & \leq & \frac{1}{2} \log \left( 1 + \frac{P_1 +
        P_2 + 2\rho_{\textnormal{max}} \sqrt{P_1 P_2}}{N} \right),
  \end{IEEEeqnarray}
  where $\rho_{\textnormal{max}}$ is the Hirschfeld-Gebelein-R\'{e}nyi
  maximal correlation between $S_{1}$ and $S_{2}$, i.e.
  \begin{equation}
    \rho_{\textnormal{max}} = \sup \E{g(S_{1}) h(S_{2})}
  \end{equation}
  where the supremum is over all functions $g(\cdot)$, $h(\cdot)$
  under which
  \begin{equation}
    \E{g(S_{1})} = \E{h(S_{2})} = 0 \qquad \text{and} \qquad
    \E{g^{2}(S_{1})} = \E{h^{2}(S_{2})} = 1.
  \end{equation}
  For the bi-variate Gaussian memoryless source, condition
  \eqref{eq:biv-RD-problem-nonGauss} reduces to
  \eqref{eq:noFB-nec-cond} because in this case
  $\rho_{\textnormal{max}}$ is equal to $\rho$ \cite[Lemma~10.2,
  p.~182]{rozanov67}.
\end{rmk}

\begin{rmk}\label{rmk:nec-cond-mac-awgn}
  The necessary condition of Theorem \ref{thm:nofb-nec-cond}
  corresponds to the necessary and sufficient condition for the
  achievability of a distortion pair $(D_1,D_2)$ when the source $\{
  (S_{1,k}, S_{2,k}) \}$ is transmitted over a point-to-point AWGN
  channel of input power constraint $P_1 + P_2 + \rho \sqrt{P_1P_2}$
  (see \eqref{eq:awgn-nec+suff-cond}). This relation is not a
  coincidence. The proof of Theorem~\ref{thm:nofb-nec-cond} (see
  Appendix~\ref{appx:prf-thm-nofb-nec-cond}) indeed consists of
  reducing the multiple-access problem to the problem of transmitting
  the source $\{ (S_{1,k}, S_{2,k}) \}$ over an AWGN channel of input
  power constraint $P_1 + P_2 + \rho \sqrt{P_1P_2}$.
\end{rmk}

We now specialize Theorem \ref{thm:nofb-nec-cond} to the symmetric
case. We combine the explicit form of the rate-distortion function in
\eqref{eq:RD1D2-solution} with \eqref{eq:noFB-nec-cond} and substitute
$(D,D)$ for $(D_1,D_2)$ to obtain:

\begin{cor}\label{cor:nofb-nec-cond-sym}
In the symmetric case
\begin{IEEEeqnarray}{rCl}
D^{\ast}(\sigma^2, \rho, P, N) \geq \left\{ \begin{array}{l l}
\sigma^2 \frac{P (1-\rho^2) +N}{2P(1 + \rho) +N} & \text{for
}\frac{P}{N} \in \left( 0, \frac{\rho}{1-\rho^2}\right] \nonumber \\[5mm]
\sigma^2 \sqrt{\frac{(1-\rho^2)N}{2P(1+\rho) +N}} & \text{for }
\frac{P}{N} > \frac{\rho}{1-\rho^2}. \nonumber
\end{array} \right.
\end{IEEEeqnarray}
\end{cor}
Corollary \ref{cor:nofb-nec-cond-sym} concludes the section on the
necessary condition for the achievability of a distortion pair
$(D_1,D_2)$. We now compare this necessary condition to several
sufficient conditions. The first sufficient condition that we consider
is based on conventional source-channel separation.

\subsection{Source-Channel Separation}\label{subsec:sep-based}

As a benchmark we now consider the set of distortion pairs that are
achieved by combining the optimal scheme for the corresponding
source-coding problem with the optimal scheme for the corresponding
channel-coding problem.

The corresponding source-coding problem is illustrated in Figure
\ref{fig:setup-oohama}.
\begin{figure}[h]
  \centering
  \psfrag{s1}[cc][cc]{$S_{1,k}$}
  \psfrag{s2}[cc][cc]{$S_{2,k}$}
  \psfrag{r1}[cc][cc]{$R_1$}
  \psfrag{r2}[cc][cc]{$R_2$}
  \psfrag{src}[cc][cc]{Source}
  \psfrag{f1}[cc][cc]{$\bar{f}_1^{(n)}(\cdot)$}
  \psfrag{f2}[cc][cc]{$\bar{f}_2^{(n)}(\cdot)$}
  \psfrag{p1}[cc][cc]{$\bar{\phi}_1^{(n)}(\cdot)$}
  \psfrag{p2}[cc][cc]{$\bar{\phi}_2^{(n)}(\cdot)$}
  \psfrag{s1h}[cc][cc]{$\hat{S}_{1,k}$}
  \psfrag{s2h}[cc][cc]{$\hat{S}_{2,k}$}
  \epsfig{file=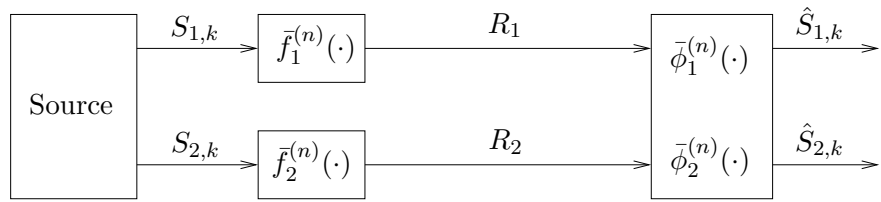, width=0.8\textwidth}
  \caption{Distributed source coding problem for a bi-variate Gaussian
    source.}
  \label{fig:setup-oohama}
\end{figure}
The two source components are observed by two separate encoders. These
two encoders wish to describe their source sequence to the common
receiver by means of individual rate-limited and error-free bit
pipes. The receiver estimates each of the sequences subject to
expected squared-error distortion. A detailed description of this
problem can be found in \cite{oohama97,
  wagner-tavildar-vishwanath05}. The associated rate-distortion region
is given in the next theorem.
\begin{thm}[Oohama \cite{oohama97}; Wagner, Tavildar, Viswanath
  \cite{wagner-tavildar-vishwanath05}]\label{thm:source-coding} 
  For the Gaussian two-terminal source coding problem (with source
  components of unit variances) a distortion-pair $(D_1,D_2)$ is
  achievable if, and only if,
  \begin{IEEEeqnarray*}{rCl}
    (R_1,R_2) \in \mathcal{R}_1(D_1) \cap \mathcal{R}_2(D_2) \cap
    \mathcal{R}_{\textnormal{sum}}(D_1,D_2),
  \end{IEEEeqnarray*} 
  where
  \begin{IEEEeqnarray*}{rCl}
    \mathcal{R}_1(D_1) & = & \left\{ (R_1,R_2): R_1\geq \frac{1}{2}
      \log_2^+ \left[ \frac{1}{D_1} (1-\rho^2(1-2^{-2R_2})) \right]
    \right\}\\[2mm]
    \mathcal{R}_2(D_2) & = & \left\{ (R_1,R_2): R_2\geq \frac{1}{2}
      \log_2^+ \left[ \frac{1}{D_2} (1-\rho^2(1-2^{-2R_1})) \right]
    \right\}\\[2mm]
    \mathcal{R}_{\textnormal{sum}}(D_1,D_2) & = & \left\{ (R_1,R_2):
      R_1+R_2\geq \frac{1}{2} \log_2^+ \left[ \frac{(1-\rho^2)
          \beta(D_1,D_2)}{2D_1D_2} \right] \right\}
  \end{IEEEeqnarray*}
  with
  \begin{IEEEeqnarray*}{rCl}
    \beta(D_1,D_2) & = & 1 + \sqrt{1 +
      \frac{4\rho^2D_1D_2}{(1-\rho^2)^2}}.
  \end{IEEEeqnarray*}
\end{thm}

The distortions achievable by source-channel separation now follow
from combining Theorem \ref{thm:source-coding} with the capacity of
the Gaussian multiple-access channel (see e.g.~\cite{ahlswede71,
  liao72}). We state here the explicit expression for the resulting
distortion pairs only for the symmetric case.

\begin{cor}\label{cor:noFB-sep-based}
  In the symmetric case, a distortion $D$ is achievable by
  source-channel separation if, and only if,
  \begin{IEEEeqnarray*}{rCl}
    D \geq \sigma^2 \frac{\sqrt{N (N+2P(1-\rho^2))}}{2P +N}.
  \end{IEEEeqnarray*}
\end{cor}

We next consider several combined source-channel coding schemes. The
first scheme is an uncoded scheme.

\subsection{Uncoded Scheme}\label{subsec:uncoded}

In this section we consider an uncoded transmission scheme, which, as
we show, is optimal below a certain SNR-threshold.

The uncoded scheme operates as follows. At every time instant $k$, Encoder
$i \in \{ 1,2 \}$ produces as channel input $X_{i,k}$ a scaled version
of the time-$k$ source output $S_{i,k}$. The corresponding scaling is
such that the average power constraint of the channel is
satisfied. That is,
\begin{IEEEeqnarray*}{rCl}
  X_{i,k}^{\textnormal{u}} & = & \sqrt{\frac{P_i}{\sigma^2}} S_{i,k}
  \qquad k \in \{ 1,2, \ldots ,n \}. 
\end{IEEEeqnarray*}
Based on the resulting time-$k$ channel output $Y_k$, the decoder then
performs an MMSE estimate $\hat{S}_{i,k}^{\textnormal{u}}$ of the
source output $S_{i,k}$, $i \in \{ 1,2 \}$, $k \in \{ 1,2, \ldots ,n
\}$. That is,
\begin{IEEEeqnarray*}{rCl}
  \hat{S}_{i,k}^{\textnormal{u}} & = & \E{ S_{i,k} | Y_k } \qquad k
  \in \{ 1,2, \ldots ,n \}.
\end{IEEEeqnarray*}
The expected distortions $(D_1^{\textnormal{u}},D_2^{\textnormal{u}})$
resulting from this uncoded scheme as well as the optimality of the
scheme below a certain SNR-threshold are given in the following
theorem.

\begin{thm}\label{thm:uncoded-main}
  The distortion pairs $(D_1^{\textnormal{u}},D_2^{\textnormal{u}})$
  resulting from the described uncoded scheme are given by
  \begin{IEEEeqnarray}{rCl}\label{eq:mac-uc-expr-D1u-D2u}
    D_1^{\textnormal{u}} = \sigma^2 \frac{(1-\rho^2)P_2 + N}{P_1 + P_2 +
      2\rho \sqrt{P_1P_2} + N} & \qquad \qquad & D_2^{\textnormal{u}} =
    \sigma^2 \frac{(1-\rho^2)P_1 + N}{P_1 + P_2 + 2\rho \sqrt{P_1P_2} +
      N}.
  \end{IEEEeqnarray}
  These distortion pairs are optimal, i.e., lie on the boundary of the
  distortion region, whenever
  \begin{IEEEeqnarray}{rCl}\label{eq:mac-SNRcond-uc-opt}
    P_2 (1-\rho^2)^2 \Big( P_1 + 2\rho \sqrt{P_1P_2} \Big) & \leq &
    N\rho^2 \Big( 2 P_2(1-\rho^2) + N \Big). \hspace{5mm}
  \end{IEEEeqnarray}
\end{thm}

\begin{proof}
  The evaluation of $(D_1^{\textnormal{u}},D_2^{\textnormal{u}})$
  leading to \eqref{eq:mac-uc-expr-D1u-D2u} is given in Appendix
  \ref{appx:prf-thm-uncoded-main}. Based on the expressions for
  $D_1^{\textnormal{u}}$ and $D_2^{\textnormal{u}}$ the optimality of
  the uncoded scheme now follows from verifying that for all $P_1$,
  $P_2$ and $N$ satisfying \eqref{eq:mac-SNRcond-uc-opt} the
  corresponding distortion pair
  $(D_1^{\textnormal{u}},D_2^{\textnormal{u}})$ satisfies the
  necessary condition \eqref{eq:noFB-nec-cond} of Theorem
  \ref{thm:nofb-nec-cond} with equality. To verify this, one can first
  verify that for all $P_1$, $P_2$ and $N$ satisfying
  \eqref{eq:mac-SNRcond-uc-opt} we have
  $(D_1^{\textnormal{u}},D_2^{\textnormal{u}}) \in \mathscr{D}_3$.
\end{proof}

\begin{rmk}
  The optimality of the uncoded scheme can also be derived in a more
  conceptual way. To see this, denote by
  $\mathscr{D}_{\textnormal{MAC}}(\sigma^2,\rho,P_1,P_2,N)$ the
  distortion region for our multiple-access problem, and by
  $\mathscr{D}_{\textnormal{PTP}}(\sigma^2,\rho,P,N)$ the distortion
  region for the point-to-point problem of Section
  \ref{sec:pt2pt}. The optimality of the uncoded scheme for the
  multiple-access problem now follows from combining the following
  three statements:\\
  \begin{itemize}
  \item[A)]
    \begin{IEEEeqnarray*}{rCl}
      \mathscr{D}_{\textnormal{MAC}}\bigl( \sigma^2,\rho,P_1,P_2,N
      \bigr) & \; \subseteq \; & \mathscr{D}_{\textnormal{PTP}}\Bigl(
      \sigma^2,\rho,P_1 + P_2 + 2 \rho \sqrt{P_1P_2},N \Bigr).\\
    \end{IEEEeqnarray*}
  \end{itemize}
  \hfill \parbox{141mm}{Statement A) is nothing but a restatement of
    Theorem~\ref{thm:nofb-nec-cond} and
    Remark~\ref{rmk:nec-cond-mac-awgn}.\\[5mm]}

  \begin{itemize}
  \item[B)] For the point-to-point problem of Section \ref{sec:pt2pt}
    with power constraint $P = P_1 + P_2 + 2 \rho \sqrt{P_1P_2}$,
    let $(D_1,D_2)$ be a distorion pair resulting from the uncoded
    scheme of Section \ref{subsec:pt2pt-uncoded}. If
    \begin{IEEEeqnarray*}{rCl}
      \frac{P_1 + P_2 + 2 \rho \sqrt{P_1P_2}}{N} & \leq & \Gamma
      (D_1,\sigma^2,\rho),
    \end{IEEEeqnarray*}
    where $\Gamma$ is the threshold function defined in
    \eqref{eq:SNR-threshold}, then $(D_1,D_2)$ lies on the boundary of
    $\mathscr{D}_{\textnormal{PTP}}(\sigma^2,\rho,P_1 + P_2 + 2 \rho
    \sqrt{P_1P_2},N)$.\\
  \end{itemize}
  \hfill \parbox{141mm}{Statement B) follows immediately by
    Proposition~\ref{prp:uncoded-AWGN-ch}.\\[5mm]}

  \begin{itemize}
  \item[C)] Let $(\tilde{D}_1^{\textnormal{u}}(\alpha,\beta),
    \tilde{D}_2^{\textnormal{u}}(\alpha,\beta))$ be the distortion
    pair resulting from the uncoded scheme for the point-to-point
    problem, and let $(D_1^{\textnormal{u}},
    D_2^{\textnormal{u}})$ be the distortion pair
    resulting from the uncoded scheme for the multiple-access
    problem. Then, if
    \begin{IEEEeqnarray*}{rCl}
      \alpha = \sqrt{\frac{P_1}{\sigma^2}} & \hspace{8mm} \text{and} \hspace{8mm}
      \beta = \sqrt{\frac{P_2}{\sigma^2}},
    \end{IEEEeqnarray*}
    then
    \begin{IEEEeqnarray*}{rCl}
      \bigl( \tilde{D}_1^{\textnormal{u}}(\alpha,\beta),
      \tilde{D}_2^{\textnormal{u}}(\alpha,\beta) \bigr) & = &
      \bigl( D_1^{\textnormal{u}}, D_2^{\textnormal{u}} \bigr).\\
    \end{IEEEeqnarray*}
  \end{itemize}
  \hfill \parbox{141mm}{Statement C) follows since in the
    multiple-access problem, the channel output}
  \begin{IEEEeqnarray*}{rCl}
    Y_k = \alpha S_{1,k} + \beta S_{2,k} + Z_k,
  \end{IEEEeqnarray*}
  \hfill \parbox{141mm}{resulting from the uncoded scheme mimics the
    channel output of the uncoded scheme for the point-to-point
    problem with power constraint $P = P_1 + P_2 + 2 \rho
    \sqrt{P_1P_2}$. Thus, while in the multiple-access problem the
    encoders cannot cooperate, the channel performs the addition for
    them. And since the reconstructors are the same in the
    multiple-access problem and the point-to-point problem, the
    resulting distortions are the same in both problems.\\[5mm]}
  Combining Statements A), B) and C), gives that if
  \begin{IEEEeqnarray}{rCl}\label{eq:mac-rmk-uc-opt}
    \frac{P_1 + P_2 + 2 \rho \sqrt{P_1P_2}}{N} & \leq & \Gamma
    (D_1^{\textnormal{u}},\sigma^2,\rho),
  \end{IEEEeqnarray}
  then $(D_1^{\textnormal{u}}, D_2^{\textnormal{u}})$ lies on the
  boundary of
  $\mathscr{D}_{\textnormal{MAC}}(\sigma^2,\rho,P_1,P_2,N)$, i.e., the
  uncoded scheme for the multiple-access problem is optimal. The
  threshold condition \eqref{eq:mac-SNRcond-uc-opt} now follows by
  \eqref{eq:mac-rmk-uc-opt} and from substituting therein the value of
  $D_1^{\textnormal{u}}$ by its explicit expression given in
  \eqref{eq:mac-uc-expr-D1u-D2u}.

\end{rmk}

As a special case of Theorem \ref{thm:uncoded-main} we obtain:
\begin{cor}\label{cor:uncoded-optimal}
  In the symmetric case,
  \begin{IEEEeqnarray}{rCl}\label{eq:snr-threshold-uncoded-optimal}
    D^{\ast}(\sigma^2, \rho, P, N) = \sigma^2 \frac{P(1-\rho^2)
      +N}{2P(1+\rho) +N}, \hspace{12mm} \text{for all} \hspace{4mm}
    \frac{P}{N} \leq \frac{\rho}{1-\rho^2}.
  \end{IEEEeqnarray}
  Moreover, for all SNRs below the given threshold, the minimal
  distortion $D^{\ast}(\sigma^2, \rho, P, N)$ is achieved by the
  uncoded scheme.
\end{cor}

The upper and lower bounds that result on
$D^{\ast}(\sigma^2,\rho,P,N)$ from our derived necessary conditions and
sufficient conditions are illustrated in Figure
\ref{fig:uc-vs-sb} for a source of correlation coefficient $\rho = 0.5$.
\begin{figure}[h]
  \centering
  \psfrag{ds}[cc][cc]{$D/\sigma^2$}
  \psfrag{pn}[cc][cc]{$P/N$}
  \psfrag{r}[cc][cc]{$\rho = 0.5$}
  \psfrag{sb}[cc][cc]{source-channel separation}
  \psfrag{uc}[cc][cc]{uncoded}
  \psfrag{lb}[cc][cc]{lower bound}
  \psfrag{0}[cc][cc]{\footnotesize $0$}
  \psfrag{0.2}[cc][cc]{\footnotesize $0.2$}
  \psfrag{0.3}[cc][cc]{\footnotesize $0.3$}
  \psfrag{0.4}[cc][cc]{\footnotesize $0.4$}
  \psfrag{0.5}[cc][cc]{\footnotesize $0.5$}
  \psfrag{0.6}[cc][cc]{\footnotesize $0.6$}
  \psfrag{0.7}[cc][cc]{\footnotesize $0.7$}
  \psfrag{0.8}[cc][cc]{\footnotesize $0.8$}
  \psfrag{0.9}[cc][cc]{\footnotesize $0.9$}
  \psfrag{1}[cc][cc]{\footnotesize $1$}
  \psfrag{1.5}[cc][cc]{\footnotesize $1.5$}
  \psfrag{2}[cc][cc]{\footnotesize $2$}
  \psfrag{2.5}[cc][cc]{\footnotesize $2.5$}
  \psfrag{3}[cc][cc]{\footnotesize $3$}
  \epsfig{file=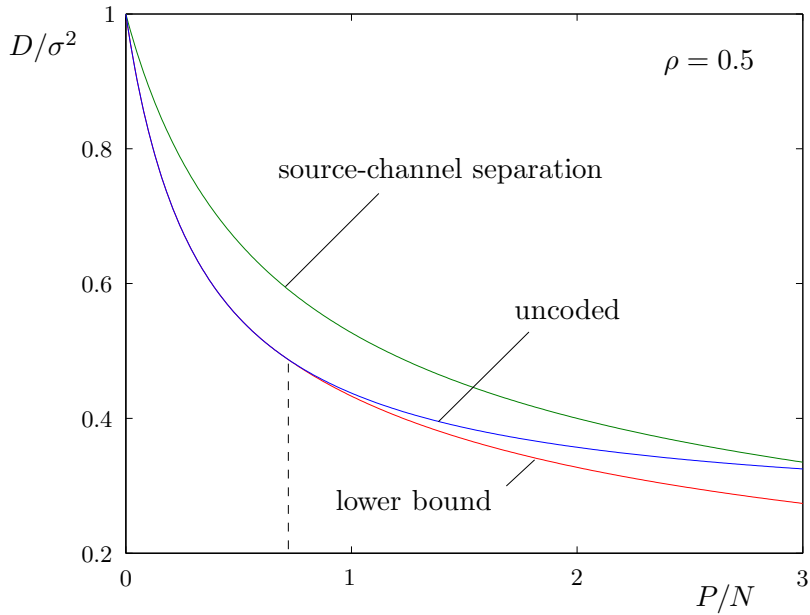, width=0.8\textwidth}
  \caption{Upper and lower bounds on $D^{\ast}(\sigma^2,\rho,P,N)$ for
    a source of correlation coefficient $\rho = 0.5$.}
  \label{fig:uc-vs-sb}
\end{figure}
For the SNRs below the threshold of
\eqref{eq:snr-threshold-uncoded-optimal} (marked by the dashed line)
the uncoded approach performs significantly better than the
separation-based approach. However, for SNRs above the threshold of
\eqref{eq:snr-threshold-uncoded-optimal} the performance of the
uncoded scheme gets successively worse. By the expressions in
\eqref{eq:snr-threshold-uncoded-optimal}, we obtain that in the
symmetric case
\begin{IEEEeqnarray}{rCl}\label{eq:mac-uc-sym-highSNR}
  \lim_{P/N \rightarrow \infty} D_i^{\textnormal{u}} & = & \sigma^2
  \frac{1 - \rho}{2}, \hspace{12mm} i \in \{ 1,2 \}.
\end{IEEEeqnarray}
That is, as $P/N \rightarrow \infty$ the distortion
$D_i^{\textnormal{u}}$ does not tend to zero. The reason is that as
the noise tends to zero, the channel output corresponding to the
uncoded scheme tends to $\alpha {\bf S}_1 + \beta {\bf S}_2$, from
which ${\bf S}_1$ and ${\bf S}_2$ cannot be recovered.

\subsection{Vector-Quantizer Scheme}\label{sec:mac-vq}

In this section, we propose a coding scheme that improves on the
uncoded scheme at high SNR. In this scheme the signal transmitted by
each encoder is a vector-quantized version of its source sequence. The
vital difference to the separation-based scheme is that the
vector-quantized sequences are not mapped to bits before they are
transmitted. Instead, the vector-quantized sequences are the channel
inputs themselves. This transfers some of the correlation from the
source to the channel inputs with the channel inputs still being from
discrete sets, thereby enabling the decoder to make distinct estimates
of ${\bf S}_1$ and of ${\bf S}_2$. For this scheme, we derive the
achievable distortions and, based on those and on the necessary
condition of Theorem \ref{thm:nofb-nec-cond}, deduce the high SNR
asymptotics of an optimal scheme.

The structure of an encoder of our scheme is illustrated in Figure
\ref{fig:vq-enc}.
\begin{figure}[h]
  \centering
  \psfrag{s}[cc][cc]{${\bf S}_i$}
  \psfrag{u}[cc][cc]{${\bf U}_i^{\ast}$}
  \psfrag{x}[cc][cc]{${\bf X}_i$}
  \psfrag{vq}[cc][cc]{\textsf{VQ}}
  \psfrag{r}[cc][cc]{rate-$R_i$}
  \psfrag{c}[cc][cc]{$\sqrt{\frac{P}{\sigma^2(1-2^{-2R_i})}}$}
  \epsfig{file=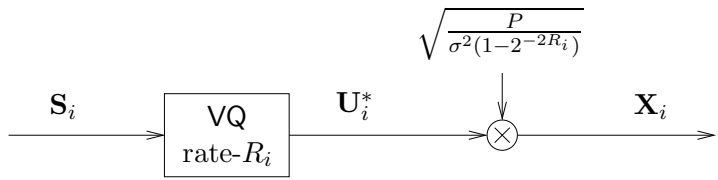, width=0.65\textwidth}
  \caption{Encoder of Vector-Quantizer Scheme.}
  \label{fig:vq-enc}
\end{figure}
First, the source sequence ${\bf S}_i$ is quantized by an optimal
rate-$R_i$ vector-quantizer. The resulting quantized sequence is
denoted by ${\bf U}_i^{\ast}$. For its transmission over the channel,
this sequence is scaled so as to satisfy the average power constraint
of \eqref{eq:power-constraint}. That is, the channel input sequence
${\bf X}_i$ is given by
\begin{IEEEeqnarray*}{rCl}
  {\bf X}_i & = & \sqrt{\frac{P}{\sigma^2(1-2^{-2R_i})}} {\bf U}_i^{\ast}
  \hspace{15mm} i \in \{ 1,2 \}.
\end{IEEEeqnarray*}
Based on the channel output ${\bf Y}$ resulting from ${\bf X}_1$ and
${\bf X}_2$, the decoder then estimates the two source sequences ${\bf
  S}_1$ and ${\bf S}_2$. It does this in two steps. First, it tries to
recover the two transmitted sequences ${\bf U}_1^{\ast}$ and ${\bf
  U}_2^{\ast}$ from the channel output sequence ${\bf Y}$ by
performing joint decoding that takes into consideration the
correlation between two transmitted sequences ${\bf U}_1^{\ast}$ and
${\bf U}_2^{\ast}$. The resulting decoded sequences are denoted by
$\hat{\bf U}_1$ and $\hat{\bf U}_2$ respectively. In the second step,
the decoder performs approximate MMSE estimates $\hat{\bf S}_i$, $i
\in \{ 1,2 \}$, of the source sequences ${\bf S}_i$ based on $\hat{\bf
  U}_1$ and $\hat{\bf U}_2$, i.e.
\begin{IEEEeqnarray*}{rCl}
  \hat{\bf S}_i & = & \gamma_{i1} \hat{\bf U}_1 + \gamma_{i2} \hat{\bf U}_2,\\[2mm]
  & \approx & \E{{\bf S}_i \big| \hat{\bf U}_1, \hat{\bf U}_2}.
\end{IEEEeqnarray*}
A detailed description of the scheme is given in Appendix
\ref{appx:prf-thm-vq}.

The distortion pairs achieved by this vector-quantizer scheme are
stated in the following theorem.

\begin{thm}\label{thm:vq-achv}
  The distortions achieved by the vector-quantizer scheme are all
  pairs $(D_1,D_2)$ satisfying
  \begin{IEEEeqnarray*}{rCl}
    D_1 &>& \sigma^2 2^{-2R_1} \cdot
    \frac{1-\rho^2(1-2^{-2R_2})}{1-\tilde{\rho}^2} \\
    D_2 &>& \sigma^2 2^{-2R_2} \cdot
    \frac{1-\rho^2(1-2^{-2R_1})}{1-\tilde{\rho}^2},
  \end{IEEEeqnarray*}
  where the rate-pair $(R_1,R_2)$ satisfies
  \begin{IEEEeqnarray}{rCl}
    R_1 &<& \frac{1}{2} \log_2 \left( \frac{P_1(1-\tilde{\rho}^2) +
        N}{N(1-\tilde{\rho}^2)} \right) \label{eq:vq-R1-constr}\\[2mm]
    R_2 &<& \frac{1}{2} \log_2 \left( \frac{P_2(1-\tilde{\rho}^2) +
        N}{N(1-\tilde{\rho}^2)} \right) \label{eq:vq-R2-constr}\\[2mm]
    R_1+R_2 &<& \frac{1}{2} \log_2 \left( \frac{P_1 + P_2 + 2\tilde{\rho}
        \sqrt{P_1P_2} + N}{N(1-\tilde{\rho}^2)} \right), \label{eq:vq-R1R2-constr}
  \end{IEEEeqnarray}
  and where
  \begin{IEEEeqnarray}{rCl}\label{eq:vq-tilde-rho}
    \tilde{\rho} & = & \rho \sqrt{(1-2^{-2R_1})(1-2^{-2R_2})}.
  \end{IEEEeqnarray}
\end{thm}

\begin{proof}
See Appendix \ref{appx:prf-thm-vq}.
\end{proof}

\begin{rmk}
  The coefficient $\tilde{\rho}$ corresponds to the asymptotic average
  correlation coefficient between two time-$k$ channel inputs
  $X_{1,k}$ and $X_{2,k}$.
\end{rmk}

Based on Theorem \ref{thm:vq-achv} we now derive two more results: we
show that for the symmetric version of our problem, source-channel
separation is suboptimal also at high SNR, and we determine the
precise high-SNR asymptotics of an optimal scheme. We begin with the
sub-optimality of source-channel separation. To this end, we restate
Theorem \ref{thm:vq-achv} more specifically for the symmetric case.
\begin{cor}\label{cor:nofb-VQ-sym}
In the symmetric case 
\begin{IEEEeqnarray*}{rCl}
D^{*}(\sigma^2, \rho, P, N) & \leq & \sigma^2 2^{-2R} \cdot \frac{1-
  \rho^2(1-2^{-2R})}{1-\rho^2(1-2^{-2R})^2},
\end{IEEEeqnarray*}
where 
\begin{IEEEeqnarray*}{rCl}
R & < & \frac{1}{4} \log_2 \left( \frac{2P (1+\rho(1-2^{-2R})) +N}{N(1-
    \rho^2(1-2^{-2R})^2)} \right).
\end{IEEEeqnarray*}
\end{cor}
By comparing the achievable distortion of the vector-quantizer scheme,
in Corollary \ref{cor:nofb-VQ-sym}, with the achievable distortion of
the separation-based scheme, in Corollary \ref{cor:noFB-sep-based}, we
obtain:
\begin{cor}
  In the symmetric case with $\rho > 0$, source-channel separation is
  suboptimal for all $P > 0$.
\end{cor}

We turn to the high-SNR asymptotics of an optimal scheme. To this end,
let $(D_1^{\ast},D_2^{\ast})$ denote an arbitrary distortion pair
resulting from an optimal scheme. For a subset of those distortion
pairs, the high SNR behavior is described in the following theorem.
\begin{thm}[High-SNR Distortion]\label{thm:nofb-asymptotics}
  The high-SNR asymptotic behavior of $(D_1^{\ast},D_2^{\ast})$ is
  given by
  \begin{IEEEeqnarray*}{rCl}
    \lim_{N \rightarrow 0} \frac{P_1 + P_2 + 2\rho
      \sqrt{P_1P_2}}{N} D_1^{\ast} D_2^{\ast} & = & \sigma^4
    (1-\rho^2),
  \end{IEEEeqnarray*}
  provided that $D_1^{\ast} \leq \sigma^2$ and $D_2^{\ast} \leq
  \sigma^2$, and that
  \begin{IEEEeqnarray}{rCl} \label{eq:lim-Di}
    \lim_{N \rightarrow 0} \frac{N}{P_1 D_1^{\ast}} = 0 & \qquad \text{and} 
    \qquad & \lim_{N \rightarrow 0} \frac{N}{P_2 D_2^{\ast}} = 0.
  \end{IEEEeqnarray}
\end{thm}

\begin{proof}
  See Appendix \ref{appx:prf-thm-asympt-nofb}.
\end{proof}

We restate Theorem \ref{thm:nofb-asymptotics} more specifically for
the symmetric case. Since there $D_1^{\ast} = D_2^{\ast} =
D^{\ast}(\sigma^2, \rho, P, N)$, condition \eqref{eq:lim-Di} is
implicitly satisfied. Thus,
\begin{cor}\label{cor:noFB-high-SNR}
  In the symmetric case
  \begin{IEEEeqnarray}{rCl}\label{eq:vq-high-snr}
    \lim_{P/N \rightarrow \infty} \sqrt{\frac{P}{N}} D^{*}(\sigma^2, \rho,
    P, N) = \sigma^2 \sqrt{\frac{1-\rho}{2}}.
  \end{IEEEeqnarray}
\end{cor}

\begin{rmk}
  Corollary \ref{cor:noFB-high-SNR} can also be deduced without
  Theorem \ref{thm:nofb-asymptotics}, by comparing the distortion of
  the vector-quantizer scheme in Corollary \ref{cor:nofb-VQ-sym} to
  the lower bound on $D^{\ast}(\sigma^2, \rho, P, N)$ in Corollary
  \ref{cor:nofb-nec-cond-sym}.
\end{rmk}

To get some understanding of the coefficient on the RHS of
\eqref{eq:vq-high-snr}, let us first rewrite \eqref{eq:vq-high-snr} as
follows
\begin{IEEEeqnarray*}{rCl}
  D^{\ast} (\sigma^2, \rho, P, N) & \approx & \sigma^2 \sqrt{\frac{N
      (1-\rho^2)}{2P(1+\rho)}} \hspace{15mm} \text{as} \hspace{3mm}
  \frac{P}{N} \gg 1.
\end{IEEEeqnarray*}
Next, let us compare this asymptotic behavior to that of two
suboptimal schemes: the best separation-based scheme and the
suboptimal separation-based scheme that completely ignores the source
correlation, i.e., the best scheme where the transmitters and the
receiver treat the two source components as if they where
independent. Denoting the distortion of the best separation-based
scheme by $D_{\textnormal{SB}}$ and the distortion of the scheme that
ignores the source correlation by $D_{\textnormal{IC}}$, gives
\begin{IEEEeqnarray*}{rCl}
  D_{\textnormal{SB}} \approx \sigma^2 \sqrt{\frac{N(1-\rho^2)}{2P}}
  & \qquad \text{and} \qquad & D_{\textnormal{IC}} \approx \sigma^2
  \sqrt{\frac{N}{2P}}, \hspace{10mm} \text{as} \hspace{3mm} \frac{P}{N}
  \gg 1.
\end{IEEEeqnarray*}
The asymptotic expression for $D_{\textnormal{SB}}$ follows by
Corollary \ref{cor:noFB-sep-based} and the asymptotic expression for
$D_{\textnormal{IC}}$ follows from combining the rate-distortion
function of a Gaussian random variable, see e.g.~\cite[Theorem 13.3.2,
p.~344]{cover-thomas91}, with the capacity region of the Gaussian
multiple-access channel, see e.g.~\cite[Section 14.3.6,
p.~403]{cover-thomas91}.

The asymptotic behavior can now be understood as follows. The
denominator under the square-root corresponds to the average power
that the scheme under discussion produces on the sum of the channel
inputs $X_{1,k} + X_{2,k}$. In the two separation-based approaches
this average power is $2P$, and in the vector-quantizer scheme this
average power is $2P(1+\rho)$ as $P/N \rightarrow \infty$. The
numerator under the square-root consists of the noise variance $N$
multiplied by a coefficient reflecting the gain due to the logical
exploitation of the source correlation. For the scheme ignoring the
source correlation this coefficient is, by definition of the scheme,
equal to 1, i.e., no gain, whereas for the best separation-based scheme
and for the vector-quantizer scheme this coefficient is equal to
$1-\rho^2$. The means by which this gain is obtained in the best
separation-based scheme and in the vector-quantizer scheme are
fundamentally different. In the separation-based scheme the gain is
achieved by a generalized form of Slepian-Wolf coding (see
\cite{oohama97}), whereas in the vector-quantizer scheme the gain is
achieved by joint-typicality decoding that takes into consideration
the correlation between the transmitted sequences ${\bf U}_1^{\ast}$
and ${\bf U}_2^{\ast}$ (see Theorem \ref{thm:vq-achv}). The
corresponding advantage of the vector-quantizer scheme is that by
performing the logical exploitation only at the receiver, it
additionally allows for exploiting the source correlation in a
physical way, i.e., by producing a power boost in the transmitted
signal pair.

\subsection{Superposition Approach}

The last scheme of this paper is a combination of the previously
considered uncoded scheme and vector-quantizer scheme. One way to
combine these schemes would be by time- and power-sharing. As stated
in Remark \ref{rmk:convexity}, this would result in a convexification
of the union of the achievable distortions of the two individual
schemes. In this section, we instead propose an approach where the two
schemes are superimposed. In the symmetric case, this approach results
in better performances than time- and power-sharing, and for all SNRs,
the resulting distortion is very close to the lower bound on
$D^{\ast}(\sigma^2, \rho, P, N)$ of Corollary
\ref{cor:nofb-nec-cond-sym}. We also point out that for the simpler
problem of transmitting a univariate memoryless Gaussian source over a
point-to-point AWGN channel subject to expected squared-error
distortion, a similar superposition approach was shown in
\cite{bross-lapidoth-tinguely06} to yield a continuum of optimal
schemes.

The superimposed scheme can be described as follows. The channel input
sequence ${\bf X}_i$ produced by Encoder $i$, $i \in \{ 1,2 \}$, is a
linear combination of the source sequence ${\bf S}_i$ and its
rate-$R_i$ vector-quantized version ${\bf U}_i^{\ast}$. That is,
\begin{IEEEeqnarray}{rCl}\label{eq:si-Xi}
  {\bf X}_i & = & \alpha_i {\bf S}_i + \beta_i {\bf U}_i^{\ast},
\end{IEEEeqnarray}
where the sequence ${\bf U}_i^{\ast}$ is obtained in exactly the same way as
in the vector-quantizer scheme, and where the coefficients $\alpha_i$ and
$\beta_i$ are chosen so that the sequence ${\bf X}_i$ satisfies the
power constraint \eqref{eq:power-constraint}, and so that the receiver
can, with high probability, recover the transmitted codeword pair
$({\bf U}_1^{\ast}, {\bf U}_2^{\ast})$. As we shall see, these two
conditions will be satisfied as long as $\alpha_i$ and $\beta_i$ satisfy to
within some $\epsilon$'s and $\delta$'s
\begin{IEEEeqnarray}{rCl}\label{eq:superimp:alpha-beta}
  \alpha_i \in \left[ 0, \frac{P_i}{\sigma^2} \right] & \qquad &
  \beta_i = \sqrt{\frac{P_i - \alpha_i^2 \sigma^2
      2^{-2R_i}}{\sigma^2(1-2^{-2R_i})}} - \alpha_i \qquad i \in \{
  1,2 \}.
\end{IEEEeqnarray}
(For a precise statement see Appendix~\ref{appx:prf-thm-supimp}).


From the resulting channel output ${\bf Y} = {\bf X}_1 + {\bf X}_2 +
{\bf Z}$, the decoder then makes a guess $(\hat{\bf U}_1, \hat{\bf
  U}_2)$ of the transmitted sequences $({\bf U}_1^{\ast}, {\bf
  U}_2^{\ast})$. This guess is obtained by joint typicality decoding
that takes into consideration the correlation between ${\bf
  U}_1^{\ast}$, ${\bf U}_2^{\ast}$, ${\bf S}_1$ and ${\bf S}_2$. From
the sequences $\hat{\bf U}_1$, $\hat{\bf U}_2$, and ${\bf Y}$, the
decoder then computes approximate MMSE estimates $\hat{\bf S}_1$ and
$\hat{\bf S}_2$ of the source sequences ${\bf S}_1$ and ${\bf S}_2$,
i.e.,
\begin{IEEEeqnarray}{rCl}
  \hat{\bf S}_i & = & \gamma_{i1} \hat{\bf U}_1 + \gamma_{i2} \hat{\bf
    U}_2 + \gamma_{i3} {\bf Y} \qquad \qquad i \in \{ 1,2 \}, \label{eq:si-reconstr-Sih}
\end{IEEEeqnarray}
where the coefficients $\gamma_{ij}$ are chosen such that $\hat{\bf
  S}_i \approx \mat{E} \big[ {\bf S}_i \big| {\bf Y}, \hat{\bf U}_1,
\hat{\bf U}_2 \big]$. To state the explicit form of coefficients
$\gamma_{ij}$, define for any rate pair $(R_1,R_2)$, where $R_i \geq
0$, the $3 \times 3$ matrix $\mat{K}(R_1,R_2)$ by
\begin{IEEEeqnarray}{rCl}\label{eq:si-def-K}
  \mat{K}(R_1,R_2) & \triangleq & \left( \begin{array}{c c c}
      \mat{k}_{11} & \mat{k}_{12} & \mat{k}_{13}\\
      \mat{k}_{12} & \mat{k}_{22} & \mat{k}_{23}\\
      \mat{k}_{13} & \mat{k}_{23} & \mat{k}_{33}
    \end{array} \right),
\end{IEEEeqnarray}
where
\begin{IEEEeqnarray*}{rCl}
  \mat{k}_{11} & = & \sigma^2 (1-2^{-2R_1})\\
  \mat{k}_{12} & = & \sigma^2 \rho (1-2^{-2R_1}) (1-2^{-2R_2})\\
  \mat{k}_{13} & = & (\alpha_1 + \beta_1 + \alpha_2 \rho) \mat{k}_{11}
  + \beta_2 \mat{k}_{12}\\
  \mat{k}_{22} & = & \sigma^2 (1-2^{-2R_2})\\
  \mat{k}_{23} & = & (\alpha_2 + \beta_2 + \alpha_1 \rho) \mat{k}_{22}
  + \beta_1 \mat{k}_{12}\\
  \mat{k}_{33} & = & \alpha_1^2 \sigma^2 + 2 \alpha_1 \beta_1
  \mat{k}_{11} + 2 \alpha_1 \alpha_2 \rho \sigma^2 + 2 \alpha_1
  \beta_2 \rho \mat{k}_{22} + \beta_1^2 \mat{k}_{11} + 2 \beta_1
  \alpha_2 \rho \mat{k}_{11}\\
  & & {} + 2 \beta_1 \beta_2 \mat{k}_{12} + 2 \alpha_2 \beta_2
  \mat{k}_{22} + \alpha_2^2 \sigma^2 + \beta_2^2 \mat{k}_{22} + N.
\end{IEEEeqnarray*}
The coefficients $\gamma_{ij}$ are then given by
\begin{IEEEeqnarray}{rCl}\label{eq:src-estmt-coeff}
  \left( \begin{array}{c}
      \gamma_{i1}\\
      \gamma_{i2}\\
      \gamma_{i3} 
    \end{array} \right) & \triangleq & \mat{K}^{-1}(R_1,R_2)
  \left( \begin{array}{c}
      \mat{c}_{i1}\\
      \mat{c}_{i2}\\
      \mat{c}_{i3}
    \end{array} \right) \qquad \qquad i \in \{ 1,2 \},
\end{IEEEeqnarray}
where
\begin{IEEEeqnarray*}{rCl}
  \mat{c}_{11} & = & \mat{k}_{11}\\
  \mat{c}_{12} & = & \rho \mat{k}_{22}\\
  \mat{c}_{13} & = & (\alpha_1 + \alpha_2 \rho) \sigma^2 + \beta_1
  \mat{k}_{11} + \beta_2 \rho \mat{k}_{22}\\
  \mat{c}_{21} & = & \rho \mat{k}_{11}\\
  \mat{c}_{22} & = & \mat{k}_{22}\\
  \mat{c}_{23} & = & (\alpha_2 + \alpha_1 \rho) \sigma^2 + \beta_1
  \rho \mat{k}_{11} + \beta_2 \mat{k}_{22}.
\end{IEEEeqnarray*}
The distortions achieved by the superimposed scheme are now given in
the following theorem.

\begin{thm}\label{thm:superimposed}
  The distortions achieved by the superposition approach are all
  pairs $(D_1,D_2)$ satisfying
  \begin{IEEEeqnarray*}{rCl}
    D_i & > & \sigma^2 - \gamma_{i1} \mat{c}_{i1} - \gamma_{i2}
    \mat{c}_{i2} - \gamma_{i3} \mat{c}_{i3} \qquad \qquad i \in \{ 1,2
    \}.
  \end{IEEEeqnarray*}
  where the rate-pair $(R_1,R_2)$ satisfies 
  \begin{IEEEeqnarray*}{rCl}
    R_1 & < & \frac{1}{2} \log_2 \left( \frac{\beta_1'^{2} \mat{k}_{11}
        (1-\tilde{\rho}^2) + N'}{N'(1-\tilde{\rho}^2)} \right)\\[2mm]
    R_2 & < & \frac{1}{2} \log_2 \left( \frac{\beta_2'^{2} \mat{k}_{22}
        (1-\tilde{\rho}^2) + N'}{N'(1-\tilde{\rho}^2)} \right)\\[2mm]
    R_1 + R_2 & < & \frac{1}{2} \log_2 \left( \frac{\beta_1'^{2} \mat{k}_{11}
        + \beta_2'^{2} \mat{k}_{22} + 2 \tilde{\rho} \beta_1' \beta_2'
          \sqrt{\mat{k}_{11} \mat{k}_{22}} + N'}{N'(1-\tilde{\rho}^2)} \right),
  \end{IEEEeqnarray*}
  for some $\alpha_1$, $\alpha_2$, $\beta_1$, and
  $\beta_2$ satisfying \eqref{eq:superimp:alpha-beta} and where
  \begin{IEEEeqnarray}{rCl}\label{eq:si-N'}
    N' & = & \alpha_1^2 \nu_1 + \alpha_2^2 \nu_2 + 2 \alpha_1
    \alpha_2 \nu_3 + N,
  \end{IEEEeqnarray}
  where
  \begin{IEEEeqnarray*}{rCl}
    \nu_1 & = & \sigma^2 - (1-a_1\tilde{\rho})^2 \mat{k}_{11}
    - 2 (1-a_1\tilde{\rho}) a_1 \mat{k}_{12} - a_1^2 \mat{k}_{22}\\[1mm]
    \nu_2 & = & \sigma^2 - (1-a_2\tilde{\rho})^2 \mat{k}_{22}
    - 2 (1-a_2\tilde{\rho}) a_2 \mat{k}_{12} - a_2^2 \mat{k}_{11}\\[1mm]
    \nu_3 & = & \rho \sigma^2 - \big( (1-a_1\tilde{\rho})
    (1-a_2\tilde{\rho}) + a_1a_2\big) \mat{k}_{12} - (1-a_1\tilde{\rho}) a_2 \mat{k}_{11}
    - (1-a_2\tilde{\rho}) a_1 \mat{k}_{22},
  \end{IEEEeqnarray*}
  with
  \begin{IEEEeqnarray}{rCl}
    \beta_1' & = & \alpha_1 (1-a_1\tilde{\rho}) + \beta_1 + \alpha_2
    a_2 \label{eq:si-beta1'}\\[2mm]
    \beta_2' & = & \alpha_2 (1-a_2\tilde{\rho}) + \beta_2 + \alpha_1
    a_1, \label{eq:si-beta2'}
  \end{IEEEeqnarray}
  and with
  \begin{IEEEeqnarray}{rCl}
    a_1 & = & \frac{\rho 2^{-2R_1} (1-2^{-2R_2})}{(1-2^{-2R_2}) -2
      \tilde{\rho}^2 \sqrt{(1-2^{-2R_1})(1-2^{-2R_2})} + \tilde{\rho}^2
      (1-2^{-2R_1})}, \label{eq:si-a1}\\[3mm]
    a_2 & = & \frac{\rho 2^{-2R_2} (1-2^{-2R_1})}{(1-2^{-2R_1}) -2
      \tilde{\rho}^2 \sqrt{(1-2^{-2R_1})(1-2^{-2R_2})} + \tilde{\rho}^2
      (1-2^{-2R_2})}. \label{eq:si-a2}
  \end{IEEEeqnarray}
\end{thm}

\begin{proof}
  See Appendix \ref{appx:prf-thm-supimp}.
\end{proof}

In the symmetric case where $P_1 = P_2 = P$, $R_1 = R_2 = R$ and where
$\alpha_1 = \alpha_2 = \alpha$ and $\beta_1 = \beta_2 = \beta$, the
matrix $\mat{K}(R,R)$ and the coefficients $\gamma_{ij}$ reduce to
\begin{IEEEeqnarray*}{rCl}
  \mat{K}(R,R) & = & \left( \begin{array}{c c c}
      \mat{k}_1 & \mat{k}_2 & \mat{k}_3\\
      \mat{k}_2 & \mat{k}_1 & \mat{k}_3\\
      \mat{k}_3 & \mat{k}_3 & \mat{k}_4
    \end{array} \right) \qquad \text{where} \quad \begin{array}{l}
    \mat{k}_1 = \sigma^2 (1-2^{-2R})\\[1mm]
    \mat{k}_2 = \sigma^2 \rho (1-2^{-2R})^2\\[1mm]
    \mat{k}_3 = (\alpha + \beta + \alpha \rho) \mat{k}_1
    + \beta \mat{k}_2\\[1mm]
    \mat{k}_4 = 2 \alpha \mat{c}_3 + 2 \beta \mat{k}_3 + N,
  \end{array}
\end{IEEEeqnarray*}
and
\begin{IEEEeqnarray*}{rCl}
  \left( \begin{array}{c}
      \gamma_1\\
      \gamma_2\\
      \gamma_3 
    \end{array} \right) & \triangleq & \mat{K}^{-1}(R,R)
  \left( \begin{array}{c}
      \mat{c}_{1}\\
      \mat{c}_{2}\\
      \mat{c}_{3}
    \end{array} \right) \qquad \text{where} \quad \begin{array}{l}
    \mat{c}_{1} = \mat{k}_{1}\\
    \mat{c}_{2} = \rho \mat{k}_{1}\\
    \mat{c}_{3} = ( \alpha \sigma^2 + \beta \mat{k}_{1}) (1 +
    \rho).
  \end{array}
\end{IEEEeqnarray*}
Thus, in the symmetric case Theorem \ref{thm:superimposed} simplifies
as follows.
\begin{cor}\label{cor:superimposed-sym}
  With the superposition approach in the symmetric case we can achieve
  the distortion
  \begin{IEEEeqnarray*}{C}
    \inf \sigma^2 - \gamma_1 \mat{c}_1 - \gamma_2 \mat{c}_2 -
    \gamma_3 \mat{c}_3,
  \end{IEEEeqnarray*}
  where the infimum is over all rates $R$ satisfying
  \begin{IEEEeqnarray*}{rCl}
    R & < & \frac{1}{4} \log_2 \left( \frac{2 \beta'^{2}
        \mat{k}_{1} (1 + \tilde{\rho}) + N'}{N'(1-\tilde{\rho}^2)}
    \right),
  \end{IEEEeqnarray*}
  for some $\alpha$ and $\beta$ satisfying
  \begin{IEEEeqnarray}{rCl}
    \alpha \in \left[ 0, \frac{P}{\sigma^2} \right] & \qquad
    \text{and} \qquad & \beta = \sqrt{\frac{P - \alpha^2 \sigma^2
        2^{-2R}}{\sigma^2(1-2^{-2R})}} - \alpha, 
  \end{IEEEeqnarray}
  and where
  \begin{IEEEeqnarray*}{rCl}
    \beta' = \alpha \left(1 + \frac{\rho
        2^{-2R}}{1-\tilde{\rho}^2}(1-\tilde{\rho}) \right) + \beta,
  \end{IEEEeqnarray*}
  and
  \begin{IEEEeqnarray*}{rCl}
    N' & = & 2 \alpha^2 (\nu_1 + \nu_3) + N,
  \end{IEEEeqnarray*}
  with
  \begin{IEEEeqnarray*}{rCl}
    \nu_1 = \sigma^2 2^{-2R} \frac{1-\rho
      \tilde{\rho}}{1-\tilde{\rho}^2} & \hspace{20mm} & \nu_3 =
    \sigma^2 \rho
    \frac{2^{-4R}}{1-\tilde{\rho}^2}.
  \end{IEEEeqnarray*}
\end{cor}

To conclude our main results we have illustrated in Figure
\ref{fig:nofb-all} all presented upper and lower bounds on
$D^{\ast}(\sigma^2,\rho,P,N)$.

\begin{figure}[h]
  \centering
  \psfrag{ds}[cc][cc]{$D/\sigma^2$}
  \psfrag{pn}[cc][cc]{$P/N$}
  \psfrag{r}[cc][cc]{$\rho = 0.5$}
  \psfrag{sb}[cc][cc]{source-channel separation}
  \psfrag{vq}[cc][cc]{vector-quantizer}
  \psfrag{uc}[cc][cc]{uncoded}
  \psfrag{si}[cc][cc]{super-imposed}
  \psfrag{lb}[cc][cc]{lower bound}
  \psfrag{0}[cc][cc]{\footnotesize $0$}
  \psfrag{0.2}[cc][cc]{\footnotesize $0.2$}
  \psfrag{0.3}[cc][cc]{\footnotesize $0.3$}
  \psfrag{0.4}[cc][cc]{\footnotesize $0.4$}
  \psfrag{0.5}[cc][cc]{\footnotesize $0.5$}
  \psfrag{0.6}[cc][cc]{\footnotesize $0.6$}
  \psfrag{0.7}[cc][cc]{\footnotesize $0.7$}
  \psfrag{0.8}[cc][cc]{\footnotesize $0.8$}
  \psfrag{0.9}[cc][cc]{\footnotesize $0.9$}
  \psfrag{1}[cc][cc]{\footnotesize $1$}
  \psfrag{1.5}[cc][cc]{\footnotesize $1.5$}
  \psfrag{2}[cc][cc]{\footnotesize $2$}
  \psfrag{2.5}[cc][cc]{\footnotesize $2.5$}
  \psfrag{3}[cc][cc]{\footnotesize $3$}
  \psfrag{4}[cc][cc]{\footnotesize $4$}
  \epsfig{file=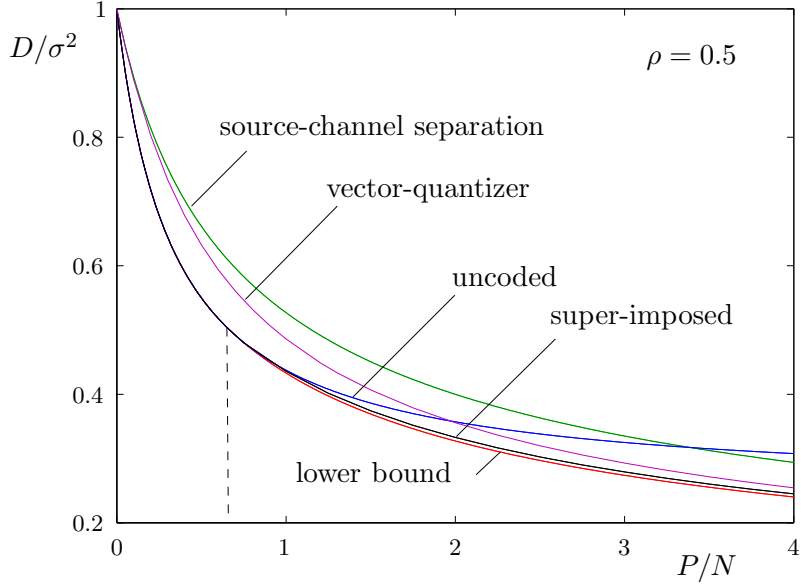, width=0.8\textwidth}
  \caption{Upper and lower bounds on $D^{\ast}(\sigma^2,\rho,P,N)$ for
    a source of correlation coefficient $\rho = 0.5$.}
  \label{fig:nofb-all}
\end{figure}

\section{Summary}

We studied the power versus distortion trade-off for the distributed
transmission of a memoryless bi-variate Gaussian source over a
two-to-one average-power limited Gaussian multiple-access channel. In
this problem, each of two separate transmitters observes a different
component of a memoryless bi-variate Gaussian source. The two
transmitters then describe their source component to a common
receiver via a Gaussian multiple-access channel with average-power
constraints on each channel input sequences. From the output of the
multiple-access channel, the receiver wishes to reconstruct each
source component with the least possible expected squared-error
distortion. Our interest was in characterizing the distortion pairs
that are simultaneously achievable on the two source components. These
pairs are a function of the power constraints and the variance of the
additive channel noise, as well as of the source variance and of the
correlation coefficient between the two source components.

We first considered a different (non-distributed) problem, which was
the point-to-point analog of our multiple-access problem. More
precisely, we studied the power versus distortion trade-off for the
transmission of a memoryless bi-variate Gaussian source over the AWGN
channel, subject to expected squared-error distortion on each source
component. For this problem, we determined the set of achievable
distortion pairs by deriving the explicit expression for the
rate-distortion function of a memoryless bi-variate Gaussian
source. Moreover, we showed that below a certain SNR-threshold an
uncoded transmission scheme is optimal.


For the multiple-access problem we then derived:\\

\begin{itemize}
\item A necessary condition for the achievability of a distortion pair
  (Theorem~\ref{thm:nofb-nec-cond}). This condition was obtained by
  reducing the multiple-access problem to a point-to-point
  problem. The key step was to upper bound the maximal correlation
  between two simultaneous channel inputs by using a result from
  maximum correlation theory.\\
\item The optimality of an uncoded transmission scheme below a certain
  SNR-threshold (Theorem~\ref{thm:uncoded-main}). In the symmetric
  case, this result becomes (Corollary~\ref{cor:uncoded-optimal})
  \begin{IEEEeqnarray*}{rCl}
    D^{\ast}(\sigma^2,\rho,P,N) = \sigma^2
    \frac{P(1-\rho^2)+N}{2P(1+\rho)+N}, & \hspace{15mm} \text{for all}
    \hspace{3mm} \frac{P}{N} \leq \frac{\rho}{1-\rho^2}.
  \end{IEEEeqnarray*}
  The strength of the underlying uncoded scheme is that it translates
  the entire source correlation onto the channel inputs, and thereby
  boosts the received power of the transmitted signal pair. Its
  weakness is that it allows the receiver to recover only the sum of
  the channel inputs.\\
\item A sufficient condition based on a ``vector-quantizer'' scheme
  (Theorem~\ref{thm:vq-achv}). The motivation behind this scheme was
  to overcome the weakness of the uncoded scheme. To this end, rather
  than transmitting the source components in an uncoded manner, the
  scheme transmits a scaled version of the optimally vector-quantized
  source components (without channel coding).\\ 
\item The precise high-SNR asymptotics of an optimal transmission
  scheme, which in the symmetric case are given by (Corollary
  \ref{cor:noFB-high-SNR})
  \begin{IEEEeqnarray*}{rCl}
    \lim_{P/N \rightarrow \infty} \sqrt{\frac{P}{N}}
    D^{\ast}(\sigma^2,\rho,P,N) & = & \sigma^2
    \sqrt{\frac{1-\rho}{2}}.
  \end{IEEEeqnarray*}
  This result follows from the ``vector-quantizer'' scheme
  (Theorem~\ref{thm:vq-achv}) and the necessary condition of Theorem
  \ref{thm:nofb-nec-cond}.\\ 
\item The suboptimality, in the symmetric case, of source-channel
  separation at all SNRs. This follows from comparing the best
  separation-based approach (Corollary~\ref{cor:noFB-sep-based}) with
  the achievable distortions from the ``vector-quantizer'' scheme
  (Corollary~\ref{cor:nofb-VQ-sym}).\\
\item A sufficient condition based on a superposition of the uncoded
  scheme and the vector-quantizer scheme
  (Theorem~\ref{thm:superimposed}). In the symmetric case this
  superposition approach was shown to be optimal or close to optimal
  at all SNRs.\\
\end{itemize}

The presented sufficient conditions indicate that for the efficient
exploitation of the source correlation it is necessary not only to
exploit the source correlation in a logical way, e.g.~by
Slepian-Wolf-like strategies, but to additionally exploit the source
correlation in a physical way. In the considered schemes, this is done
by translating the source correlation onto the channel inputs. The
logical exploitation of the source correlation is then performed at
the receiver-side, e.g.~by joint-typicality decoding taking into
consideration the correlation between the channel inputs.


\newpage

\appendix

\section{Proof of Theorem
  \ref{thm:rd1d2-main}} \label{appx:prf-thm-RD1D2}

Theorem \ref{thm:rd1d2-main} gives the expression of the
rate-distortion function $R_{S_1,S_2}(D_1,D_2)$. A single-letter
expression of this function, in the form of an optimization problem,
follows from \cite[Theorem 2, p.~856]{elgamal_cover82} and is
\begin{equation}\label{eq:problem}
  R_{S_1,S_2}(D_1,D_2) = \min_{\substack{P_{\widehat{S}_1,
        \widehat{S}_2 | S_1,S_2}:\\
      \E{(S_1 - \widehat{S}_1)^2} \leq D_1 \\
      \E{(S_2 - \widehat{S}_2)^2} \leq D_2}} I(S_1,S_2 ;
  \widehat{S}_1, \widehat{S}_2 ).
\end{equation}
Thus, to prove Theorem \ref{thm:rd1d2-main} it remains to solve
\eqref{eq:problem} for all distortion pairs $(D_1,D_2) \in
(0,\sigma^2] \times (0,\sigma^2]$. One way of doing this was presented
in \cite{xiao-luo05}. Here, we present an alternative approach. The
derivation is split in two parts. First we derive
$R_{S_1,S_2}(D_1,D_2)$ for $(D_1,D_2) \in \mathscr{D}_1$, and then for
$(D_1,D_2) \in (0,\sigma^2] \times (0,\sigma^2] \setminus
\mathscr{D}_1$. Before starting with the derivations, we remark:

\begin{rmk}
  The restriction to $\E{S_1^2} = \E{S_2^2} = \sigma^2$ and $\rho \in
  [0,1]$ incurs no loss in generality in the evaluation of
  $R_{S_1,S_2}(D_1,D_2)$ since the distortion region scales linearly
  with the source variance, and since the distortion region is the
  same for correlation coefficients $-\rho$ and $\rho$.
\end{rmk}

\subsection{$R_{S_1,S_2}(D_1,D_2)$ for the Region $\mathscr{D}_1$}\label{sec:prf-regD3}

For pairs $(D_1,D_2)$ in $\mathscr{D}_1$ the evaluation of
$R_{S_1,S_2}(D_1,D_2)$ is very simple. The pairs $(D_1,D_2) \in
\mathscr{D}_1$ are pairs where the larger of the two distortions can
be achieved from the description of the other source component without
any additional information simply by exploiting the correlation
between $S_1$ and $S_2$. For example, since $S_2 = (S_2 -\rho S_1) +
\rho S_1$, from an optimally vector-quantized version of $S_1$ that
yields a distortion $D_1$, a distortion $D_2 = \sigma^2 (1-\rho^2) +
\rho^2 D_1$ can be achieved on $S_2$ by picking the reconstruction
$\hat{S}_2$ of $S_2$ as a scaled version of the reconstruction
$\hat{S}_1$ of $S_1$. Thus, a necessary and sufficient condition to
achieve a distortion pair $(D_1,\sigma^2 (1-\rho^2) + \rho^2 D_1)$ is
\begin{displaymath}
  R = \frac{1}{2} \log_2^+ \left( \frac{\sigma^2}{D_1} \right).
\end{displaymath}
By symmetry, the above argument also works when the roles of $S_1$ and
$S_2$ reversed. Hence,
\begin{align}
  R_{S_1,S_2}(D_1,D_2) &= \max \left\{ R_{S_1}(D_1), R_{S_2}(D_2)
  \right\}\nonumber\\ 
  &= \frac{1}{2} \log_2^+ \left( \frac{\sigma^2}{D_{\textnormal{min}}}
  \right) \qquad \qquad \text{if } (D_1,D_2) \in
  \mathscr{D}_1.\label{eq:RD1D2-D3} 
\end{align}

\subsection{$R_{S_1,S_2}(D_1,D_2)$ for the Regions $\mathscr{D}_2$ and
  $\mathscr{D}_3$}\label{sec:prf-regD1D2}

To solve \eqref{eq:problem} for $(D_1,D_2) \in (0,\sigma^2] \times
(0,\sigma^2] \setminus \mathscr{D}_1$ we propose the following
algorithm. First, we scale the source component $S_2$ by a factor $c
\in \Reals^+$, where $\Reals^+$ is the set of positive reals numbers,
and then we decorrelate $(S_1,cS_2)$ to a pair $(V_1,V_2)$, by means of the
unitary linear transformation:
\begin{equation}\label{eq:decorrelation}
  \left( \begin{array}{c}
      V_1\\
      V_2
  \end{array} \right)
= 
\left( \begin{array}{c c}
      a_1 & -a_2\\
      a_2 & a_1
  \end{array} \right)
\left( \begin{array}{c c}
      1 & 0\\
      0 & c
  \end{array} \right)
\left( \begin{array}{c}
      S_1\\
      S_2
  \end{array} \right),
\end{equation}
where $\trans{(a_1,a_2)}$, $\trans{(-a_2,a_1)}$ are normalized
eigenvectors of the covariance matrix of $(S_1,cS_2)$. Finally, we
apply reverse waterfilling with rate $R$ on $(V_1,V_2)$, as described
in \cite[pp.~347]{cover-thomas91}. The resulting reconstruction of
$(V_1,V_2)$ is denoted by $(\hat{V}_1,\hat{V}_2)$. From
$(\hat{V}_1,\hat{V}_2)$ the reconstruction $(\hat{S}_1,\hat{S}_2)$ of
$(S_1,S_2)$ is then obtained by
\begin{equation}\label{eq:recorrelation}
  \left( \begin{array}{c}
      \hat{S}_1\\
      \hat{S}_2
    \end{array} \right)
  = \mat{B}
  \left( \begin{array}{c}
      \hat{V}_1\\
      \hat{V}_2
    \end{array} \right),
\end{equation}
where the matrix $\mat{B}$ is the reverse transformation
of \eqref{eq:decorrelation}, i.e.
\begin{equation}
  \mat{B} = 
  \left( \begin{array}{c c}
      1 & 0\\
      0 & \frac{1}{c}
    \end{array} \right)
  \left( \begin{array}{c c}
      a_1 & a_2\\
      -a_2 & a_1
    \end{array} \right).
\end{equation}

The probability law of $(S_1,S_2,\hat{S}_1,\hat{S}_2)$ that results
from the algorithm is a solution to the optimization problem in
\eqref{eq:problem}. This is stated more precisely in the following
theorem.

\vspace{5mm}

\begin{thm}\label{thm:basic-method}
  \hfill
  \renewcommand{\labelenumi}{\roman{enumi})}
  \begin{enumerate}
  \item For every distortion pair $(D_1,D_2) \in (0,\sigma^2] \times
    (0,\sigma^2] \setminus \mathscr{D}_1$ there exists some rate
    $R^{\ast} \in \Reals^+$ and some scaling coefficient $c^{\ast} \in
    \Reals^+$ such that the distortion pair resulting from the above
    algorithm is $(D_1,D_2)$.
  \item The probability law of $(S_1,S_2,\hat{S}_1,\hat{S}_2)$
    resulting from this algorithm is a solution to the optimization
    problem in \eqref{eq:problem}.
  \item For every $(D_1,D_2) \in (0,\sigma^2] \times (0,\sigma^2]
    \setminus \mathscr{D}_1$
    \begin{equation}\label{eq:method-main}
      R_{S_1,S_2}(D_1,D_2) = \frac{1}{2} \log_2 \left(
        \frac{\det{\cov{S}}}{\det{\left( \cov{S} - \cov{\hat{S}} \right)}} \right),
    \end{equation}
    where $\cov{\hat{S}}$ is the covariance matrix of the random vector
    \eqref{eq:recorrelation} corresponding to the above $R^{\ast}$ and
    $c^{\ast}$.
  \end{enumerate}
\end{thm}

\begin{proof}
  See Section \ref{sec:prf-thm-basic}.
\end{proof}

Based on Theorem \ref{thm:basic-method}, we now evaluate
$R_{S_1,S_2}(D_1,D_2)$ for all $(D_1,D_2) \in (0,\sigma^2] \times
(0,\sigma^2] \setminus \mathscr{D}_1$ by expressing $\cov{\hat{S}}$ in
terms of $(D_1,D_2)$. This is done in the following lemma.

\begin{lm}\label{lm:KShat-new}
  For $(D_1,D_2) \in (0,\sigma^2] \times (0,\sigma^2] \setminus
  \mathscr{D}_1$, the covariance matrix $\cov{\hat{S}}$ in
  \eqref{eq:method-main} is given by
  \begin{equation}\label{eq:cvsh}
    \cov{\hat{S}} = \left\{ \begin{array}{l l}
        \left( \begin{array}{c c}
            \sigma^2-D_1 & \rho \sigma^2\\[5mm]
            \rho \sigma^2 & \sigma^2 -D_2
          \end{array} \right) & \text{if } (D_1,D_2) \in \mathscr{D}_2\\[8mm]
        \left( \begin{array}{c c}
            \sigma^2-D_1 & \sqrt{(\sigma^2-D_1)(\sigma^2-D_2)}\\[5mm]
            \sqrt{(\sigma^2-D_1)(\sigma^2-D_2)} & \sigma^2 -D_2
          \end{array} \right) & \text{if } (D_1,D_2) \in
        \mathscr{D}_3.
      \end{array} \right.
  \end{equation}
\end{lm}

\begin{proof}
  See Section \ref{sec:RD1D2-prf-lm-D2D3}
\end{proof}

Combining Lemma \ref{lm:KShat-new} with Theorem \ref{thm:basic-method}
gives
\begin{equation}\label{eq:RS1S2-woD3}
  R_{S_1,S_2}(D_1,D_2) = \left\{ \begin{array}{l l}
      \frac{1}{2} \log_2^+ \left( \frac{\sigma^4 (1-\rho^2)}{D_1 D_2}
      \right) & \text{if } (D_1,D_2) \in \mathscr{D}_2\\[4mm]
      \frac{1}{2} \log_2^+ \left( \frac{\sigma^4 (1-\rho^2)}{D_1D_2 -
          \left( \rho \sigma^2 - \sqrt{(\sigma^2-D_1)(\sigma^2-D_2)}
          \right)^2} \right) & \text{if } (D_1,D_2) \in \mathscr{D}_3.
  \end{array} \right.
\end{equation}
The complete expression of $R_{S_1,S_2}(D_1,D_2)$ for all $(D_1,D_2)
\in (0,\sigma^2] \times (0,\sigma^2]$ is now obtained by combining
\eqref{eq:RS1S2-woD3} with \eqref{eq:RD1D2-D3}. \hfill $\Box$\\

To conclude this section, we briefly point out that the above
described algorithm easily extends to multi-variate Gaussians with more
than two components.

\begin{rmk}
  The extension of the above algorithm to multivariate Gaussians with
  more than two components is straight forward. For a source with
  $\nu$ components $(S_1,S_2, \ldots ,S_{\nu})$, the scaling needs to
  be applied to $\nu-1$ components, e.g.~each of the components $S_2$,
  \ldots ,$S_{\nu}$ is scaled with a respective coefficient $c_{2}$,
  \ldots $c_{\nu}$ in $\Reals^+$. The tuple $(S_1, c_2 S_2, \ldots ,
  c_{\nu}S_{\nu})$ is then unitarily decorrelated to the tuple
  $(V_1,V_2, \ldots ,V_{\nu})$ on which reverse waterfilling is
  applied. The reconstructions $(\hat{S}_1, \hat{S}_2, \ldots
  ,\hat{S}_{\nu})$ then follow from $(\hat{V}_1, \hat{V}_2, \ldots
  ,\hat{V}_{\nu})$ by the reverse of the transformation that mapped
  $(S_1, S_2, \ldots , S_{\nu})$ to $(V_1,V_2, \ldots ,V_{\nu})$. The
  corresponding extension of Theorem \ref{thm:basic-method} follows
  easily from the proof in Section \ref{sec:prf-thm-basic}.
\end{rmk}

\subsection{Proof of Theorem \ref{thm:basic-method}}\label{sec:prf-thm-basic}

To prove Theorem \ref{thm:basic-method} we show that for every
distortion pair $(D_1,D_2) \in (0,\sigma^2] \times (0,\sigma^2]
\setminus \mathscr{D}_1$ there exists a rate $R^{\ast} \in \Reals^+$
and a scaling coefficient $c^{\ast} \in \Reals^+$ such that $\E{ (S_1
  - \hat{S}_1)^2} = D_1$, $\E{( S_2 - \hat{S}_2 )^2} = D_2$, and
$R^{\ast} = R_{S_1,S_2}(D_1,D_2)$. The proof is then completed by
computing $I(S_1,S_2;\hat{S}_1,\hat{S}_2)$ and showing that it is
equal to the RHS of \eqref{eq:method-main}.

We begin by introducing the distortion regions $\mathscr{D}(R)$ and
$\mathscr{D}_c(R)$. The region $\mathscr{D}(R)$ is the set of all
pairs $(D_1,D_2)$ that are achievable with rate $R$ for the source
pair $(S_1,S_2)$. Similarly, the region $\mathscr{D}_c(R)$ is the set
of all pairs $(D_{1},D_{2})$ that are achievable with rate $R$ on the
scaled source pair $(S_1,cS_2)$. The following two remarks state some
properties of these two regions.

\begin{rmk}\label{rmk:prop-DR-DcR}
  The regions $\mathscr{D}(R)$ and $\mathscr{D}_c(R)$ satisfy the
  following properties:
  \begin{enumerate}
    \renewcommand{\labelenumi}{\roman{enumi})}
  \item $\mathscr{D}(R)$ and $\mathscr{D}_c(R)$ are convex.
  \item $\mathscr{D}_c(R)$ is a linearly scaled version of
    $\mathscr{D}(R)$. The scaling is in the dimension of the
    $D_2$-axis and of factor $c^2$.
  \item For a source satisfying $\Var{S_1} = \Var{S_2}$, the region
    $\mathscr{D}(R)$ is symmetric with respect to the line $D_1=D_2$.
  \end{enumerate}
\end{rmk}

\begin{proof}[Proof of Remark \ref{rmk:prop-DR-DcR}]
  Part i) follows by a time-sharing argument.  Part ii) follows by
  showing that if a distortion pair $(D_1,D_2)$ is achievable with
  rate $R$ on $(S_1,S_2)$, then also $(D_1,c^2D_2)$ is achievable with
  rate $R$ on $(S_1,cS_2)$, and vice versa. This follows since if a
  reconstruction pair $(\hat{S}_1, \hat{S}_2)$ results in distortions
  $(D_1,D_2)$ on $(S_1,S_2)$, then the scaled reconstructions
  $(\hat{S}_1, c\hat{S}_2)$ result in distortions $(D_1,c^2D_2)$ on
  $(S_1,cS_2)$.  Part iii) follows since for $(S_1, S_2)$ jointly
  Gaussian with same variances, the distribution between $S_1$ and
  $S_2$ is perfectly symmetric. Hence, if with rate $R$ the pair
  $(D_1,D_2) = (a,b)$ is achievable, then also the pair $(D_1,D_2) =
  (b,a)$ is achievable.
\end{proof}

\begin{rmk}\label{rmk:D1D1-rw}
  The scaled reconstruction pair $(\hat{S}_1,c\hat{S}_2)$, where
  $(\hat{S}_1,\hat{S}_2)$ is the result from our algorithm at rate
  $R$, yields the expected distortion pair $(D_{1},D_{2})$ of minimal
  sum $D_{1} + D_{2}$ in $\mathscr{D}_c(R)$.
\end{rmk}

\begin{proof}[Proof of Remark \ref{rmk:D1D1-rw}]
  We denote by $\Delta_1$ and $\Delta_2$ the distortion on $V_1$ and
  $V_2$, i.e.
  \begin{displaymath}
    \Delta_i = \E{(V_i - \hat{V}_i)^2} \qquad i \in \{ 1,2 \}.
  \end{displaymath}
  By definition of the reverse waterfilling solution, the
  reconstruction pair $(\hat{V}_1,\hat{V}_2)$ achieves the distortion
  pair $(\Delta_1,\Delta_2)$ on $(V_1,V_2)$ of minimal sum $\Delta_1 +
  \Delta_2$ among all rate-$R$ achievable pairs
  $(\Delta_1,\Delta_2)$. Since $(S_1,cS_2)$ relates to $(V_1,V_2)$ by
  the same unitary transformation that relates
  $(\hat{S}_1,c\hat{S}_2)$ to $(\hat{V}_1,\hat{V}_1)$, the sum of the
  distortions $D_{1} + D_{2}$ on $(S_1,cS_2)$ equals $\Delta_1 +
  \Delta_2$. Hence, if $(\Delta_1,\Delta_2)$ is the pair of minimal
  sum among all rate-$R$ achievable pairs on $(V_1,V_2)$, then $D_{1},
  D_{2}$ is the pair of minimal sum among all rate-$R$ achievable
  pairs on $(S_1,cS_2)$, i.e., in $\mathscr{D}_c(R)$.
\end{proof}

We are now ready to start with the proof of Theorem
\ref{thm:basic-method}.

\begin{proof}[Proof of Theorem \ref{thm:basic-method}]

  We first prove the Parts i) and ii) together. To this end, we begin
  by arguing that for every boundary point $(D_1,D_2)$ of
  $\mathscr{D}(R)$ that falls out of $\mathscr{D}_1$ and satisfies
  $D_1 \leq D_2$, there exists some $c \in (0,1]$ such that our
  algorithm yields $(D_1,D_2)$. To begin, consider $c = 1$. For this
  case, the distortion pair resulting from our scheme is the boundary
  point of $\mathscr{D}(R)$ that satisfies $D_1 = D_2$. This can be
  seen by noticing that for $c = 1$ the regions $\mathscr{D}(R)$ and
  $\mathscr{D}_c(R)$ coincide. Thus, by Remark \ref{rmk:D1D1-rw}, the
  distortion pair $(D_1,D_2)$ resulting on $(S_1,S_2)$ is the one of
  smallest sum $D_1 + D_2$ in $\mathscr{D}(R)$. This distortion pair
  is the point on the boundary of $\mathscr{D}(R)$ that satisfies $D_1
  = D_2$, since, by Remark \ref{rmk:prop-DR-DcR}, the region
  $\mathscr{D}(R)$ is convex and symmetric with respect to the line
  satisfying $D_1=D_2$.

  Next, consider $0 < c < 1$. The key idea of our algorithm is
  illustrated in Figure \ref{fig:idea}.
  \begin{figure}[h]
    \centering
    \psfrag{d1}[cc][cc]{$D_1$}
    \psfrag{d2}[cc][cc]{$D_2$}
    \psfrag{d1c}[cc][cc]{$D_{1}^{\ast}$}
    \psfrag{d2c}[cc][cc]{$c^2D_{2}^{\ast}$}
    \psfrag{cd2c}[cc][cc]{$D_{2}^{\ast}$}
    \psfrag{dr}[cc][cc]{\large $\mathscr{D}(R)$}
    \psfrag{dcr}[cc][cc]{\large $\mathscr{D}_c(R)$}
    \psfrag{xi}[cc][cc]{\large $\xi$}
    \psfrag{a}[cc][cc]{a)}
    \psfrag{b}[cc][cc]{b)}
    \epsfig{file=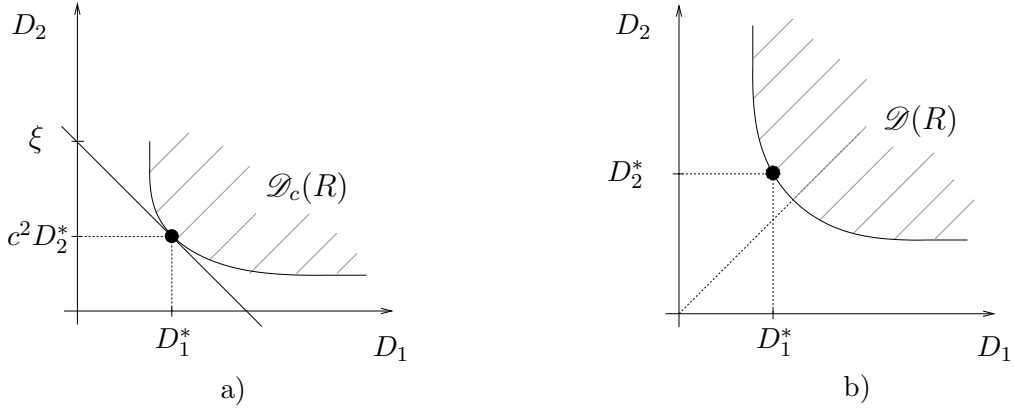, width=0.9\textwidth}
    \caption{Source scaling combined with reverse waterfilling.}
    \label{fig:idea}
  \end{figure}
  Subplot a) shows the distortion region $\mathscr{D}_c(R)$, and
  Subplot b) shows the distortion region $\mathscr{D}(R)$. The
  distortion pair $(D_1^{\ast},D_2^{\ast})$ resulting from our
  algorithm on $(S_1,S_2)$ and the corresponding distortion pair
  $(D_1^{\ast},c^2D_2^{\ast})$ on $(S_1,cS_2)$ are marked with a dot
  in Subplot b) and Subplot a) respectively. By Remark
  \ref{rmk:D1D1-rw} the distortion pair $(D_{1}^{\ast},
  c^2D_{2}^{\ast})$ is the point of smallest sum in
  $\mathscr{D}_c(R)$. Hence, graphically this point is characterized
  as the point on the boundary of $\mathscr{D}_c(R)$ for which the
  straight line of slope $-1$ containing it has the smallest
  ordinate\footnote{ This characterization follows since for any $\xi
    \in \Reals$ the set of pairs $(D_1,D_2)$ with sum $D_1+D_2 = \xi$
    is the straight line of slope $-1$ and ordinate $\xi$. Hence, if
    the ordinate $\xi$ of the straight line is minimized then the sum
    $D_1+D_2$ is minimized.}  $\xi$.  By this graphical
  characterization of $(D_{1}^{\ast}, c^2D_{2}^{\ast})$ it can now be
  seen that the smaller $c$ gets, i.e., the more $\mathscr{D}_c(R)$ is
  shrunk along the $D_2$-axis, the smaller the associated distortion
  $D_{1}^{\ast}$ gets. And as $c \rightarrow 0$, the distortion
  $D_{1}^{\ast}$ tends to the minimal value of $\sigma^2
  2^{-2R}$. Thus, by the linear relationship between
  $\mathscr{D}_c(R)$ and $\mathscr{D}(R)$, and by the convexity of
  $\mathscr{D}(R)$, it follows that for $c \in (0,1]$ our algorithm
  can achieve any $(D_1,D_2)$ on the boundary of $\mathscr{D}(R)$ for
  which $\sigma^22^{-2R} < D_1 \leq D_2$.

  For $c > 1$, it can be shown by similar arguments that our algorithm
  can achieve any $(D_1,D_2)$ on the boundary of $\mathscr{D}(R)$ for which $\sigma^2
  2^{-2R} < D_2 < D_1$, i.e., for which $D_2 < D_1 < \sigma^2 (1-\rho^2) +
  \rho^2 D_1$. Hence, for $c \in \Reals^+$ our algorithm can achieve
  any boundary point of $\mathscr{D}(R)$ in $(0,\sigma^2] \times
  (0,\sigma^2] \setminus \mathscr{D}_1$.

  To complete the proof of Parts i) and ii), we now show that with the
  appropriate $R$, the boundary of $\mathscr{D}(R)$ can indeed cover
  any point in $[0,\sigma^2] \times [0,\sigma^2] \setminus
  \mathscr{D}_1$, and that for each such boundary point of
  $\mathscr{D}(R)$ the rate-distortion function satisfies
  $R_{S_1,S_2}(D_1,D_2) = R$.

  \begin{lm}\label{lm:prp-reg-DR}
    For every distortion pair $(D_1,D_2) \in (0,\sigma^2] \times
    (0,\sigma^2] \setminus \mathscr{D}_1$, there exists some rate
    $R^{\ast} \in \Reals^+$ such that $(D_1,D_2)$ is in the boundary
    of $\mathscr{D}(R^{\ast})$. Furthermore, for each such boundary
    point
    \begin{displaymath}
      R_{S_1,S_2}(D_1,D_2) = R^{\ast}.
    \end{displaymath}
  \end{lm}

  \begin{proof}[Proof of Lemma \ref{lm:prp-reg-DR}]
    We know that each distortion pair $(D_1,D_2)$ resulting from our
    algorithm at rate $R$ lies on the boundary of $\mathscr{D}(R)$,
    and, by \eqref{eq:decorrelation} and \eqref{eq:recorrelation}, is
    given by
    \begin{equation}\label{eq:Di-Delta12}
      D_1^{\ast} =  a_1^2 \Delta_1 + a_2^2 \Delta_2, \qquad \qquad
      D_2^{\ast} = \frac{a_2^2}{c^2} \Delta_1 + \frac{a_1^2}{c^2} \Delta_2.
    \end{equation}
    Since the distortion pair $(\Delta_1, \Delta_2)$ results from
    reverse waterfilling, either both of the distortions are strictly
    and continuously decreasing in $R$, or one of them is constant
    while the other is strictly and continuously decreasing in
    $R$. Thus, since for every $c \in \Reals^+$, the coefficients
    $a_1$ and $a_2$ are both non-zero and constant, both distortions
    $D_1^{\ast}$ and $D_2^{\ast}$ are strictly and continuously
    decreasing in $R$.  Hence, for every $\epsilon > 0$, no boundary
    point $(D_1,D_2)$ of $\mathscr{D}(R)$ belongs to
    $\mathscr{D}(R-\epsilon)$, and therefore, every boundary point
    $(D_1,D_2)$ of $\mathscr{D}(R)$ satisfies
    \begin{displaymath}
      R_{S_1,S_2}(D_1,D_2) = R.
    \end{displaymath}
    The proof of of Lemma \ref{lm:prp-reg-DR} is now completed by
    recalling that for each $R \in \Reals^+$ the region
    $\mathscr{D}(R)$ is convex, and that for every fixed $c \in
    \Reals^+$ the corresponding boundary point in $\mathscr{D}(R)$ is
    evolving continuously in $R$. Thus, since for $R=0$ the distortion
    region is given by $\mathscr{D}(R) = (\sigma^2,\sigma^2)$ and for
    $R \rightarrow \infty$ the distortion region is given by
    $\mathscr{D}(R) = [0,\sigma^2] \times [0,\sigma^2]$, it follows
    that with the appropriate $R \in \Reals^+$ the boundary of
    $\mathscr{D}(R)$ can cover any point in $[0,\sigma^2] \times
    [0,\sigma^2]$.
  \end{proof}

  We now turn to the proof of Part iii). The proof consists of
  computing $I(S_1,S_2;\hat{S}_1,\hat{S}_2)$ and showing that it is
  equal to the RHS of \eqref{eq:method-main}. We use the shorthand
  notation $\tilde{\vect{V}}$ for $(V_1,V_2,\hat{V}_1,\hat{V}_2)$ and
  $\tilde{\vect{S}}$ for $(S_1,S_2,\hat{S}_1,\hat{S}_2)$. Thus,
  \begin{equation}\label{eq:S-tilde_to_V-tilde}
    \tilde{\vect{S}} = \left( \begin{array}{c c}
        \mat{B} & \mat{0}\\
        \mat{0} & \mat{B}
      \end{array} \right)
    \tilde{\vect{V}}.
  \end{equation}
  Since $\tilde{\vect{S}}$ and $\tilde{\vect{V}}$ are related by a
  linear transformation and since $\tilde{\vect{V}}$ is a zero-mean
  Gaussian vector, it follows that also $\tilde{\vect{S}}$ is a
  zero-mean Gaussian vector. Thus,
  \begin{equation}\label{eq:I-S-tilde}
    I(S_1,S_2;\hat{S}_1,\hat{S}_2) = \frac{1}{2} \log_2 \left(
      \frac{\det{\cov{S}} \det{\cov{\hat{S}}}}{\det{\cov{\tilde{S}}}}
    \right),
  \end{equation}
  where $\cov{\tilde{S}}$ is the covariance matrix of
  $\tilde{\vect{S}}$. The determinant $\det{\cov{\tilde{S}}}$ in
  \eqref{eq:I-S-tilde} is now expressed in terms of $\cov{S}$ and
  $\cov{\hat{S}}$. Since the reconstruction pair
  $(\hat{S}_1,\hat{S}_2)$ achieves a boundary point of
  $\mathscr{D}(R)$, it must satisfy the orthogonality principle and
  thus $\mat{K}_{S\hat{S}} = \cov{\hat{S}}$, where
  $\mat{K}_{S\hat{S}}$ denoted the cross-covariance matrix between
  $(S_1,S_2)$ and $(\hat{S}_1,\hat{S}_2)$. Hence,
  \begin{equation}\label{eq:cov_S-tilde}
    \cov{\tilde{S}} = \left( \begin{array}{c c}
        \cov{S} & \cov{\hat{S}}\\
        \cov{\hat{S}} & \cov{\hat{S}}
      \end{array} \right).
  \end{equation}
  Applying Schur's complement \cite[pp.~21]{horn-johnson85} to
  \eqref{eq:cov_S-tilde} gives
  \begin{equation}
    \det{\cov{\tilde{S}}} = \det{\cov{\hat{S}}} \cdot \det{ \left(
        \cov{S} - \cov{\hat{S}} \right)}.\label{eq:Schur}
  \end{equation}
  Combining \eqref{eq:Schur} with \eqref{eq:I-S-tilde}, and using that
  the law of the $(S_1,S_2,\hat{S}_1,\hat{S}_2)$ at hand is a solution to
  the optimization problem in \eqref{eq:problem}, gives
  \begin{equation}\label{eq:RD2}
    R_{S_1,S_2}(D_1,D_2) = \frac{1}{2} \log_2 \left(
      \frac{\det{\cov{S}}}{\det{\left( \cov{S} - \cov{\hat{S}}
          \right)}} \right),
  \end{equation}
  for all $(D_1,D_2) \in (0,\sigma^2] \times (0,\sigma^2] \setminus
  \mathscr{D}_1$. 
\end{proof}

\subsection{Proof of Lemma \ref{lm:KShat-new}}\label{sec:RD1D2-prf-lm-D2D3}

Lemma \ref{lm:KShat-new} expresses $\cov{\hat{S}}$ in terms of
$(D_1,D_2)$. To obtain the stated expression, we first use that
$\cov{\hat{S}} = \mat{B} \cov{\hat{V}} \trans{\mat{B}}$ and begin by
writing $\cov{\hat{V}}$ in terms of the variances of $V_1$ and $V_2$,
and in terms of the distortions on those two components. To this end,
we denote by $\lambda_1$ and $\lambda_2$ the variances of $V_1$ and
$V_2$, i.e.
\begin{displaymath}
  \lambda_i = \Var{V_i} \qquad i \in \{ 1,2 \}.
\end{displaymath}
By \eqref{eq:decorrelation}, the variances $\lambda_1$ and $\lambda_2$
are the eigenvalues of the covariance matrix of $(S_1,c S_2)$, which
are given by
\begin{align}
\lambda_1 &= \frac{\sigma^2}{2} \left( 1 + c^2 -
  \sqrt{1-2c^2(1-2\rho^2) +c^4} \right), \nonumber\\[-2mm]
& \label{eq:eigenvalues}\\[-2mm]
\lambda_2 &= \frac{\sigma^2}{2} \left( 1 + c^2 +
  \sqrt{1-2c^2(1-2\rho^2) +c^4} \right). \nonumber
\end{align}
The covariance matrix $\cov{\hat{V}}$ can now easily be expressed in
terms of $\lambda_1$, $\lambda_2$, $\Delta_1$, and $\Delta_2$. Since
for every $c \in \Reals^+$ the expressions in \eqref{eq:eigenvalues}
yield $\lambda_1 \leq \lambda_2$ , the reverse waterfilling on
$(V_1,V_2)$ satisfies
\begin{equation}\label{eq:Delta1-rw}
  \Delta_1 = \left\{ \begin{array}{l l}
      \Delta_2 & \text{if } 0 \leq \Delta_2 \leq \lambda_1\\[2mm]
      \lambda_1 & \text{if } \lambda_1 < \Delta_2 \leq \lambda_2.
    \end{array} \right.
\end{equation}
The corresponding covariance matrix $\cov{\hat{V}}$ is thus
\begin{equation}\label{eq:cov-vhat}
\cov{\hat{V}} = \left\{ \begin{array}{l l}
    \left( \begin{array}{c c}
        \lambda_1 - \Delta_2 & 0\\
        0 & \lambda_2 - \Delta_2
      \end{array} \right) & \text{if } 0 \leq \Delta_2 \leq \lambda_1\\[7mm]
    \left( \begin{array}{c c}
        0 & 0\\
        0 & \lambda_2 - \Delta_2
      \end{array} \right) & \text{if } \lambda_1 < \Delta_2 \leq \lambda_2.
  \end{array} \right.
\end{equation}
Based on \eqref{eq:cov-vhat}, we now express $\cov{\hat{S}} = \mat{B}
\cov{\hat{V}} \trans{\mat{B}}$ in terms of $D_1$ and $D_2$. To this
end, we denote the set of all distortion pairs $(D_1,D_2)$ resulting
from $0 \leq \Delta_2 < \lambda_1$ by $\mathscr{A}$ and the set of all
distortion pairs resulting from $\lambda_1 \leq \Delta_2 \leq
\lambda_2$ by $\mathscr{B}$. The expressions for $\mathscr{A}$ and
$\mathscr{B}$, and for $\cov{\hat{S}}$, in terms of $D_1$ and $D_2$,
are now given in the following two lemmas.

\begin{lm}\label{lm:KShat}
  For $(D_1,D_2) \in (0,\sigma^2] \times (0,\sigma^2] \setminus
  \mathscr{D}_1$, the covariance matrix $\cov{\hat{S}}$ in
  \eqref{eq:method-main} is given by
  \begin{equation}\label{eq:cvsh2}
    \cov{\hat{S}} = \left\{ \begin{array}{l l}
        \left( \begin{array}{c c}
            \sigma^2-D_1 & \rho \sigma^2\\[5mm]
            \rho \sigma^2 & \sigma^2 -D_2
          \end{array} \right) & \text{if } (D_1,D_2) \in \mathscr{A}\\[8mm]
        \left( \begin{array}{c c}
            \sigma^2-D_1 & \sqrt{(\sigma^2-D_1)(\sigma^2-D_2)}\\[5mm]
            \sqrt{(\sigma^2-D_1)(\sigma^2-D_2)} & \sigma^2 -D_2
          \end{array} \right) & \text{if } (D_1,D_2) \in
        \mathscr{B}.
      \end{array} \right.
  \end{equation}
\end{lm}

\begin{proof}
  See Section \ref{subsec:RD1D2-prf-KShat}.
\end{proof}

\begin{lm}\label{lm:reg-D2D3}
  The regions $\mathscr{A}$ and $\mathscr{B}$ are given by
  \begin{IEEEeqnarray*}{rCl}
    \mathscr{A} = \mathscr{D}_2 & \qquad \text{and} \qquad &
    \mathscr{B} = \mathscr{D}_3.
  \end{IEEEeqnarray*}
\end{lm}

\begin{proof}
  See Section \ref{subsec:RD1D2-prf-regD2D3}.
\end{proof}

Combining Lemma \ref{lm:KShat} and Lemma \ref{lm:reg-D2D3} yields
Lemma \ref{lm:KShat-new}. \hfill $\Box$

\begin{prp}\label{prp:opt-lin-comb}
  Every $(D_1,D_2)$ in $\mathscr{D}_1 \cup \mathscr{D}_3$ can be
  achieved with rate $R_{S_1,S_2}(D_1,D_2)$ by optimally describing a
  linear combination of the sequences ${\bf S}_1$ and ${\bf S}_2$.
\end{prp}
\begin{proof}
  We prove Proposition \ref{prp:opt-lin-comb} for the regions
  $\mathscr{D}_1$ and $\mathscr{D}_3$ separately. For $\mathscr{D}_1$
  the proof follows immediately from Section \ref{sec:prf-regD3} where
  it is shown that every $(D_1,D_2)$ in $\mathscr{D}_1$ is achieved
  with rate $R_{S_1,S_2}(D_1,D_2)$ either by optimally describing
  ${\bf S}_1$ or by optimally describing ${\bf S}_2$. For
  $\mathscr{D}_3$ the proof follows from combining \eqref{eq:cov-vhat}
  with Lemma \ref{lm:reg-D2D3}, from which it follows that every
  $(D_1,D_2)$ in $\mathscr{D}_3$ is achieved with rate
  $R_{S_1,S_2}(D_1,D_2)$ by optimally describing the sequence of
  corresponding $V_2$-components of $({\bf S}_1, {\bf S}_2)$. This
  sequence, by the definition of $V_2$, is a linear combination of
  ${\bf S}_1$ and ${\bf S}_2$.
\end{proof}

It remains to prove Lemma \ref{lm:KShat} and Lemma
\ref{lm:reg-D2D3}. We begin with Lemma \ref{lm:reg-D2D3}.

\subsubsection{Proof of Lemma \ref{lm:reg-D2D3}} \label{subsec:RD1D2-prf-regD2D3}

Lemma \ref{lm:reg-D2D3} determines the sets $\mathscr{A}$ and
$\mathscr{B}$ in terms of the distortions $D_1$ and $D_2$. To prove
this lemma we first derive the expression for $\mathscr{A}$. The
expression for $\mathscr{B}$ will then be deduced by noticing that
$\mathscr{A}$ and $\mathscr{B}$ form a partition of $(0,\sigma^2]
\times (0,\sigma^2] \setminus \mathscr{D}_1$.

The region $\mathscr{A}$ is defined as the set of all $(D_1,D_2)$
deriving from $0 \leq \Delta_2 \leq \lambda_1$. To translate this
condition to an expression in terms of $D_1$ and $D_2$, we express
$\Delta_2$ in terms of $D_1$ and $D_2$.

For $0 \leq \Delta_2 \leq \lambda_1$ however, the reverse
waterfilling solution in \eqref{eq:Delta1-rw} yields $\Delta_1 =
\Delta_2$ such that \eqref{eq:Di-Delta12} simplifies to
\begin{equation}\label{eq:D1D2-regD1}
  D_1 = \Delta_2 \qquad \qquad D_2 = \frac{1}{c^2} \Delta_2.
\end{equation}
Thus, the region $\mathscr{A}$ can be restated as the set of all pairs
$(D_1,D_2) \in (0,\sigma^2] \times (0,\sigma^2]$ satisfying $0 \leq
D_1 \leq \lambda_1$ and $D_1/D_2 = c^2$. Writing out $\lambda_1$ as in
\eqref{eq:eigenvalues} and substituting therein $c^2$ by $D_1/D_2$
yields
\begin{equation}\label{eq:D2-leq-f(D1)}
  0 \leq D_2 < (\sigma^2(1-\rho^2) - D_1) \frac{\sigma^2}{\sigma^2-D_1}.
\end{equation}
On the other hand, substituting the RHS of \eqref{eq:eigenvalues} in
$D_1 = \lambda_1$ leads to
\begin{displaymath}
  \max_{(D_1,D_2) \in \mathscr{D}_2} D_1 = \sigma^2 (1-\rho^2),
\end{displaymath}
where the maximum is obtained as $c \rightarrow \infty$. Thus, the
region $\mathscr{A}$ can finally be restated as
\begin{displaymath}
  \mathscr{A} = \bigg\{ (D_1,D_2): 0 \leq D_1 \leq \sigma^2
  (1-\rho^2), 0 \leq D_2 < (\sigma^2(1-\rho^2) - D_1)
  \frac{\sigma^2}{\sigma^2-D_1} \bigg\},
\end{displaymath}
i.e., $\mathscr{A} = \mathscr{D}_2$.

We now turn to the evaluation of $\mathscr{B}$. As can be verified
by help of \eqref{eq:Di-Delta12} the two sets $\mathscr{A}$ and
$\mathscr{B}$ form a partition of $(0,\sigma^2] \times
(0,\sigma^2] \setminus \mathscr{D}_1$. Thus, the region
$\mathscr{B}$ is given by $(0,\sigma^2] \times (0,\sigma^2]
\setminus \left( \mathscr{A} \cup \mathscr{D}_1 \right)$. Hence,
\begin{align*}
  \mathscr{B} = \Bigg\{ (D_1,D_2): \: &0 \leq D_1 \leq
  \sigma^2(1-\rho^2),\\
  &(\sigma^2(1-\rho^2) - D_1) \frac{\sigma^2}{\sigma^2-D_1} \leq D_2 <
  \sigma^2 (1-\rho^2) + \rho^2 D_1;\\[4mm]
  &\sigma^2(1-\rho^2) < D_1 \leq \sigma^2,\\
  &\frac{D_1 - \sigma^2(1-\rho^2)}{\rho^2} <  D_2 < \sigma^2
  (1-\rho^2) + \rho^2 D_1 \Bigg\},
\end{align*}
i.e., $\mathscr{B} = \mathscr{D}_3$. \hfill $\Box$

\subsubsection{Proof of Lemma
  \ref{lm:KShat}} \label{subsec:RD1D2-prf-KShat}

Lemma \ref{lm:KShat} gives an explicit expression of $\cov{\hat{S}}$
for $(D_1,D_2) \in \mathscr{A}$ and for $(D_1,D_2) \in
\mathscr{B}$. The derivations are based on the expression for
$\cov{\hat{V}}$ in \eqref{eq:cov-vhat}. Combining \eqref{eq:cov-vhat}
with $\cov{\hat{S}} = \mat{B} \cov{\hat{V}} \trans{\mat{B}}$ gives
\begin{equation} \label{eq:cov_S-hat}
  \cov{\hat{S}} = \left\{ \begin{array}{l l}
      \left( \begin{array}{c c}
          a_1^2 \lambda_1 + a_2^2 \lambda_2 - \Delta_2 &
          \frac{a_1a_2}{c}(\lambda_2-\lambda_1)\\[3mm] 
          \frac{a_1a_2}{c}(\lambda_2-\lambda_1) & \frac{1}{c^2}
          (a_2^2\lambda_1 + a_1^2\lambda_2 - \Delta_2)
        \end{array} \right) & \text{if } (D_1,D_2) \in \mathscr{A}\\[8mm]
      \left( \begin{array}{c c}
          a_2^2(\lambda_2 - \Delta_2) & \frac{a_1a_2}{c}(\lambda_2 -
          \Delta_2)\\[3mm] 
          \frac{a_1a_2}{c}(\lambda_2 - \Delta_2) & \frac{a_1^2}{c^2}
          (\lambda_2 - \Delta_2)
        \end{array} \right) & \text{if } (D_1,D_2) \in \mathscr{B}.
    \end{array} \right.
\end{equation}

For $(D_1,D_2) \in \mathscr{A}$, we now express the variables in
\eqref{eq:cov_S-hat} in terms of $D_1$ and $D_2$. From
\eqref{eq:D1D2-regD1} it follows that
\begin{equation}\label{eq:Delta2-c2-regD1}
  \Delta_2 = D_1, \qquad \text{and} \qquad c^2 = \frac{D_1}{D_2}.
\end{equation}
Furthermore, the coefficients $a_1$ and $a_2$ are determined by $a_1^2
+ a_2^2 = 1$ and
\begin{align}
  a_1^2 &= \frac{2c^2\rho^2}{1-2c^2(1-2\rho^2) + c^4 +(1-c^2)
    \sqrt{1-2c^2(1-2\rho^2)+c^4}}. \label{eq:eigenvector}
\end{align}
Combining \eqref{eq:Delta2-c2-regD1}, \eqref{eq:eigenvector} and the
expressions for $\lambda_1$ and $\lambda_2$ in
\eqref{eq:eigenvalues} with \eqref{eq:cov_S-hat}, yields
\begin{equation}\label{eq:KShat-D1}
  \cov{\hat{S}} = \left( \begin{array}{c c}
      \sigma^2-D_1 & \rho \sigma^2\\[5mm]
      \rho \sigma^2 & \sigma^2 -D_2
    \end{array} \right) \qquad \text{if } (D_1,D_2) \in \mathscr{A}.
\end{equation}

We now turn to the evaluation of $\cov{\hat{S}}$ for the region
$\mathscr{B}$. For this region, no calculations are needed. It
suffices to notice that since every optimal reconstruction pair
$(\hat{S}_1, \hat{S}_2)$ satisfies the orthogonality principle, the
main diagonal terms of the covariance matrix $\cov{\hat{S}}$ are
\begin{displaymath}
  a_2^2(\lambda_2 - \Delta_2) = \sigma^2 - D_1 \qquad \text{and}
  \qquad \frac{a_1^2}{c^2} (\lambda_2 - \Delta_2) = \sigma^2 - D_2,
\end{displaymath}
and that the anti-diagonal terms are both equal to the square-root
of the product of the two main diagonal terms and thus are
\begin{displaymath}
  \frac{a_1a_2}{c}(\lambda_2 - \Delta_2) =
  \sqrt{(\sigma^2-D_1)(\sigma^2-D_2)}.
\end{displaymath}
The covariance matrix $\cov{\hat{S}}$ for $\mathscr{B}$ is therefore
\begin{equation}\label{eq:KShat-D2}
  \cov{\hat{S}} = \left( \begin{array}{c c}
      \sigma^2-D_1 & \sqrt{(\sigma^2-D_1)(\sigma^2-D_2)}\\[5mm]
      \sqrt{(\sigma^2-D_1)(\sigma^2-D_2)} & \sigma^2 -D_2
    \end{array} \right) \qquad \text{if } (D_1,D_2) \in \mathscr{B}.
\end{equation}

\hfill $\Box$

\section{Proof of Proposition
  \ref{prp:uncoded-AWGN-ch}} \label{appx:prf-unc-AWGN}

Proposition \ref{prp:uncoded-AWGN-ch} pertains to the point-to-point
problem of Section \ref{sec:pt2pt}, in which the source pair $ \{ (
S_{1,k}, S_{2,k} ) \}$ is to be transmitted over an AWGN channel. It
states that for an achievable distortion pair $(D_1,D_2)$ for which
the SNR of the channel satisfies $P/N \leq \Gamma
(D_1,\sigma^2,\rho)$, there exist $\alpha^{\ast}, \beta^{\ast} \geq 0$
such that
\begin{IEEEeqnarray*}{rCl}
  \tilde{D}_1^{\textnormal{u}}(\alpha^{\ast},\beta^{\ast}) \leq D_1 & \qquad
  \text{and} \qquad & \tilde{D}_2^{\textnormal{u}}(\alpha^{\ast},\beta^{\ast})
  \leq D_2.
\end{IEEEeqnarray*}

The essence of Proposition \ref{prp:uncoded-AWGN-ch} is that the
uncoded scheme proposed in Section \ref{subsec:pt2pt-uncoded} achieves
every distortion pair $(D_1,D_2)$ in $\mathscr{D}_1 \cup
\mathscr{D}_3$ with the least possible transmission power, i.e., with
the $P$ for which
\begin{IEEEeqnarray*}{rCl}
  R_{S_1,S_2}(D_1,D_2) & = & \frac{1}{2} \log_2 \left( 1 + \frac{P}{N}
  \right).
\end{IEEEeqnarray*}
In Proposition \ref{prp:uncoded-AWGN-ch}, the condition $(D_1,D_2) \in
\mathscr{D}_1 \cup \mathscr{D}_3$ is merely expressed in form of the
threshold $\Gamma (D_1,\sigma^2,\rho)$ on $P/N$.

We start the proof by showing that the uncoded scheme indeed achieves
every $(D_1,D_2) \in \mathscr{D}_1 \cup \mathscr{D}_3$ with the least
possible transmission power, respectively at the smallest $P/N$. To
this end, let $\Psi(D_1,D_2)$ be the smallest $P/N$ at which
$(D_1,D_2)$ is achievable, i.e.
\begin{IEEEeqnarray*}{rCl}
  R_{S_1,S_2}(D_1,D_2) & = & \frac{1}{2} \log_2 \left( 1 +
    \Psi(D_1,D_2) \right).
\end{IEEEeqnarray*}
We now show that for every $(D_1,D_2) \in \mathscr{D}_1 \cup
\mathscr{D}_3$, there exist $\alpha^{\ast}$, $\beta^{\ast}$ such that
the distortions resulting from the uncoded scheme at $P/N =
\Psi(D_1,D_2)$ are $(\tilde{D}_1^{\textnormal{u}}(\alpha,\beta),
\tilde{D}_2^{\textnormal{u}}(\alpha,\beta)) = (D_1, D_2)$. To show
this, we rely on Proposition \ref{prp:opt-lin-comb}
(p.~\pageref{prp:opt-lin-comb}) and on the result of
\cite{goblick65}. Proposition \ref{prp:opt-lin-comb} states for the
corresponding source coding problem that if $(D_1,D_2)$ is in
$\mathscr{D}_1 \cup \mathscr{D}_3$ then $R_{S_1,S_2}(D_1,D_2)$ can be
achieved by optimally describing a linear combination of the sequences
${\bf S}_1$ and ${\bf S}_2$. The result of \cite{goblick65} states
that the minimum expected squared-error transmission of a Gaussian
source over a AWGN channel is achieved by uncoded transmission. Thus,
by combining Proposition \ref{prp:opt-lin-comb} with the result of
\cite{goblick65} and using that since $\{ (S_{1,k},S_{2,k}) \}$ are
jointly Gaussian, each of their linear combination $\alpha {\bf S}_1 +
\beta {\bf S}_2$ is also Gaussian, it follows that every distortion
pair $(D_1,D_2) \in \mathscr{D}_1 \cup \mathscr{D}_3$ is achieved at
$P/N = \Psi (D_1,D_2)$, by sending over the channel
\begin{IEEEeqnarray*}{rCl}
  X_k^{\textnormal{u}}(\alpha,\beta) & = & \sqrt{\frac{P}{\sigma^2
      (\alpha^2 + 2 \rho \alpha \beta + \beta^2)}} \left( \alpha
    S_{1,k} + \beta S_{2,k} \right) \qquad k \in \{ 1,2, \ldots ,n \},
\end{IEEEeqnarray*}
with the appropriate $\alpha,\beta \geq 0$.

It remains to derive the threshold function $\Gamma$. To this end,
first notice that for an arbitrary fixed $D_1 \in [0,\sigma^2]$, the
smaller the associated $D_2$ gets, the larger $\Psi(D_1,D_2)$ becomes,
i.e., for a fixed $D_1$ the function $\Psi(D_1,D_2)$ is decreasing in
$D_2$. Now, for every $D_1 \in [0,\sigma^2]$, let $\bar{D}_2(D_1)$ be
the smallest $D_2$ such that $(D_1,D_2) \in \mathscr{D}_1 \cup
\mathscr{D}_3$. Then, for every $D_1 \in [0,\sigma^2]$
\begin{IEEEeqnarray*}{rCl}
  \Gamma (D_1,\sigma^2,\rho) & = & \Psi(D_1,\bar{D}_2(D_1)).
\end{IEEEeqnarray*}
Hence, it remains to evaluate $\Psi(D_1,\bar{D}_2(D_1))$ for every $D_1
\in [0,\sigma^2]$. We have
\begin{IEEEeqnarray}{rCl}\label{eq:D2min-uncoded}
  \bar{D}_2(D_1) & = & \left\{ \begin{array}{l l}
      \left( \sigma^2 (1-\rho^2) - D_1 \right)
      \frac{\sigma^2}{\sigma^2 - D_1} & \text{if } 0 \leq D_1 \leq
      \sigma^2(1-\rho^2),\\[2mm]
      0 & \text{if } D_1 > \sigma^2(1-\rho^2).
    \end{array}
 \right.
\end{IEEEeqnarray}
For $D_1 > \sigma^2 (1-\rho^2)$ it immediately follows that $\Gamma
(D_1, \sigma^2, \rho) = \infty$. For $0 \leq D_1 \leq \sigma^2
(1-\rho^2)$ the value of $\Psi(D_1,\bar{D}_2(D_1))$, and hence the
value of $\Gamma(D_1,\sigma^2,\rho)$ follows from solving
\begin{IEEEeqnarray}{rCl}\label{eq:cond-D1D2-inD3}
 \frac{1}{2} \log_2^+ \left( \frac{\sigma^4 (1-\rho^2)}{D_1\bar{D}_2 - \left(
              \rho \sigma^2 - \sqrt{(\sigma^2-D_1)(\sigma^2-\bar{D}_2)} \right)^2}
        \right) & = & \frac{1}{2} \log_2 \left( 1 + \frac{P}{N}
        \right),
\end{IEEEeqnarray}
where we have used the shorthand notation $\bar{D}_2$ for
$\bar{D}_2(D_1)$. From \eqref{eq:D2min-uncoded}, we now get
\begin{IEEEeqnarray*}{rCl}
  \rho \sigma^2 - \sqrt{(\sigma^2 - D_1)(\sigma^2-\bar{D}_2)} & = & 0.
\end{IEEEeqnarray*}
Thus, \eqref{eq:cond-D1D2-inD3} reduces to
\begin{IEEEeqnarray}{rCl}
  \frac{\sigma^4(1-\rho^2)}{D_1 \bar{D}_2} & = & 1 + \frac{P}{N}.
\end{IEEEeqnarray}
which, by \eqref{eq:D2min-uncoded}, can be rewritten as
\begin{IEEEeqnarray}{rCl}\label{eq:}
  \frac{P}{N} & = & \frac{\sigma^4(1-\rho^2) - 2 \sigma^2 D_1
    (1-\rho^2) + D_1^2}{D_1 \left( \sigma^2 (1-\rho^2) - D_1 \right)}.
\end{IEEEeqnarray}
This is the threshold given in Proposition
\ref{prp:uncoded-AWGN-ch} for $0 \leq D_1 \leq
\sigma^2(1-\rho^2)$.

To conclude the proof, we justify the restriction to $\alpha, \beta
\geq 0$. This restriction is made because from the expressions for
$\tilde{D}_1^{\textnormal{u}}(\alpha,\beta)$ and
$\tilde{D}_2^{\textnormal{u}}(\alpha,\beta)$ it follows that it incurs
no loss in performance. This is so, since $\rho \geq 0$, and thus the
uncoded transmission scheme with the choice of $( \alpha, \beta )$
such that $\alpha \cdot \beta < 0$ yields a distortion that is
uniformly worse than the choice of $(| \alpha | , | \beta |)$, and
every distortion pair achievable with $\alpha, \beta < 0$, is also
achievable with $(| \alpha |, | \beta |)$. Thus, without loss in
performance, we can limit ourselves to $\alpha, \beta \geq 0$. \hfill
$\Box$

\section{Proof of Theorem
  \ref{thm:nofb-nec-cond}} \label{appx:prf-thm-nofb-nec-cond}

Theorem \ref{thm:nofb-nec-cond} applies to the multiple-access problem
without feedback. For this problem it gives a necessary condition for
the achievability of a distortion pair $(D_1,D_2)$. We begin with a
reduction.


To state the proof in more detail, we make the following reduction.
\begin{rdc}\label{rdc:zero-avg}
  There is no loss in optimality in restricting the encoding
  functions to satisfy
  \begin{IEEEeqnarray}{rCl}\label{eq:red-avg-zero-signal}
    \E{X_{i,k}} & = & 0 \qquad \text{for } i \in \{ 1,2 \}, \text{ and
    all } k \in \Integers.
  \end{IEEEeqnarray}
\end{rdc}

\begin{proof}
  We show that for every achievable tuple $(D_1,D_2, \sigma^2,
  \sigma^2, \rho, P_1, P_2, N)$, there exists a scheme with encoding
  functions satisfying \eqref{eq:red-avg-zero-signal} that achieves
  this tuple. To this end, let $(D_1,D_2, \sigma^2, \sigma^2, \rho,
  P_1, P_2, N)$ be an arbitrary achievable tuple. Further, let $\{
  f_1^{(n)}\}$, $\{ f_2^{(n)}\}$, $\{ \phi^{(n)}\}$ be sequences of
  encoding and decoding functions achieving this tuple. If the
  encoding functions $\{ f_1^{(n)}\}$, $\{ f_2^{(n)}\}$ do not satisfy
  \eqref{eq:red-avg-zero-signal}, then they can be adapted as
  follows. Before sending the codewords over the channel, the mean of
  the codewords is subtracted so as to satisfy
  \eqref{eq:red-avg-zero-signal}. And at the channel output this
  subtraction is corrected by adding this term to the received
  sequence before decoding.
\end{proof}

In view of Reduction \ref{rdc:zero-avg}, we restrict ourselves, for
the remainder of this proof to encoding functions that satisfy
\eqref{eq:red-avg-zero-signal}. The key element in the proof of
Theorem \ref{thm:nofb-nec-cond} is the following.
\begin{lm}\label{lm:ub-power}
  Any scheme satisfying condition \eqref{eq:red-avg-zero-signal} and
  the original power constraints \eqref{eq:power-constraint}, also satisfies
\begin{IEEEeqnarray}{rCl}\label{eq:ub-power}
  \frac{1}{n} \sum_{k=1}^n \E{(X_{1,k} + X_{2,k})^2} & \leq & P_1 +
  P_2 + 2 \rho \sqrt{P_1P_2}.
\end{IEEEeqnarray}
\end{lm}

\begin{proof}
  See Appendix \ref{sec:prf-ub-power}.
\end{proof}

Based on Lemma \ref{lm:ub-power}, the proof of Theorem
\ref{thm:nofb-nec-cond} is now obtained by relaxing the original
problem as follows. First, the power constraint of
\eqref{eq:power-constraint} is replaced by the power constraint of
\eqref{eq:ub-power}. Then, under the power constraint of
\eqref{eq:ub-power}, the two transmitters are allowed to fully
cooperate. These two relaxations reduce the original multiple-access
problem to a point-to-point problem where the source sequence $\{
(S_{1,k}, S_{2,k}) \}$ is to be transmitted over an AWGN channel of
power constraint $P_1 + P_2 + 2 \rho \sqrt{P_1 P_2}$ and noise
variance $N$. For this point-to-point problem, a necessary condition
for the achievability of a distortion pair $(D_1,D_2)$ follows by
source-channel separation, and is
\begin{IEEEeqnarray}{rCl}\label{eq:nec-cond}
  R_{S_1,S_2}(D_1,D_2) & \leq & \frac{1}{2} \log_2 \left( 1 +
    \frac{P_1 + P_2 + 2 \rho \sqrt{P_1 P_2}}{N} \right).
\end{IEEEeqnarray}
It is now easy to conclude that \eqref{eq:nec-cond} is also a
necessary condition for the achievability of a distortion pair
$(D_1,D_2)$ in the original multiple-access problem. This simply
follows since \eqref{eq:nec-cond} is a necessary condition for the
achievability of a distortion pair $(D_1,D_2)$ in a relaxed version of
the multiple-access problem. This concludes the proof of Theorem
\ref{thm:nofb-nec-cond}. \hfill $\Box$

\subsection{Proof of Lemma \ref{lm:ub-power}}\label{sec:prf-ub-power}

The key to Lemma \ref{lm:ub-power} is as follows:
\begin{lm}\label{lm:max-EX1X2}
  For any coding scheme with encoding functions of the form
  \eqref{eq:encoders-nofb} that satisfy the power constraints
  \eqref{eq:power-constraint} and condition
  \eqref{eq:red-avg-zero-signal} of Reduction \ref{rdc:zero-avg}, and
  where the encoder input sequences are jointly Gaussian as in
  \eqref{eq:source-law} with non-negative correlation coefficient
  $\rho$ and equal variances $\sigma_1^2 =\sigma_2^2 = \sigma^2$
  (Reduction \ref{rdc:source-normalization}), any time-$k$ encoder
  outputs $X_{1,k}$ and $X_{2,k}$ satisfy
  \begin{IEEEeqnarray}{rCl}\label{eq:max-EX1X2}
    \E{X_{1,k} X_{2,k}} & \leq & \rho \sqrt{\E{X_{1,k}^2}}
    \sqrt{\E{X_{2,k}^2}}.
  \end{IEEEeqnarray}
\end{lm}

\begin{proof}
  Lemma \ref{lm:max-EX1X2} follows from two results from Maximum
  Correlation Theory. These results are stated now.
  \begin{thm}[Witsenhausen \cite{witsenhausen75}]\label{thm:witsenhausen}
    Consider a sequence of pairs of random variables $\{
    (W_{1,k},W_{2,k}) \}$, where the pairs are independent (not
    necessarily identically distributed). Then,
    \begin{IEEEeqnarray}{rCl}\label{eq:witsenhausen}
      \sup_{g_1^{(n)},g_2^{(n)}} \E{g_1^{(n)}({\bf W}_1) g_2^{(n)}({\bf
          W}_2)} & \leq & \sup_{\substack{1 \leq k \leq n\\ g_{1,k},
          g_{2,k}}} \E{g_{1,k}(W_{1,k}) g_{2,k}(W_{2,k})},
    \end{IEEEeqnarray}
    where the supremum on the LHS of \eqref{eq:witsenhausen} is taken
    over all functions $g_i^{(n)}: \Reals^n \rightarrow \Reals$, satisfying
    \begin{IEEEeqnarray*}{rCl}
      \E{g_i^{(n)} \left( {\bf W}_i \right)} = 0 \qquad \E{\left( g_i^{(n)}
          \left( {\bf W}_i \right) \right)^2} = 1 \qquad i \in \{ 1,2 \},
    \end{IEEEeqnarray*}
    and the supremum on the RHS of \eqref{eq:witsenhausen} is taken over
    all functions $g_{i,k}: \Reals \rightarrow \Reals$, satisfying
    \begin{IEEEeqnarray*}{rCl}
      \E{g_{i,k}^{(n)} \left( W_{i,k} \right)} = 0 \qquad \E{\left(
          g_{i,k}^{(n)} \left( W_{i,k} \right) \right)^2} = 1 \qquad i
      \in \{ 1,2 \}.
    \end{IEEEeqnarray*}
  \end{thm}
  
  \begin{proof}
    See \cite[Theorem 1, p.~105]{witsenhausen75}.
  \end{proof}

  \begin{lm}\label{lm:max-corr-gauss}
    Consider two jointly Gaussian random variables $W_{1,k}$ and $W_{2,k}$ with
    correlation coefficient $\rho_k$. Then,
    \begin{IEEEeqnarray*}{rCl}
      \sup_{g_{1,k},g_{2,k}} \E{g_{1,k}(W_{1,k})g_{2,k}(W_{2,k})} = |
      \rho_k |,
    \end{IEEEeqnarray*}
    where the supremum is taken over all functions $g_{i,k}: \Reals
    \rightarrow \Reals$, satisfying
    \begin{IEEEeqnarray*}{rCl}
      \E{g_{i,k}(W_{i,k})} = 0 & \qquad & \E{(g_{i,k}(W_{i,k}))^2} = 1
      \qquad i \in \{ 1,2 \}.
    \end{IEEEeqnarray*}
  \end{lm}
  
  \begin{proof}
    See \cite[Lemma 10.2, p.~182]{rozanov67}.
  \end{proof}

  Lemma \ref{lm:max-EX1X2} is now merely a consequence of Theorem
  \ref{thm:witsenhausen} and Lemma \ref{lm:max-corr-gauss} applied to
  our setup. To see this, substitute ${\bf W}_1$ and ${\bf W}_2$ by the
  source sequences ${\bf S}_1$ and ${\bf S}_2$, and let the functions
  $g_1^{(n)} (\cdot)$ and $g_2^{(n)} (\cdot)$ be the encoding
  sub-functions that produce the time-$k$ channel inputs $X_{1,k}$ and
  $X_{2,k}$, i.e., $g_i^{(n)} ({\bf S}_i) = X_{i,k}$.  Then, for every $k
  \in \{ 1,2, \ldots ,n\}$,
  \begin{IEEEeqnarray}{rCl}
    \frac{\E{X_{1,k} X_{2,k}}}{\sqrt{\E{X_{1,k}^2}}
      \sqrt{\E{X_{2,k}^2}}} & \leq & \sup_{g_1^{(n)},g_2^{(n)}}
    \E{g_1^{(n)}({\bf S}_1) g_2^{(n)}({\bf S}_2)} \nonumber\\
    & \stackrel{a)}{\leq} & \sup_{\substack{1 \leq k \leq n\\ g_{1,k}, g_{2,k}}}
    \E{g_{1,k}(S_{1,k}) g_{2,k}(S_{2,k})} \nonumber\\
    & \stackrel{b)}{\leq} & \rho,
  \end{IEEEeqnarray}
  where $a)$ follows from Theorem \ref{thm:witsenhausen} and $b)$
  follows from Lemma \ref{lm:max-corr-gauss} and from our assumption
  that $\rho \geq 0$ (Reduction \ref{rdc:source-normalization}). Thus,
  for every time $k$,
  \begin{IEEEeqnarray}{rCl}\label{eq:bdEX1X2}
    \E{X_{1,k} X_{2,k}} & \leq & \rho \sqrt{\E{X_{1,k}^2}}
    \sqrt{\E{X_{2,k}^2}},
  \end{IEEEeqnarray}
  which is the bound of Lemma \ref{lm:max-EX1X2}. \hfill $\Box$

  Using Lemma \ref{lm:max-EX1X2} we can now prove the bound of Lemma
  \ref{lm:ub-power} as follows:
  \begin{IEEEeqnarray}{rCl}\label{eq:bd-sum-power}
    \frac{1}{n} \sum_{k=1}^n \E{(X_{1,k} + X_{2,k})^2} & = & \frac{1}{n}
    \sum_{k=1}^n \E{X_{1,k}^2} + \frac{1}{n} \sum_{k=1}^n \E{X_{1,k}^2}
    + 2 \frac{1}{n} \sum_{k=1}^n \E{X_{1,k}X_{2,k}}\nonumber\\
    & \leq & P_1 + P_2 + 2 \frac{1}{n} \sum_{k=1}^n \E{X_{1,k}X_{2,k}}\nonumber\\
    & \stackrel{a)}{\leq} & P_1 + P_2 + 2 \rho \frac{1}{n} \sum_{k=1}^n
    \sqrt{\E{X_{1,k}^2}} \sqrt{\E{X_{2,k}^2}}\nonumber\\
    & \stackrel{b)}{\leq} & P_1 + P_2 + 2 \rho \frac{1}{n} \sqrt{\sum_{k=1}^n
      \E{X_{1,k}^2}} \sqrt{\sum_{k=1}^n \E{X_{2,k}^2}}\nonumber\\
    & \leq & P_1 + P_2 + 2 \rho \frac{1}{n} \sqrt{nP_1} \sqrt{nP_2}\nonumber\\
    & = & P_1 + P_2 + 2 \rho \sqrt{P_1P_2},
  \end{IEEEeqnarray}
  where Inequality $a)$ follows by Lemma \ref{lm:max-EX1X2} and from our
  assumption that $\rho \geq 0$, and where Inequality $b)$ follows by
  Cauchy-Schwarz. This concludes the proof of Lemma
  \ref{lm:ub-power}.
\end{proof}

\section{Distortions $(D_1^{\textnormal{u}}, D_2^{\textnormal{u}})$ of
the Uncoded Scheme} \label{appx:prf-thm-uncoded-main}

The expression for $D_i^{\textnormal{u}}$, $i \in \{ 1,2 \}$, is
obtained as follows
\begin{IEEEeqnarray*}{rCl}
  D_i^{\textnormal{u}} & = & \frac{1}{n} \sum_{k=1}^n \E{(S_{i,k} -
    \hat{S}_{i,k}^{\textnormal{u}})^2}\\[2mm] 
  & = & \frac{1}{n} \left( \E{S_{i,k}^2} - 2
    \E{S_{i,k}\hat{S}_{i,k}^{\textnormal{u}}} + \E{ \left(
        \hat{S}_{i,k}^{\textnormal{u}} \right)^2} \right)\\[2mm]
  & \stackrel{a)}{=} & \frac{1}{n} \left( \E{S_{i,k}^2} -
    \E{ \left( \hat{S}_{i,k}^{\textnormal{u}} \right)^2} \right)\\[2mm]
  & \stackrel{b)}{=} & \frac{1}{n} \left( \E{S_{i,k}^2} -
    \frac{(\E{S_{i,k} Y_k})^2}{\E{Y_k^2}} \right)\\[2mm]
  & \stackrel{c)}{=} & \sigma^2 - \sigma^2 \frac{P_1 + 2 \rho
    \sqrt{P_1P_2} + \rho^2 P_2}{P_1 + 2 \rho \sqrt{P_1 P_2} + P_2 + N}\\[2mm]
  & = & \sigma^2 \frac{P_1 (1-\rho^2) + N}{P_1 + P_2 + 2 \rho
    \sqrt{P_1 P_2} + N},
\end{IEEEeqnarray*}
where in $a)$ we have used that $\hat{S}_{i,k}^{\textnormal{u}} = \E{
  S_{i,k} | Y_k }$ satisfies the Orthogonality Principle; in $b)$ we
have used the explicit form of the conditional mean for jointly
Gaussians
\begin{IEEEeqnarray*}{rCl}
  \hat{S}_{i,k}^{\textnormal{u}} & = & \E{ S_{i,k} | Y_k } =
  \frac{\E{S_{1,k} Y_k}}{\E{Y_k^2}} \cdot Y_k;
\end{IEEEeqnarray*}
and in $c)$ we have used the calculation
\begin{IEEEeqnarray*}{rCl}
  \hspace{35mm} \frac{(\E{S_{1,k} Y_k})^2}{\E{Y_k^2}} & = & \frac{ \sigma^2 \left(
      \sqrt{P_1} + \rho \sqrt{P_2} \right)^2 }{P_1 + P_2 + 2 \rho
    \sqrt{P_1P_2} + N}. \hspace{34mm} \qed
\end{IEEEeqnarray*}


\section{Proof of Theorem \ref{thm:vq-achv}} \label{appx:prf-thm-vq}

In this appendix we analyze the distortions achievable by the
vector-quantizer scheme that was presented in Section
\ref{sec:mac-vq}. To start, we give a thorough description of the
corresponding coding scheme.

\subsection{Coding Scheme}

Fix some $\epsilon > 0$ and rates $R_1$ and $R_2$.\\[3mm]
{\bf Code Construction:} Two codebooks $\mathcal{C}_1$ and
$\mathcal{C}_2$ are generated independently. Codebook $\mathcal{C}_i$,
$i \in \{1,2\}$, consists of $2^{nR_i}$ codewords $\{ {\bf U}_i(1),
{\bf U}_i(2), \ldots ,{\bf U}_i(2^{nR_i}) \}$. The codewords are drawn
independently uniformly over the surface of the centered
$\Reals^n$-sphere $\mathcal{S}_i$ of radius $r_i =
\sqrt{n\sigma^2(1-2^{-2R_i})}$.\\[3mm]
{\bf Encoding:} Based on the observed source sequence ${\bf S}_i$ each
encoder produces its channel input ${\bf X}_i$ by first
vector-quantizing the source sequence ${\bf S}_i$ to a codeword ${\bf
  U}_i^{\ast} \in \mathcal{C}_i$ and then scaling ${\bf U}_i^{\ast}$
to satisfy the average power constraint. To describe the
vector-quantizer precisely, denote for every ${\bf w}, {\bf v} \in
\Reals^n$ where neither ${\bf w}$ nor ${\bf v}$ are the zero-sequence,
the angle between ${\bf w}$ and ${\bf v}$ by $\sphericalangle ({\bf
  w},{\bf v})$, i.e.
\begin{IEEEeqnarray}{rCl}
  \cos \sphericalangle ({\bf w},{\bf v}) & \triangleq & 
    \frac{\inner{{\bf w}}{{\bf v}}}{\|{\bf w}\| \|{\bf v}\|}.
\end{IEEEeqnarray}
Let
\begin{IEEEeqnarray}{rCl}\label{eq:encoding}
  \mathcal{F}({\bf s}_i,\mathcal{C}_i) & = & \left\{ {\bf u}_i \in
    \mathcal{C}_i: \sqrt{1-2^{-2R_i}}(1-\epsilon) \leq \cos \sphericalangle
    ({\bf s}_i, {\bf U}_i) \leq \sqrt{1-2^{-2R_i}}(1+\epsilon) \right\}. \hspace{6mm}
\end{IEEEeqnarray}
If $\mathcal{F}({\bf s}_i,\mathcal{C}_i) \neq \emptyset$, the
vector-quantizer output ${\bf U}_i^{\ast}$ is the codeword ${\bf
  U}_i(j) \in \mathcal{F}({\bf s}_i,\mathcal{C}_i)$ that minimizes
$|\cos \sphericalangle ({\bf u}_i(j),{\bf s}_i) -
\sqrt{1-2^{-2R_i}}|$, and if $\mathcal{F}({\bf s}_i,\mathcal{C}_i) =
\emptyset$ the vector-quantizer output ${\bf U}_i^{\ast}$ is the
all-zero sequence. Thus,
\begin{IEEEeqnarray*}{rCl}
  {\bf U}_i^{\ast} & = & \left\{ \begin{array}{l l}
      \argmin_{\substack{{\bf U}_i \in \mathcal{C}_i:\\[1mm]
          {\bf U}_i \in \mathcal{F}({\bf s}_i,\mathcal{C}_i)}} \Big|\cos
      \sphericalangle({\bf u}_i(j),{\bf S}_i) - \sqrt{1-2^{-2R_i}}\Big| &
      \text{if } \mathcal{F}({\bf s}_i,\mathcal{C}_i) \neq \emptyset,\\[5mm]
      {\bf 0} & \text{otherwise}.
    \end{array} \right.
\end{IEEEeqnarray*}
More formally, ${\bf U}_i^{\ast}$ should be written as
${\bf U}_i^{\ast}({\bf S}_i, \mathcal{C}_i)$, but we shall usually
make these dependencies implicit. The channel input is now given by
\begin{IEEEeqnarray}{rCl}
  {\bf X}_i & = & \alpha_i {\bf U}_i^{\ast} \qquad \qquad \text{where
  }\alpha_i = \sqrt{\frac{P_i}{\sigma^2 (1-2^{-2R_i})}} \qquad i \in
  \{ 1,2 \}.
\end{IEEEeqnarray}
Since the codebook $\mathcal{C}_i$ is drawn over the centered
$\Reals^n$-sphere of radius $r_i = \sqrt{\sigma^2(1-2^{-2R_i})}$, each
channel input ${\bf X}_i$ satisfies the average power constraint
individually.\\[3mm]
{\bf Reconstruction:} The receiver's estimate $(\hat{\bf S}_1,
\hat{\bf S}_2)$ of the source pair $({\bf S}_1, {\bf S}_2)$ is derived
from the channel output ${\bf Y}$ in two steps. First, the receiver
makes a guess $(\hat{\bf U}_1, \hat{\bf U}_2)$ of the pair $({\bf
  U}_1^{\ast}, {\bf U}_2^{\ast})$ by choosing among all ``jointly
typical pairs'' $({\bf U}_1, {\bf U}_2) \in \mathcal{C}_1 \times
\mathcal{C}_2$ the pair whose linear combination $\alpha_1 {\bf U}_1 +
\alpha_2 {\bf U}_2$ has the smallest distance to the received sequence
${\bf Y}$. More precisely,
\begin{IEEEeqnarray}{rCl}\label{eq:decoding}
  (\hat{\bf U}_1, \hat{\bf U}_2) & = & \argmin_{\substack{({\bf
        U}_1, {\bf U}_2) \in \mathcal{C}_1 \times \mathcal{C}_2:\\
      \left| \tilde{\rho} - \cos \sphericalangle ({\bf u}_1, {\bf
          u}_2) \right| \leq 7 \epsilon}} \| {\bf Y} - (\alpha_1{\bf
    U}_1 + \alpha_2{\bf U}_2) \|^2,
\end{IEEEeqnarray}
where
\begin{IEEEeqnarray*}{rCl}
  \tilde{\rho} & = & \rho \sqrt{(1-2^{-2R_1})(1-2^{-2R_2})}.
\end{IEEEeqnarray*}
If the channel output ${\bf Y}$ and the codebooks $\mathcal{C}_1$ and
$\mathcal{C}_2$ are such that there does not exist a pair $({\bf
  U}_1, {\bf U}_2) \in \mathcal{C}_1 \times \mathcal{C}_2$ that
satisfies
\begin{IEEEeqnarray}{rCl}\label{eq:vq-jnt-typ-U1U2}
  \left| \tilde{\rho} - \cos \sphericalangle ({\bf u}_1, {\bf u}_2)
  \right| & \leq & 7 \epsilon,
\end{IEEEeqnarray}
then $\hat{\bf U}_1$ and $\hat{\bf U}_2$ are chosen to be all-zero.

In the second step, the receiver computes the estimates $(\hat{\bf
  S}_1, \hat{\bf S}_2)$ from the guess $(\hat{\bf U}_1, \hat{\bf
  U}_2)$ by setting
\begin{IEEEeqnarray}{rCl}
  \hat{\bf S}_1 & = & \gamma_{11} \hat{\bf U}_1 + \gamma_{12} \hat{\bf
    U}_2 \label{eq:lin-est-S1h}\\
  \hat{\bf S}_2 & = & \gamma_{21} \hat{\bf U}_2 + \gamma_{22} \hat{\bf U}_1,
\end{IEEEeqnarray}
where
\begin{IEEEeqnarray}{rCl}
  \gamma_{11} = \frac{1-\rho^2(1-2^{-2R_2})}{1-\tilde{\rho}^2} & \qquad
  \qquad & \gamma_{12} = \rho 2^{-2R_1} \label{eq:vq-gamma1}\\[3mm]
  \gamma_{21} = \frac{1-\rho^2(1-2^{-2R_1})}{1-\tilde{\rho}^2} & &
  \gamma_{22} = \rho 2^{-2R_2}. \label{eq:vq-gamma2}
\end{IEEEeqnarray}
Note that
\begin{IEEEeqnarray}{rCl}\label{eq:vq-bounds-gamma}
  0 < \gamma_{i1} \leq 1 & \quad \text{and} \quad & 0 < \gamma_{i2}
  \leq \rho, \qquad i \in \{ 1,2 \}.
\end{IEEEeqnarray}


\subsection{Expected Distortion}

To analyze the expected distortion we use a genie-aided argument. We
first show that, under certain rate constraints, the asymptotic
normalized distortion of the proposed scheme remains the same when a
certain help from a genie is provided. To derive the achievable
distortions it then suffices to analyze the genie-aided version.

\subsubsection{Genie-Aided Scheme}

In the genie-aided scheme, the genie's help is provided to the
decoder.  An illustration of this genie-aided decoder is given in
Figure \ref{fig:vq-dec-genie}.
\begin{figure}[h]
  \centering
  \psfrag{y}[cc][cc]{${\bf Y}$}
  \psfrag{u1}[cc][cc]{$\hat{\bf U}_1$}
  \psfrag{u2}[cc][cc]{$\hat{\bf U}_2$}
  \psfrag{u1g}[cc][cc]{${\bf U}_1^{\ast}$}
  \psfrag{u2g}[cc][cc]{${\bf U}_2^{\ast}$}
  \psfrag{s1}[cc][cc]{$\hat{\bf S}_1^{\textnormal{G}}$}
  \psfrag{s2}[cc][cc]{$\hat{\bf S}_2^{\textnormal{G}}$}
  \epsfig{file=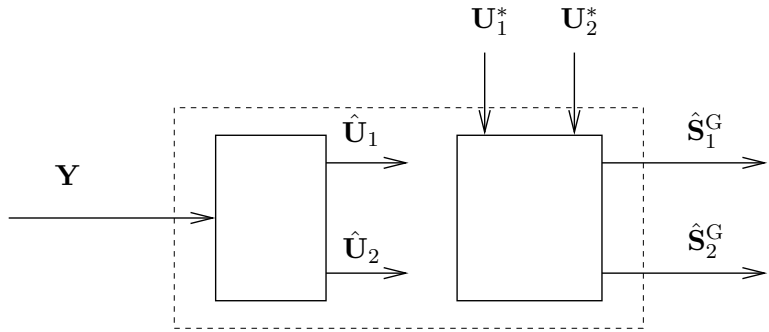, width=0.7\textwidth}
  \caption{Genie-aided decoder.}
  \label{fig:vq-dec-genie}
\end{figure}
The genie provides the decoder with the codeword pair $({\bf
  U}_1^{\ast}, {\bf U}_2^{\ast})$. The decoder then estimates the
source pair $({\bf S}_1, {\bf S}_2)$ based on $({\bf U}_1^{\ast}, {\bf
  U}_2^{\ast})$ and ignores the guess $(\hat{\bf U}_1, \hat{\bf U}_2)$
produced in the first decoding step.  The estimate of this genie-aided
decoder is denoted by $(\hat{\bf S}_1^{\textnormal{G}}, \hat{\bf
  S}_2^{\textnormal{G}})$, where
\begin{IEEEeqnarray}{rCl}
  \hat{\bf S}_1^{\textnormal{G}} & = & \gamma_{11} {\bf U}_1^{\ast} +
  \gamma_{12} {\bf U}_2^{\ast} \label{eq:vq-genie-lin-est-S1h}\\
  \hat{\bf S}_2^{\textnormal{G}} & = & \gamma_{21} {\bf U}_2^{\ast} +
  \gamma_{22} {\bf U}_1^{\ast},
\end{IEEEeqnarray}
with $\gamma_{11}$, $\gamma_{12}$, $\gamma_{21}$, $\gamma_{22}$ as in
\eqref{eq:vq-gamma1} and \eqref{eq:vq-gamma2}. Under certain rate
constraints, the normalized asymptotic distortion of this genie-aided
scheme is the same as for the originally proposed scheme. This is
stated more precisely in the following proposition.

\begin{prp}\label{prp:vq-D1-eql-genie}
  For every $\delta > 0$ and $0 < \epsilon < 0.3$ there exists an
  $n'(\delta,\epsilon)>0$ such that for all $n > n'(\delta,\epsilon)$,
  \begin{IEEEeqnarray*}{rCl}
    \frac{1}{n}\E{\| {\bf S}_1 - \hat{\bf S}_1 \|^2} & \leq &
    \frac{1}{n}\E{\| {\bf S}_1 - \hat{\bf S}_1^{\textnormal{G}} \|^2}
    + 2 \sigma^2 \left( \epsilon + \left( 44 \sqrt{1+\epsilon} + 61
      \right) \delta \right).
  \end{IEEEeqnarray*}
  whenever $(R_1,R_2)$ is in the rate region $\mathcal{R}(\epsilon)$
  given by
  \begin{IEEEeqnarray*}{rClrCl}
    \mathcal{R}(\epsilon) & = & \bigg\{ & R_1 & \leq & \frac{1}{2} \log_2
    \left( \frac{P_1 (1-\tilde{\rho}^2) + N}{N(1-\tilde{\rho}^2)} -
      \kappa_1 \epsilon \right),\\[2mm]
    & & & R_2 & \leq & \frac{1}{2} \log_2
    \left( \frac{P_2 (1-\tilde{\rho}^2) + N}{N(1-\tilde{\rho}^2)}
      - \kappa_2 \epsilon \right),\\[2mm]
    & & & R_1 + R_2 & \leq & \frac{1}{2} \log_2 \left( \frac{P_1 + P_2 +
        2\tilde{\rho} \sqrt{P_1 P_2} + N }{N (1-\tilde{\rho}^2)} - \kappa_3
      \epsilon \right) \bigg\},
  \end{IEEEeqnarray*}
  where $\kappa_1$, $\kappa_2$ and $\kappa_3$ depend only on $P_1$,
  $P_2$, $N$, $\zeta_1$ and $\zeta_2$, where 
  \begin{IEEEeqnarray*}{rCl}
    \zeta_1 = \frac{N\tilde{\rho}}{P_1(1-\tilde{\rho}^2) + N}
    \sqrt{\frac{P_1}{P_2}} & \qquad \text{and} \qquad &
    \zeta_2 = \frac{P_1(1-\tilde{\rho}^2)}{P_1(1-\tilde{\rho}^2) + N}.
  \end{IEEEeqnarray*}
\end{prp}

\begin{proof}
  See Section \ref{sec:vq-prf-prp-genie}.
\end{proof}

\begin{cor}\label{cor:vq-genie}
  If $(R_1,R_2)$ satisfy
  \begin{IEEEeqnarray*}{rCl}
    R_1 & < & \frac{1}{2} \log_2 \left( \frac{P_1 (1-\tilde{\rho}^2) +
        N}{N(1-\tilde{\rho}^2)} \right)\\[2mm]
    R_2 & < & \frac{1}{2} \log_2 \left( \frac{P_2 (1-\tilde{\rho}^2) +
        N}{N(1-\tilde{\rho}^2)} \right)\\[2mm]
    R_1 + R_2 & < & \frac{1}{2} \log_2 \left( \frac{P_1 + P_2 +
        2\tilde{\rho} \sqrt{P_1 P_2} + N }{N (1-\tilde{\rho}^2)} \right),
  \end{IEEEeqnarray*}
  then
  \begin{IEEEeqnarray*}{rCl}
    \varlimsup_{n \rightarrow \infty} \frac{1}{n}\E{\| {\bf S}_1 -
      \hat{\bf S}_1 \|^2} & \leq & \varlimsup_{n \rightarrow
      \infty}\frac{1}{n}\E{\| {\bf S}_1 - \hat{\bf S}_1^{\textnormal{G}}
      \|^2}.
  \end{IEEEeqnarray*}
\end{cor}

\begin{proof}
  Follows from Proposition \ref{prp:vq-D1-eql-genie} by first letting
  $n \rightarrow \infty$ and then $\epsilon \rightarrow 0$ and $\delta
  \rightarrow 0$.
\end{proof}

By Corollary \ref{cor:vq-genie}, to analyze the distortion achievable
by our scheme it suffices to analyze the genie-aided scheme. This is
done in Section \ref{sec:vq-ub-D1}.

\subsection{Proof of Proposition
  \ref{prp:vq-D1-eql-genie}}\label{sec:vq-prf-prp-genie}

The main step in the proof of Proposition \ref{prp:vq-D1-eql-genie} is
to show that for every $(R_1,R_2) \in \mathcal{R}(\epsilon)$ and
sufficiently large $n$, the probability of a decoding error, and thus
the probability of $\hat{\bf S}_1 \neq \hat{\bf
  S}_1^{\textnormal{G}}$, can be made very small. This step is done in
the following section. The proof of Proposition
\ref{prp:vq-D1-eql-genie} is then completed in Section
\ref{sec:prf-prp-genie}.


\subsubsection{Upper Bound on Probability of a Decoding Error} \label{sec:vq-ub-PrE4}

In this section we show that for every $(R_1,R_2) \in
\mathcal{R}(\epsilon)$ and sufficiently large $n$, the probability of
a decoding error, and thus the probability of $\hat{\bf S}_1 \neq
\hat{\bf S}_1^{\textnormal{G}}$, can be made very small. The hitch is
that to upper bound the probability of a decoding error for the
proposed scheme, we cannot proceed by the method conventionally used
for the multiple-access channel. The reason is that in the conventional
analysis of the multiple-access channel it is assumed that the
probability of the codewords ${\bf u}_i(j)$ does not depend on the
realization of the codebook $\mathcal{C}_i$. However, in our combined
source-channel coding scheme, the probability of encoder $i \in \{ 1,2
\}$ producing the channel input of index $j \in \{ 1,2, \ldots
,2^{nR_i}\}$ depends not only on the source sequence ${\bf s}_i$, but
also on the realization of $\mathcal{C}_i$. Another reason the
conventional analysis fails is that, conditional on the
codebooks $\mathcal{C}_1$ and $\mathcal{C}_2$, the indices produces by
the vector-quantizers are dependent.

To address these difficulties, we proceed by a geometric approach. To
this end, we introduce an error event related to a decoding error at
the receiver. This event is denoted by $\mathcal{E}_{\hat{\bf U}}$ and consists of
all tuples $({\bf s}_1, {\bf s}_2, \mathcal{C}_1, \mathcal{C}_2, {\bf
  z})$ for which there exists a pair $(\tilde{\bf u}_1, \tilde{\bf
  u}_2) \neq ({\bf u}_1^{\ast}, {\bf u}_2^{\ast})$ in $\mathcal{C}_1
\times \mathcal{C}_2$ that satisfies Condition
\eqref{eq:vq-jnt-typ-U1U2} of the reconstructor, and for which the
Euclidean distance between $\alpha_1 \tilde{\bf u}_1 + \alpha_2
\tilde{\bf u}_2$ and ${\bf y}$ is smaller or equal to the Euclidean
distance between $\alpha_1 {\bf u}_1^{\ast} + \alpha_2 {\bf
  u}_2^{\ast}$ and ${\bf y}$. More formally, $\mathcal{E}_{\hat{\bf U}} =
\mathcal{E}_{\hat{\bf U}_1} \cup \mathcal{E}_{\hat{\bf U}_2} \cup
\mathcal{E}_{(\hat{\bf U}_1,\hat{\bf U}_2)}$ where
\begin{IEEEeqnarray}{l}
  \mathcal{E}_{\hat{\bf U}_1} = \bigg\{ ({\bf s}_1, {\bf s}_2, \mathcal{C}_1,
  \mathcal{C}_2, {\bf z}) : \exists \tilde{\bf u}_1 \in \mathcal{C}_1
  \setminus \{ {\bf u}_1^{\ast} \} \text{ s.t. } \left| \tilde{\rho} -
    \cos \sphericalangle (\tilde{\bf u}_1, {\bf u}_2^{\ast}) \right|
  \leq 7 \epsilon, \nonumber\\
  \hspace{35mm} \text{and} \quad \| {\bf y} - (\alpha_1 \tilde{\bf
    u}_1 + \alpha_2 {\bf u}_2^{\ast}) \|^2 \leq \| {\bf y} - (\alpha_1
  {\bf u}_1^{\ast} + \alpha_2 {\bf u}_2^{\ast})
  \|^2 \bigg\} \quad \label{eq:E41}\\[3mm]
  \mathcal{E}_{\hat{\bf U}_2} = \bigg\{ ({\bf s}_1, {\bf s}_2, \mathcal{C}_1,
  \mathcal{C}_2, {\bf z}) : \exists \tilde{\bf u}_2 \in \mathcal{C}_2
  \setminus \{ {\bf u}_2^{\ast} \} \text{ s.t. } \left| \tilde{\rho} -
    \cos \sphericalangle ({\bf u}_1^{\ast},\tilde{\bf u}_2) \right|
  \leq 7 \epsilon, \nonumber\\
  \hspace{35mm} \text{and} \quad \| {\bf y} - (\alpha_1 {\bf
    u}_1^{\ast} + \alpha_2 \tilde{\bf u}_2) \|^2 \leq \| {\bf y} -
  (\alpha_1 {\bf u}_1^{\ast} + \alpha_2 {\bf u}_2^{\ast})
  \|^2 \bigg\} \label{eq:E42}\\[3mm]
  \mathcal{E}_{(\hat{\bf U}_1,\hat{\bf U}_2)} = \bigg\{ ({\bf s}_1, {\bf s}_2, \mathcal{C}_1,
  \mathcal{C}_2, {\bf z}) : \exists \tilde{\bf u}_1 \in \mathcal{C}_1
  \setminus \{ {\bf u}_1^{\ast} \} \; \; \text{and} \; \; \exists
  \tilde{\bf u}_2 \in \mathcal{C}_2 \setminus \{ {\bf
    u}_2^{\ast} \} \text{ s.t.} \nonumber\\
  \hspace{55mm} \left| \tilde{\rho} - \cos \sphericalangle
    (\tilde{\bf u}_1, \tilde{\bf u}_2 ) \right| \leq 7 \epsilon, \nonumber\\
  \hspace{35mm} \text{and} \quad\| {\bf y} - (\alpha_1 \tilde{\bf
    u}_1 + \alpha_2 \tilde{\bf u}_2) \|^2 \leq \| {\bf y} - (\alpha_1
  {\bf u}_1^{\ast} + \alpha_2 {\bf u}_2^{\ast}) \|^2 \bigg\}, \label{eq:E43}
\end{IEEEeqnarray}
and where ${\bf y} \triangleq \alpha_1 {\bf u}_1^{\ast} + \alpha_1 {\bf
  u}_1^{\ast} + {\bf z}$. Note that a decoding error occurs only if
$({\bf s}_1, {\bf s}_2, \mathcal{C}_1, \mathcal{C}_2, {\bf z}) \in
\mathcal{E}_{\hat{\bf U}}$. The main result of this section can now be stated as
follows.

\begin{lm}\label{lm:vq-Pr-E4}
  For every $\delta > 0$ and $0.3 > \epsilon > 0$, there exists an
  $n_4'(\delta,\epsilon) \in \Naturals$ such that for all $n >
  n_4'(\delta,\epsilon)$
  \begin{IEEEeqnarray*}{rCl}
    \qquad \qquad \Prv{\mathcal{E}_{\hat{\bf U}}} & < & 11 \delta, \qquad \qquad \text{whenever }
    (R_1,R_2) \in \mathcal{R}(\epsilon).
  \end{IEEEeqnarray*}
\end{lm}

To prove Lemma \ref{lm:vq-Pr-E4}, we introduce three auxiliary error
events. The first auxiliary event $\mathcal{E}_{\bf S}$ corresponds to an
atypical source output. More precisely,
  \begin{IEEEeqnarray}{rCl}
    \mathcal{E}_{\bf S} & = \bigg\{ ({\bf s}_1,{\bf s}_2) \in \Reals^n
    \times \Reals^n : \bigg| \frac{1}{n} \| {\bf s}_1\|^2 - \sigma^2
    \bigg| > \epsilon \sigma^2 \quad & \text{or} \quad \bigg|
    \frac{1}{n} \| {\bf s}_2\|^2 -
    \sigma^2 \bigg| > \epsilon \sigma^2 \nonumber\\
    & & \text{or} \quad \left| \cos \sphericalangle ({\bf s}_1, {\bf
        s}_2) - \rho \right| > \epsilon \rho
    \bigg\}.\hspace{5mm}\label{eq:vq-E1-def}
\end{IEEEeqnarray}
The second auxiliary event is denoted by $\mathcal{E}_{\bf Z}$ and corresponds to
an atypical behavior of the additive noise:
\begin{IEEEeqnarray}{rCl}
  \mathcal{E}_{\bf Z} & = \bigg\{ ({\bf s}_1, {\bf s}_2, \mathcal{C}_1,
  \mathcal{C}_2, {\bf z}) : \bigg| \frac{1}{n} \| {\bf z} \|^2 -N
  \bigg| > \epsilon N \; \; &\text{ or } \; \; \frac{1}{n} |
  \inner{\alpha_1{\bf u}_1^{\ast}({\bf s}_1, \mathcal{C}_1)}{{\bf z}}
  | > \sqrt{P_1 N} \epsilon \nonumber \hspace{8mm}\\
  & & \text{ or } \; \; \frac{1}{n} | \inner{\alpha_2 {\bf
      u}_2^{\ast}({\bf s}_2, \mathcal{C}_2)}{{\bf z}} | > \sqrt{P_2 N}
  \epsilon\bigg\}. \hspace{8mm} \label{eq:vq-E3-def}
\end{IEEEeqnarray}
Finally, the third auxiliary event is denoted by $\mathcal{E}_{\bf X}$
and corresponds to irregularities at the encoders. That is, the event
that one of the codebooks contains no codeword satisfying Condition
\eqref{eq:encoding} of the vector-quantizer, or that the two quantized
sequences ${\bf u}_1^{\ast}$ and ${\bf u}_2^{\ast}$ have an
``atypical'' angle to each other. More formally, $\mathcal{E}_{\bf
  X} = \mathcal{E}_{{\bf X}_1} \cup \mathcal{E}_{{\bf X}_2} \cup
\mathcal{E}_{({\bf X}_1,{\bf X}_2)}$ where
\begin{IEEEeqnarray*}{l}
  \mathcal{E}_{{\bf X}_1} = \bigg\{ ( {\bf s}_1, {\bf s}_2,
  \mathcal{C}_1, \mathcal{C}_2 ) : \nexists {\bf u}_1 \in
  \mathcal{C}_1 \text{ s.t.~ } \left| \sqrt{1-2^{-2R_1}} -
  \cos \sphericalangle ({\bf s}_1, {\bf u}_1) \right| \leq
  \epsilon \sqrt{1-2^{-2R_1}} \bigg\}\\[3mm]
  \mathcal{E}_{{\bf X}_2} = \bigg\{ ( {\bf s}_1, {\bf s}_2,
  \mathcal{C}_1, \mathcal{C}_2 ) : \nexists {\bf u}_2 \in
  \mathcal{C}_2 \text{ s.t.~ } \left| \sqrt{1-2^{-2R_2}} -
  \cos \sphericalangle ({\bf s}_2, {\bf u}_2) \right| \leq
  \epsilon \sqrt{1-2^{-2R_2}} \bigg\}\\[3mm]
  \mathcal{E}_{({\bf X}_1,{\bf X}_2)} = \bigg\{ ({\bf s}_1, {\bf s}_2, \mathcal{C}_1,
  \mathcal{C}_2) : \left| \tilde{\rho} - \cos \sphericalangle ({\bf
        u}_1^{\ast}({\bf s}_1, \mathcal{C}_1), {\bf
        u}_2^{\ast}({\bf s}_2, \mathcal{C}_2) ) \right| > 7 \epsilon
  \bigg\}.
\end{IEEEeqnarray*}
To prove Lemma \ref{lm:vq-Pr-E4} we now start with the decomposition
\begin{IEEEeqnarray}{rCl}\label{eq:vq-Pr-E4-decomp}
  \Prv{\mathcal{E}_{\hat{\bf U}}} & = & \Prv{\mathcal{E}_{\hat{\bf U}}
    \cap \mathcal{E}_{\bf S}^c \cap \mathcal{E}_{\bf X}^c \cap
    \mathcal{E}_{\bf Z}^c} + \Prv{\mathcal{E}_{\hat{\bf U}} |
    \mathcal{E}_{\bf S} \cup \mathcal{E}_{\bf X} \cup \mathcal{E}_{\bf
      Z}}
  \Prv{\mathcal{E}_{\bf S} \cup \mathcal{E}_{\bf X} \cup
    \mathcal{E}_{\bf Z}} \nonumber\\ 
  & \leq & \Prv{\mathcal{E}_{\hat{\bf U}} \cap \mathcal{E}_{\bf S}^c
    \cap \mathcal{E}_{\bf X}^c \cap \mathcal{E}_{\bf Z}^c} +
  \Prv{\mathcal{E}_{\bf S}}
  + \Prv{\mathcal{E}_{\bf X}} + \Prv{\mathcal{E}_{\bf Z}} \nonumber\\
  & \leq & \Prv{\mathcal{E}_{\hat{\bf U}_1} \cap \mathcal{E}_{\bf S}^c
    \cap \mathcal{E}_{\bf X}^c \cap \mathcal{E}_{\bf Z}^c} +
  \Prv{\mathcal{E}_{\hat{\bf U}_2}
    \cap \mathcal{E}_{\bf S}^c \cap \mathcal{E}_{\bf X}^c \cap
    \mathcal{E}_{\bf Z}^c} \nonumber\\ 
  & & {} + \Prv{\mathcal{E}_{(\hat{\bf U}_1,\hat{\bf U}_2)} \cap
    \mathcal{E}_{\bf S}^c \cap \mathcal{E}_{\bf X}^c \cap
    \mathcal{E}_{\bf Z}^c} + \Prv{\mathcal{E}_{\bf S}} +
  \Prv{\mathcal{E}_{\bf X}} + \Prv{\mathcal{E}_{\bf Z}},
\end{IEEEeqnarray}
where we have used the shorthand notation $\Prv{\mathcal{E}_{\nu}}$
for $\Prv{ ({\bf S}_1, {\bf S}_2, \mathcal{C}_1, \mathcal{C}_2, {\bf
    Z}) \in \mathcal{E}_{\nu}}$, and where $\mathcal{E}_{\nu}^c$
denotes the complement of $\mathcal{E}_{\nu}$. Lemma \ref{lm:vq-Pr-E4}
now follows from upper bounding the probability terms on the RHS of
\eqref{eq:vq-Pr-E4-decomp}.

\begin{lm}\label{lm:vq-Pr-E1}
  For every $\delta > 0$ and $\epsilon > 0$ there exists an
  $n_1'(\delta,\epsilon) \in \Naturals$ such that for all $n >
  n_1'(\delta,\epsilon)$
  \begin{IEEEeqnarray*}{rCl}
    \Prv{\mathcal{E}_{\bf S}} & < & \delta.
  \end{IEEEeqnarray*}
\end{lm}

\begin{proof}
  The proof follows by the weak law of large numbers.
\end{proof}

\begin{lm}\label{lm:vq-Pr-E3}
  For every $\epsilon > 0$ and $\delta > 0$ there exists an
  $n_3'(\delta,\epsilon) \in \Naturals$ such that for all $n >
  n_3'(\delta,\epsilon)$
  \begin{IEEEeqnarray*}{rCl}
    \Prv{\mathcal{E}_{\bf Z}} & < & \delta.
  \end{IEEEeqnarray*}
\end{lm}

\begin{proof}
  The proof follows by the weak law of large numbers and since for
  every $\epsilon > 0$
  \begin{IEEEeqnarray*}{rCl}
    \sup_{\substack{{\bf u} \in \Reals^n:\\
        \|{\bf u}\| = \sqrt{n \sigma^2 (1-2^{-2R_i})}}}
    \Prv{\frac{1}{n} | \inner{\alpha_i {\bf u}}{{\bf z}} | > \sqrt{P_i
        N} \epsilon} \longrightarrow 0 
    & \qquad & \text{as } n \rightarrow \infty,
  \end{IEEEeqnarray*}
  where $i \in \{ 1,2 \}$.
\end{proof}

\begin{lm}\label{lm:vq-Pr-E2}
  For every $\delta > 0$ and $0.3 > \epsilon > 0$ there exists an
  $n_2'(\delta,\epsilon) \in \Naturals$ such that for all $n >
  n_2'(\delta,\epsilon)$
  \begin{IEEEeqnarray*}{rCl}
    \Prv{\mathcal{E}_{\bf X}} & < & 6 \delta.
  \end{IEEEeqnarray*}
\end{lm}

\begin{proof}
  This result has nothing to do with the channel; it is a result from
  rate-distortion theory. A proof for our setting is given in Section
  \ref{sec:prf-lm-Pr-E2}.
\end{proof}

\begin{lm}\label{lm:vq-dec-error-proba}
  For every $\delta > 0$ and every $\epsilon > 0$ there exists some
  $n_4''(\delta,\epsilon) \in \Naturals$ such that for all $n >
  n_4''(\delta,\epsilon)$
  \begin{IEEEeqnarray}{rCll}
    \Prv{\mathcal{E}_{\hat{\bf U}_1} \cap \mathcal{E}_{\bf S}^c \cap
      \mathcal{E}_{\bf X}^c \cap \mathcal{E}_{\bf Z}^c} & \leq &
    \delta, \qquad & \text{if } R_1 < \frac{1}{2} \log_2 \left(
      \frac{P_1(1-\tilde{\rho}^2)+N}{N(1-\tilde{\rho}^2)} -
      \kappa_1 \epsilon \right)\label{eq:Pr-E32-E33c-E1c-E2c}\\[3mm]
    \Prv{\mathcal{E}_{\hat{\bf U}_2} \cap \mathcal{E}_{\bf S}^c \cap
      \mathcal{E}_{\bf X}^c \cap \mathcal{E}_{\bf Z}^c} & \leq &
    \delta, & \text{if } R_2 < \frac{1}{2} \log_2 \left(
      \frac{P_2(1-\tilde{\rho}^2)+N}{N(1-\tilde{\rho}^2)}
      - \kappa_2 \epsilon \right) \label{eq:Pr-E33-E32c-E1c-E2c}\\[3mm]
    \Prv{\mathcal{E}_{(\hat{\bf U}_1,\hat{\bf U}_2)} \cap
      \mathcal{E}_{\bf S}^c \cap \mathcal{E}_{\bf X}^c \cap
      \mathcal{E}_{\bf Z}^c} & \leq & \delta, & \text{if } R_1 + R_2 <
    \frac{1}{2} \log_2 \left( \frac{P_1 + P_2 +
        2\tilde{\rho}\sqrt{P_1P_2} +N}{N(1-\tilde{\rho}^2)} - \kappa_3
      \epsilon \right), \qquad \label{eq:Pr-E32-E33-E1c-E2c}
  \end{IEEEeqnarray}
  where $\kappa_1$, $\kappa_2$, and $\kappa_3$ are positive constants
  determined by $P_1$, $P_2$, and $N$.
\end{lm}

The proof of Lemma \ref{lm:vq-dec-error-proba} requires some
preliminaries. To this end, define
\begin{IEEEeqnarray}{rCl}\label{eq:w}
  {\bf w} ({\bf s}_1, {\bf s}_2, \mathcal{C}_1, \mathcal{C}_2, {\bf
    z}) & = & \zeta_1 ({\bf y} - \alpha_2 {\bf u}_2^{\ast}) +
  \zeta_2 \alpha_2 {\bf u}_2^{\ast},
\end{IEEEeqnarray}
where
\begin{IEEEeqnarray}{rCl}\label{eq:vq-dec-err-zeta}
  \zeta_1 = \frac{N\tilde{\rho}}{P_1 (1-\tilde{\rho}^2) + N}
  \sqrt{\frac{P_1}{P_2}} & \qquad \text{and} \qquad \zeta_2 =
  \frac{P_1 (1-\tilde{\rho}^2)}{P_1 (1-\tilde{\rho}^2) + N}.
\end{IEEEeqnarray}
In the remainder we shall use the shorthand notation ${\bf w}$ instead
of ${\bf w} ({\bf s}_1, {\bf s}_2, \mathcal{C}_1, \mathcal{C}_2, {\bf
  z})$. We now start with a lemma that will be used to prove
\eqref{eq:Pr-E32-E33c-E1c-E2c}.

\begin{lm}\label{lm:ub-E41-B41}
  Let $\varphi_j \in [0,\pi]$ be the angle between ${\bf w}$ and
  ${\bf u}_1(j)$, and let the set $\mathcal{E}_{\hat{\bf U}_1}'$ be
  defined as
  \begin{IEEEeqnarray}{rCl}
    \mathcal{E}_{\hat{\bf U}_1}' & \triangleq \Bigg\{ ({\bf s}_1, {\bf s}_2,
      \mathcal{C}_1, \mathcal{C}_2, {\bf z}): & \exists {\bf u}_1(j)
      \in \mathcal{C}_1 \setminus \{ {\bf u}_1^{\ast} \} \text{ s.t. }
      \nonumber\\
      & & \hspace{1cm} \cos \varphi_j \geq
      \sqrt{\frac{P_1(1-\tilde{\rho}^2)
          +N\tilde{\rho}^2}{P_1(1-\tilde{\rho}^2)+N} - \kappa''
        \epsilon} \Bigg\}, \label{eq:B41}
  \end{IEEEeqnarray}
  where $\kappa''$ is a positive constant determined by $P_1$, $P_2$,
  $N$, $\zeta_1$ and $\zeta_2$. Then,
  \begin{IEEEeqnarray*}{rCl}
    \mathcal{E}_{\hat{\bf U}_1} \cap \mathcal{E}_{\bf S}^c \cap
      \mathcal{E}_{\bf X}^c \cap \mathcal{E}_{\bf Z}^c & \subseteq &
    \mathcal{E}_{\hat{\bf U}_1}' \cap \mathcal{E}_{\bf S}^c \cap \mathcal{E}_{\bf X}^c \cap
      \mathcal{E}_{\bf Z}^c,
  \end{IEEEeqnarray*}
  and, in particular
  \begin{IEEEeqnarray*}{rCl}
    \Prv{\mathcal{E}_{\hat{\bf U}_1} \cap \mathcal{E}_{\bf S}^c \cap
      \mathcal{E}_{\bf X}^c \cap \mathcal{E}_{\bf Z}^c} & \leq &
    \Prv{\mathcal{E}_{\hat{\bf U}_1}' \cap \mathcal{E}_{\bf S}^c \cap
      \mathcal{E}_{\bf X}^c \cap \mathcal{E}_{\bf Z}^c}.
  \end{IEEEeqnarray*}
\end{lm}

\begin{proof}
  We first recall that for the event $\mathcal{E}_{\hat{\bf U}_1}$ to occur, there must
  exist a codeword ${\bf u}_1(j) \in \mathcal{C}_1 \setminus \{ {\bf
    u}_1^{\ast} \}$ that satisfies
  \begin{IEEEeqnarray}{r rCl}
    & \left| \tilde{\rho} - \cos \sphericalangle ({\bf u}_1(j), {\bf
        u}_2^{\ast})  \right| & < & 7 \epsilon, \label{eq:lm-E41-cond1}\\[1mm]
    \text{and }\qquad & & & \nonumber\\[1mm]
    & \| {\bf y} - (\alpha_1 {\bf u}_1(j) + \alpha_2 {\bf u}_2^{\ast})
    \|^2 & \leq & \| {\bf y} - (\alpha_1 {\bf u}_1^{\ast} + \alpha_2 {\bf
      u}_2^{\ast}) \|^2. \label{eq:lm-E41-cond2}
  \end{IEEEeqnarray}
  The proof is now based on a sequence of statements related to
  Condition \eqref{eq:lm-E41-cond1} and Condition
  \eqref{eq:lm-E41-cond2}:
  \begin{itemize}
  \item[A)] For every $({\bf s}_1,{\bf s}_2, \mathcal{C}_1,
    \mathcal{C}_2, {\bf z}) \in \mathcal{E}_{\bf X}^c$ and every ${\bf u}
    \in \mathcal{S}_1$, where $\mathcal{S}_1$ is the surface area
    of the codeword sphere of $\mathcal{C}_1$ defined in the code
    construction,
    \begin{IEEEeqnarray}{rCl}\label{eq:vq-dec-err-stmtA}
      \left| \tilde{\rho} - \cos \sphericalangle ( {\bf u}, {\bf
          u}_2^{\ast}) \right| < 7 \epsilon & \quad \Rightarrow \quad
      & \left| n \tilde{\rho} \sqrt{P_1 P_2} - \inner{\alpha_1 {\bf
            u}}{ \alpha_2 {\bf u}_2^{\ast}} \right| \leq 7 n \sqrt{P_1
        P_2} \epsilon.\quad \\ \nonumber
    \end{IEEEeqnarray}
  \end{itemize}
  \hfill \parbox{14.1cm}{ Statement A) follows by rewriting $\cos
    \sphericalangle ( {\bf u}, {\bf u}_2^{\ast})$ as $\inner{{\bf
        u}}{{\bf u}_2^{\ast}}/(\|{\bf u}\| \|{\bf u}_2^{\ast}\|)$, and
    then multiplying the inequality on the LHS of
    \eqref{eq:vq-dec-err-stmtA} by $\| \alpha_1 {\bf u} \| \cdot \|
    \alpha_2 {\bf u}_2^{\ast}\|$ and recalling that $\| \alpha_1 {\bf
      u} \| = \sqrt{nP_1}$ and that $\| \alpha_2 {\bf u}_2^{\ast}\| =
    \sqrt{nP_2}$.\\[3mm]}
  \begin{itemize}
  \item[B)] For every $({\bf s}_1,{\bf s}_2, \mathcal{C}_1,
    \mathcal{C}_2, {\bf z}) \in \mathcal{E}_{\bf X}^c \cap \mathcal{E}_{\bf Z}^c$
    and every ${\bf u} \in \mathcal{S}_1$
    \begin{IEEEeqnarray}{rCl} \label{eq:vq-dec-err-stmtB}
      \| {\bf y} - (\alpha_1 {\bf u} + \alpha_2 {\bf u}_2^{\ast})
      \|^2 & \leq & \| {\bf y} - (\alpha_1 {\bf u}_1^{\ast} + \alpha_2
      {\bf u}_2^{\ast}) \|^2 \nonumber\\[3mm]
      & & \Rightarrow \quad \inner{{\bf y} - \alpha_2 {\bf
          u}_2^{\ast}}{\alpha_1 {\bf u}} \geq nP_1 -
      n\sqrt{P_1N}\epsilon.\\ \nonumber
    \end{IEEEeqnarray}
  \end{itemize}
  \hfill \parbox{14.1cm}{ Statement B) follows from rewriting the
    inequality on the LHS of \eqref{eq:vq-dec-err-stmtB} as $\| ({\bf
      y} - \alpha_2 {\bf u}_2^{\ast}) - \alpha_1 {\bf u}_1(j) \|^2
    \leq \| ({\bf y} - \alpha_2 {\bf u}_2^{\ast}) - \alpha_1 {\bf
      u}_1^{\ast} \|^2$ or equivalently as
    \begin{IEEEeqnarray}{rCl}
      \inner{{\bf y} - \alpha_2 {\bf u}_2^{\ast}}{\alpha_1 {\bf u}_1(j)} &
      \geq & \inner{{\bf y} - \alpha_2 {\bf u}_2^{\ast}}{\alpha_1 {\bf
          u}_1^{\ast}} \nonumber\\
      & = & \inner{\alpha_1 {\bf u}_1^{\ast} + {\bf z}}{\alpha_1 {\bf
          u}_1^{\ast}} \nonumber\\
      & = & \| \alpha_1 {\bf u}_1^{\ast} \|^2 + \inner{{\bf z}}{\alpha_1
        {\bf u}_1^{\ast}}. \label{eq:dec-error-cd2-intmed}
    \end{IEEEeqnarray}
    It now follows from the equivalence of the first inequality in
    \eqref{eq:vq-dec-err-stmtB} with \eqref{eq:dec-error-cd2-intmed}
    that for $({\bf s}_1, {\bf s}_2, \mathcal{C}_1, \mathcal{C}_2, {\bf
      z}) \in \mathcal{E}_{\bf Z}^c$, the first inequality in
    \eqref{eq:vq-dec-err-stmtB} can only hold if
    \begin{IEEEeqnarray}{rCl}
      \inner{{\bf y} - \alpha_2 {\bf u}_2^{\ast}}{\alpha_1 {\bf u}} &
      \geq &  nP_1 - n\sqrt{P_1N}\epsilon, \label{eq:dec-error-cond2}
    \end{IEEEeqnarray}
    thus establishing B).\\[3mm]}

  \begin{itemize}
  \item[C)] For every $({\bf s}_1,{\bf s}_2, \mathcal{C}_1,
    \mathcal{C}_2, {\bf z}) \in \mathcal{E}_{\bf X}^c \cap \mathcal{E}_{\bf Z}^c$
    and every ${\bf u} \in \mathcal{S}_1$,
    \begin{IEEEeqnarray}{C}\label{eq:vq-dec-err-stmtC}
      \left( \left| \tilde{\rho} - \cos \sphericalangle ({\bf u}, {\bf
            u}_2^{\ast}) \right| < 7 \epsilon \quad \text{and} \quad
        \| {\bf y} - (\alpha_1 {\bf u} + \alpha_2 {\bf u}_2^{\ast})
        \|^2 \leq \| {\bf y} - (\alpha_1 {\bf u}_1^{\ast} + \alpha_2
        {\bf u}_2^{\ast})
        \|^2 \right) \nonumber\\[6mm]
      \Rightarrow \\[6mm]
      \bigg( \| \alpha_1 {\bf u} - {\bf w}\|^2 \leq nP_1 - 2 \left(
        \zeta_1 n \left( P_1 - \sqrt{P_1N} \epsilon \right) + \zeta_2
        n\sqrt{P_1 P_2}(\tilde{\rho}-7\epsilon) \right) + \| {\bf w}
      \|^2 \bigg). \nonumber \\ \nonumber
    \end{IEEEeqnarray}
  \end{itemize}
  \hfill \parbox{14.1cm}{ Statement C) is obtained as follows:
  \begin{IEEEeqnarray*}{rCl}
    \| \alpha_1 {\bf u} - {\bf w}\|^2 & = & \| \alpha_1 {\bf
      u} \|^2 -2 \inner{\alpha_1 {\bf u}}{{\bf w}} + \| {\bf
      w} \|^2 \\
    & = & \| \alpha_1 {\bf u} \|^2 -2 \left( \zeta_1
      \inner{\alpha_1 {\bf u}}{{\bf y} - \alpha_2 {\bf u}_2^{\ast}}
      + \zeta_2 \inner{\alpha_1 {\bf u}}{\alpha_2 {\bf u}_2^{\ast}}
    \right) + \| {\bf w} \|^2 \\
    & \stackrel{a)}{\leq} & nP_1 - 2 \left( \zeta_1 n \left( P_1 -
        \sqrt{P_1N} \epsilon \right) + \zeta_2 n\sqrt{P_1
        P_2}(\tilde{\rho}-7\epsilon) \right) + \| {\bf w} \|^2,
  \end{IEEEeqnarray*}
  where in $a)$ we have used Statement A) and Statement B).\\[3mm]}

  \begin{itemize}  
  \item[D)] For every $({\bf s}_1,{\bf s}_2, \mathcal{C}_1,
    \mathcal{C}_2, {\bf z}) \in \mathcal{E}_{\bf X}^c \cap \mathcal{E}_{\bf Z}^c$
    \begin{IEEEeqnarray}{rCl}\label{eq:vq-dec-err-stmtD}
      \| {\bf w} \|^2 & \leq & n \left( \zeta_1^2 P_1 + 2 \zeta_1 \zeta_2
        \sqrt{P_1P_2}\tilde{\rho} + \zeta_2^2 \left( P_1 + N \right) + \kappa
        \epsilon \right),
    \end{IEEEeqnarray}
    where $\kappa$ depends on $P_1$, $P_2$, $N$, $\zeta_1$ and
    $\zeta_2$ only.\\[3mm]
  \end{itemize}
  \hfill \parbox{14.1cm}{ Statement D) is obtained as follows
  \begin{IEEEeqnarray*}{rCl}
    \| {\bf w} \|^2 & = & \zeta_1^2 \| \alpha_2 {\bf u}_2^{\ast} \|^2 + 2
    \zeta_1 \zeta_2 \inner{\alpha_2 {\bf u}_2^{\ast}}{{\bf y} - \alpha_2
      {\bf u}_2^{\ast}} + \zeta_2^2 \| {\bf y} - \alpha_2 {\bf u}_2^{\ast}
    \|^2\\
    & = & \zeta_1^2 n P_2 + 2 \zeta_1 \zeta_2 \left( \inner{\alpha_2 {\bf
          u}_2^{\ast}}{\alpha_1 {\bf u}_1^{\ast}} - \inner{\alpha_2 {\bf
          u}_2^{\ast}}{{\bf z}} \right) + \zeta_2^2 \left( \| \alpha_1
      {\bf u}_1^{\ast} \|^2 + 2\inner{\alpha_1 {\bf u}_1^{\ast}}{{\bf z}} +
      \| {\bf z} \|^2 \right)\\
    & \stackrel{a)}{\leq} & \zeta_1^2 n P_1 + 2 \zeta_1 \zeta_2 \left( n
      \sqrt{P_1P_2}(\tilde{\rho}+7\epsilon) + n \sqrt{P_2N}\epsilon
    \right) + \zeta_2^2 \left( nP_1 + 2 n\sqrt{P_1N}\epsilon +
      nN(1+\epsilon) \right)\\
    & \leq & n \left( \zeta_1^2 P_1 + 2 \zeta_1 \zeta_2
      \sqrt{P_1P_2}\tilde{\rho} + \zeta_2^2 \left( P_1 + N \right) + \kappa
      \epsilon \right),
  \end{IEEEeqnarray*}
  where in $a)$ we have used that $({\bf s}_1,{\bf s}_2,
  \mathcal{C}_1, \mathcal{C}_2, {\bf z}) \in \mathcal{E}_{\bf
    Z}^c$.\\[3mm]}

  \begin{itemize}
  \item[E)] For every $({\bf s}_1,{\bf s}_2, \mathcal{C}_1,
    \mathcal{C}_2, {\bf z}) \in \mathcal{E}_{\bf X}^c \cap \mathcal{E}_{\bf Z}^c$
    and an arbitrary ${\bf u} \in \mathcal{S}_1$,
    \begin{IEEEeqnarray}{C}\label{eq:vq-dec-err-stmtE}
      \left( \left| \tilde{\rho} - \cos \sphericalangle ({\bf u}, {\bf
        u}_2^{\ast}) \right| <
      7 \epsilon \quad \text{and} \quad \| {\bf y} - (\alpha_1 {\bf
        u} + \alpha_2 {\bf u}_2^{\ast}) \|^2 \leq \| {\bf y} - (\alpha_1
      {\bf u}_1^{\ast} + \alpha_2 {\bf u}_2^{\ast})
      \|^2 \right) \nonumber\\[6mm]
      \Rightarrow \\[6mm]
      \bigg( \| \alpha_1 {\bf u} - {\bf w}\|^2 \leq \Upsilon
      (\epsilon) \bigg), \nonumber
    \end{IEEEeqnarray}
    where
    \begin{IEEEeqnarray*}{rCl}
      \Upsilon (\epsilon) & = & n \frac{P_1 N (1-\tilde{\rho}^2)}{P_1
        (1-\tilde{\rho}^2) + N} + n \kappa' \epsilon,
    \end{IEEEeqnarray*}
    and where $\kappa'$ only depends on $P$, $N_1$, $N_2$, $\zeta_1$
    and $\zeta_2$.\\
  \end{itemize}
  \hfill \parbox{14.1cm}{ Statement E) follows from combining
    Statement C) with Statement D) and the explicit values of
    $\zeta_1$ and $\zeta_2$ given in
    \eqref{eq:vq-dec-err-zeta}.\\[3mm]}

  \begin{itemize}
  \item[F)] For every ${\bf u} \in \mathcal{S}_1$, denote
    by $\varphi \in [0,\pi]$ the angle between ${\bf u}$ and ${\bf
      w}$, and let  
    \begin{IEEEeqnarray*}{rCl}
      \mathcal{B}({\bf s}_1, {\bf s}_2, {\bf u}_1^{\ast}, {\bf
        u}_2^{\ast}, {\bf z}) & \triangleq \Bigg\{ {\bf u} \in
      \mathcal{S}_1^{(n)}: & \cos \varphi \geq \sqrt{\frac{P_1(1-\tilde{\rho}^2)
          +N\tilde{\rho}^2}{P_1(1-\tilde{\rho}^2)+N} - \kappa''
        \epsilon} \Bigg\},
    \end{IEEEeqnarray*}
    where $\kappa''$ only depends on $P$, $N_1$, $N_2$, $\zeta_1$ and
    $\zeta_2$, and where we assume $\epsilon$ sufficiently small such
    that
    \begin{IEEEeqnarray*}{rCl}
      \frac{P_1(1-\tilde{\rho}^2)
        +N\tilde{\rho}^2}{P_1(1-\tilde{\rho}^2)+N} - \kappa'' \epsilon
      & > & 0.
    \end{IEEEeqnarray*}
    Then, for every $({\bf s}_1,{\bf s}_2, \mathcal{C}_1,
    \mathcal{C}_2, {\bf z}) \in \mathcal{E}_{\bf X}^c \cap \mathcal{E}_{\bf Z}^c$,
    \begin{IEEEeqnarray}{C}\label{eq:vq-dec-err-stmtF}
      \left( \left| \tilde{\rho} - \cos \sphericalangle ({\bf u}, {\bf
            u}_2^{\ast}) \right| < 7 \epsilon \quad \text{and} \quad
        \| {\bf y} - (\alpha_1 {\bf u} + \alpha_2 {\bf u}_2^{\ast})
        \|^2 \leq \| {\bf y} - (\alpha_1 {\bf u}_1^{\ast} + \alpha_2
        {\bf u}_2^{\ast})
        \|^2 \right) \nonumber\\[6mm]
      \Rightarrow \qquad {\bf u} \in \mathcal{B}({\bf s}_1, {\bf s}_2,
      {\bf u}_1^{\ast}, {\bf u}_2^{\ast}, {\bf z}).\\ \nonumber
    \end{IEEEeqnarray}
  \end{itemize}
  Statement F) follows from Statement E) by noting that if ${\bf w}
  \neq {\bf 0}$ and $1-\Upsilon(\epsilon)/(nP_1)>0$, then
  \begin{IEEEeqnarray*}{rCl}
    \left. \begin{array}{l}
        \| \alpha_1 {\bf u} \|^2 = n P_1\\[2mm]
        \| \alpha_1 {\bf u} - {\bf w} \|^2 \leq \Upsilon(\epsilon)
      \end{array} \right\} \quad \Rightarrow \quad \cos
    \sphericalangle ({\bf u}, {\bf w}) \geq
    \sqrt{1-\frac{\Upsilon(\epsilon)}{nP_1}}.
  \end{IEEEeqnarray*}
  To see this, first note that for every $\alpha_1 {\bf u}$, where
  ${\bf u} \in \mathcal{S}_1$, satisfying the condition on the LHS of
  \eqref{eq:vq-dec-err-stmtF} lies within a sphere of radius
  $\sqrt{\Upsilon (\epsilon)}$ centered at ${\bf w}$. And for every
  ${\bf u} \in \mathcal{S}_1$ we have that $\alpha_1 {\bf u}$ also
  lies on the centered $\Reals^n$-sphere of radius
  $\sqrt{nP_1}$. Hence, every ${\bf u} \in \mathcal{S}_1^{(n)}$
  satisfying the condition on the LHS of \eqref{eq:vq-dec-err-stmtF}
  lies in the intersection of these two regions, which is a polar cap
  on the centered sphere of radius $\sqrt{nP_1}$. An illustration of
  such a polar cap is given in Figure \ref{fig:polar-cap}.
  \begin{figure}[h]
    \centering
    \psfrag{1}[cr][cc]{$\sqrt{nP_1}$}
    \psfrag{rn}[cc][cc]{$\Reals^n$}
    \psfrag{p}[cc][cc]{$\varphi$}
    \epsfig{file=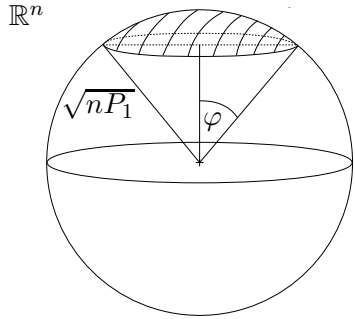, width=0.3\textwidth}
    \caption{Polar cap of half angle $\varphi$ on an $\Reals^n$-sphere of
      radius $\sqrt{nP_1}$.}
    \label{fig:polar-cap}
  \end{figure}
  The area of this polar cap is outer bounded as follows. Let ${\bf
    r}$ be an arbitrary point on the boundary of this polar cap. The
  half-angle of the polar cap would be maximized if ${\bf w}$ and
  ${\bf r} - {\bf w}$ would lie perpendicular to each other, as is
  illustrated in Subplot b) of Figure \ref{fig:cap-max}.
  \begin{figure}[h]
      \centering
      \psfrag{0}[cc][cc]{${\bf 0}$}
      \psfrag{w}[cc][cc]{${\bf w}$}
      \psfrag{r}[cc][cc]{${\bf r}$}
      \psfrag{p1}[cc][cc]{$\sqrt{nP_1}$}
      \psfrag{ps}[cc][cc]{$\sqrt{\Upsilon(\epsilon)}$}
      \psfrag{a}[cc][cc]{a)}
      \psfrag{b}[cc][cc]{b)}
      \psfrag{c}[cc][cc]{c)}
      \epsfig{file=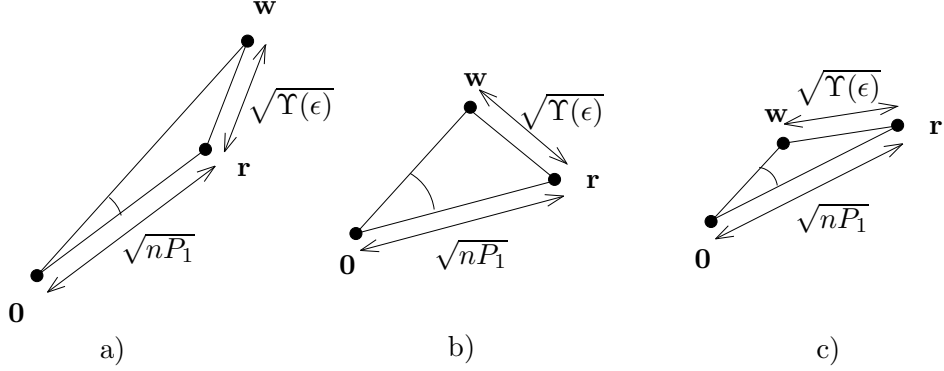, width=0.85\textwidth}
      \caption{Half-angle of cap for different constellations of ${\bf w}$
        and ${\bf r}$.}
      \label{fig:cap-max}
    \end{figure}
    Hence, every ${\bf u} \in \mathcal{S}_1^{(n)}$ satisfying the
    upper conditions of \eqref{eq:vq-dec-err-stmtF} also satisfies
    \begin{IEEEeqnarray*}{rCl}
      \cos \varphi & \geq & \sqrt{1 - \frac{\Upsilon (\epsilon)}{nP_1}}\\
      & = &  \sqrt{\frac{P_1(1-\tilde{\rho}^2) +
          N\tilde{\rho}^2}{P_1(1-\tilde{\rho}^2)+N} - \kappa'' \epsilon},
    \end{IEEEeqnarray*}
    where we assume $\epsilon$ sufficiently small such that
    $1-\Upsilon(\epsilon)/(nP_1) > 0$ and where $\kappa'' =
    \kappa'/P_1$.\\[3mm]

  The proof of Lemma \ref{lm:ub-E41-B41} is now concluded by noticing
  that the set $\mathcal{E}_{\hat{\bf U}_1}'$, defined in \eqref{eq:B41}, is the
  set of tuples $({\bf s}_1, {\bf s}_2, \mathcal{C}_1, \mathcal{C}_2,
  {\bf z})$ for which there exists a ${\bf u}_1(j) \in \mathcal{C}_1
  \setminus \{ {\bf u}_1^{\ast} \}$ such that ${\bf u}_1(j) \in
  \mathcal{B}({\bf s}_1, {\bf s}_2, {\bf u}_1^{\ast}, {\bf
    u}_2^{\ast}, {\bf z})$. Thus, by Statement F) and by the
  definition of $\mathcal{E}_{\hat{\bf U}_1}$ in \eqref{eq:E41} it follows that
  \begin{IEEEeqnarray*}{rCl}
    \mathcal{E}_{\hat{\bf U}_1} \cap \mathcal{E}_{\bf Z}^c \cap \mathcal{E}_{\bf S}^c \cap
    \mathcal{E}_{\bf X}^c & \subseteq & \mathcal{E}_{\hat{\bf U}_1}' \cap
    \mathcal{E}_{\bf Z}^c \cap \mathcal{E}_{\bf S}^c \cap \mathcal{E}_{\bf X}^c, 
  \end{IEEEeqnarray*}
  and therefore
  \begin{IEEEeqnarray*}{rCl}
    \qquad \qquad \qquad \; \, \Prv{\mathcal{E}_{\hat{\bf U}_1} \cap
      \mathcal{E}_{\bf Z}^c \cap \mathcal{E}_{\bf S}^c \cap \mathcal{E}_{\bf X}^c} & \leq
    & \Prv{\mathcal{E}_{\hat{\bf U}_1}' \cap \mathcal{E}_{\bf Z}^c
      \cap \mathcal{E}_{\bf S}^c \cap 
      \mathcal{E}_{\bf X}^c}. \qquad \qquad \qquad \; \, \qedhere
  \end{IEEEeqnarray*}
\end{proof}

We now state one more lemma that will be used for the proof of
\eqref{eq:Pr-E32-E33c-E1c-E2c}.

\begin{lm}\label{lm:caps-single-error}
  For every $\Delta \in (0,1]$, let the set $\mathcal{G}$ be given by
  \begin{IEEEeqnarray*}{rCl}
    \mathcal{G} & = & \left\{ ({\bf s}_1, {\bf s}_2, \mathcal{C}_1,
      \mathcal{C}_2, {\bf z}) : \exists {\bf u}_1(j) \in \mathcal{C}_1
      \setminus \{ {\bf u}_1^{\ast} \} \text{ s.t.~} \cos
      \sphericalangle ({\bf w}, {\bf u}_1(j)) \geq
      \Delta \right\},
  \end{IEEEeqnarray*}
  where ${\bf w}$ is defined in \eqref{eq:w}.  Then,
  \begin{IEEEeqnarray}{rCl}\label{eq:lm-wrong-codeword}
    \left( R_1 < - \frac{1}{2} \log_2 (1 - \Delta^2) \right) & \quad
    \Rightarrow \quad & \left( \lim_{n \rightarrow \infty}
      \Prv{\mathcal{G} | \mathcal{E}_{{\bf X}_1}^c} = 0, \quad
      \epsilon > 0 \right).
  \end{IEEEeqnarray}
\end{lm}

\begin{proof}
  The proof follows from upper bounding in every point on
  $\mathcal{S}_1$ the density of every ${\bf u}_1(j) \in \mathcal{C}_1
  \setminus \{ {\bf u}_1^{\ast} \}$ and then using a standard argument
  from sphere-packing. The proof is given in Section
  \ref{app:proof-dec-single-error}.
\end{proof}

We next state two lemmas for the proof of
\eqref{eq:Pr-E32-E33-E1c-E2c}. These lemmas are similar to
Lemma~\ref{lm:ub-E41-B41} and Lemma~\ref{lm:caps-single-error}.

\begin{lm}\label{lm:ub-E43-B43}
  For every sufficiently small $\epsilon > 0$, define the set
  $\mathcal{E}_{(\hat{\bf U}_1,\hat{\bf U}_2)}'$ as
  \begin{IEEEeqnarray*}{rCl}
    \mathcal{E}_{(\hat{\bf U}_1,\hat{\bf U}_2)}' & \triangleq & \big\{
    ({\bf s}_1, {\bf s}_2, \mathcal{C}_1, \mathcal{C}_2, {\bf z}):
    \exists {\bf u}_1(j) \in \mathcal{C}_1 \setminus \{ {\bf
      u}_1^{\ast} \} \quad \text{and} \quad \exists {\bf u}_2(\ell)
    \in
    \mathcal{C}_2 \setminus \{ {\bf u}_2^{\ast} \} \text{ s.t.~}\\[2mm]
    & & \quad \cos \sphericalangle ({\bf u}_1(j), {\bf u}_2(\ell))
    \geq \tilde{\rho} - 7 \epsilon \quad \text{and} \quad \cos
    \sphericalangle ({\bf y},\alpha_1 {\bf u}_1(j) + \alpha_2 {\bf
      u}_2(\ell)) \geq \Lambda(\epsilon) \big\},
  \end{IEEEeqnarray*}
  where
  \begin{IEEEeqnarray*}{rCl}
    \Lambda (\epsilon) & = & \sqrt{\frac{P_1 + P_2 +
        2\tilde{\rho}\sqrt{P_1P_2} - \xi'\epsilon}{P_1 + P_2 +
        2\tilde{\rho} \sqrt{P_1P_2} +N + \xi_2\epsilon}},
  \end{IEEEeqnarray*}
  and where $\xi'$ and $\xi_2$ depend only on $P_1$, $P_2$ and
  $N$. Then, for every sufficiently small $\epsilon > 0$
  \begin{IEEEeqnarray*}{rCl}
    \mathcal{E}_{(\hat{\bf U}_1,\hat{\bf U}_2)} \cap \mathcal{E}_{\bf
      Z}^c \cap \mathcal{E}_{\bf S}^c \cap \mathcal{E}_{\bf X}^c &
    \subseteq & \mathcal{E}_{(\hat{\bf U}_1,\hat{\bf U}_2)}' \cap
    \mathcal{E}_{\bf Z}^c \cap \mathcal{E}_{\bf S}^c \cap
    \mathcal{E}_{\bf X}^c,
  \end{IEEEeqnarray*}
  and, in particular
  \begin{IEEEeqnarray*}{rCl}
    \Prv{\mathcal{E}_{(\hat{\bf U}_1,\hat{\bf U}_2)} \cap
      \mathcal{E}_{\bf Z}^c \cap \mathcal{E}_{\bf S}^c \cap
      \mathcal{E}_{\bf X}^c} & \leq & \Prv{\mathcal{E}_{(\hat{\bf
          U}_1,\hat{\bf U}_2)}' \cap \mathcal{E}_{\bf Z}^c \cap
      \mathcal{E}_{\bf S}^c \cap \mathcal{E}_{\bf X}^c}.
  \end{IEEEeqnarray*}
\end{lm}

\begin{proof}
  We first recall that for the event $\mathcal{E}_{(\hat{\bf
      U}_1,\hat{\bf U}_2)}$ to occur, there must exist codewords ${\bf
    u}_1(j) \in \mathcal{C}_1 \setminus \{ {\bf u}_1^{\ast} \}$ and
  ${\bf u}_2(\ell) \in \mathcal{C}_2 \setminus \{ {\bf u}_2^{\ast} \}$
  such that
  \begin{IEEEeqnarray}{rrCl}
    & \left| \tilde{\rho} - \cos \sphericalangle ({\bf u}_1(j), {\bf
      u}_2(\ell)) \right| & < &
    7 \epsilon ,\label{eq:lm-E43-cond1}\\[1mm]
    \text{and } \qquad & & & \nonumber\\[1mm]
    & \| {\bf y} - (\alpha_1 {\bf u}_1(j) + \alpha_2 {\bf u}_2(\ell)) \|^2
    & \leq & \| {\bf y} - (\alpha_1 {\bf u}_1^{\ast} + \alpha_2 {\bf
      u}_2^{\ast}) \|^2.\label{eq:lm-E43-cond2}
  \end{IEEEeqnarray}
  The proof is now based on a sequence of statements related to
  Condition \eqref{eq:lm-E43-cond1} and Condition
  \eqref{eq:lm-E43-cond2}.


  \begin{itemize}
  \item[A)] For every $({\bf s}_1, {\bf s}_2, \mathcal{C}_1,
    \mathcal{C}_2, {\bf z}) \in \mathcal{E}_{\bf X}^c \cap \mathcal{E}_{\bf Z}^c$,
    \begin{IEEEeqnarray}{C}
      \| {\bf y} - (\alpha_1 {\bf u}_1(j) + \alpha_2 {\bf u}_2(\ell))
      \|^2 \leq \| {\bf y} - (\alpha_1 {\bf u}_1^{\ast} + \alpha_2
      {\bf u}_2^{\ast}) \|^2 \nonumber\\[3mm]
      \Rightarrow \label{eq:vq-E43-cond-B}\\[3mm]
      \inner{{\bf y}}{\alpha_1 {\bf u}_1(j) + \alpha_2 {\bf
          u}_2(\ell)} \geq n \left( P_1 + P_2 +
        2\tilde{\rho}\sqrt{P_1P_2} - \xi_1\epsilon \right), \nonumber
    \end{IEEEeqnarray}
    where $\xi_1$ only depends on $P_1$, $P_2$ and $N$.\\
  \end{itemize}
  \hfill \parbox{14.1cm}{ Statement A) follows by rewriting the LHS of
    \eqref{eq:vq-E43-cond-B} as
    \begin{IEEEeqnarray}{rCl}
      2 \inner{{\bf y}}{\alpha_1 {\bf u}_1(j) + \alpha_2 {\bf u}_2(\ell)}
      & \geq & 2 \inner{{\bf y}}{\alpha_1 {\bf u}_1^{\ast} + \alpha_2 {\bf
          u}_2^{\ast}} + \| \alpha_1 {\bf u}_1(j) + \alpha_2 {\bf
        u}_2(\ell) \|^2 - \| \alpha_1 {\bf u}_1^{\ast} + \alpha_2 {\bf
        u}_2^{\ast} \|^2 \nonumber\\ 
      & = & \| \alpha_1 {\bf u}_1^{\ast} + \alpha_2 {\bf u}_2^{\ast} \|^2
      + 2 \inner{{\bf z}}{\alpha_1 {\bf u}_1^{\ast} + \alpha_2 {\bf
          u}_2^{\ast}} + \| \alpha_1 {\bf u}_1(j) + \alpha_2 {\bf
        u}_2(\ell) \|^2 \nonumber\\
      & \stackrel{a)}{\geq} & 2n \left( P_1 + 2 \tilde{\rho}\sqrt{P_1P_2}(1-7\epsilon) +
        P_2 + \sqrt{P_1N}\epsilon + \sqrt{P_2N} \epsilon \right) \nonumber\\
      & = & 2n \left( P_1 + P_2 + 2\tilde{\rho}\sqrt{P_1P_2} -
        \xi_1\epsilon \right), \label{eq:two-dec-errors-cond2}
    \end{IEEEeqnarray}
    where in $a)$ we have used that $({\bf s}_1, {\bf s}_2,
    \mathcal{C}_1, \mathcal{C}_2, {\bf z}) \in \mathcal{E}_{\bf X}^c \cap
    \mathcal{E}_{\bf Z}^c$ and that $\| \alpha_1 {\bf u}_1(j) + \alpha_2 {\bf
      u}_2(\ell) \|^2 \geq 0$.\\[3mm]}

  \begin{itemize}
  \item[B)] For every $({\bf s}_1, {\bf s}_2, \mathcal{C}_1,
    \mathcal{C}_2, {\bf z}) \in \mathcal{E}_{\bf X}^c \cap \mathcal{E}_{\bf Z}^c$,
    \begin{IEEEeqnarray*}{rCl}
      \| {\bf y} \|^2 & \leq & n \left( P_1 +
        2\tilde{\rho}\sqrt{P_1P_2} + P_2 + N + \xi_2 \epsilon \right),
    \end{IEEEeqnarray*}
    where $\xi_2$ only depends on $P_1$, $P_2$ and $N$.\\
  \end{itemize}
  \hfill \parbox{14.1cm}{ Statement B) is obtained as follows:
    \begin{IEEEeqnarray*}{rCl}
      \| {\bf y} \|^2 & = & \| \alpha_1 {\bf u}_1^{\ast} \|^2 + 2
      \inner{\alpha_1 {\bf u}_1^{\ast}}{\alpha_2 {\bf u}_2^{\ast}} + \|
      \alpha_2 {\bf u}_2^{\ast} \|^2 + 2 \left( \inner{\alpha_1 {\bf
            u}_1^{\ast}}{{\bf z}} + \inner{\alpha_2 {\bf u}_2^{\ast}}{{\bf
            z}} \right) + \| {\bf z} \|^2\\
      & \stackrel{a)}{\leq} & nP_1 + 2n \tilde{\rho} \sqrt{P_1P_2} (1+7\epsilon) + nP_2 +
      2n\sqrt{P_1N} \epsilon + 2n \sqrt{P_2N} \epsilon + nN(1+\epsilon)\\
      & \leq & n \left( P_1 + 2\tilde{\rho}\sqrt{P_1P_2} + P_2 + N + \xi_2
        \epsilon \right),
    \end{IEEEeqnarray*}
    where in $a)$ we have used that $({\bf s}_1, {\bf s}_2,
    \mathcal{C}_1, \mathcal{C}_2, {\bf z}) \in \mathcal{E}_{\bf X}^c \cap
    \mathcal{E}_{\bf Z}^c$.\\[3mm]}

  \begin{itemize}
  \item[C)] For every $({\bf s}_1, {\bf s}_2, \mathcal{C}_1,
    \mathcal{C}_2, {\bf z})$,
    \begin{IEEEeqnarray}{C}
      \bigg| \tilde{\rho} - \inner{\frac{{\bf u}_1(j)}{\| {\bf u}_1(j)
          \|}}{\frac{{\bf u}_2(\ell)}{\| {\bf u}_2(\ell) \|}} \bigg| <
      7 \epsilon \nonumber\\[4mm]
      \Rightarrow \label{eq:vq-E43-cond-D}\\[3mm]
      \| \alpha_1 {\bf u}_1(j) + \alpha_2 {\bf u}_2(\ell) \|^2 \leq n
      \left( P_1 + 2\tilde{\rho}\sqrt{P_1P_2} + P_2 + \xi_3 \epsilon
      \right). \nonumber\\ \nonumber
    \end{IEEEeqnarray}
  \end{itemize}
  \hfill \parbox{14.1cm}{ Statement C) follows by
    \begin{IEEEeqnarray*}{rCl}
      \| \alpha_1 {\bf u}_1(j) + \alpha_2 {\bf u}_2(\ell) \|^2 & = & \|
      \alpha_1 {\bf u}_1(j) \|^2 + 2\inner{\alpha_1 {\bf u}_1(j)}{\alpha_2
        {\bf u}_2(\ell)} + \| \alpha_2 {\bf u}_2(\ell) \|^2\\
      & \stackrel{a)}{\leq} & nP_1 + 2n\tilde{\rho} \sqrt{P_1P_2} (1+7\epsilon) + nP_2\\
      & = & n \left( P_1 + 2\tilde{\rho}\sqrt{P_1P_2} + P_2 + \xi_3 \epsilon
      \right).
    \end{IEEEeqnarray*}
    where in $a)$ we have used that multiplying the inequality on the
    LHS of \eqref{eq:vq-E43-cond-D} by $\| \alpha_1 {\bf u}_1(j) \|
    \cdot \| \alpha_2 {\bf u}_2(\ell) \|$ and recalling that $\|
    \alpha_1 {\bf u}_1(j) \| \leq \sqrt{nP_1}$ and that $\| \alpha_2
    {\bf u}_2(\ell) \| \leq \sqrt{nP_2}$ gives
    \begin{IEEEeqnarray*}{rCl}
      \big| n \sqrt{P_1P_2} \tilde{\rho} - \inner{\alpha_1 {\bf
          u}_1(j)}{\alpha_2 {\bf u}_2(\ell)} \big| & < & 7 n
      \sqrt{P_1P_2} \epsilon,
    \end{IEEEeqnarray*}
    and thus
    \begin{IEEEeqnarray*}{rCl}
      \inner{\alpha_1 {\bf u}_1(j)}{\alpha_2 {\bf u}_2(\ell)} & < & n
      \sqrt{P_1P_2} \tilde{\rho} (1+7\epsilon),
    \end{IEEEeqnarray*}
    thus establishing C).\\[3mm]}

  \begin{itemize}
  \item[D)] For every $({\bf s}_1, {\bf s}_2, \mathcal{C}_1,
    \mathcal{C}_2, {\bf z}) \in \mathcal{E}_{\bf X}^c \cap \mathcal{E}_{\bf Z}^c$,
    \begin{IEEEeqnarray*}{rl}
      \Big( & \left| \tilde{\rho} - \cos \sphericalangle ({\bf
          u}_1(j), {\bf u}_2(\ell)) \right| < 7 \epsilon\\
        & \qquad \qquad \qquad \text{and} \quad \| {\bf y}
        - (\alpha_1 {\bf u}_1(j) + \alpha_2 {\bf u}_2(\ell)) \|^2 \leq
        \| {\bf y} - (\alpha_1 {\bf u}_1^{\ast} + \alpha_2 {\bf
          u}_2^{\ast}) \|^2 \Big)\\[4mm]
      & \qquad \qquad \Rightarrow \qquad \cos \sphericalangle ({\bf
        y}, \alpha_1 {\bf u}_1(j)+\alpha_2 {\bf u}_2(\ell)) \geq \Lambda(\epsilon).\\
    \end{IEEEeqnarray*}
  \end{itemize}
  \hfill \parbox{14.1cm}{ Statement D) follows by rewriting $\cos
    \sphericalangle ({\bf y}, \alpha_1 {\bf u}_1(j)+\alpha_2 {\bf
      u}_2(\ell))$ as 
    \begin{IEEEeqnarray*}{rCl}
      \cos \sphericalangle ({\bf y}, \alpha_1 {\bf u}_1(j)+\alpha_2 {\bf
        u}_2(\ell)) & = & \frac{\inner{{\bf y}}{\alpha_1 {\bf u}_1(j) +
          \alpha_2 {\bf u}_2(\ell)}}{\| {\bf y} \| \cdot \| \alpha_1
        {\bf u}_1(j) + \alpha_2 {\bf u}_2(\ell) \|},
    \end{IEEEeqnarray*}
    and then lower bounding $\inner{{\bf y}}{\alpha_1 {\bf u}_1(j) +
      \alpha_2 {\bf u}_2(\ell)}$ using A) and upper bounding $\| {\bf y}
    \|$ and $\| \alpha_1 {\bf u}_1(j) + \alpha_2 {\bf u}_2(\ell) \|$
    using B) and C) respectively. This, yields that for every $({\bf
      s}_1, {\bf s}_2, \mathcal{C}_1, \mathcal{C}_2, {\bf z}) \in
    \mathcal{E}_{\bf X}^c \cap \mathcal{E}_{\bf Z}^c$,
    \begin{IEEEeqnarray*}{l}
      \cos \sphericalangle ({\bf y}, \alpha_1 {\bf u}_1(j)+\alpha_2 {\bf
        u}_2(\ell)) \\
      \hspace{20mm} \geq \frac{P_1 + P_2 +
        2\tilde{\rho}\sqrt{P_1P_2} - \xi_1\epsilon}{\sqrt{P_1 + P_2 +
          2\tilde{\rho} \sqrt{P_1P_2} +N + \xi_2\epsilon} \sqrt{P_1 + P_2
          + 2\tilde{\rho}\sqrt{P_1P_2} + \xi_3\epsilon}}\\[3mm]
      \hspace{20mm} \geq \sqrt{\frac{P_1 + P_2 + 2\tilde{\rho}\sqrt{P_1P_2} -
          \xi'\epsilon}{P_1 + P_2 + 2\tilde{\rho} \sqrt{P_1P_2} +N
          + \xi_2\epsilon}}\\[3mm]
      \hspace{20mm} = \Lambda (\epsilon).\\
    \end{IEEEeqnarray*}}
  Lemma \ref{lm:ub-E43-B43} now follows by D) which gives
  \begin{IEEEeqnarray*}{rCl}
    \mathcal{E}_{(\hat{\bf U}_1,\hat{\bf U}_2)} \cap \mathcal{E}_{\bf
      Z}^c \cap \mathcal{E}_{\bf S}^c \cap \mathcal{E}_{\bf X}^c &
    \subseteq & \mathcal{E}_{(\hat{\bf U}_1,\hat{\bf U}_2)}' \cap
    \mathcal{E}_{\bf Z}^c \cap \mathcal{E}_{\bf S}^c \cap
    \mathcal{E}_{\bf X}^c,
  \end{IEEEeqnarray*}
  and therefore
  \begin{IEEEeqnarray*}{rCl}
    \qquad \qquad \; \; \; \; \; \Prv{\mathcal{E}_{(\hat{\bf
          U}_1,\hat{\bf U}_2)} | \mathcal{E}_{\bf Z}^c \cap
      \mathcal{E}_{\bf S}^c \cap \mathcal{E}_{\bf X}^c} & \leq &
    \Prv{\mathcal{E}_{(\hat{\bf U}_1,\hat{\bf U}_2)}' |
      \mathcal{E}_{\bf Z}^c \cap \mathcal{E}_{\bf S}^c \cap
      \mathcal{E}_{\bf X}^c}. \qquad \qquad \; \; \; \; \;
    \qedhere
  \end{IEEEeqnarray*}
\end{proof}

We now state the second lemma needed for the proof of
\eqref{eq:Pr-E32-E33-E1c-E2c}.

\begin{lm}\label{lm:caps-double-error}
  For every $\Theta \in (0,1]$ and $\Delta \in (0,1]$, let the set
  $\mathcal{G}$ be given by
  \begin{IEEEeqnarray*}{rCl}
    \mathcal{G} & = \big\{ ({\bf s}_1, {\bf s}_2, \mathcal{C}_1,
    \mathcal{C}_2, {\bf z}): & \exists {\bf u}_1(j) \in \mathcal{C}_1
    \setminus \left\{ {\bf u}_1^{\ast} \right\}, {\bf u}_2(\ell) \in
    \mathcal{C}_2 \setminus \left\{ {\bf u}_2^{\ast} \right\} \text{ s.t.~}\\
      & & \cos \sphericalangle
      ({\bf u}_1(j), {\bf u}_2(\ell)) \geq
      \Theta, \cos \sphericalangle ({\bf y}, \alpha_1 {\bf u}_1(j) +
      \alpha_2{\bf u}_2(\ell)) \geq \Delta \big\}.\label{eq:pr1-wrong-cdw-pair}
  \end{IEEEeqnarray*}
  Then,
  \begin{IEEEeqnarray}{C}
    \left(R_1 + R_2 < - \frac{1}{2} \log_2 \left( (1-\Theta^2)
      (1-\Delta^2)\right) \right) \hspace{50mm} \nonumber\\
  \hspace{30mm}\Rightarrow \quad \left(\lim_{n \rightarrow \infty} \Prv{\mathcal{G} |
      \mathcal{E}_{{\bf X}_1}^c \cap \mathcal{E}_{{\bf X}_2}^c} = 0,
    \quad \epsilon > 0 \right).
  \end{IEEEeqnarray}
\end{lm}

\begin{proof}
  The proof follows from upper bounding in every point on
  $\mathcal{S}_i$, $i \in \{ 1,2 \}$, the density of every ${\bf
    u}_i(j) \in \mathcal{C}_i \setminus \{ {\bf u}_i^{\ast} \}$ and
  then using a standard argument from sphere-packing. The proof is
  given in Section \ref{app:proof-dec-double-error}.
\end{proof}

\begin{proof}[Proof of Lemma \ref{lm:vq-dec-error-proba}]
  We first prove \eqref{eq:Pr-E32-E33c-E1c-E2c}.
  \begin{IEEEeqnarray}{rCl}
    \Prv{\mathcal{E}_{\hat{\bf U}_1} \cap \mathcal{E}_{\bf Z}^c \cap \mathcal{E}_{\bf S}^c
      \cap \mathcal{E}_{\bf X}^c} & \stackrel{a)}{\leq} &
    \Prv{\mathcal{E}_{\hat{\bf U}_1}' \cap \mathcal{E}_{\bf Z}^c \cap
      \mathcal{E}_{\bf S}^c \cap \mathcal{E}_{\bf X}^c}\nonumber\\
    & \stackrel{b)}{\leq} & \Prv{\mathcal{E}_{\hat{\bf U}_1}' \Big|
      \mathcal{E}_{{\bf X}_1}^c },\label{eq:E41-B41}
  \end{IEEEeqnarray}
  where $a)$ follows by Lemma \ref{lm:ub-E41-B41} and $b)$ follows
  because $\mathcal{E}_{{\bf X}}^c \subseteq \mathcal{E}_{{\bf X}_1}^c$. The proof of
  \eqref{eq:Pr-E32-E33c-E1c-E2c} is now completed by combining
  \eqref{eq:E41-B41} with Lemma \ref{lm:caps-single-error}. This gives
  that for every $\delta > 0$ and every $\epsilon > 0$ there exists
  some $n'_{41}(\delta,\epsilon)$ such that for all $n >
  n'_{41}(\delta,\epsilon)$ we have $\Prv{ \mathcal{E}_{\hat{\bf U}_1} \cap
    \mathcal{E}_{\bf Z}^c \cap \mathcal{E}_{\bf S}^c \cap \mathcal{E}_{\bf X}^c } <
  \delta$ whenever
  \begin{IEEEeqnarray*}{rCl}
    R_1 & < & - \frac{1}{2} \log_2 \left( \frac{N
        (1-\tilde{\rho}^2)}{P_1 (1-\tilde{\rho}^2) + N} + \kappa''
      \epsilon \right)\\[2mm]
    & \leq & \frac{1}{2} \log_2 \left( \frac{P_1 (1-\tilde{\rho}^2) + N}{N
        (1-\tilde{\rho}^2)} - \kappa_1 \epsilon \right),
  \end{IEEEeqnarray*}
  where $\kappa_1$ is a positive constant determined by $P_1$, $P_2$,
  $N$, $\zeta_1$ and $\zeta_2$. A similar argument establishes
  \eqref{eq:Pr-E33-E32c-E1c-E2c}.


  We turn to the proof of \eqref{eq:Pr-E32-E33-E1c-E2c}.
  \begin{IEEEeqnarray}{rCl}
    \Prv{\mathcal{E}_{(\hat{\bf U}_1,\hat{\bf U}_2)} \cap
      \mathcal{E}_{\bf Z}^c \cap \mathcal{E}_{\bf S}^c \cap
      \mathcal{E}_{\bf X}^c} & \stackrel{a)}{\leq} &
    \Prv{\mathcal{E}_{(\hat{\bf U}_1,\hat{\bf U}_2)}' \cap
      \mathcal{E}_{\bf Z}^c \cap
      \mathcal{E}_{\bf S}^c \cap \mathcal{E}_{\bf X}^c}\nonumber\\
    & \stackrel{b)}\leq & \Prv{\mathcal{E}_{(\hat{\bf U}_1,\hat{\bf U}_2)}' |
      \mathcal{E}_{{\bf X}_1}^c \cap \mathcal{E}_{{\bf
          X}_2}^c},\label{eq:E43-decomp}
  \end{IEEEeqnarray}
  where $a)$ follows by Lemma \ref{lm:ub-E43-B43} and $b)$ follows
  because $\mathcal{E}_{{\bf X}}^c \subseteq \left( \mathcal{E}_{{\bf
        X}_1}^c \cap \mathcal{E}_{{\bf X}_2}^c\right)$. The proof of
  \eqref{eq:Pr-E32-E33-E1c-E2c} is now completed by combining
  \eqref{eq:E43-decomp} with Lemma \ref{lm:caps-double-error}, which
  gives that for every $\delta > 0$ and every $\epsilon > 0$ there
  exists some $n'_{43}(\delta,\epsilon)$ such that for all $n >
  n'_{43}(\delta,\epsilon)$ we have $\Prv{ \mathcal{E}_{(\hat{\bf
        U}_1,\hat{\bf U}_2)} \cap \mathcal{E}_{\bf Z}^c \cap
    \mathcal{E}_{\bf S}^c \cap \mathcal{E}_{\bf X}^c } < \delta$
  whenever
  \begin{IEEEeqnarray*}{rCl}
    R_1 + R_2 & < & \frac{1}{2} \log_2 \left( \frac{P_1 + P_2 +
        2\tilde{\rho} \sqrt{P_1P_2} +N +\xi_2 \epsilon}{ \left( N +
          (\xi' + \xi_2) \epsilon \right) \left( 1- \tilde{\rho}^2
          +\xi' \epsilon \right)} \right) \\[3mm]
    & \leq & \frac{1}{2} \log_2 \left( \frac{P_1 + P_2 + 2
        \tilde{\rho}\sqrt{P_1P_2}+N }{ N (1-\tilde{\rho}^2)} -
      \kappa_3 \epsilon \right),
  \end{IEEEeqnarray*}
  where $\kappa_3$ is is a positive constant determined by $P_1$,
  $P_2$ and $N$.
\end{proof}

The proof of Lemma \ref{lm:vq-Pr-E4} now follows straight forwardly.

\begin{proof}[Proof of Lemma \ref{lm:vq-Pr-E4}]
  Combining \eqref{eq:vq-Pr-E4-decomp} with Lemma \ref{lm:vq-Pr-E1},
  Lemma \ref{lm:vq-Pr-E3}, Lemma \ref{lm:vq-Pr-E2} and Lemma
  \ref{lm:vq-dec-error-proba}, yields that for every $\delta>0$ and
  $0.3 > \epsilon > 0$ there exists some $n_4'(\delta,\epsilon) \in
  \Naturals$ such that for all $n > n_4'(\delta,\epsilon)$
  \begin{IEEEeqnarray*}{rCl}
    \hspace{48mm}
    \Prv{\mathcal{E}_{\hat{\bf U}}} \leq 11 \delta & \qquad \qquad & \text{if }
    (R_1,R_2) \in \mathcal{R}(\epsilon). \hspace{20mm} \qedhere
  \end{IEEEeqnarray*}
\end{proof}


\subsubsection{Concluding the Proof of Proposition
  \ref{prp:vq-D1-eql-genie}} \label{sec:prf-prp-genie} 

We start with four lemmas. The first lemma upper bounds the impact of
atypical source outputs on the expected distortion.

\begin{lm}\label{lm:vq-D-E1}
  For every $\epsilon > 0$
  \begin{IEEEeqnarray*}{rCl}
    \frac{1}{n} \E{\|{\bf S}_1\|^2 \Big| \mathcal{E}_{\bf S}}
    \Prv{\mathcal{E}_{\bf S}} & \leq & \sigma^2 \left( \epsilon +
      \Prv{\mathcal{E}_{\bf S}} \right).
  \end{IEEEeqnarray*}
\end{lm}

\begin{proof}
  \begin{IEEEeqnarray*}{rCll}
    \hspace{19mm} \frac{1}{n} \E{\|{\bf S}_1\|^2 \Big| \mathcal{E}_{\bf S}}
    \Prv{\mathcal{E}_{\bf S}} & = & \frac{1}{n} \E{\|{\bf S}_1\|^2 } -
    \frac{1}{n} \E{\|{\bf S}_1\|^2 \Big| \mathcal{E}_{\bf S}^c}
    \Prv{\mathcal{E}_{\bf S}^c}\\[1mm]
    & \leq & \sigma^2 - \sigma^2 (1-\epsilon)
    \Prv{\mathcal{E}_{\bf S}^c} & \\[1mm]
    & = & \sigma^2 - \sigma^2 (1-\epsilon) \left( 1 -
      \Prv{\mathcal{E}_{\bf S}} \right) & \\[1mm]
    & = & \sigma^2 \epsilon + \sigma^2 (1-\epsilon)
    \Prv{\mathcal{E}_{\bf S}} & \\[1mm]
    & \leq & \sigma^2 \left( \epsilon + \Prv{\mathcal{E}_{\bf S}} \right). &
    \hspace{19mm} \qedhere
  \end{IEEEeqnarray*}
\end{proof}

The second lemma gives upper bounds on norms related to the
reconstructions $\hat{\bf s}_1$ and $\hat{\bf s}_1^{\textnormal{G}}$.

\begin{lm}\label{lm:vq-ub-norm-s1h}
  Let the reconstructions $\hat{\bf s}_1$ and $\hat{\bf
    s}_1^{\textnormal{G}}$ be as defined in \eqref{eq:lin-est-S1h} and
  \eqref{eq:vq-genie-lin-est-S1h}. Then,
  \begin{IEEEeqnarray*}{rCl}
    \|\hat{\bf s}_1\|^2 \leq 4 n \sigma^2 \qquad & \qquad
    \|\hat{\bf s}_1^{\textnormal{G}}\|^2 \leq 4 n \sigma^2
    \qquad & \qquad \|\hat{\bf s}_1^{\textnormal{G}} - \hat{\bf s}_1\|^2
    \leq 16 n \sigma^2.
  \end{IEEEeqnarray*}
\end{lm}

\begin{proof}
  We start by upper bounding the squared norm of $\hat{\bf s}_1$
  \begin{IEEEeqnarray*}{rCl}
    \|\hat{\bf s}_1\|^2 & = & \| \gamma_{11} \hat{\bf u}_1 + \gamma_{12}
    \hat{\bf u}_2 \|^2\\
    & = & \gamma_{11}^2 \|\hat{\bf u}_1\|^2 + 2 \gamma_{11} \gamma_{12}
    \inner{\hat{\bf u}_1}{\hat{\bf u}_2} + \gamma_{12}^2 \|\hat{\bf
      u}_2\|^2\\
    & \leq & \gamma_{11}^2 \|\hat{\bf u}_1\|^2 + 2 \gamma_{11} \gamma_{12}
    \|\hat{\bf u}_1\| \|\hat{\bf u}_2\| + \gamma_{12}^2 \|\hat{\bf
      u}_2\|^2\\
    & = & \left( \gamma_{11} \|\hat{\bf u}_1\| +\gamma_{12} \|\hat{\bf
        u}_2\| \right)^2\\
    & \stackrel{a)}{\leq} & n \sigma^2 (1+\rho)^2\\
    & \leq & 4 n \sigma^2,
  \end{IEEEeqnarray*}
  where in $a)$ we have used \eqref{eq:vq-bounds-gamma}, i.e., that
  $\gamma_{11} < 1$ and $\gamma_{12} < \rho$, and that $\|\hat{\bf
    u}_i\| \leq \sqrt{n\sigma^2}$, $i \in \{ 1,2 \}$. The upper bound
  on the squared norm of $\hat{\bf s}_1^{\textnormal{G}}$ is obtained
  similarly. Its proof is therefore omitted. The upper bound on the
  squared norm of the difference between $\hat{\bf s}_1$ and $\hat{\bf
    s}_1^{\textnormal{G}}$ now follows easily:
  \begin{IEEEeqnarray*}{rCll}
      \qquad \qquad \qquad \qquad \quad \|\hat{\bf
        s}_1^{\textnormal{G}} - \hat{\bf s}_1\|^2 & \leq & \|\hat{\bf
        s}_1^{\textnormal{G}}\|^2 + 2 \|\hat{\bf
        s}_1^{\textnormal{G}}\| \|\hat{\bf s}_1\| + \|\hat{\bf s}_1\|^2 & \\
      & = & \left( \|\hat{\bf s}_1^{\textnormal{G}}\| + \|\hat{\bf s}_1\|
      \right)^2 & \\
      & \leq & 16 n \sigma^2. & \qquad \qquad \qquad \qquad
      \quad \qedhere
    \end{IEEEeqnarray*}
\end{proof}

The next two lemmas are used directly in the upcoming proof of
Proposition \ref{prp:vq-D1-eql-genie}. They rely on Lemma
\ref{lm:vq-D-E1} and Lemma \ref{lm:vq-ub-norm-s1h}.

\begin{lm}\label{lm:vq-bd-inner-S1-S1h-S1G}
  \begin{IEEEeqnarray*}{rCl}
    \frac{1}{n} \E{\inner{{\bf S}_1}{\hat{\bf S}_1^{\textnormal{G}} -
        \hat{\bf S}_1}} & \leq & \sigma^2 \left( \epsilon + 17
      \Prv{\mathcal{E}_{\bf S}} + 4 \sqrt{1+\epsilon} \Prv{\mathcal{E}_{\hat{\bf U}}}
    \right).
  \end{IEEEeqnarray*}
\end{lm}

\begin{proof}
  
  \begin{IEEEeqnarray}{rCl}
    \frac{1}{n} \E{\inner{{\bf S}_1}{\hat{\bf S}_1^{\textnormal{G}}
        - \hat{\bf S}_1}} & = & \frac{1}{n} \E{\inner{{\bf
          S}_1}{\hat{\bf S}_1^{\textnormal{G}} - \hat{\bf S}_1}
      \Big| \mathcal{E}_{\bf S}}
    \Prv{\mathcal{E}_{\bf S}} \nonumber\\
    & & {} + \frac{1}{n} \E{\inner{{\bf S}_1}{\hat{\bf
          S}_1^{\textnormal{G}} - \hat{\bf S}_1} \Big|
      \mathcal{E}_{\bf S}^c
      \cap \mathcal{E}_{\hat{\bf U}}} \Prv{\mathcal{E}_{\bf S}^c \cap
      \mathcal{E}_{\hat{\bf U}}} \nonumber\\ 
    & & {} + \frac{1}{n} \underbrace{\E{\inner{{\bf S}_1}{\hat{\bf
          S}_1^{\textnormal{G}} - \hat{\bf S}_1} \Big|
      \mathcal{E}_{\bf S}^c \cap \mathcal{E}_{\hat{\bf U}}^c}}_{=0} \Prv{\mathcal{E}_{\bf S}^c
      \cap \mathcal{E}_{\hat{\bf U}}^c} \nonumber\\[3mm]
    & \stackrel{a)}{\leq} & \frac{1}{n} \E{\|{\bf S}_1\|^2 +
      \|\hat{\bf S}_1^{\textnormal{G}} - \hat{\bf S}_1\|^2 \Big|
      \mathcal{E}_{\bf S}} \Prv{\mathcal{E}_{\bf S}} \nonumber\\
    & & {} + \frac{1}{n} \E{\|{\bf S}_1\| \|\hat{\bf
        S}_1^{\textnormal{G}} - \hat{\bf S}_1\| \Big|
      \mathcal{E}_{\bf S}^c \cap \mathcal{E}_{\hat{\bf U}}}
    \Prv{\mathcal{E}_{\hat{\bf U}}} \nonumber\\[3mm]
    & \stackrel{b)}{\leq} & \frac{1}{n} \E{\|{\bf S}_1\|^2 \Big| \mathcal{E}_{\bf S}}
    \Prv{\mathcal{E}_{\bf S}} + 16 \sigma^2 \Prv{\mathcal{E}_{\bf S}} \nonumber\\
    & & {} + \sqrt{\sigma^2(1+\epsilon)} \sqrt{16\sigma^2}
    \Prv{\mathcal{E}_{\hat{\bf U}}} \nonumber\\[3mm]
    & \stackrel{c)}{\leq} & \sigma^2 (\epsilon + \Prv{\mathcal{E}_{\bf S}})
    + 16 \sigma^2 \Prv{\mathcal{E}_{\bf S}} \nonumber\\
    & & {} + 4 \sigma^2 \sqrt{1+\epsilon}
    \Prv{\mathcal{E}_{\hat{\bf U}}} \nonumber\\[3mm]
    & \leq & \sigma^2 \left( \epsilon + 17 \Prv{\mathcal{E}_{\bf S}} + 4
      \sqrt{1+\epsilon} \Prv{\mathcal{E}_{\hat{\bf U}}} \right). \label{eq:vq-inner-genie}
  \end{IEEEeqnarray}
  In the first equality the third expectation term equals zero because
  by $\mathcal{E}_{\hat{\bf U}}^c$ we have $\| \hat{\bf
    s}_1^{\textnormal{G}} - \hat{\bf s}_1 \| = 0$ and by
  $\mathcal{E}_{\bf S}^c$ the norm $\|{\bf s}_1\|$ is bounded.
  In $a)$ we have used two inequalities: in the first term, the
  inner product is upper bounded by using that for any two vectors
  ${\bf v} \in \Reals^n$ and ${\bf w} \in \Reals^n$
  \begin{IEEEeqnarray}{rCl}\label{eq:vq-bound inner}
    |\inner{{\bf v}}{{\bf w}}| & \leq & \| {\bf v} \| \cdot \| {\bf w}
    \| \nonumber\\
    & \leq & \frac{1}{2} \left( \| {\bf v} \|^2 + \| {\bf w} \|^2
    \right) \nonumber\\
    & \leq & \| {\bf v} \|^2 + \| {\bf w} \|^2.
  \end{IEEEeqnarray}
  The second term is upper bounded by the Cauchy-Schwarz inequality
  and by $\Prv{\mathcal{E}_{\bf S}^c \cap \mathcal{E}_{\hat{\bf U}}} \leq
  \Prv{\mathcal{E}_{\hat{\bf U}}}$.  In $b)$ we have used
  Lemma~\ref{lm:vq-ub-norm-s1h} and in $c)$ we have used
  Lemma~\ref{lm:vq-D-E1}.
\end{proof}

\begin{lm}\label{lm:vq-bd-S1h-S1G}
  \begin{IEEEeqnarray*}{rCl}
    \frac{1}{n} \E{\|\hat{\bf S}_1\|^2 - \|\hat{\bf
        S}_1^{\textnormal{G}}\|^2} & \leq & 8 \sigma^2
    \Prv{\mathcal{E}_{\hat{\bf U}}}.
  \end{IEEEeqnarray*}
\end{lm}

\begin{proof}
  \begin{IEEEeqnarray*}{rCl}
    \frac{1}{n} \E{\|\hat{\bf S}_1\|^2 - \|\hat{\bf
        S}_1^{\textnormal{G}}\|^2} & = & \frac{1}{n} \E{\|\hat{\bf S}_1\|^2
      - \|\hat{\bf S}_1^{\textnormal{G}}\|^2 \big| \mathcal{E}_{\hat{\bf U}}}
    \Prv{\mathcal{E}_{\hat{\bf U}}}\\
    & & {} + \frac{1}{n} \E{\|\hat{\bf S}_1\|^2 - \|\hat{\bf
        S}_1^{\textnormal{G}}\|^2 \big| \mathcal{E}_{\hat{\bf U}}^c}
    \Prv{\mathcal{E}_{\hat{\bf U}}^c}\\
    & \stackrel{a)}{\leq} & \frac{1}{n} \E{\|\hat{\bf S}_1\|^2 + \|\hat{\bf
        S}_1^{\textnormal{G}}\|^2 \big| \mathcal{E}_{\hat{\bf U}}}
    \Prv{\mathcal{E}_{\hat{\bf U}}}\\[2mm]
    & \stackrel{b)}{\leq} & 8 \sigma^2 \Prv{\mathcal{E}_{\hat{\bf U}}},
  \end{IEEEeqnarray*}
  where $a)$ follows since conditional on $\mathcal{E}_{\hat{\bf U}}^c$ we have
  $\hat{\bf s}_1 = \hat{\bf s}_1^{\textnormal{G}}$ and therefore
  $\|\hat{\bf s}_1\|^2 - \|\hat{\bf s}_1^{\textnormal{G}}\|^2 = 0$,
  and where $b)$ follows by Lemma \ref{lm:vq-ub-norm-s1h}.
\end{proof}


\begin{proof}[Proof of Proposition \ref{prp:vq-D1-eql-genie}]
  
  We show that the asymptotic normalized distortion resulting from the
  proposed vector-quantizer scheme is the same as the asymptotic
  normalized distortion resulting from the genie-aided version of this
  scheme.
  
  \begin{IEEEeqnarray}{l}
    \frac{1}{n}\E{\| {\bf S}_1 - \hat{\bf S}_1 \|^2} -
    \frac{1}{n}\E{\| {\bf S}_1 - \hat{\bf S}_1^{\textnormal{G}} \|^2} \nonumber\\[2mm]
    \hspace{35mm} = \frac{1}{n} \bigg( \E{\|{\bf S}_1\|^2} - 2 \E{\inner{{\bf S}_1}{\hat{\bf
          S}_1}} + \E{\|\hat{\bf S}_1\|^2} \nonumber\\
    \hspace{44mm} {} - \E{\|{\bf S}_1\|^2}  + 2 \E{\inner{{\bf S}_1}{\hat{\bf
          S}_1^{\textnormal{G}}}} - \E{\|\hat{\bf
        S}_1^{\textnormal{G}}\|^2} \bigg) \nonumber\\[2mm]
    \hspace{35mm} = 2 \frac{1}{n} \E{\inner{{\bf S}_1}{\hat{\bf
          S}_1^{\textnormal{G}} - \hat{\bf S}_1}} + \frac{1}{n}
    \E{\|\hat{\bf S}_1\|^2 - \|\hat{\bf S}_1^{\textnormal{G}}\|^2} \nonumber\\[2mm]
    \hspace{35mm} \stackrel{a)}{\leq} 2 \sigma^2 \left( \epsilon + 17
      \Prv{\mathcal{E}_{\bf S}} + 4 \sqrt{1+\epsilon} \Prv{\mathcal{E}_{\hat{\bf U}}}
    \right) \nonumber\\
    \hspace{39mm} {} + 8 \sigma^2 \Prv{\mathcal{E}_{\hat{\bf U}}} \nonumber\\[2mm]
    \hspace{35mm} = 2 \sigma^2 \left( \epsilon + 17 \Prv{\mathcal{E}_{\bf S}} + 4
      \left( \sqrt{1+\epsilon} + 1 \right) \Prv{\mathcal{E}_{\hat{\bf U}}}
    \right), \qquad \label{eq:vq-diff-genie}
  \end{IEEEeqnarray}
  where in step $a)$ we have used Lemma
  \ref{lm:vq-bd-inner-S1-S1h-S1G} and Lemma
  \ref{lm:vq-bd-S1h-S1G}. Combining \eqref{eq:vq-diff-genie} with
  Lemma \ref{lm:vq-Pr-E1} and Lemma \ref{lm:vq-Pr-E4} gives that for
  every $\delta > 0$ and $0.3 > \epsilon > 0$, there exists an
  $n'(\delta,\epsilon) > 0$ such that for all $(R_1,R_2) \in
  \mathcal{R}(\epsilon)$ and $n > n'(\delta,\epsilon)$
  \begin{IEEEeqnarray*}{rCl}
    \hspace{12mm} \frac{1}{n}\E{\| {\bf S}_1 - \hat{\bf S}_1 \|^2} -
    \frac{1}{n}\E{\| {\bf S}_1 - \hat{\bf S}_1^{\textnormal{G}} \|^2}
    & < & 2 \sigma^2 \left( \epsilon + \left(
        44 \sqrt{1+\epsilon} + 61 \right) \delta
    \right). \hspace{12mm} \qedhere
  \end{IEEEeqnarray*}
\end{proof}


\subsection{Upper Bound on Expected Distortion}\label{sec:vq-ub-D1}

We now derive an upper bound on the achievable distortion for the
proposed vector-quantizer scheme. By Corollary \ref{cor:vq-genie}, it
suffices to analyze the genie-aided scheme. Using that $\hat{\bf
  S}_1^{\textnormal{G}} = \gamma_{11} {\bf U}_1^{\ast} + \gamma_{12}
{\bf U}_2^{\ast}$, we have
\begin{IEEEeqnarray}{rCl}\label{eq:distortion-general}
  \frac{1}{n} \E{\| {\bf S}_1 - \hat{\bf S}_1^{\textnormal{G}} \|^2} &
  = & \frac{1}{n} \Big( \E{\| {\bf S}_1 \|^2} -2\gamma_{11}
  \E{\inner{{\bf S}_1}{{\bf U}_1^{\ast}}} -2\gamma_{12} \E{\inner{{\bf
        S}_1}{{\bf U}_2^{\ast}}} \nonumber \\
  & & \quad \; + \gamma_{11}^2 \E{\| {\bf U}_1^{\ast} \|^2} + 2
  \gamma_{11} \gamma_{12} \E{\inner{{\bf U}_1^{\ast}}{{\bf U}_2^{\ast}}} +
  \gamma_{12}^2 \E{\| {\bf U}_2^{\ast} \|^2} \Big) \nonumber \\
  & = & \sigma^2 -2\gamma_{11} \frac{1}{n}\E{\inner{{\bf
        S}_1}{{\bf U}_1^{\ast}}} -2\gamma_{12} \frac{1}{n}\E{\inner{{\bf
        S}_1}{{\bf U}_2^{\ast}}} \nonumber \\
  & & \quad \; + \gamma_{11}^2 \sigma^2(1-2^{-2R_1}) + 2 \gamma_{11}
  \gamma_{12} \frac{1}{n}\E{\inner{{\bf U}_1^{\ast}}{{\bf U}_2^{\ast}}}\nonumber\\
  & & \quad + \gamma_{12}^2 \sigma^2(1-2^{-2R_2}),
\end{IEEEeqnarray}
where in the last equality all expected squared norms have been
replaced by their explicit values, i.e., $\E{\| {\bf S}_1 \|^2} =
n\sigma^2$ and $\E{\| {\bf U}_i \|^2} = n\sigma^2 (1-2^{-2R_i})$ for
$i \in \{1,2\}$. The remaining expectations of the inner products are
bounded in the following three lemmas.

\begin{lm}\label{lm1:vq-bd-D1genie}
  For every $\delta > 0$ and $0.3 > \epsilon > 0$ and every positive integer $n$
  \begin{IEEEeqnarray}{rCl}
    \frac{1}{n} \E{\inner{{\bf S}_1}{{\bf U}_1^{\ast}}} & \geq &
    \sigma^2 (1-2^{-2R_1})(1-2\epsilon) (1 - 7\delta ).
  \end{IEEEeqnarray}
\end{lm}

\begin{proof}
  \begin{IEEEeqnarray*}{rCl}
    \frac{1}{n} \E{\inner{{\bf S}_1}{{\bf U}_1^{\ast}}} & = &
    \frac{1}{n} \underbrace{\E{\|{\bf S}_1\| \|{\bf U}_1^{\ast}\|
        \cos \sphericalangle ({\bf S}_1,{\bf U}_1^{\ast}) \Big|
        \mathcal{E}_{\bf S} \cup \mathcal{E}_{\bf X}}}_{\geq 0} \cdot
    \Prv{\mathcal{E}_{\bf S} \cup \mathcal{E}_{\bf X}}\\
    & & {} + \frac{1}{n} \E{\|{\bf S}_1\| \|{\bf U}_1^{\ast}\| \cos
      \sphericalangle ({\bf S}_1,{\bf U}_1^{\ast})
      \Big| \mathcal{E}_{\bf S}^c \cap \mathcal{E}_{\bf X}^c} \cdot
    \Prv{\mathcal{E}_{\bf S}^c \cap \mathcal{E}_{\bf X}^c}\\[3mm]
    & \geq & \frac{1}{n} \sqrt{n \sigma^2(1-\epsilon)}
      \sqrt{n\sigma^2(1-2^{-2R_1})} \sqrt{1-2^{-2R_1}}(1-\epsilon)
    \Prv{\mathcal{E}_{\bf S}^c \cap \mathcal{E}_{\bf X}^c}\\[3mm]
    & \geq & \sigma^2 (1-2^{-2R_1})(1-\epsilon)^2 \left( 1 -
      \Prv{\mathcal{E}_{\bf S} \cup \mathcal{E}_{\bf X}} \right)\\
    & \geq & \sigma^2 (1-2^{-2R_1})(1-2\epsilon) \left( 1 -
      \Prv{\mathcal{E}_{\bf S}} - \Prv{\mathcal{E}_{\bf X}} \right),
  \end{IEEEeqnarray*}
  where in the first equality the first expectation term is
  non-negative because conditioned on $\mathcal{E}_{\bf X}$ either
  ${\bf U}_1^{\ast} = {\bf 0}$ or, if ${\bf U}_1^{\ast} \neq {\bf 0}$,
  then $\cos \left( \sphericalangle ({\bf S}_1,{\bf U}_1^{\ast})
  \right) > 0$.

  By Lemma \ref{lm:vq-Pr-E1} and Lemma \ref{lm:vq-Pr-E2} it now
  follows that for every $\delta > 0$ and $0.3 > \epsilon > 0$ there
  exists an $n'(\delta,\epsilon) \in \Naturals$ such that for all $n >
  n'(\delta,\epsilon)$
  \begin{IEEEeqnarray*}{rCl}
    \hspace{30mm} \frac{1}{n} \E{\inner{{\bf S}_1}{{\bf
          U}_1^{\ast}}} & \geq & \sigma^2 (1-2^{-2R_1})(1-2\epsilon) (1
    - 7\delta ). \hspace{30mm} \qedhere
  \end{IEEEeqnarray*}
\end{proof}

\begin{lm}\label{lm2:vq-bd-D1genie}
  For every $\delta > 0$ and $0.3 > \epsilon > 0$, there exists an
  $n_2'(\delta,\epsilon) \in \Naturals$ such that for all $n > n_2'(\delta,\epsilon)$
  \begin{IEEEeqnarray*}{rCl}
    \frac{1}{n} \E{\inner{{\bf U}_1^{\ast}}{{\bf U}_2^{\ast}}} & \leq
    &\sigma^2 6 \delta + \sigma^2 \rho (1-2^{-2R_1})
    (1-2^{-2R_2}) (1+7\epsilon).
  \end{IEEEeqnarray*}
\end{lm}

\begin{proof}
  \begin{IEEEeqnarray*}{rCl}
    \frac{1}{n} \E{\inner{{\bf U}_1^{\ast}}{{\bf U}_2^{\ast}}} & = &
    \frac{1}{n} \E{\inner{{\bf U}_1^{\ast}}{{\bf U}_2^{\ast}} \big| \mathcal{E}_{\bf X}}
    \Prv{\mathcal{E}_{\bf X}} + \frac{1}{n} \E{\inner{{\bf U}_1^{\ast}}{{\bf
          U}_2^{\ast}} \big| \mathcal{E}_{\bf X}^c} \Prv{\mathcal{E}_{\bf X}^c}\\[2mm]
    & \leq & \frac{1}{n} \E{\|{\bf U}_1^{\ast}\| \|{\bf U}_2^{\ast}\| \big|
      \mathcal{E}_{\bf X}} \Prv{\mathcal{E}_{\bf X}} + \frac{1}{n} \E{\inner{{\bf
          U}_1^{\ast}}{{\bf U}_2^{\ast}} \big| \mathcal{E}_{\bf X}^c}\\[2mm]
    & \leq & \sigma^2 \sqrt{(1-2^{-2R_1})(1-2^{-2R_2})}
    \Prv{\mathcal{E}_{\bf X}}\\
    & & {} + \frac{1}{n} \E{ \tilde{\rho}(1+7\epsilon) \sqrt{n \sigma^2
        (1-2^{-2R_1})} \sqrt{n \sigma^2 (1-2^{-2R_2})} \bigg|
      \mathcal{E}_{\bf X}^c}\\[2mm]
    & \leq & \sigma^2 \Prv{\mathcal{E}_{\bf X}} + \sigma^2 \rho
    (1-2^{-2R_1}) (1-2^{-2R_2}) (1+7\epsilon).
  \end{IEEEeqnarray*}
  Thus, it follows by Lemma \ref{lm:vq-Pr-E2} that for every $\delta >
  0$ and $0.3 > \epsilon > 0$ there exists an $n_2'(\delta,\epsilon)
  \in \Naturals$ such that for all $n > n_2'(\delta,\epsilon)$
  \begin{IEEEeqnarray*}{rCl}
    \hspace{21mm} \E{\inner{{\bf U}_1^{\ast}}{{\bf U}_2^{\ast}}} & \leq & \sigma^2
    6 \delta + \sigma^2 \rho (1-2^{-2R_1}) (1-2^{-2R_2})
    (1+7\epsilon). \hspace{21mm} \qedhere
  \end{IEEEeqnarray*}
\end{proof}

\begin{lm} \label{lm3:vq-bd-D1genie}
  For every $\delta > 0$ and $0.3 > \epsilon > 0$, there exists an
  $n'(\delta,\epsilon) \in \Naturals$ such that for all $n > n'(\delta,\epsilon)$
  \begin{IEEEeqnarray*}{rCl}
    \frac{1}{n} \E{\inner{{\bf S}_1}{{\bf U}_2^{\ast}}} & \geq &
    \sigma^2 \rho (1-2^{-2R_2}) (1-\epsilon)^3 - \sigma^2 \left(
      \epsilon + 21 \delta +  6 \delta \epsilon \right).
  \end{IEEEeqnarray*}
\end{lm}

\begin{proof}
  We begin with the following decomposition:
  \begin{IEEEeqnarray}{rCl}\label{eq:decmp-ES1U2*}
    \frac{1}{n} \E{\inner{{\bf S}_1}{{\bf U}_2^{\ast}}} & = &
    \frac{1}{n} \E{\inner{{\bf S}_1}{{\bf U}_2^{\ast}} \big|
      \mathcal{E}_{\bf S} \cup \mathcal{E}_{{\bf X}_2}}
    \Prv{\mathcal{E}_{\bf S} \cup \mathcal{E}_{{\bf X}_2}} \nonumber\\
    && {} + \frac{1}{n} \E{\inner{{\bf S}_1}{{\bf U}_2^{\ast}} \big|
      \mathcal{E}_{\bf S}^c \cap \mathcal{E}_{{\bf X}_2}^c}
    \Prv{\mathcal{E}_{\bf S}^c \cap \mathcal{E}_{{\bf X}_2}^c}.
  \end{IEEEeqnarray}
  The first term on the RHS of \eqref{eq:decmp-ES1U2*} is lower bounded as follows.
  \begin{IEEEeqnarray}{rCl}
    \frac{1}{n} \E{\inner{{\bf S}_1}{{\bf U}_2^{\ast}} \big|
      \mathcal{E}_{\bf S} \cup \mathcal{E}_{{\bf X}_2}} & &
    \Prv{\mathcal{E}_{\bf S} \cup \mathcal{E}_{{\bf X}_2}}\nonumber\\
    & & \stackrel{a)}{\geq} - \frac{1}{n} \E{\|{\bf S}_1\|^2 + \|{\bf U}_2^{\ast}\|^2
      \big| \mathcal{E}_{\bf S} \cup \mathcal{E}_{{\bf X}_2}} \Prv{\mathcal{E}_{\bf S} \cup
      \mathcal{E}_{{\bf X}_2}} \nonumber\\[2mm]
    & & \stackrel{b)}{\geq} - \frac{1}{n} \Big( \E{\|{\bf S}_1\|^2 \big|
      \mathcal{E}_{\bf S}} \Prv{\mathcal{E}_{\bf S}} \nonumber\\
    & & \qquad \quad {} + \E{\|{\bf S}_1\|^2 \big| \mathcal{E}_{\bf S}^c
      \cap \mathcal{E}_{{\bf X}_2}}
    \Prv{\mathcal{E}_{\bf S}^c \cap \mathcal{E}_{\bf X}} \nonumber\\[2mm]
    & & \qquad \quad {} + \|{\bf U}_2^{\ast}\|^2 (\Prv{\mathcal{E}_{\bf S}} +
    \Prv{\mathcal{E}_{\bf X}})\Big) \nonumber\\[2mm]
    & & \stackrel{c)}{\geq} - \Big( \sigma^2 \left( \epsilon + \Prv{\mathcal{E}_{\bf S}}
    \right) + \sigma^2 (1+\epsilon) \Prv{\mathcal{E}_{\bf X}} \nonumber\\[2mm]
    & & \hspace{9mm} {} + \sigma^2 (1-2^{-2R_2}) \left(
      \Prv{\mathcal{E}_{\bf S}} + \Prv{\mathcal{E}_{\bf X}}\right)\Big) \nonumber\\[3mm]
    & & \geq - \sigma^2 \big( \epsilon + 2\Prv{\mathcal{E}_{\bf S}} + (2 +
      \epsilon)\Prv{\mathcal{E}_{\bf X}} \big), \label{eq:bd-trm1-decmpES1U2*}
  \end{IEEEeqnarray}
  where in $a)$ we have used \eqref{eq:vq-bound inner}, in $b)$ we
  have used that $\mathcal{E}_{\bf X} \supseteq \mathcal{E}_{{\bf
      X}_2}$, and in $c)$ we have used Lemma \ref{lm:vq-D-E1}.

  We now turn to lower bounding the second term on the RHS of
  \eqref{eq:decmp-ES1U2*}. The probability term is lower bounded as follows
  \begin{IEEEeqnarray}{rCl}
    \Prv{\mathcal{E}_{\bf S}^c \cap \mathcal{E}_{{\bf X}_2}^c} & = & 1 -
    \Prv{\mathcal{E}_{\bf S} \cup \mathcal{E}_{{\bf X}_2}} \nonumber\\ 
    & \geq & 1 - \left( \Prv{\mathcal{E}_{\bf S}} +
      \Prv{\mathcal{E}_{\bf X}}\right). \label{eq:bd-trm21-decmpES1U2*}
  \end{IEEEeqnarray}
  To lower bound the expectation term, we represent ${\bf u}_i^{\ast}$
  as a scaled version of ${\bf s}_i$ corrupted by an additive
  ``quantization noise'' ${\bf v}_i$. More precisely,
  \begin{IEEEeqnarray}{rCl}\label{eq:decomp-u^ast}
    {\bf u}_i^{\ast} & = & \nu_i {\bf s}_i + {\bf v}_i \qquad \qquad
    \text{where} \qquad \nu_i = \frac{\| {\bf u}_i^{\ast} \|}{\| {\bf
        s}_i \|} \cos \sphericalangle ({\bf s}_i, {\bf u}_i^{\ast})
    \qquad i \in \{ 1,2 \}.
  \end{IEEEeqnarray}
  With this choice of $\nu_i$, the vector ${\bf v}_i$ is always
  orthogonal to ${\bf s}_i$. By \eqref{eq:decomp-u^ast}, the inner
  product $\inner{{\bf S}_1}{{\bf U}_2^{\ast}}$ can now be rewritten
  as $\nu_2 \inner{{\bf S}_1}{{\bf S}_2} + \inner{{\bf S}_1}{{\bf
      V}_2}$. Hence,
  \begin{IEEEeqnarray}{rl}
    \E{\inner{{\bf S}_1}{{\bf U}_2^{\ast}} \Big| \mathcal{E}_{\bf S}^c
      \cap
      \mathcal{E}_{{\bf X}_2}^c } & \nonumber\\
    & \hspace{-2cm} \stackrel{a)}{=} \textsf{E}_{{\bf S}_1,{\bf
        S}_2}\Bigg[ \textsf{E}_{\mathscr{C}_1,\mathscr{C}_2} \left[
      \nu_2 \inner{{\bf s}_1}{{\bf s}_2} \Big| ({\bf S}_1,{\bf S}_2) =
      ({\bf s}_1, {\bf s}_2), \mathcal{E}_{\bf S}^c \cap
      \mathcal{E}_{{\bf X}_2}^c \right] \nonumber\\
    & \hspace{-2cm} \qquad \qquad \; + \underbrace{
      \textsf{E}_{\mathscr{C}_1,\mathscr{C}_2} \left[ \inner{{\bf
            s}_1}{{\bf V}_2} \Big| ({\bf S}_1, {\bf S}_2) = ({\bf
          s}_1, {\bf s}_2), \mathcal{E}_{\bf S}^c \cap
        \mathcal{E}_{{\bf X}_2}^c \right]}_{=0} \Bigg] \nonumber\\[2mm]
    & \hspace{-2cm} = \textsf{E}_{{\bf S}_1,{\bf S}_2}\Bigg[ \frac{\|
      {\bf U}_2^{\ast} \|}{\| {\bf S}_2 \|} \inner{{\bf S}_1}{{\bf
        S}_2} \textsf{E}_{\mathscr{C}_1,\mathscr{C}_2} \left[ \cos
      \sphericalangle ({\bf s}_2, {\bf U}_2^{\ast}) \big| ({\bf S}_1,
      {\bf S}_2) = ({\bf s}_1, {\bf s}_2), \mathcal{E}_{\bf S}^c \cap
      \mathcal{E}_{{\bf X}_2}^c \right] \Bigg] \nonumber\\[2mm] 
    & \hspace{-2cm} \stackrel{b)}{\geq} \textsf{E}_{{\bf S}_1,{\bf
        S}_2}\Bigg[ \| {\bf U}_2^{\ast} \| \| {\bf S}_1 \| \cos \left(
      \sphericalangle ({\bf S}_1, {\bf S}_2) \right)
    \sqrt{1-2^{-2R_2}}(1-\epsilon) \Bigg|
    \mathcal{E}_{\bf S}^c \cap \mathcal{E}_{{\bf X}_2}^c \Bigg] \nonumber\\[2mm]
    & \hspace{-2cm} \stackrel{c)}{\geq} \sqrt{n \sigma^2
      (1-2^{-2R_2})} \sqrt{n \sigma^2 (1-\epsilon)} \rho (1-\epsilon)
    \sqrt{1-2^{-2R_2}}(1-\epsilon) \nonumber\\[2mm]
    & \hspace{-2cm} \geq n \rho \sigma^2 (1-2^{-2R_2})
    (1-\epsilon)^3, \label{eq:bd-trm22-decmpES1U2*}
  \end{IEEEeqnarray}
  where we have denoted by $\mathscr{C}_i$ the random codebook of user
  $i \in \{ 1,2 \}$, and where in $a)$ the second expectation term is
  zero because for every $({\bf s}_1, {\bf s}_2) \in \mathcal{E}_{\bf
    S}^c$
  \begin{IEEEeqnarray*}{rCl}
    \mat{E}_{\mathscr{C}_2} \left[ \inner{{\bf s}_1}{{\bf V}_2} \Big|
      ({\bf S}_1, {\bf S}_2) = ({\bf s}_1, {\bf s}_2), \mathcal{E}_{{\bf
          X}_2}^c \right] & = & 0.
  \end{IEEEeqnarray*}
  This holds since in the expectation over the codebooks
  $\mathscr{C}_2$ with conditioning on $\mathcal{E}_{{\bf X}_2}^c$,
  for every ${\bf v}_2 \in \Reals^n$ the sequences ${\bf v}_2$ and
  $-{\bf v}_2$ are equiprobable and thus their inner products with
  ${\bf s}_1$ cancel off each other. Inequality b) follows from lower
  bounding $\cos \sphericalangle ({\bf s}_2, {\bf U}_2^{\ast})$
  conditioned on $\mathcal{E}_{\bf X}^c$ combined with the fact that
  conditioned on $\mathcal{E}_{\bf S}^c$ the term $\cos
  \sphericalangle ({\bf S}_1,{\bf S}_2)$ is positive. Inequality c)
  follows from lower bounding $\| {\bf S}_1 \|$ and $\cos
  \sphericalangle ({\bf S}_1, {\bf S}_2)$ conditioned on
  $\mathcal{E}_{\bf S}^c$.

  Combining \eqref{eq:decmp-ES1U2*} with \eqref{eq:bd-trm1-decmpES1U2*},
  \eqref{eq:bd-trm21-decmpES1U2*} and \eqref{eq:bd-trm22-decmpES1U2*} gives
  \begin{IEEEeqnarray*}{rCl}
    \frac{1}{n} \E{\inner{{\bf S}_1}{{\bf U}_2^{\ast}}} & \geq & -
    \sigma^2 \left( \epsilon + 2\Prv{\mathcal{E}_{\bf S}} + (2 +
      \epsilon)\Prv{\mathcal{E}_{\bf X}} \right)\\
    & & {} + \sigma^2 \rho (1-2^{-2R_2}) (1-\epsilon)^3 \left( 1 -
      \left( \Prv{\mathcal{E}_{\bf S}} + \Prv{\mathcal{E}_{\bf X}}\right)
    \right)\\
    & \geq & \sigma^2 \rho (1-2^{-2R_2}) (1-\epsilon)^3 - \sigma^2 \left(
      \epsilon + 3\Prv{\mathcal{E}_{\bf S}} + (3 +
      \epsilon)\Prv{\mathcal{E}_{\bf X}} \right).
  \end{IEEEeqnarray*}
  Thus, by Lemma \ref{lm:vq-Pr-E1} and Lemma \ref{lm:vq-Pr-E2} it
  follows that for every $\delta > 0$ and $0.3 > \epsilon > 0$ there exists
  an $n'(\delta,\epsilon) \in \Naturals$ such that for all $n >
  n'(\delta,\epsilon)$
  \begin{IEEEeqnarray*}{rCl}
    \qquad \qquad \; \; \; \frac{1}{n} \E{\inner{{\bf S}_1}{{\bf
          U}_2^{\ast}}} & \geq & \sigma^2 \rho (1-2^{-2R_2})
    (1-\epsilon)^3 - \sigma^2 \left( \epsilon + 21 \delta +  6 \delta
      \epsilon \right). \qquad \qquad \; \; \; \qedhere
  \end{IEEEeqnarray*}
\end{proof}

The distortion of the genie-aided scheme is now upper bounded as
follows
\begin{IEEEeqnarray*}{rCl}
  \frac{1}{n} \E{\| {\bf S}_1 - \hat{\bf S}_1^{\textnormal{G}} \|^2}
  & = & \sigma^2 -2\gamma_{11} \frac{1}{n}\E{\inner{{\bf
        S}_1}{{\bf U}_1^{\ast}}} -2\gamma_{12} \frac{1}{n}\E{\inner{{\bf
        S}_1}{{\bf U}_2^{\ast}}} \\
  & & {} + \gamma_{11}^2 \sigma^2(1-2^{-2R_1}) + 2 \gamma_{11}
  \gamma_{12} \frac{1}{n}\E{\inner{{\bf U}_1^{\ast}}{{\bf
        U}_2^{\ast}}} \\
  & & {} + \gamma_{12}^2 \sigma^2(1-2^{-2R_2})\\[3mm]
  & \stackrel{a)}{\leq} & \sigma^2 2^{-2R_1}
  \frac{1-\rho^2(1-2^{-2R_2})}{1-\tilde{\rho}^2} +
  \xi'(\delta,\epsilon),
\end{IEEEeqnarray*}
where in $a)$ we have used Lemma \ref{lm1:vq-bd-D1genie}, Lemma
\ref{lm2:vq-bd-D1genie} and Lemma \ref{lm3:vq-bd-D1genie}, and where
\begin{IEEEeqnarray*}{rCl}
  \lim_{\delta,\epsilon \rightarrow 0} \xi'(\delta,\epsilon) & = & 0.
\end{IEEEeqnarray*}

\subsection{Proofs of Lemma \ref{lm:vq-Pr-E2}, Lemma
  \ref{lm:caps-single-error} and Lemma \ref{lm:caps-double-error}}

The proofs in this section rely on bounds from the geometry of sphere
packing. To this end, we denote by $C_n(\varphi)$ the surface area of
a polar cap of half angle $\varphi$ on an $\Reals^n$-sphere of unit
radius. An illustration of $C_n(\varphi)$ is given in Figure
\ref{fig:srf-plr-cp}.
\begin{figure}[h]
 \centering
 \psfrag{1}[cc][cc]{$1$}
 \psfrag{rn}[cc][cc]{$\Reals^n$}
 \psfrag{cn}[cc][cc]{$C_n(\varphi)$}
 \psfrag{p}[cc][cc]{$\varphi$}
 \epsfig{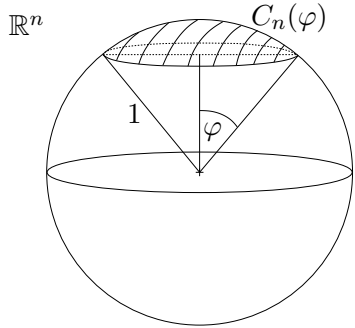}
 \caption{Polar cap of half angle $\varphi$.}
 \label{fig:srf-plr-cp}
\end{figure}
Upper and lower bounds on the surface area $C_n(\varphi)$ are given in
the following lemma.
\begin{lm}\label{lm:bounds-polar-caps}
For any $\varphi \in [0,\pi/2]$,
\begin{IEEEeqnarray*}{rCCCl}
  \frac{\Gamma \left( \frac{n}{2} +1\right) \sin^{(n-1)}\varphi}{n
    \Gamma \left( \frac{n+1}{2}\right) \sqrt{\pi} \cos \varphi} \left( 1 -
    \frac{1}{n} \tan^2 \varphi \right) & \leq &
  \frac{C_n(\varphi)}{C_n(\pi)} & \leq & \frac{\Gamma \left(
      \frac{n}{2} +1\right) \sin^{(n-1)}\varphi}{n \Gamma \left(
      \frac{n+1}{2}\right) \sqrt{\pi} \cos \varphi}.
\end{IEEEeqnarray*}
\end{lm}
\begin{proof}
  See \cite[Inequality (27)]{shannon58}.
\end{proof}

The ratio of the two gamma functions that appears in the upper bound
and the lower bound of Lemma \ref{lm:bounds-polar-caps} has the
following asymptotic series.
\begin{lm}\label{lm:gamma-series}
  \begin{IEEEeqnarray*}{rCl}
    \frac{\Gamma \left( x + \frac{1}{2} \right)}{\Gamma (x)} & = &
    \sqrt{x} \left( 1 - \frac{1}{8x} + \frac{1}{128x^2} +
      \frac{5}{1024x^3} - \frac{21}{32768x^4} + \ldots \right),
  \end{IEEEeqnarray*}
  and in particular
  \begin{IEEEeqnarray*}{rCl}
    \lim_{x \rightarrow \infty} \frac{\Gamma \left( x + \frac{1}{2}
      \right)}{\Gamma (x) \sqrt{x}} & = & 1.
  \end{IEEEeqnarray*}
\end{lm}
\begin{proof}
  We first note that
  \begin{IEEEeqnarray}{rCl}\label{eq:ratio-gamma1}
    \frac{\Gamma \left(x + \frac{1}{2}\right)}{\Gamma (x)} & = &
    \frac{(2x-1)!!}{2^x (x-1)!} \sqrt{\pi} \nonumber\\[3mm]
    & = & \frac{x}{4^x} \binom{2x}{x} \sqrt{\pi},
  \end{IEEEeqnarray}
  where $\xi !!$ denotes the double factorial of $\xi$. The proof now
  follows by combining \eqref{eq:ratio-gamma1} with
  \begin{IEEEeqnarray*}{rCl}
    \binom{2x}{x} & = & \frac{4^x}{\sqrt{\pi x}} \left( 1 -
      \frac{1}{8x} + \frac{1}{128x^2} + \frac{5}{1024x^3} -
      \frac{21}{32768x^4} + \ldots \right),
  \end{IEEEeqnarray*}
  which is given in \cite[Problem 9.60, p.~495]{graham-knuth-patashnik94}.
\end{proof}

Before starting with the proofs of this section, we give one more
lemma. To this end, whenever the vector-quantizer of Encoder $1$ does
not produce the all-zero sequence, denote by $\varsigma_1({\bf s}_1,
\mathcal{C}_1)$ the index of ${\bf u}_1^{\ast}$ in its codebook
$\mathcal{C}_1$. And whenever the vector-quantizer of Encoder $1$
produces the all-zero sequence, let $\varsigma_1({\bf s}_1,
\mathcal{C}_1) = 0$. Further, let $\lambda_1(\cdot)$ denote the
measure on the codeword sphere $\mathcal{S}_1$ induced by the uniform
distribution, and let $f^{\lambda_1}(\cdot)$ denote the density on
$\mathcal{S}_1$ with respect to $\lambda_1(\cdot)$. Similarly, for
Encoder 2 define $\varsigma_2({\bf s}_2, \mathcal{C}_2)$ and
$f^{\lambda_2}(\cdot)$.

\begin{lm}\label{lm:density-wrong-codeword}
  Conditional on $\varsigma_1({\bf S}_1, \mathscr{C}_1) = 1$, the
  density of ${\bf U}_1 (j)$ is upper bounded for every $j \in \{ 2,3,
  \ldots , 2^{nR_1}\}$ and at every point ${\bf u} \in \mathcal{S}_1$
  by twice the uniform density:
  \begin{IEEEeqnarray*}{rCl}
    f^{\lambda_1} \left( {\bf U}_1(j) = {\bf u} | \varsigma_1({\bf
        S}_1, \mathscr{C}_1) = 1 \right) & \leq & 2 \cdot
    \frac{1}{r_1^{n-1}C_n(\pi)}, \quad {\bf u} \in \mathcal{S}_1, \; j
    \in \{ 2,3, \ldots , 2^{nR_1}\}.
  \end{IEEEeqnarray*}
  And similarly for Encoder 2.
\end{lm}

\begin{proof}

  We first write the conditional density as an average over $\cos
  \sphericalangle ({\bf S}_1, {\bf U}_1(1))$. Since conditioned on
  $\varsigma_1({\bf S}_1, \mathscr{C}_1) = 1$ we have $\cos
  \sphericalangle ({\bf S}_1, {\bf U}_1(1)) \in
  [\sqrt{1-2^{-2R_1}}(1-\epsilon), \sqrt{1-2^{-2R_1}}(1+\epsilon)]$,
  we obtain
  \begin{IEEEeqnarray}{l}\label{eq:vq-density-uij-1}
    f^{\lambda_1} \left( {\bf U}_1(j) = {\bf u} \big|
      \varsigma_1({\bf S}_1, \mathscr{C}_1) = 1 \right) \nonumber\\[2mm]
    \hspace{2mm}= \int_{{\bf s}_1 \in \Reals^n}
    \int_{\sqrt{1-2^{-2R_1}}(1-\epsilon)}^{\sqrt{1-2^{-2R_1}}(1+\epsilon)}
    f^{\lambda_1} \left({\bf U}_1(j) = {\bf u} \big| {\bf S}_1 = {\bf
        s}_1, \varsigma_1({\bf s}_1, \mathscr{C}_1) = 1,
      \cos \sphericalangle ({\bf s}_1, {\bf U}_1(1)) = a \right)\nonumber\\
    \hspace{45mm} f ({\bf S}_1 = {\bf s}_1, \cos \sphericalangle ({\bf
      s}_1, {\bf U}_1(1)) = a \big| \varsigma_1({\bf S}_1,
    \mathscr{C}_1) = 1 ) \d a \d {\bf s}_1. \qquad
  \end{IEEEeqnarray}
  The proof now follows by upper bounding the conditional density
  \begin{IEEEeqnarray*}{C}
    f^{\lambda_1} \left({\bf U}_1(j) = {\bf u} \big| {\bf S}_1 = {\bf
        s}_1, \varsigma_1({\bf s}_1, \mathscr{C}_1) = 1, \cos
      \sphericalangle ({\bf s}_1, {\bf U}_1(1)) = a \right).
  \end{IEEEeqnarray*}
  To this end, define for every $a \in
  [\sqrt{1-2^{-2R_1}}(1-\epsilon), \sqrt{1-2^{-2R_1}}(1+\epsilon)]$
  \begin{IEEEeqnarray*}{rCl}
    \mathscr{D}_a({\bf s}_1) & \triangleq & \left\{ {\bf u} \in
      \mathcal{S}_1: \left| \cos \sphericalangle ({\bf s}_1, {\bf u})
        - \sqrt{1-2^{-2R_1}} \right| \leq \left| a - \sqrt{1-2^{-2R_1}}
    \right| \right\},
  \end{IEEEeqnarray*}
  and
  \begin{IEEEeqnarray*}{rCl}
    \mathscr{D}_a^c({\bf s}_1) & \triangleq & \left\{ {\bf u} \in
      \mathcal{S}_1: \left| \cos \sphericalangle ({\bf s}_1, {\bf u})
        - \sqrt{1-2^{-2R_1}} \right| > \left| a - \sqrt{1-2^{-2R_1}}
      \right| \right\}.
  \end{IEEEeqnarray*}
  The conditional density can now be upper bounded by distinguishing
  between ${\bf u} \in \mathscr{D}_a({\bf s}_1)$ and ${\bf u} \in
  \mathscr{D}_a^c({\bf s}_1)$. If ${\bf u} \in \mathscr{D}_a({\bf
    s}_1)$, then the conditional density is zero because the fact that
  $\varsigma_1({\bf s}_1, \mathscr{C}_1)$ is $1$ implies that
  \begin{IEEEeqnarray*}{rCl}
    \left| \cos \sphericalangle ({\bf s}_1, {\bf U}_1(j)) -
      \sqrt{1-2^{-2R_1}} \right| & > & \left| a - \sqrt{1-2^{-2R_1}}
    \right|, \qquad \forall j \in \{ 2,3, \ldots ,2^{nR_1} \}.
  \end{IEEEeqnarray*}
  And if ${\bf u} \in \mathscr{D}_a^c({\bf s}_1)$ the conditional
  density is uniform over $\mathscr{D}_a^c({\bf s}_1)$, i.e.
  \begin{IEEEeqnarray*}{rCl}
    f^{\lambda_1} \left({\bf U}_1(j) = {\bf u} \big| {\bf S}_1 = {\bf
        s}_1, \varsigma_1({\bf s}_1,\mathscr{C}_1) = 1, \cos
      \sphericalangle ({\bf s}_1, {\bf U}_1(1)) = a \right) & = &
    \upsilon, \qquad {\bf u} \in \mathscr{D}_a^c({\bf s}_1),
  \end{IEEEeqnarray*}
  for some $\upsilon > 0$. Thus,
  \begin{IEEEeqnarray}{C} \label{eq:prf-vq-ub-cond-densty}
    f^{\lambda_1} \left({\bf U}_1(j) = {\bf u} \big| {\bf S}_1 = {\bf
        s}_1, \varsigma_1({\bf s}_1, \mathscr{C}_1) = 1, \cos
      \sphericalangle ({\bf s}_1, {\bf U}_1(1)) = a \right) \leq
    \upsilon, \nonumber\\[-1mm]
    \\[-1mm]
    \hspace{20mm}\forall {\bf u} \in \mathcal{S}_1, {\bf s}_1 \in
    \Reals^n, a \in [\sqrt{1-2^{-2R_1}}(1-\epsilon),
    \sqrt{1-2^{-2R_1}}(1+\epsilon)]. \nonumber
  \end{IEEEeqnarray}
  It now remains to upper bound $\upsilon$. To this end, notice that
  the surface area of $\mathscr{D}_a({\bf s}_1)$ never exceeds half
  the surface area of $\mathcal{S}_1$. This follows since
  $\sqrt{1-2^{-2R_1}}(1-\epsilon) > 0$, and therefore every ${\bf u}
  \in \mathscr{D}_a({\bf s}_1)$ satisfies $|\sphericalangle ({\bf
    s}_1,{\bf u})| < \pi/2$. Hence, the surface area of
  $\mathscr{D}_a^c({\bf s}_1)$ is always larger than half the surface
  area of $\mathcal{S}_1$ and therefore
  \begin{IEEEeqnarray}{C}\label{eq:vq-density-uij-2}
    \upsilon \leq 2 \cdot \frac{1}{r_1^{n-1}C_n(\pi)}.
  \end{IEEEeqnarray}
  Combining \eqref{eq:vq-density-uij-2} with
  \eqref{eq:prf-vq-ub-cond-densty} and \eqref{eq:vq-density-uij-1}
  proves the lemma.
\end{proof}

\subsubsection{Proof of Lemma \ref{lm:vq-Pr-E2}}\label{sec:prf-lm-Pr-E2}

We begin with the following decomposition
\begin{IEEEeqnarray*}{rCl}
  \Prv{\mathcal{E}_{\bf X}} & = & \Prv{\mathcal{E}_{\bf X} \cap \mathcal{E}_{\bf S}} +
  \Prv{\mathcal{E}_{\bf X} \cap \mathcal{E}_{\bf S}^c}\\
  & \leq & \Prv{\mathcal{E}_{\bf S}} + \Prv{\mathcal{E}_{{\bf X}_1} \cap
    \mathcal{E}_{\bf S}^c} + \Prv{\mathcal{E}_{{\bf X}_2} \cap \mathcal{E}_{\bf S}^c} +
  \Prv{\mathcal{E}_{({\bf X}_1,{\bf X}_2)} \cap \mathcal{E}_{{\bf X}_2}^c \cap
    \mathcal{E}_{{\bf X}_1}^c \cap \mathcal{E}_{\bf S}^c}\\
  & \leq & \Prv{\mathcal{E}_{\bf S}} + \Prv{\mathcal{E}_{{\bf X}_1}} +
  \Prv{\mathcal{E}_{{\bf X}_2}} + \Prv{\mathcal{E}_{({\bf X}_1,{\bf X}_2)} \cap
    \mathcal{E}_{{\bf X}_2}^c \cap \mathcal{E}_{{\bf X}_1}^c \cap \mathcal{E}_{\bf S}^c}.
\end{IEEEeqnarray*}
The proof of Lemma \ref{lm:vq-Pr-E2} now follows by showing that for
every $\delta > 0$ and $0.3 > \epsilon > 0$ there exists an
$n_2'(\delta,\epsilon) > 0$ such that for all $n > n_2'(\delta,\epsilon)$
\begin{IEEEeqnarray}{rCll}
  \Prv{\mathcal{E}_{{\bf X}_i}} & \leq & \delta, \qquad & i \in \{ 1,2 \}
  \label{eq:Pr-E2i}\\
  \Prv{\mathcal{E}_{({\bf X}_1,{\bf X}_2)} \cap \mathcal{E}_{\bf S}^c
    \cap \mathcal{E}_{{\bf X}_1}^c 
    \cap \mathcal{E}_{{\bf X}_2}^c} & \leq & 3
  \delta. \label{eq:Pr-E23-E21c-E22c}
\end{IEEEeqnarray}

\paragraph{Proof of \eqref{eq:Pr-E2i}:}\label{app:proof-RD}

We give the proof for $\mathcal{E}_{{\bf X}_1}$. Due to the symmetry
the proof for $\mathcal{E}_{{\bf X}_2}$ then follows by similar
arguments. 
Let $\mathcal{E}_{{\bf X}_1}(j)$ be the event that ${\bf U}_1(j)$ does not
have a typical angle to ${\bf S}_1$, i.e.
\begin{IEEEeqnarray*}{rCl}
  \mathcal{E}_{{\bf X}_1}(j) & = & \left\{ ({\bf s}_1, {\bf s}_2,
    \mathcal{C}_1, \mathcal{C}_2) : \Big| \cos \sphericalangle ({\bf
      u}_1(j),{\bf s}_1) - \sqrt{1-2^{-2R_1}} \Big| > \epsilon
    \sqrt{1-2^{-2R_1}} \right\}.
\end{IEEEeqnarray*}
Then,
\begin{IEEEeqnarray}{rCl}\label{eq:dec-E2i}
  \Prv{\mathcal{E}_{{\bf X}_1}} & = & \Prv{\mathcal{E}_{{\bf X}_1}
    \big| {\bf S}_1 = {\bf s}_1} \nonumber\\[2mm]
  & = & \Prv{\bigcap_{j=1}^{2^{nR_1}}
    \mathcal{E}_{{\bf X}_1}(j) \Bigg| {\bf S}_1 = {\bf
      s}_1}\nonumber\\
  & = & \prod_{j=1}^{2^{nR_1}}
  \Prv{\mathcal{E}_{{\bf X}_1}(j) \big| {\bf S}_1 = {\bf
      s}_1}\nonumber\\
  & \stackrel{a)}{=} & \prod_{j=1}^{2^{nR_1}}
  \Prv{\mathcal{E}_{{\bf X}_1}(j) }\nonumber\\
  & \stackrel{b)}{=} & \left( \Prv{\mathcal{E}_{{\bf X}_1}(1)}
  \right)^{2^{nR_1}}\nonumber\\[3mm]
  & = & \left(1 - \Prv{\mathcal{E}_{{\bf X}_1}^c(1)}\right)^{2^{nR_1}},
\end{IEEEeqnarray}
where in $a)$ we have used that the probability of $\mathcal{E}_{{\bf
    X}_1}(j)$ does not depend on ${\bf S}_1 = {\bf s}_1$, and in $b)$
we have used that all ${\bf U}_1(j)$ have the same distribution. To
upper bound \eqref{eq:dec-E2i} we now rewrite $\mathcal{E}_{{\bf
    X}_1}^c(1)$ as
\begin{IEEEeqnarray*}{rCl}
  \mathcal{E}_{{\bf X}_1}^c(1) & = & \left\{ ({\bf s}_1, {\bf s}_2,
    \mathcal{C}_1, \mathcal{C}_2) : \Big| \cos \sphericalangle ({\bf
      u}_1(1), {\bf s}_1) - \sqrt{1-2^{-2R_1}} \Big| \leq \epsilon
    \sqrt{1-2^{-2R_1}}
  \right\}\\
  & = & \left\{ ({\bf s}_1, {\bf s}_2, \mathcal{C}_1, \mathcal{C}_2) :
    \sqrt{1-2^{-2R_1}}(1-\epsilon) \leq \cos \sphericalangle ({\bf
      u}_1(1), {\bf s}_1) \leq
    \sqrt{1-2^{-2R_1}}(1+\epsilon) \right\}\\
  & = & \left\{ ({\bf s}_1, {\bf s}_2, \mathcal{C}_1, \mathcal{C}_2) :
    \cos \theta_{1,\text{max}} \leq \cos \sphericalangle ({\bf
      u}_1(1), {\bf s}_1) \leq \cos \theta_{1,\text{min}} \right\},
\end{IEEEeqnarray*}
where we have used the notation
\begin{IEEEeqnarray*}{rCl}
  \cos \theta_{1,\text{max}} \triangleq \sqrt{1-2^{-2R_1}} (1-\epsilon) &
  \quad \text{and} \quad & \cos \theta_{1,\text{min}} \triangleq
  \sqrt{1-2^{-2R_1}} (1+\epsilon).
\end{IEEEeqnarray*}
Hence, since ${\bf U}_1(1)$ is generated independently of ${\bf S}_1$
and distributed uniformly on $\mathcal{S}_1$,
\begin{IEEEeqnarray}{rCl}\label{eq:Pr-E2ic}
  \Prv{\mathcal{E}_{{\bf X}_1}^c(1)} & = & \frac{C_n(\theta_{1,\text{max}}) -
    C_n(\theta_{1,\text{min}})}{C_n(\pi)}.
\end{IEEEeqnarray}
Combining \eqref{eq:Pr-E2ic} with \eqref{eq:dec-E2i} gives
\begin{IEEEeqnarray}{rCl}\label{eq:deriv-bd-E2i}
  \Prv{\mathcal{E}_{{\bf X}_1}} & = & \left( 1 -
    \frac{C_n(\theta_{1,\text{max}}) -
      C_n(\theta_{1,\text{min}})}{C_n(\pi)} \right)^{2^{nR_1}}\nonumber\\
  & \stackrel{a)}{\leq} & \left( \exp \left( -
      \frac{C_n(\theta_{1,\text{max}}) -
        C_n(\theta_{1,\text{min}})}{C_n(\pi)} \right)
  \right)^{2^{nR_1}}\nonumber\\
  & \stackrel{b)}{\leq} & \exp \left( - 2^{nR_1} \frac{\Gamma \left(
        \frac{n}{2} +1\right)}{n \Gamma \left( \frac{n+1}{2}\right)
      \sqrt{\pi}} \left( \frac{\sin^{(n-1)}\theta_{1,\text{max}}}{\cos
        \theta_{1,\text{max}}} \left( 1 - \frac{1}{n} \tan^2
        \theta_{1,\text{max}} \right) -
      \frac{\sin^{(n-1)}\theta_{1,\text{min}}}{ \cos
        \theta_{1,\text{min}}} \right) \right)\nonumber\\
  & = & \exp \Bigg( - \frac{\Gamma \left( \frac{n}{2} +1\right)}{n
    \Gamma \left( \frac{n+1}{2}\right) \sqrt{\pi}} \Bigg( \frac{ 2^{n
      (R_1 + \log_2 (\sin \theta_{1,\text{max}}))}}{\sin
    \theta_{1,\text{max}} \cos \theta_{1,\text{max}}} \left( 1 -
    \frac{1}{n} \tan^2
    \theta_{1,\text{max}} \right) - \nonumber\\
  & & \qquad \qquad \qquad \qquad \qquad \qquad \qquad \qquad \qquad
  \qquad \qquad - \frac{2^{n (R_1 + \log_2 (\sin
      \theta_{1,\text{min}}))}}{\sin \theta_{1,\text{min}} \cos
    \theta_{1,\text{min}}} \Bigg) \Bigg)
\end{IEEEeqnarray}
where in $a)$ we have used that $1-x \leq \exp(-x)$, and in $b)$ we
have lower bounded $C_n(\theta_{1,\text{max}})/C_n(\pi)$ and upper
bounded $C_n(\theta_{1,\text{min}})/C_n(\pi)$ according to Lemma
\ref{lm:bounds-polar-caps}. It now follows from sphere-packing and
-covering, see e.g.~\cite{wyner67}, that for every $\epsilon > 0$ we
have $\Prv{\mathcal{E}_{{\bf X}_1}} \rightarrow 0$ as $n \rightarrow
\infty$. More precisely, this holds since the exponent on the RHS of
\eqref{eq:deriv-bd-E2i} grows exponentially in $n$. This follows since
on the one hand for large $n$
\begin{IEEEeqnarray*}{rCl}
  \frac{\Gamma \left(
        \frac{n}{2} +1\right)}{n \Gamma \left( \frac{n+1}{2}\right)
      \sqrt{\pi}} & \approx & \frac{1}{\sqrt{n 2\pi}},
\end{IEEEeqnarray*}
and on the other hand the term
\begin{IEEEeqnarray*}{rCl}
\frac{ 2^{n (R_1 + \log_2 (\sin \theta_{1,\text{max}}))}}{\sin
  \theta_{1,\text{max}} \cos \theta_{1,\text{max}}} \left( 1 -
  \frac{1}{n} \tan^2 \theta_{1,\text{max}} \right) - \frac{2^{n (R_1 +
    \log_2 (\sin \theta_{1,\text{min}}))}}{\sin \theta_{1,\text{min}}
  \cos \theta_{1,\text{min}}}
\end{IEEEeqnarray*}
grows exponentially in $n$. The latter holds since first of all
\begin{IEEEeqnarray*}{rCl}
  \left( 1 - \frac{1}{n} \tan^2 \theta_{1,\text{max}} \right) &
  \approx & 1 \qquad \text{for large } n,
\end{IEEEeqnarray*}
second, the denominators of the fractions are independent of $n$, and
third since
\begin{IEEEeqnarray*}{rCl}
 R_1 + \log_2 \left( \sin \theta_{1,\max} \right) \geq R_1 + \log_2
 \left( \sin \theta_{1,\min} \right),
\end{IEEEeqnarray*}
with $R_1 + \log_2 (\sin \theta_{1,\text{max}}) > 0$. That $R_1 +
\log_2 (\sin \theta_{1,\text{max}}) > 0$ can be seen as follows.
\begin{IEEEeqnarray*}{rCl}
  - \log_2 (\sin \theta_{1,\text{max}}) & = & -\log_2 ( \sqrt{1 -
    \cos^2 \theta_{1,\text{max}}})\\
& \stackrel{a)}{=} & -\frac{1}{2} \log_2 \left(
    2^{-2R_1} + \epsilon (2-\epsilon) (1-2^{-2R_1}) \right)\\
& < & -\frac{1}{2} \log_2 \left( 2^{-2R_1} \right)\\
& = & R_1,
\end{IEEEeqnarray*}
where in $a)$ we have used the definition of $\cos
\theta_{1,\text{max}}$. \hfill $\qed$

\paragraph{Proof of \eqref{eq:Pr-E23-E21c-E22c}:}\label{app:prf-typ-u1-u2}

By the notation in \eqref{eq:decomp-u^ast} we have
\begin{IEEEeqnarray}{rCl}
  \cos \sphericalangle ({\bf u}_1^{\ast}, {\bf u}_2^{\ast}) & = &
  \frac{\inner{{\bf u}_1^{\ast}}{{\bf u}_2^{\ast}}}{\| {\bf u}_1^{\ast} \|
  \|{\bf u}_2^{\ast}\|}\nonumber\\
  & = & \frac{1}{\| {\bf u}_1^{\ast}\| \| {\bf u}_2^{\ast} \|} \left(
    \nu_1 \nu_2 \inner{{\bf s}_1}{{\bf s}_2} + \nu_1 \inner{{\bf
        s}_1}{{\bf v}_2} + \nu_2 \inner{{\bf v}_1}{{\bf s}_2} +
    \inner{{\bf v}_1}{{\bf v}_2} \right), \qquad \label{eq:prf-inner-Ustar}
\end{IEEEeqnarray}
where we recall that $\nu_1$ is a function of $\| {\bf s}_1 \|$ and
$\cos \sphericalangle ({\bf s}_1, {\bf u}_1^{\ast})$ and similarly
$\nu_2$ is a function of $\| {\bf s}_2 \|$ and $\cos \sphericalangle
({\bf s}_2, {\bf u}_2^{\ast})$. Now, define the four events
\begin{IEEEeqnarray*}{rCl}
  \mathcal{A}_1 & = & \left\{ ({\bf s}_1, {\bf s}_2, \mathcal{C}_1,
    \mathcal{C}_2): \left| \tilde{\rho} - \frac{\nu_1 \nu_2}{\| {\bf
          u}_1^{\ast} \| \| {\bf u}_2^{\ast} \|} \inner{{\bf
          s}_1}{{\bf s}_2} \right| > 4 \epsilon \right\}\\
  \mathcal{A}_2 & = & \left\{ ({\bf s}_1, {\bf s}_2, \mathcal{C}_1,
    \mathcal{C}_2): \left| \frac{\nu_1}{\| {\bf u}_1^{\ast} \| \| {\bf
          u}_2^{\ast} \|} \inner{{\bf s}_1}{{\bf v}_2} \right| >
    \epsilon \right\}\\
  \mathcal{A}_3 & = & \left\{ ({\bf s}_1, {\bf s}_2, \mathcal{C}_1,
    \mathcal{C}_2): \left| \frac{\nu_2}{\| {\bf u}_1^{\ast} \| \| {\bf
          u}_2^{\ast} \|} \inner{{\bf s}_2}{{\bf v}_1} \right| >
    \epsilon \right\}\\
  \mathcal{A}_4 & = & \left\{ ({\bf s}_1, {\bf s}_2, \mathcal{C}_1,
    \mathcal{C}_2): \left| \frac{1}{\| {\bf u}_1^{\ast}
        \| \| {\bf u}_2^{\ast} \|} \inner{{\bf v}_1}{{\bf v}_2}
    \right| > \epsilon \right\}.
\end{IEEEeqnarray*}
Note that by \eqref{eq:prf-inner-Ustar}, $\mathcal{E}_{({\bf X}_1,{\bf X}_2)} \subset
(\mathcal{A}_1 \cup \mathcal{A}_2 \cup \mathcal{A}_3 \cup
\mathcal{A}_4)$. Thus,
\begin{IEEEeqnarray}{rCl}
  \Prv{\mathcal{E}_{({\bf X}_1,{\bf X}_2)} \cap \mathcal{E}_{\bf S}^c
    \cap \mathcal{E}_{{\bf X}_1}^c 
    \cap \mathcal{E}_{{\bf X}_2}^c} & \leq & \Prv{\mathcal{A}_1 \cap
    \mathcal{E}_{\bf S}^c \cap \mathcal{E}_{{\bf X}_1}^c \cap \mathcal{E}_{{\bf X}_2}^c} +
  \Prv{\mathcal{A}_2 \cap \mathcal{E}_{\bf S}^c \cap \mathcal{E}_{{\bf X}_1}^c
    \cap \mathcal{E}_{{\bf X}_2}^c} \nonumber\\
& & + \Prv{\mathcal{A}_3 \cap
    \mathcal{E}_{\bf S}^c \cap \mathcal{E}_{{\bf X}_1}^c \cap \mathcal{E}_{{\bf X}_2}^c} +
  \Prv{\mathcal{A}_4 \cap \mathcal{E}_{\bf S}^c \cap \mathcal{E}_{{\bf X}_1}^c
    \cap \mathcal{E}_{{\bf X}_2}^c} \nonumber\\
& \leq & \Prv{\mathcal{A}_1 | \mathcal{E}_{\bf S}^c \cap
  \mathcal{E}_{{\bf X}_1}^c \cap \mathcal{E}_{{\bf X}_2}^c} + \Prv{\mathcal{A}_2
  | \mathcal{E}_{\bf S}^c } + \Prv{\mathcal{A}_3 | \mathcal{E}_{\bf S}^c } \nonumber\\
& & + \Prv{\mathcal{A}_4 | \mathcal{E}_{\bf S}^c }.\label{eq:E23-dec}
\end{IEEEeqnarray}
The four terms on the RHS of \eqref{eq:E23-dec} are now bounded in the
following two lemmas.
\begin{lm}\label{lm:VQ-innerU-s1s2}
  For $\epsilon \leq 0.3$
  \begin{IEEEeqnarray*}{rCl}
    \Prv{\mathcal{A}_1 | \mathcal{E}_{\bf S}^c \cap
      \mathcal{E}_{{\bf X}_1}^c \cap \mathcal{E}_{{\bf X}_2}^c} & = & 0.
  \end{IEEEeqnarray*}
\end{lm}

\begin{proof}
  We first note that the term in the definition of $\mathcal{A}_1$ can
  be rewritten as
  \begin{IEEEeqnarray}{rCl}\label{eq:prf-lm3-rw-A1}
    \frac{\nu_1 \nu_2}{\| {\bf u}_1^{\ast} \| \| {\bf u}_2^{\ast} \|}
    \inner{{\bf s}_1}{{\bf s}_2} & = & \cos \sphericalangle ({\bf
      s}_1, {\bf u}_1^{\ast}) \cos \sphericalangle ({\bf s}_2, {\bf
      u}_2^{\ast}) \cos \sphericalangle ( {\bf s}_1, {\bf s}_2).
  \end{IEEEeqnarray}
  We can now upper and lower bound the RHS of \eqref{eq:prf-lm3-rw-A1}
  for $({\bf s}_1, {\bf s}_2, \mathcal{C}_1, \mathcal{C}_2) \in
  \mathcal{E}_{\bf S}^c \cap \mathcal{E}_{{\bf X}_1}^c \cap
  \mathcal{E}_{{\bf X}_2}^c$ by noticing that $({\bf s}_1, {\bf s}_2,
  \mathcal{C}_1, \mathcal{C}_2) \in \mathcal{E}_{\bf S}^c$ implies
  \begin{IEEEeqnarray*}{rCl}
    \left| \cos \sphericalangle ( {\bf s}_1, {\bf s}_2) - \rho
    \right| & < & \rho \epsilon,
  \end{IEEEeqnarray*}
  that $({\bf s}_1, {\bf s}_2, \mathcal{C}_1,
  \mathcal{C}_2) \in \mathcal{E}_{{\bf X}_1}^c$ implies
  \begin{IEEEeqnarray*}{rCl}
    \left| \sqrt{1-2^{-2R_1}} - \cos \sphericalangle ({\bf s}_1, {\bf
        u}_1^{\ast}) \right| & < & \epsilon \sqrt{1-2^{-2R_1}},
  \end{IEEEeqnarray*}
  and that $({\bf s}_1, {\bf s}_2,
  \mathcal{C}_1, \mathcal{C}_2) \in \mathcal{E}_{{\bf X}_2}^c$ implies
  \begin{IEEEeqnarray*}{rCl}
    \left| \sqrt{1-2^{-2R_2}} - \cos \sphericalangle ({\bf s}_2, {\bf
        u}_2^{\ast}) \right| & < & \epsilon \sqrt{1-2^{-2R_2}}.
  \end{IEEEeqnarray*}
  Hence, combined with \eqref{eq:prf-lm3-rw-A1} this gives
  \begin{IEEEeqnarray*}{rCl}
    \hspace{10mm} \tilde{\rho} (1-\epsilon)^3 & \leq \frac{\nu_1 \nu_2}{\| {\bf
        u}_1^{\ast} \| \| {\bf u}_2^{\ast} \|} \inner{{\bf s}_1}{{\bf
        s}_2} \leq & \tilde{\rho} (1+\epsilon)^3, \hspace{10mm} ({\bf s}_1,
    {\bf s}_2, \mathcal{C}_1, \mathcal{C}_2) \in \mathcal{E}_{{\bf
        S}}^c \cap \mathcal{E}_{{\bf X}_1}^c \cap \mathcal{E}_{{\bf X}_2}^c.
  \end{IEEEeqnarray*}
  The LHS can be lower bounded by $\tilde{\rho} (1-3\epsilon) \leq
  \tilde{\rho}(1-\epsilon)^3$, and the RHS can be upper bounded by
  $\tilde{\rho} (1+\epsilon)^3 \leq \tilde{\rho}(1+4\epsilon)$ whenever
  $\epsilon \leq 0.3$. Hence, for $\epsilon \leq 0.3$
  \begin{IEEEeqnarray*}{rCl}
    \hspace{40mm} \left| \tilde{\rho} - \frac{\nu_1 \nu_2}{\| {\bf u}_1^{\ast} \| \|
        {\bf u}_2^{\ast} \|} \inner{{\bf s}_1}{{\bf s}_2} \right| & \leq &
    4 \tilde{\rho} \epsilon \leq 4 \epsilon. \hspace{39mm} \qedhere\\[2mm]
  \end{IEEEeqnarray*}
\end{proof}

\begin{lm}\label{lm:VQ-innerU-sv}
  For every $\delta > 0$ and $\epsilon > 0$ there exists an
  $n_{\mathcal{A}}'(\delta,\epsilon)$ such that for all $n >
  n_{\mathcal{A}}'(\delta,\epsilon)$
\begin{IEEEeqnarray*}{rCl}
  \Prv{\mathcal{A}_2 | \mathcal{E}_{\bf S}^c} < \delta, & \qquad \qquad
  \Prv{\mathcal{A}_3 | \mathcal{E}_{\bf S}^c} < \delta, & \qquad \qquad
  \Prv{\mathcal{A}_4 | \mathcal{E}_{\bf S}^c} < \delta.
\end{IEEEeqnarray*}
\end{lm}
\begin{proof}
  We start derivation of the bound on $\mathcal{A}_2$. To this end, we
  first upper bound the inner product between ${\bf s}_1$ and ${\bf
    v}_2$. Let ${\bf s}_{1,\textnormal{P}}$ denote the projection of ${\bf s}_1$
  onto the subspace of $\Reals^n$ that is orthogonal to ${\bf
    s}_2$, and that thus contains ${\bf v}_2$. Hence,
  \begin{IEEEeqnarray}{rCl}
    \left| \frac{\nu_1}{\| {\bf u}_1^{\ast} \| \| {\bf u}_2^{\ast} \|}
      \inner{{\bf s}_1}{{\bf v}_2} \right| & \stackrel{a)}{=} & \left|
      \cos \sphericalangle ({\bf s}_1, {\bf u}_1^{\ast})
      \inner{\frac{{\bf s}_1}{\| {\bf s}_1
          \|}}{\frac{{\bf v}_2}{\| {\bf u}_2^{\ast} \|}} \right| \nonumber\\
    & \stackrel{b)}{\leq} & \left| \cos \sphericalangle ({\bf s}_1,
      {\bf u}_1^{\ast}) \right| \left| \inner{\frac{{\bf s}_1}{\| {\bf
            s}_1 \|}}{\frac{{\bf v}_2}{\| {\bf
            v}_2 \|}} \right| \nonumber\\
    & \leq & \left| \inner{\frac{{\bf s}_1}{\| {\bf s}_1
          \|}}{\frac{{\bf
            v}_2}{\| {\bf v}_2 \|}} \right| \nonumber\\
    & = & \left| \inner{\frac{{\bf s}_{1,\textnormal{P}}}{\| {\bf
            s}_1 \|}}{\frac{{\bf v}_2}{\| {\bf v}_2 \|}} \right| \nonumber\\
    & \leq & \left| \inner{\frac{{\bf s}_{1,P}}{\| {\bf s}_{1,\textnormal{P}}
          \|}}{\frac{{\bf v}_2}{\| {\bf v}_2 \|}}
    \right|\nonumber\\[3mm]
    & = & \left| \cos \sphericalangle ({\bf s}_{1,\textnormal{P}}, {\bf v}_2)
    \right|,\label{eq:lm4-bd1}
  \end{IEEEeqnarray}
  where a) follows by the definition of $\nu_1$ and b) follows since by
  the definition of ${\bf v}_2$ we have $\| {\bf v}_2 \| \leq \| {\bf
    u}_2^{\ast} \|$. By \eqref{eq:lm4-bd1} it now follows that
  \begin{IEEEeqnarray*}{rCl}
    \Prv{\mathcal{A}_2 | \mathcal{E}_{\bf S}^c } & \leq & \Prv{({\bf S}_1,
      {\bf S}_2, \mathscr{C}_1,\mathscr{C}_2): \left| \cos \sphericalangle ({\bf
          S}_{1,\textnormal{P}}, {\bf V}_2) \right| >
      \epsilon \big| \mathcal{E}_{\bf S}^c }\\[2mm]
    & \leq & \Prv{ ({\bf S}_1, {\bf S}_2, \mathscr{C}_1,\mathscr{C}_2): \left|
        \frac{\pi}{2} - \sphericalangle ({\bf S}_{1,\textnormal{P}}, {\bf V}_2)
      \right| > \epsilon \Big| \mathcal{E}_{\bf S}^c
    }\\[2mm]
    & = & {\textsf E}_{{\bf S}_1, {\bf S}_2} \left[
      {\textnormal{Pr}}_{\mathscr{C}_1,\mathscr{C}_2} \bigg(
          \left| \frac{\pi}{2} - 
            \sphericalangle ({\bf s}_{1,\textnormal{P}}, {\bf V}_2) \right| >
          \epsilon \bigg| ({\bf S}_1, {\bf S}_2) = ({\bf s}_1,
        {\bf s}_2), \mathcal{E}_{\bf S}^c \bigg) \right],
  \end{IEEEeqnarray*}
  where in the last line we have denoted by
  ${\textnormal{Pr}}_{\mathscr{C}_1, \mathscr{C}_2} \left( \cdots |
    \cdots \right)$ the conditional probability of the codebooks
  $\mathscr{C}_1$ and $\mathscr{C}_2$ being such that $| \pi/2 -
  \sphericalangle ({\bf s}_{1,\textnormal{P}}, {\bf V}_2)| > \epsilon$
  given $({\bf S}_1, {\bf S}_2) = ({\bf s}_1, {\bf s}_2)$ and $({\bf
    s}_1, {\bf s}_2) \in \mathcal{E}_{\bf S}^c$. To conclude our bound
  we now notice that conditioned on $({\bf S}_1, {\bf S}_2) = ({\bf
    s}_1, {\bf s}_2)$, the random vector ${\bf V}_2/\| {\bf V}_2 \|$
  is distributed uniformly on the surface of the centered
  $\Reals^{n-1}$-sphere of unit radius that lies in the subspace that
  is orthogonal to ${\bf s}_2$. Hence,
  \begin{IEEEeqnarray*}{rCl}
    \Prv{\mathcal{A}_2 | \mathcal{E}_{\bf S}^c } & \leq & {\textsf E}_{{\bf
        S}_1, {\bf S}_2} \left[
      \frac{2C_{n-1}(\pi/2 - \epsilon)}{C_{n-1}(\pi)} \Big|
      \mathcal{E}_{\bf S}^c \right]\\[2mm]
    & \leq & \frac{2C_{n-1}(\pi/2-\epsilon)}{C_{n-1}(\pi)}\\[2mm]
    & \leq & \frac{2\Gamma \left( \frac{n+1}{2} \right)}{(n-1) \Gamma
      \left( \frac{n}{2}\right) \sqrt{\pi}}
    \frac{\sin^{(n-2)}(\pi/2 - \epsilon)}{\cos (\pi/2-\epsilon)}\\[2mm]
    & \leq & \frac{2\Gamma \left( \frac{n+1}{2} \right)}{(n-1) \Gamma
      \left( \frac{n}{2}\right) \sqrt{\pi} \cos (\pi/2-\epsilon)}.
  \end{IEEEeqnarray*}
  Upper bounding the ratio of Gamma functions by the asymptotic series
  of Lemma \ref{lm:gamma-series}, gives for every $\epsilon > 0$ that
  $\Prv{\mathcal{A}_2 | \mathcal{E}_{\bf S}^c} \rightarrow 0$ as $n
  \rightarrow \infty$. By similar arguments it also follows that
  $\Prv{\mathcal{A}_3 | \mathcal{E}_{\bf S}^c} \rightarrow 0$ as $n
  \rightarrow \infty$.

  To conclude the proof of Lemma \ref{lm:VQ-innerU-sv}, we derive the
  bound on $\mathcal{A}_4$. The derivations are similar to those for
  $\mathcal{A}_2$. First, define by ${\bf v}_{1,\textnormal{P}}$ the
  projection of ${\bf v}_1$ onto the subspace of $\Reals^n$ that is
  orthogonal to ${\bf s}_2$. As in \eqref{eq:lm4-bd1} we can show that
  \begin{IEEEeqnarray}{rCl}
    \left| \frac{1}{\| {\bf u}_1^{\ast} \| \| {\bf u}_2^{\ast} \|}
      \inner{{\bf v}_1}{{\bf v}_2} \right|
    & \leq & \left| \cos \sphericalangle ({\bf v}_{1,\textnormal{P}}, {\bf v}_2)
    \right|, \label{eq:lm4-bd3}
  \end{IEEEeqnarray}
  from which it then follows, using $| \cos x| \leq |\pi/2 - x |$, that
  \begin{IEEEeqnarray*}{rCl}
    \Prv{\mathcal{A}_4 | \mathcal{E}_{\bf S}^c } & \leq & 
    {\textsf E}_{{\bf S}_1, {\bf S}_2, \mathscr{C}_1} \bigg[
    {\textnormal{Pr}}_{\mathscr{C}_2} \bigg( \left| \frac{\pi}{2} -
      \sphericalangle ({\bf v}_{1,\textnormal{P}}, {\bf V}_2) \right|
    > \epsilon \bigg| ({\bf S}_1, {\bf S}_2) = ({\bf s}_1, {\bf s}_2),
    \mathscr{C}_1 = \mathcal{C}_1, \mathcal{E}_{\bf S}^c \bigg) \bigg].
  \end{IEEEeqnarray*}
  The desired bound now follows from noticing that conditioned on
  $({\bf S}_1, {\bf S}_2) = ({\bf s}_1, {\bf s}_2)$ and $\mathscr{C}_1
  = \mathcal{C}_1$, the random vector ${\bf V}_2/\| {\bf V}_2 \|$ is
  distributed uniformly on the surface of the centered
  $\Reals^{n-1}$-sphere of unit radius that lies in the subspace that
  is orthogonal to ${\bf s}_2$. Hence, similarly as in the derivation
  for $\mathcal{A}_2$
  \begin{IEEEeqnarray*}{rCl}
    \Prv{\mathcal{A}_4 | \mathcal{E}_{\bf S}^c } & \leq & {\textsf E}_{{\bf
        S}_1, {\bf S}_2, \mathscr{C}_1} \left[
      \frac{2C_{n-1}(\pi/2 - \epsilon)}{C_{n-1}(\pi)} \bigg|
      \mathcal{E}_{\bf S}^c \right]\\[3mm]
    & \leq & \frac{2\Gamma \left( \frac{n+1}{2} \right)}{(n-1) \Gamma
      \left( \frac{n}{2}\right) \sqrt{\pi} \cos (\pi/2-\epsilon)}.
  \end{IEEEeqnarray*}
  Upper bounding the ratio of Gamma functions by the asymptotic series
  of Lemma \ref{lm:gamma-series}, gives for every $\epsilon > 0$ that
  $\Prv{\mathcal{A}_4 | \mathcal{E}_{\bf S}^c} \rightarrow 0$ as $n
  \rightarrow \infty$.
\end{proof}

Combining Lemma \ref{lm:VQ-innerU-s1s2} and Lemma
\ref{lm:VQ-innerU-sv} with \eqref{eq:E23-dec} gives that for every
$\delta > 0$ and $0.3 > \epsilon > 0$ there exists an
  $n_{\mathcal{A}}'(\delta,\epsilon)$ such that for all $n >
  n_{\mathcal{A}}'(\delta,\epsilon)$
\begin{IEEEeqnarray*}{rCl}
  \hspace{39mm} \Prv{\mathcal{E}_{({\bf X}_1,{\bf X}_2)} \cap
    \mathcal{E}_{\bf S}^c \cap \mathcal{E}_{{\bf X}_1}^c \cap \mathcal{E}_{{\bf X}_2}^c} &
  \leq & 3 \delta. \hspace{39mm} \qed
\end{IEEEeqnarray*}

\subsubsection{Proof of Lemma \ref{lm:caps-single-error}}\label{app:proof-dec-single-error}

The proof follows from upper bounding $\Prv{\mathcal{G} |
  \mathcal{E}_{{\bf X}_1}^c}$ as a function of $R_1$. First, note that
\begin{IEEEeqnarray}{rCl}\label{eq:lm5-symm}
  \Prv{\mathcal{G} | \mathcal{E}_{{\bf X}_1}^c} & = & \Prv{\mathcal{G} | {\bf
      U}_1^{\ast} \neq {\bf 0}} \nonumber\\
  & = & \Prv{\mathcal{G} | \varsigma_1({\bf S}_1, \mathscr{C}_1) = 1},
\end{IEEEeqnarray}
where the second equality holds because the conditional distribution
of the codewords conditional on ${\bf u}_1^{\ast} \neq {\bf 0}$ is
invariant with respect to permutations of the indexing of the
codewords.
The desired upper bound is now obtained by decomposing $\mathcal{G}$
into sub-events $\mathcal{G}_j$, $j \in \{ 2,3, \ldots
,2^{nR_1} \}$, where
\begin{IEEEeqnarray*}{rCl}
  \mathcal{G}_j & \triangleq & \left\{ ({\bf s}_1, {\bf s}_2,
    \mathcal{C}_1, \mathcal{C}_2, {\bf z}): \cos \sphericalangle ({\bf
    w}, {\bf u}_1(j)) \geq \Delta \right\}.
\end{IEEEeqnarray*}
By \eqref{eq:lm5-symm} we now have 
\begin{IEEEeqnarray}{rCl}
  \Prv{\mathcal{G} | \mathcal{E}_{{\bf X}_1}^c } & = &
  \Prv{\bigcup_{j=2}^{2^{nR_1}} \mathcal{G}_j \Bigg| \varsigma_1({\bf
      S}_1, \mathscr{C}_1) = 1}
  \nonumber\\[2mm]
  & \leq & \sum_{j=2}^{2^{nR_1}} \Prv{\mathcal{G}_j | \varsigma_1({\bf
      S}_1, \mathscr{C}_1) = 1} \nonumber\\[2mm]
  & < & 2^{nR_1} \Prv{\mathcal{G}_2 | \varsigma_1({\bf S}_1,
    \mathscr{C}_1) = 1} \nonumber\\[3mm]
  & \leq & 2^{nR_1} \cdot 2 \frac{C_n(\arccos
    \Delta)}{C_n(\pi)}, \label{eq:bd1-single-error}
\end{IEEEeqnarray}
where in the third step we have used that $\Prv{\mathcal{G}_j |
  \varsigma_i({\bf s}_i, \mathcal{C}_i) = 1}$ is the same for all $j
\in \{ 2,3, \ldots 2^{nR_1} \}$ because the conditional distribution
of ${\bf u}_1(j)$ given $\varsigma_i({\bf s}_i, \mathcal{C}_i) = 1$
does not depend on $j \in \left\{ 2,3, \ldots ,2^{nR_1}\right\}$ and
where in the last step we have upper bounded the density of ${\bf
  U}_1(2)$, conditional on $\varsigma_1({\bf S}_1, \mathscr{C}_1) =
1$, by Lemma \ref{lm:density-wrong-codeword}. Thus, combining
\eqref{eq:bd1-single-error} with Lemma \ref{lm:bounds-polar-caps} gives
\begin{IEEEeqnarray*}{rCl}
  \Prv{\mathcal{G} | \mathcal{E}_{{\bf X}_1}^c} & \leq & 2^{nR_1} \cdot
  2 \frac{\Gamma \left( \frac{n}{2} +1\right)
    (1-\Delta^2)^{(n-1)/2}}{n \Gamma \left( \frac{n+1}{2}\right)
    \sqrt{\pi} \Delta}\\[2mm]
  & = & \frac{2 \Gamma \left( \frac{n}{2}
      +1\right)}{n \Gamma \left( \frac{n+1}{2}\right) \sqrt{\pi}
    \Delta \sqrt{1-\Delta^2}} 2^{n (R_1 + 1/2 \log_2(1-\Delta^2))}
\end{IEEEeqnarray*}
Replacing the ratio of the Gamma functions by the asymptotic series of
Lemma \ref{lm:gamma-series} establishes
\eqref{eq:lm-wrong-codeword}. \hfill $\qed$

\subsubsection{Proof of Lemma \ref{lm:caps-double-error}}
\label{app:proof-dec-double-error}


The proof follows by upper bounding $\Prv{\mathcal{G} |
  \mathcal{E}_{{\bf X}_1}^c \cap \mathcal{E}_{{\bf X}_2}^c}$ as a
function of $R_1 + R_2$. To this end, define
\begin{IEEEeqnarray*}{rCl}
  \tilde{\varsigma}({\bf s}_1, {\bf s}_2, \mathcal{C}_1,
  \mathcal{C}_2) & \triangleq & \left( \varsigma_1({\bf
      s}_1,\mathcal{C}_1), \varsigma_2({\bf s}_2,\mathcal{C}_2)\right).
\end{IEEEeqnarray*}
By a symmetry argument, which is similar to the one in the proof of
Lemma \ref{lm:caps-single-error}, we obtain
\begin{IEEEeqnarray}{rCl}\label{eq:dbl-err-symm}
  \Prv{\mathcal{G} | \mathcal{E}_{{\bf X}_1}^c \cap \mathcal{E}_{{\bf X}_2}^c} & = &
  \Prv{\mathcal{G} | \tilde{\varsigma}({\bf S}_1, {\bf S}_2,
    \mathscr{C}_1, \mathscr{C}_2) = (1,1)}.
\end{IEEEeqnarray}
The desired upper bound is now obtained by decomposing $\mathcal{G}$
into subevents $\mathcal{G}_{j,\ell}$, where
\begin{IEEEeqnarray*}{rCl}
  \mathcal{G}_{j,\ell} & = & \left\{ ({\bf s}_1, {\bf s}_2, \mathcal{C}_1,
    \mathcal{C}_2, {\bf z}): \cos \sphericalangle ({\bf u}_1(j), {\bf u}_2(\ell)) \geq
  \Theta, \cos \sphericalangle ({\bf y}, \alpha_1 {\bf u}_1(j) +
  \alpha_2 {\bf u}_2(\ell)) \geq \Delta \right\},
\end{IEEEeqnarray*}
for $j \in \{ 2,3, \ldots 2^{nR_1}\}$ and $\ell \in \{ 2,3, \ldots 2^{nR_2}\}$.
Hence, by \eqref{eq:dbl-err-symm}
\begin{IEEEeqnarray}{rCl}
  \Prv{\mathcal{G} | \mathcal{E}_{{\bf X}_1}^c \cap \mathcal{E}_{{\bf X}_2}^c} & = &
  \Prv{\bigcup_{j=2}^{2^{nR_1}} \bigcup_{\ell=2}^{2^{nR_2}}
    \mathcal{G}_{j,\ell} \Bigg| \tilde{\varsigma}({\bf S}_1, {\bf S}_2,
    \mathscr{C}_1, \mathscr{C}_2) = (1,1) }\nonumber\\
  & \leq & \sum_{j=2}^{2^{nR_1}} \sum_{\ell=2}^{2^{nR_2}}
  \Prv{\mathcal{G}_{j,\ell} \big| \tilde{\varsigma}({\bf S}_1, {\bf
      S}_2, \mathscr{C}_1, \mathscr{C}_2) = (1,1)}\nonumber\\
  & \stackrel{a)}{<} & 2^{n(R_1+R_2)} \Prv{\mathcal{G}_{2,2} |
    \tilde{\varsigma}({\bf S}_1, {\bf S}_2,
    \mathscr{C}_1, \mathscr{C}_2) = (1,1)},\label{eq:bd-jnt-err-prb-1}
\end{IEEEeqnarray}
where $a)$ follows since conditioned on $\tilde{\varsigma}({\bf S}_1,
{\bf S}_2, \mathscr{C}_1, \mathscr{C}_2) = (1,1)$, the laws of
$\sphericalangle ({\bf U}_1(j), {\bf U}_2(\ell))$ and $\sphericalangle
({\bf Y}, \alpha_1 {\bf U}_1(j) + \alpha_2 {\bf U}_2(\ell))$ do not
depend on $j \in \{ 2,3, \ldots , 2^{nR_1} \}$ or $\ell \in \{
2,3, \ldots , 2^{nR_2} \}$. We now rewrite the probability
$\Prv{\mathcal{G}_{2,2} | \tilde{\varsigma}({\bf S}_1, {\bf S}_2,
  \mathscr{C}_1, \mathscr{C}_2) = (1,1)}$:
\begin{IEEEeqnarray*}{l}
  \Prv{\mathcal{G}_{2,2} | \tilde{\varsigma}({\bf S}_1, {\bf S}_2,
    \mathscr{C}_1, \mathscr{C}_2) = (1,1)}\\[3mm]
  \hspace{5mm} = \text{Pr} \big[ ({\bf S}_1, {\bf S}_2, {\bf U}_1(1), {\bf
    U}_2(1), {\bf U}_1(2), {\bf U}_2(2), {\bf Z}) \text{ are such
    that } \mathcal{G}_{2,2} \text{ occurs}\\
  \hspace{90mm} \big| \tilde{\varsigma}({\bf
    S}_1, {\bf S}_2, \mathscr{C}_1, \mathscr{C}_2) = (1,1) \big]\\[3mm]
  \hspace{5mm} = \int_{\begin{subarray}{l}
      ({\bf s}_1, {\bf s}_2) \in \Reals^n \times \Reals^n\\
      ({\bf u}_1, {\bf u}_2) \in \mathcal{S}_1 \times \mathcal{S}_2\\
      {\bf z} \in \Reals^n
    \end{subarray}} f \big( ({\bf S}_1, {\bf S}_2, {\bf U}_1(1), {\bf
    U}_2(1), {\bf Z})= ( {\bf s}_1, {\bf s}_2, {\bf u}_1, {\bf u}_2,
  {\bf z}) \\[-3mm]
  \hspace{90mm} \big| \tilde{\varsigma}({\bf S}_1, {\bf S}_2, \mathscr{C}_1,
  \mathscr{C}_2) = (1,1)\big)\\[3mm]
  \hspace{15mm} \text{Pr} \big[ \cos \sphericalangle ({\bf U}_1(2), {\bf U}_2(2)) \geq
  \Theta, \cos \sphericalangle ({\bf y}, \alpha_1 {\bf U}_1(2) + \alpha_2
  {\bf U}_2(2)) \geq \Delta\\[1mm]
  \hspace{25mm} \big| \tilde{\varsigma}({\bf S}_1, {\bf S}_2, \mathscr{C}_1,
  \mathscr{C}_2) = (1,1), ({\bf S}_1, {\bf S}_2, {\bf U}_1(1), {\bf
    U}_2(1), {\bf Z})= ( {\bf s}_1, {\bf s}_2, {\bf u}_1, {\bf u}_2,
  {\bf z}) \big]\\[1mm]
  \hspace{15mm}\d ({\bf s}_1, {\bf s}_2, {\bf u}_1, {\bf u}_2,
  {\bf z})\\[5mm]
  \hspace{5mm} = \int_{\begin{subarray}{l}
      ({\bf s}_1, {\bf s}_2) \in \Reals^n \times \Reals^n\\
      ({\bf u}_1, {\bf u}_2) \in \mathcal{S}_1 \times \mathcal{S}_2
    \end{subarray}} f \big( ({\bf S}_1, {\bf S}_2, {\bf U}_1(1), {\bf
    U}_2(1))= ( {\bf s}_1, {\bf s}_2, {\bf u}_1, {\bf u}_2) \\[-3mm]
  \hspace{90mm} \big| \tilde{\varsigma}({\bf S}_1, {\bf S}_2, \mathscr{C}_1,
  \mathscr{C}_2) = (1,1)\big)\\[3mm]
  \hspace{15mm} \text{Pr} \big[ \cos \sphericalangle ({\bf U}_1(2), {\bf U}_2(2)) \geq
  \Theta, \cos \sphericalangle ({\bf y}, \alpha_1 {\bf U}_1(2) + \alpha_2
  {\bf U}_2(2)) \geq \Delta\\[1mm]
  \hspace{25mm} \big| \tilde{\varsigma}({\bf S}_1, {\bf S}_2, \mathscr{C}_1,
  \mathscr{C}_2) = (1,1), ({\bf S}_1, {\bf S}_2, {\bf U}_1(1), {\bf
    U}_2(1))= ( {\bf s}_1, {\bf s}_2, {\bf u}_1, {\bf u}_2) \big]\\[1mm]
  \hspace{15mm}\d ({\bf s}_1, {\bf s}_2, {\bf u}_1, {\bf u}_2),
\end{IEEEeqnarray*}
where in the last step we have used that the probability term does not
depend on ${\bf z}$. To upper bound the integral we now upper bound
this probability term.
\begin{IEEEeqnarray}{l}
  \text{Pr} \Big[ \cos \sphericalangle ({\bf U}_1(2), {\bf U}_2(2)) \geq
  \Theta, \cos \sphericalangle ({\bf y}, \alpha_1 {\bf U}_1(2) + \alpha_2
  {\bf U}_2(2)) \geq \Delta\nonumber\\[1mm]
  \hspace{25mm} \Big| \tilde{\varsigma}({\bf S}_1, {\bf S}_2, \mathscr{C}_1,
  \mathscr{C}_2) = (1,1), ({\bf S}_1, {\bf S}_2, {\bf U}_1(1), {\bf
    U}_2(1))= ( {\bf s}_1, {\bf s}_2, {\bf u}_1, {\bf u}_2) \Big]\nonumber\\[3mm]
  \hspace{5mm} = \int_{\begin{subarray}{l}
      (\tilde{\bf u}_1, \tilde{\bf u}_2) \in
      \mathcal{S}_1 \times \mathcal{S}_2:\nonumber\\
      \cos \sphericalangle (\tilde{\bf u}_1, \tilde{\bf u}_2) \geq
      \Theta,\nonumber\\
      \cos \sphericalangle ({\bf y}, \alpha_1 \tilde{\bf u}_1 + \alpha_2
      \tilde{\bf u}_2) \geq \Delta
    \end{subarray}} \hspace{-4mm}f^{\lambda_1 \times \lambda_2}
  \big( ({\bf U}_1(2), {\bf U}_2(2)) = (\tilde{\bf u}_1, \tilde{\bf
    u}_2) \big| \tilde{\varsigma}({\bf S}_1, {\bf S}_2, \mathscr{C}_1,
  \mathscr{C}_2) = (1,1),\nonumber\\[-6mm]
  \hspace{60mm} ({\bf S}_1, {\bf S}_2, {\bf U}_1(1), {\bf
    U}_2(1))= ( {\bf s}_1, {\bf s}_2, {\bf u}_1, {\bf u}_2) \big)
  \d (\tilde{\bf u}_1, \tilde{\bf u}_2)\nonumber\\[5mm]
  \hspace{5mm} = \int_{\begin{subarray}{l}
      (\tilde{\bf u}_1, \tilde{\bf u}_2) \in
      \mathcal{S}_1 \times \mathcal{S}_2:\nonumber\\
      \cos \sphericalangle (\tilde{\bf u}_1, \tilde{\bf u}_2) \geq
      \Theta,\nonumber\\
      \cos \sphericalangle ({\bf y}, \alpha_1 \tilde{\bf u}_1 + \alpha_2
      \tilde{\bf u}_2) \geq \Delta
    \end{subarray}}
  \hspace{-4mm}f^{\lambda_1}
  \big( {\bf U}_1(2) = \tilde{\bf u}_1 \big|
  \tilde{\varsigma}({\bf S}_1, {\bf S}_2, \mathscr{C}_1,
  \mathscr{C}_2) = (1,1),\nonumber\\[-4mm]
  \hspace{68mm}({\bf S}_1, {\bf S}_2, {\bf 
    U}_1(1), {\bf U}_2(1)) = ( {\bf s}_1, {\bf s}_2, {\bf u}_1, {\bf
    u}_2 ) \big)\nonumber\\[3mm]
  \hspace{39mm} \cdot f^{\lambda_2}
  \big( {\bf U}_2(2) = \tilde{\bf u}_2 \big|
  \tilde{\varsigma}({\bf S}_1, {\bf S}_2, \mathscr{C}_1,
  \mathscr{C}_2) = (1,1),\nonumber\\
  \hspace{68mm}({\bf S}_1, {\bf S}_2, {\bf 
    U}_1(1), {\bf U}_2(1)) = ( {\bf s}_1, {\bf s}_2, {\bf u}_1, {\bf
    u}_2 ) \big) \d (\tilde{\bf u}_1,
  \tilde{\bf u}_2)\nonumber\\[5mm]
  \hspace{5mm} = \int_{\begin{subarray}{l}
      (\tilde{\bf u}_1, \tilde{\bf u}_2) \in
      \mathcal{S}_1 \times \mathcal{S}_2:\nonumber\\
      \cos \sphericalangle (\tilde{\bf u}_1, \tilde{\bf u}_2) \geq
      \Theta,\nonumber\\
      \cos \sphericalangle ({\bf y}, \alpha_1 \tilde{\bf u}_1 + \alpha_2
      \tilde{\bf u}_2) \geq \Delta
    \end{subarray}}
  \hspace{-4mm} f^{\lambda_1}
  \left( {\bf U}_1(2) = \tilde{\bf u}_1 |
    \varsigma_1({\bf S}_1, \mathscr{C}_1,) = 1, ({\bf S}_1, {\bf
      U}_1(1))= ( {\bf s}_1, {\bf u}_1 ) \right)\nonumber\\[3mm]
  \hspace{39mm} \cdot f^{\lambda_2}
  \left( {\bf U}_2(2) = \tilde{\bf u}_2 |
    \varsigma_2({\bf S}_2, \mathscr{C}_2,) = 1, ({\bf S}_2, {\bf
      U}_2(1))= ( {\bf s}_2, {\bf u}_2 ) \right) \d (\tilde{\bf u}_1,
  \tilde{\bf u}_2)\nonumber\\[7mm]
  \hspace{5mm} \stackrel{a)}{\leq} \int_{\begin{subarray}{l}
      (\tilde{\bf u}_1, \tilde{\bf u}_2) \in
      \mathcal{S}_1 \times \mathcal{S}_2:\nonumber\\
      \cos \sphericalangle (\tilde{\bf u}_1, \tilde{\bf u}_2) \geq
      \Theta,\nonumber\\
      \cos \sphericalangle ({\bf y}, \alpha_1 \tilde{\bf u}_1 + \alpha_2
      \tilde{\bf u}_2) \geq \Delta
    \end{subarray}} \frac{2}{r_1^{n-1} C_n(\pi)} \cdot
  \frac{2}{r_2^{n-1} C_n(\pi)} \d (\tilde{\bf u}_1, \tilde{\bf
    u}_2)\nonumber\\[8mm]
  \hspace{5mm} \stackrel{b)}{=} 4 \frac{C_n(\arccos \Theta)}{C_n(\pi)}
  \cdot \frac{C_n(\arccos \Delta)}{C_n(\pi)},\label{eq:bd-jnt-err-prb-2}
\end{IEEEeqnarray}
where in $a)$ we have used Lemma \ref{lm:density-wrong-codeword} and
in $b)$ we have used that under distributions of ${\bf U}_1(2)$ and
${\bf U}_2(2)$ that are independent of ${\bf y}$ and uniform over
$\mathcal{S}_1$ and $\mathcal{S}_2$ respectively, the angles
$\sphericalangle ({\bf U}_1(j), {\bf U}_2(\ell))$ and $\sphericalangle
({\bf y}, \alpha_1 {\bf U}_1(j) + \alpha_2 {\bf U}_2(\ell))$ are
independent.
Thus, combining \eqref{eq:bd-jnt-err-prb-2} with
\eqref{eq:bd-jnt-err-prb-1} gives
\begin{IEEEeqnarray}{rCl}
  \Prv{\mathcal{G} | \mathcal{E}_{{\bf X}_1}^c \cap \mathcal{E}_{{\bf X}_2}^c}
  & < & 2^{n(R_1+R_2)} \cdot 4 \frac{C_n(\arccos
    \Theta)}{C_n(\pi)} \cdot \frac{C_n(\arccos
    \Delta)}{C_n(\pi)}. \label{eq:bd-two-cap}
\end{IEEEeqnarray}
And combining \eqref{eq:bd-two-cap} with Lemma
\ref{lm:bounds-polar-caps} gives
\begin{IEEEeqnarray*}{l}
  \Prv{\mathcal{G} | \mathcal{E}_{{\bf X}_1}^c \cap \mathcal{E}_{{\bf
        X}_2}^c}\\[3mm]
  \hspace{8mm} < 2^{n(R_1+R_2)} \cdot 4 \frac{\Gamma \left( \frac{n}{2}
      +1 \right) (1-\Theta^2)^{(n-1)/2}}{n \Gamma \left(
      \frac{n+1}{2}\right) \sqrt{\pi} \Theta} \cdot \frac{\Gamma
    \left( \frac{n}{2} +1\right) (1-\Delta^2)^{(n-1)/2}}{n
    \Gamma \left( \frac{n+1}{2}\right) \sqrt{\pi} \Delta}\\[2mm]
  \hspace{8mm} = 4 \cdot \left( \frac{\Gamma \left( \frac{n}{2} +1\right)}{n
      \Gamma \left( \frac{n+1}{2}\right) \sqrt{\pi}} \right)^2
  \frac{1}{\Theta \sqrt{1-\Theta^2} \Delta \sqrt{1-\Delta^2}} 2^{n
    \left( R_1 + R_2 + \frac{1}{2}\log_2 \left((1-\Theta^2)
        (1-\Delta^2)\right) \right)}.
\end{IEEEeqnarray*}
Replacing the ratios of the Gamma-functions by their asymptotic series
in Lemma \ref{lm:gamma-series} finally establishes
\eqref{eq:pr1-wrong-cdw-pair}. \hfill $\qed$


\section{Proof of Theorem
  \ref{thm:nofb-asymptotics}} \label{appx:prf-thm-asympt-nofb}

The high-SNR asymptotics for the multiple-access problem without
feedback can be obtained from the necessary condition for the
achievability of a distortion pair $(D_1,D_2)$ in Theorem
\ref{thm:nofb-nec-cond}, and from the sufficient conditions for the
achievability of a distortion pair $(D_1,D_2)$ deriving from the
vector-quantizer scheme in Theorem \ref{thm:vq-achv}.

By Theorem \ref{thm:vq-achv} it follows that any distortion pair
$(\bar{D}_1,\bar{D}_2)$ satisfying $\bar{D}_1 \leq \sigma^2$,
$\bar{D}_2 \leq \sigma^2$ and
\begin{IEEEeqnarray}{rCl}
  \bar{D}_1 & \geq & \sigma^2 \frac{N}{P_1}\label{eq:mac-hsnr-achv-D1}\\
  \bar{D}_2 & \geq & \sigma^2 \frac{N}{P_2}\label{eq:mac-hsnr-achv-D2}\\
  \bar{D}_1 \bar{D}_2 & = & \sigma^4
  \frac{N (1-\breve{\rho}^2)}{P_1 + P_2 +
    2\breve{\rho}\sqrt{P_1P_2}}, \label{eq:mac-hsnr-achv-D1D2}
\end{IEEEeqnarray}
where
\begin{IEEEeqnarray}{rCl}\label{eq:mac-hsnr-rho-breve}
  \breve{\rho} & = & \rho \sqrt{\left( 1 -
      \frac{\bar{D}_1}{\sigma^2}\right) \left( 1 -
      \frac{\bar{D}_2}{\sigma^2}\right) },
\end{IEEEeqnarray}
is achievable. If
\begin{IEEEeqnarray}{rCl}\label{eq:D1D2-asymptotic}
  \lim_{N \rightarrow 0} \frac{N}{P_1 \bar{D}_1} = 0 & \hspace{15mm}
  \text{and} \hspace{15mm} & \lim_{N \rightarrow 0} \frac{N}{P_2
    \bar{D}_2} = 0,
\end{IEEEeqnarray}
then for $N$ sufficiently small, \eqref{eq:mac-hsnr-achv-D1} and
\eqref{eq:mac-hsnr-achv-D2} are satisfied. Consequently, for $N$
sufficiently small any pair satisfying \eqref{eq:mac-hsnr-achv-D1D2}
and \eqref{eq:D1D2-asymptotic} is achievable. We next show that if the
pair $(\bar{D}_1, \bar{D}_2)$ satisfies \eqref{eq:mac-hsnr-achv-D1D2}
and \eqref{eq:D1D2-asymptotic}, then $\breve{\rho} \rightarrow \rho$
as $N \rightarrow 0$. To show this, we note that if $(\bar{D}_1,
\bar{D}_2)$ satisfies \eqref{eq:mac-hsnr-achv-D1D2} then
\begin{IEEEeqnarray}{rCl}\label{eq:macfb-hsnr-ub-Dibar}
  \bar{D}_2 \leq \sigma^4 \frac{N}{P_1\bar{D}_1}, & \hspace{15mm} \text{and}
  \hspace{15mm} & \bar{D}_1 \leq \sigma^4 \frac{N}{P_2\bar{D}_2}.
\end{IEEEeqnarray}
Combining \eqref{eq:macfb-hsnr-ub-Dibar} with
\eqref{eq:mac-hsnr-rho-breve} gives that if in addition to
\eqref{eq:mac-hsnr-achv-D1D2} the pair $(\bar{D}_1, \bar{D}_2)$ also
satisfies \eqref{eq:macfb-hsnr-ub-Dibar}, then $\breve{\rho}
\rightarrow \rho$ as $N \rightarrow 0$. Thus, if $(\bar{D}_1,
\bar{D}_2)$ satisfies \eqref{eq:mac-hsnr-achv-D1D2} and
\eqref{eq:D1D2-asymptotic}, then
\begin{IEEEeqnarray}{rCl}\label{eq:mac-hsnr-achv}
  \lim_{N \rightarrow 0} \frac{P_1 + P_2 + 2\rho \sqrt{P_1P_2}}{N}
  \bar{D}_1 \bar{D}_2 & \leq & \sigma^4 (1-\rho^2).
\end{IEEEeqnarray}

Now, let $(D_1^{\ast}(\sigma^2,\rho,P_1,P_2,N),
D_2^{\ast}(\sigma^2,\rho,P_1,P_2,N))$ be a distortion pair resulting
from an optimal scheme and let $(D_1^{\ast}, D_2^{\ast})$ be the
shorthand notation for this distortion pair. By Theorem
\ref{thm:nofb-nec-cond} we have that
\begin{IEEEeqnarray}{rCl}\label{eq:mac-hsnr1}
  R_{S_1,S_2}(D_1^{\ast},D_2^{\ast}) & \leq & \frac{1}{2} \log_2 \left( 1+ \frac{P_1
      + P_2 + 2 \rho \sqrt{P_1P_2}}{N} \right).
\end{IEEEeqnarray}
If $(D_1^{\ast}, D_2^{\ast})$ satisfies
\begin{IEEEeqnarray}{rCl}\label{eq:D1D2*-asymptotic}
  \lim_{N \rightarrow 0} \frac{N}{P_1 D_1^{\ast}} = 0 & \hspace{15mm}
  \text{and} \hspace{15mm} & \lim_{N \rightarrow 0} \frac{N}{P_2
    D_2^{\ast}} = 0,
\end{IEEEeqnarray}
then for $N$ sufficiently small
\begin{IEEEeqnarray}{rCl}\label{eq:mac-hsnr-RS1S2-D2}
  R_{S_1,S_2}(D_1^{\ast},D_2^{\ast}) & = & \frac{1}{2} \log_2^+ \left(
    \frac{\sigma^4(1-\rho^2)}{D_1^{\ast}D_2^{\ast}} \right),
\end{IEEEeqnarray}
by Theorem \ref{thm:rd1d2-main} and because $(D_1^{\ast},D_2^{\ast})
\in \mathscr{D}_2$. From \eqref{eq:mac-hsnr1} and
\eqref{eq:mac-hsnr-RS1S2-D2} we thus get that if
$(D_1^{\ast},D_2^{\ast})$ satisfies \eqref{eq:D1D2*-asymptotic}, then
\begin{IEEEeqnarray}{rCl}\label{eq:mac-hsnr-opt2}
  \lim_{N \rightarrow 0} \frac{P_1 + P_2 + 2\rho
    \sqrt{P_1P_2}}{N} D_1^{\ast} D_2^{\ast} & \geq & \sigma^4
  (1-\rho^2).
\end{IEEEeqnarray}
Combining \eqref{eq:mac-hsnr-achv} with \eqref{eq:mac-hsnr-opt2}
yields Theorem \ref{thm:nofb-asymptotics}. \hfill $\Box$

\section{Proof of Theorem \ref{thm:superimposed}} \label{appx:prf-thm-supimp}

Our analysis of the expected distortion for the superimposed scheme is
based on a genie-aided argument, similar as in the analysis of the
vector-quantizer scheme. This argument is described more precisely
now.

\subsection{Genie-Aided Scheme}

In our genie-aided argument, the genie assists the decoder. An
illustration of this decoder is given in Figure
\ref{fig:si-dec-genie}.
\begin{figure}[h]
  \centering
  \psfrag{y}[cc][cc]{${\bf Y}$}
  \psfrag{u1}[cc][cc]{$\hat{\bf U}_1$}
  \psfrag{u2}[cc][cc]{$\hat{\bf U}_2$}
  \psfrag{u1g}[cc][cc]{${\bf U}_1^{\ast}$}
  \psfrag{u2g}[cc][cc]{${\bf U}_2^{\ast}$}
  \psfrag{s1}[cc][cc]{$\hat{\bf S}_1^{\textnormal{G}}$}
  \psfrag{s2}[cc][cc]{$\hat{\bf S}_2^{\textnormal{G}}$}
  \epsfig{file=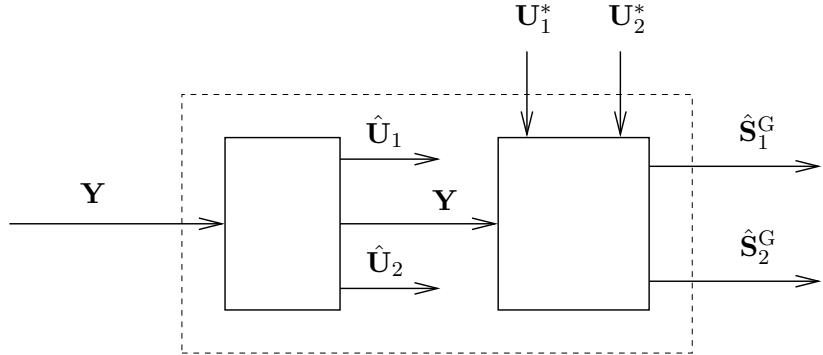, width=0.75\textwidth}
  \caption{Genie-aided decoder.}
  \label{fig:si-dec-genie}
\end{figure}
In addition to the channel output ${\bf Y}$ that is observed
originally, the decoder is now also provided with the transmitted
codeword pair $({\bf U}_1^{\ast}, {\bf U}_2^{\ast})$. Based on $({\bf
  U}_1^{\ast}, {\bf U}_2^{\ast})$ and ${\bf Y}$, the decoder then
estimates the source pair $({\bf S}_1, {\bf S}_2)$ and thereby ignores
the guess $(\hat{\bf U}_1, \hat{\bf U}_2)$ produced in the first step
of the original decoder. The estimate of this genie-aided decoder is
denoted by $(\hat{\bf S}_1^{\textnormal{G}}, \hat{\bf
  S}_2^{\textnormal{G}})$ and is given by
\begin{IEEEeqnarray}{rCl}
  \hat{\bf S}_1^{\textnormal{G}} & = & \gamma_{11} {\bf U}_1^{\ast} +
  \gamma_{12} {\bf U}_2^{\ast} + \gamma_{13} {\bf Y}\label{eq:si-genie-lin-est-S1h}\\[2mm]
  \hat{\bf S}_2^{\textnormal{G}} & = & \gamma_{21} {\bf U}_2^{\ast} +
  \gamma_{22} {\bf U}_1^{\ast} + \gamma_{23} {\bf Y},
\end{IEEEeqnarray}
where the coefficients $\gamma_{ij}$ are as defined in
\eqref{eq:src-estmt-coeff}. We now show that under certain rate
constraints, the normalized asymptotic distortion of this genie-aided
scheme is the same as for the originally proposed scheme. The key
argument is stated in the following proposition.

\begin{prp}\label{prp:si-D1-eql-genie}
  For every $\delta > 0$ and $0 < \epsilon < 0.3$ there exists an
  $n'(\delta,\epsilon)>0$ such that for all $n > n'(\delta,\epsilon)$,
  \begin{IEEEeqnarray}{rCl}\label{eq:prp-D1genie}
    \frac{1}{n}\E{\| {\bf S}_1 - \hat{\bf S}_1 \|^2} & \leq &
    \frac{1}{n}\E{\| {\bf S}_1 - \hat{\bf S}_1^{\textnormal{G}} \|^2}
    + \xi_1' \delta + \xi_2' \epsilon,
  \end{IEEEeqnarray}
  whenever $(R_1,R_2)$ is in the rate region $\mathcal{R}'(\epsilon)$
  given by
  \begin{IEEEeqnarray*}{rClrCl}
    \mathcal{R}'(\epsilon) & = & \bigg\{ & R_1 & \leq & \frac{1}{2} \log_2
    \left( \frac{\beta_1'^2 \|{\bf U}_1\|^2 (1-\tilde{\rho}^2) +
        N'}{N'(1-\tilde{\rho}^2)} - \kappa_1 \epsilon \right),\\[2mm]
    & & & R_2 & \leq & \frac{1}{2} \log_2
    \left( \frac{\beta_2'^2 \|{\bf U}_2\|^2 (1-\tilde{\rho}^2) +
        N'}{N'(1-\tilde{\rho}^2)} - \kappa_2 \epsilon \right),\\[2mm]
    & & & R_1 + R_2 & \leq & \frac{1}{2} \log_2 \left(
      \frac{\beta_1'^2 \|{\bf U}_1\|^2 + \beta_1'^2 \|{\bf U}_1\|^2 +
        2\tilde{\rho} \beta_1' \beta_2' \|{\bf U}_1\| \|{\bf U}_2\| +
        N' }{N' (1-\tilde{\rho}^2)} - \kappa_3 \epsilon \right) \bigg\},
  \end{IEEEeqnarray*}
  where in \eqref{eq:prp-D1genie} $\xi_1'$ and $\xi_2'$ depend only on
  $\sigma^2$, $\gamma_{13}$, $P_1$, $P_2$ and $N$, and where in the
  expression of $\mathcal{R}'(\epsilon)$ the terms $\kappa_1$,
  $\kappa_2$ and $\kappa_3$ depend only on $P_1$, $P_2$, $N'$, and
  $\tilde{\rho}$, and where $N'$ and $\beta_1'$, $\beta_2'$ are as
  given in \eqref{eq:si-N'}, \eqref{eq:si-beta1'} and
  \eqref{eq:si-beta2'} respectively.
\end{prp}

\begin{proof}
  See Section \ref{sec:si-prf-prp-genie}.
\end{proof}

From Proposition \ref{prp:si-D1-eql-genie} it now follows easily that
the expected distortion asymptotically achievable by the genie-aided
scheme is the same as the expected distortion achievable by the
original scheme.

\begin{cor}\label{cor:si-genie}
  If $(R_1,R_2)$ satisfy
  \begin{IEEEeqnarray*}{rCl}
    R_1 & < & \frac{1}{2} \log_2
    \left( \frac{\beta_1'^2 \|{\bf U}_1\|^2 (1-\tilde{\rho}^2) +
        N'}{N'(1-\tilde{\rho}^2)} \right),\\[2mm]
    R_2 & < & \frac{1}{2} \log_2
    \left( \frac{\beta_2'^2 \|{\bf U}_2\|^2 (1-\tilde{\rho}^2) +
        N'}{N'(1-\tilde{\rho}^2)} \right),\\[2mm]
    R_1 + R_2 & < & \frac{1}{2} \log_2 \left(
      \frac{\beta_1'^2 \|{\bf U}_1\|^2 + \beta_1'^2 \|{\bf U}_1\|^2 +
        2\tilde{\rho} \beta_1' \beta_2' \|{\bf U}_1\| \|{\bf U}_2\| +
        N' }{N' (1-\tilde{\rho}^2)} \right),
  \end{IEEEeqnarray*}
  then
  \begin{IEEEeqnarray*}{rCl}
    \varlimsup_{n \rightarrow \infty} \frac{1}{n}\E{\| {\bf S}_1 -
      \hat{\bf S}_1 \|^2} & \leq & \varlimsup_{n \rightarrow
      \infty}\frac{1}{n}\E{\| {\bf S}_1 - \hat{\bf S}_1^{\textnormal{G}}
      \|^2}.
  \end{IEEEeqnarray*}
\end{cor}

\begin{proof}
  Follows from Proposition \ref{prp:si-D1-eql-genie} by first letting
  $n \rightarrow \infty$ and then $\epsilon \rightarrow 0$ and $\delta
  \rightarrow 0$.
\end{proof}

It follows by Corollary \ref{cor:si-genie} that to analyze the
distortion achievable by our scheme it suffices to analyze the
genie-aided scheme. This is done in Section \ref{sec:si-ub-D1}.

\subsection{Proof of Proposition \ref{prp:si-D1-eql-genie}}\label{sec:si-prf-prp-genie}

The proof of Proposition \ref{prp:si-D1-eql-genie} consists of upper
bounding the difference between ${\frac{1}{n} \mat{E} \big[ \| {\bf
    S}_1 - \hat{\bf S}_1 \|^2 \big]}$ and $\frac{1}{n} \mat{E} \big[
\| {\bf S}_1 - \hat{\bf S}_1^{\textnormal{G}} \|^2 \big]$. Since the
two estimates $\hat{\bf S}_1$ and $\hat{\bf S}_1^{\textnormal{G}}$
differ only if $(\hat{\bf U}_1, \hat{\bf U}_2) \neq ({\bf U}_1^{\ast},
{\bf U}_2^{\ast})$, the main step is to upper bound the
probability of a decoding error. This is what we do now.

Let the error event $\mathcal{E}_{\hat{\bf U}}$ be as defined in
\eqref{eq:E41} -- \eqref{eq:E43} for the vector-quantizer scheme. The
probability of $\mathcal{E}_{\hat{\bf U}}$ is upper bounded in the
following lemma.
\begin{lm}\label{lm:si-pr-dec-err}
For every $\delta > 0$ and $0 < \epsilon < 0.3$, there exists an
$n_4'(\delta,\epsilon) \in \Naturals$ such that for all $n >
n_4'(\delta,\epsilon)$
\begin{IEEEeqnarray*}{rCl}
  \Prv{\mathcal{E}_{\hat{\bf U}}} < 11 \delta & \qquad \qquad &
  \text{whenever } (R_1,R_2) \in \mathcal{R}'(\epsilon).
\end{IEEEeqnarray*}
\end{lm}

\begin{proof}
  The proof follows from restating the decoding problem for the
  superimposed scheme in the form of the decoding problem for the
  vector-quantizer scheme. That is, we seek to rewrite the channel
  output in the form
  \begin{IEEEeqnarray}{rCl}\label{eq:si-vq-rep-Y}
    {\bf Y} 
    & = & \beta_1' {\bf U}_1^{\ast} + \beta_2' {\bf
      U}_2^{\ast} + {\bf Z}',
  \end{IEEEeqnarray}
  with an additive noise sequence ${\bf Z}'$ that satisfies the
  properties needed for the analysis of the vector-quantizer
  scheme. This representation is obtained by first rewriting the
  source sequences as
  \begin{IEEEeqnarray}{rCl}
    {\bf S}_1 & = & (1 - a_1 \tilde{\rho}) {\bf U}_1^{\ast} + a_1 {\bf
      U}_2^{\ast} + {\bf W}_1 \label{eq:si-decomp-S1}\\
    {\bf S}_2 & = & (1 - a_2 \tilde{\rho}) {\bf U}_2^{\ast} + a_2 {\bf
      U}_1^{\ast} + {\bf W}_2, \label{eq:si-decomp-S2}
  \end{IEEEeqnarray}
  where $a_1$ is defined in \eqref{eq:si-a1}, $a_2$ is defined in
  \eqref{eq:si-a2}, and $\tilde{\rho}$ is defined in
  \eqref{eq:vq-tilde-rho}. Combining \eqref{eq:si-decomp-S1} and
  \eqref{eq:si-decomp-S2} with the expressions for ${\bf X}_1$ and
  ${\bf X}_2$ in \eqref{eq:si-Xi} and with ${\bf Y} = {\bf X}_1 + {\bf
    X}_2 + {\bf Z}$ yields the desired form of \eqref{eq:si-vq-rep-Y}
  with
  \begin{IEEEeqnarray*}{rCl}
    \beta_1' & = & \left( \alpha_1 (1 - a_1 \tilde{\rho}) + \beta_1 +
      \alpha_2 a_2 \right)\\
    \beta_2' & = & \left( \alpha_2 (1 - a_2 \tilde{\rho}) + \beta_2 +
      \alpha_1 a_1 \right),
  \end{IEEEeqnarray*}
  and with
  \begin{IEEEeqnarray*}{rCl}
    {\bf Z}' & = & \alpha_1 {\bf W}_1 + \alpha_2 {\bf W}_2 + {\bf Z}.
  \end{IEEEeqnarray*}
  For the additive noise sequence ${\bf Z}'$ it can now be verified
  that for every $\delta > 0$ and $\epsilon > 0$ there exists an
  $n'(\delta, \epsilon)>0$, such that for $N'$ as in \eqref{eq:si-N'}
  and for all $n > n'(\delta, \epsilon)$ we have that
  \begin{IEEEeqnarray}{rCl}\label{eq:noise-norm}
    \Prv{ \left| \frac{1}{n} \| {\bf Z}' \|^2 - N' \right| \leq N'
      \epsilon } & > & 1 - \delta,
  \end{IEEEeqnarray}
  and that
  \begin{IEEEeqnarray}{rCl}\label{eq:indep-noise}
    \Prv{ \big| \inner{{\bf U}_i^{\ast}}{{\bf Z}'} \big| \leq n
      \sqrt{\sigma^2(1-2^{-2R_i})N'} \epsilon} & > & 1 - \delta, \qquad i \in \{ 1,2 \}.
  \end{IEEEeqnarray}
  Condition \eqref{eq:indep-noise} follows since for $a_1$ and $a_2$,
  given in \eqref{eq:si-a1} and \eqref{eq:si-a2}, for sufficiently large
  $n$, we have with high probability that
  \begin{IEEEeqnarray*}{rCl}
    \inner{{\bf U}_i^{\ast}}{{\bf W}_j} & \approx & 0 \qquad \forall
    i,j \in \{ 1,2 \}.
  \end{IEEEeqnarray*}
  Conditions \eqref{eq:noise-norm} and \eqref{eq:indep-noise} are
  precisely those needed in the proof of the achievable rates for the
  vector-quantizer scheme. Hence, the upper bound on the probability
  of a decoding error in the vector-quantizer scheme given in Lemma
  \ref{lm:vq-Pr-E4} can be adopted to the superimposed scheme. This
  yields Lemma \ref{lm:si-pr-dec-err}.
\end{proof}

To ease the upper bounding of the difference between ${\frac{1}{n}
  \mat{E} \big[ \| {\bf S}_1 - \hat{\bf S}_1 \|^2 \big]}$ and
$\frac{1}{n} \mat{E} \big[ \| {\bf S}_1 - \hat{\bf
  S}_1^{\textnormal{G}} \|^2 \big]$ we now state three more lemmas
which upper bound different norms and inner products involving ${\bf
  S}_1$, $\hat{\bf S}_1$ and $\hat{\bf S}_1^{\textnormal{G}}$. The
first lemma gives an upper bound on the squared norm of $\hat{\bf S}_1
- \hat{\bf S}_1^{\textnormal{G}}$.

\begin{lm}\label{lm:si-ub-norm-s1h}
  Let the reconstructions $\hat{\bf S}_1$ and $\hat{\bf
    S}_1^{\textnormal{G}}$ be as defined in \eqref{eq:si-reconstr-Sih}
  and \eqref{eq:si-genie-lin-est-S1h}. Then, with probability one
  \begin{IEEEeqnarray*}{rCl}
    \| \hat{\bf S}_1 - \hat{\bf S}_1^{\textnormal{G}}\|^2 & \leq
    & 16 n \sigma^2.
  \end{IEEEeqnarray*}
\end{lm}

\begin{proof}
\begin{IEEEeqnarray*}{rCl}
    \hspace{20mm} \| \hat{\bf S}_1 - \hat{\bf S}_1^{\textnormal{G}}\|^2 & = &
    \| \gamma_{11}(\hat{\bf U}_1 - {\bf U}_1^{\ast}) +
      \gamma_{12}(\hat{\bf U}_2 - {\bf U}_2^{\ast})\|^2 \\[2mm]
    & = & \gamma_{11}^2 \|\hat{\bf U}_1 - {\bf U}_1^{\ast}\|^2 + 2
    \gamma_{11}\gamma_{12} \inner{\hat{\bf U}_1 - {\bf
          U}_1^{\ast}}{\hat{\bf U}_2 - {\bf U}_2^{\ast}} \\
    & & {} + \gamma_{12}^2 \|\hat{\bf U}_2 - {\bf
        U}_2^{\ast}\|^2 \\[2mm]
    & \leq & \gamma_{11}^2 \underbrace{\|\hat{\bf U}_1 - {\bf
          U}_1^{\ast}\|^2}_{\leq 4 n \sigma^2} + 2
    \gamma_{11}\gamma_{12} \underbrace{\|\hat{\bf U}_1 - {\bf
          U}_1^{\ast}\| \|\hat{\bf U}_2 - {\bf U}_2^{\ast}\|}_{\leq 4
      n \sigma^2} \\
    & & {} + \gamma_{12}^2 \underbrace{\|\hat{\bf U}_2 - {\bf
          U}_2^{\ast}\|^2}_{\leq 4 n \sigma^2} \\[2mm]
    & \leq & 4 n \sigma^2 \left( \gamma_{11} + \gamma_{12}
    \right)^2 \\[2mm]
    & \leq & 16 n \sigma^2,
  \end{IEEEeqnarray*}
  where in the last step we have used that $0 \leq \gamma_{11},
  \gamma_{12} \leq 1$.
\end{proof}

For the next two lemmas, we reuse the two error events
$\mathcal{E}_{\bf S}$ and $\mathcal{E}_{\bf Z}$ which were defined in
\eqref{eq:vq-E1-def} and \eqref{eq:vq-E3-def} for the proof of the
vector-quantizer scheme. We then have:

\begin{lm}\label{lm:si-bd-inner-S1-S1h-S1G}
  For every $\epsilon > 0$
  \begin{IEEEeqnarray*}{rCl}
    \frac{1}{n} \E{\inner{{\bf S}_1}{\hat{\bf S}_1^{\textnormal{G}} -
        \hat{\bf S}_1}} & \leq & \sigma^2 \big( \epsilon +
    17 \Prv{\mathcal{E}_{\bf S}} + (17 +\epsilon)
    \Prv{\mathcal{E}_{\hat{\bf U}}} \big).
  \end{IEEEeqnarray*}
\end{lm}

\begin{proof}
  
  \begin{IEEEeqnarray}{rCl}
    \frac{1}{n} \E{\inner{{\bf S}_1}{\hat{\bf S}_1^{\textnormal{G}} -
        \hat{\bf S}_1}} & = & \frac{1}{n} \E{\inner{{\bf
          S}_1}{\hat{\bf S}_1^{\textnormal{G}} - \hat{\bf S}_1} \Big|
      \mathcal{E}_{\bf S} \cup \mathcal{E}_{\hat{\bf U}}}
    \Prv{\mathcal{E}_{\bf S} \cup \mathcal{E}_{\hat{\bf U}}} \nonumber\\
    & & {} + \frac{1}{n} \underbrace{\E{\inner{{\bf S}_1}{\hat{\bf
            S}_1^{\textnormal{G}} - \hat{\bf S}_1} \Big|
        \mathcal{E}_{\bf S}^c \cap \mathcal{E}_{\hat{\bf U}}^c}}_{=0}
    \Prv{\mathcal{E}_{\bf S}^c
      \cap \mathcal{E}_{\hat{\bf U}}^c} \nonumber\\[2mm]
    & \stackrel{a)}{\leq} & \frac{1}{n} \E{\|{\bf S}_1\|^2 +
      \|\hat{\bf S}_1^{\textnormal{G}} - \hat{\bf S}_1\|^2 \Big|
      \mathcal{E}_{\bf S} \cup \mathcal{E}_{\hat{\bf U}}}
    \Prv{\mathcal{E}_{\bf S} \cup \mathcal{E}_{\hat{\bf U}}}
    \nonumber\\[3mm]
    & = & \frac{1}{n} \E{\| {\bf S}_1 \|^2 \big| \mathcal{E}_{\bf S}}
    \Prv{\mathcal{E}_{\bf S}} + \frac{1}{n} \E{\| {\bf S}_1 \|^2 \big|
      \mathcal{E}_{\bf S}^c \cap \mathcal{E}_{\hat{\bf U}}}
    \Prv{\mathcal{E}_{\bf S}^c \cap \mathcal{E}_{\hat{\bf U}}} \nonumber\\
    & & {} + \frac{1}{n} \E{\|\hat{\bf S}_1^{\textnormal{G}} -
      \hat{\bf S}_1\|^2 \Big| \mathcal{E}_{\bf S} \cup
      \mathcal{E}_{\hat{\bf U}}} \Prv{\mathcal{E}_{\bf S} \cup
      \mathcal{E}_{\hat{\bf U}}}\nonumber\\[2mm]
    & \stackrel{b)}{\leq} & \sigma^2 (\epsilon + \Prv{\mathcal{E}_{\bf
        S}}) + \sigma^2 (1+\epsilon) \Prv{\mathcal{E}_{\hat{\bf U}}}
    \nonumber\\[2mm]
    & & {} + 16 \sigma^2 \left( \Prv{\mathcal{E}_{\bf S}} +
      \Prv{\mathcal{E}_{\hat{\bf U}}} \right) \nonumber\\[3mm]
    & \leq & \sigma^2 \big( \epsilon + 17 \Prv{\mathcal{E}_{\bf S}} +
    (17 +\epsilon) \Prv{\mathcal{E}_{\hat{\bf U}}}
    \big), \label{eq:si-inner-genie}
  \end{IEEEeqnarray}
  where in the first equality, the second expectation term equals zero
  because by $\mathcal{E}_{\hat{\bf U}}^c$ we have $\| \hat{\bf
    s}_1^{\textnormal{G}} - \hat{\bf s}_1 \| = 0$ and by
  $\mathcal{E}_{\bf S}^c$ the norm $\|{\bf s}_1\|$ is bounded. In $a)$
  we have used \eqref{eq:vq-bound inner}, and in $b)$ we have used
  Lemma~\ref{lm:vq-D-E1}, Lemma~\ref{lm:si-ub-norm-s1h} and the fact
  that conditioned on $\mathcal{E}_{\bf S}^c$ we have $\| {\bf s}_1
  \|^2 \leq n \sigma^2 (1+\epsilon)$.
\end{proof}

\begin{lm}\label{lm:si-bd-S1h-S1G}
  For every $\epsilon > 0$
  \begin{IEEEeqnarray*}{rCl}
    \frac{1}{n} \E{\|\hat{\bf S}_1\|^2 - \|\hat{\bf
        S}_1^{\textnormal{G}}\|^2} & \leq & 4 \Big(
    \sigma^2(1+4\gamma_{13}) + 9 \gamma_{13} (P_1+P_2+N) (1 +
    \epsilon) \Big) \Prv{\mathcal{E}_{\hat{\bf U}}}\\
    & & {} + 36 \gamma_{13} \Big( (P_1 + P_2)\Prv{\mathcal{E}_{\bf S}}
    + N \Prv{\mathcal{E}_{\bf Z}} + (P_1 + P_2 + N) \epsilon \Big).
  \end{IEEEeqnarray*}
\end{lm}

\begin{proof}
  \begin{IEEEeqnarray}{rCl}
    \frac{1}{n} \E{\|\hat{\bf S}_1\|^2 - \|\hat{\bf
        S}_1^{\textnormal{G}}\|^2} & = & \frac{1}{n} \E{\|\hat{\bf S}_1\|^2
      - \|\hat{\bf S}_1^{\textnormal{G}}\|^2 \Big| \mathcal{E}_{\hat{\bf U}}}
    \Prv{\mathcal{E}_{\hat{\bf U}}} \nonumber\\
    & & {} + \frac{1}{n} \E{\|\hat{\bf S}_1\|^2 - \|\hat{\bf
        S}_1^{\textnormal{G}}\|^2 \Big| \mathcal{E}_{\hat{\bf U}}^c}
    \Prv{\mathcal{E}_{\hat{\bf U}}^c} \nonumber\\[2mm]
    & \leq & \frac{1}{n} \E{\|\hat{\bf S}_1\|^2 - \|\hat{\bf
        S}_1^{\textnormal{G}}\|^2 \Big| \mathcal{E}_{\hat{\bf U}}}
    \Prv{\mathcal{E}_{\hat{\bf U}}}, \label{eq:si-diff-sh-shG0}
  \end{IEEEeqnarray}
  where the last inequality follows since conditional on
  $\mathcal{E}_{\hat{\bf U}}^c$ we have $\hat{\bf S}_1 = \hat{\bf
    S}_1^{\textnormal{G}}$ and therefore $\|\hat{\bf S}_1\|^2 -
  \|\hat{\bf S}_1^{\textnormal{G}}\|^2 = 0$. To upper bound the RHS of
  \eqref{eq:si-diff-sh-shG0}, we now upper bound the difference $\|
  \hat{\bf S}_1 \|^2 - \| \hat{\bf S}_1^{\textnormal{G}} \|^2$:
  \begin{IEEEeqnarray}{rCl}
    \| \hat{\bf S}_1 \|^2 - \| \hat{\bf S}_1^{\textnormal{G}} \|^2 & =
    & \gamma_{11}^2 \| \hat{\bf U}_1 \|^2 + 2 \gamma_{11} \gamma_{12}
    \inner{\hat{\bf U}_1}{\hat{\bf U}_2} + 2 \gamma_{11} \gamma_{13}
    \inner{\hat{\bf U}_1}{{\bf Y}} \nonumber\\
    & & {} + \gamma_{12}^2 \| \hat{\bf U}_2 \|^2 + 2 \gamma_{12} \gamma_{13}
    \inner{\hat{\bf U}_2}{{\bf Y}} + \gamma_{13}^2 \| {\bf Y} \|^2 \nonumber\\
    & & {} - \gamma_{11}^2 \| {\bf U}_1^{\ast} \|^2 - 2 \gamma_{11} \gamma_{12}
    \inner{{\bf U}_1^{\ast}}{{\bf U}_2^{\ast}} - 2 \gamma_{11} \gamma_{13}
    \inner{{\bf U}_1^{\ast}}{{\bf Y}} \nonumber\\
    & & {} - \gamma_{12}^2 \| {\bf U}_2^{\ast} \|^2 - 2 \gamma_{12} \gamma_{13}
    \inner{{\bf U}_2^{\ast}}{{\bf Y}} - \gamma_{13}^2 \| {\bf
      Y}\|^2 \nonumber\\[3mm]
    & = & 2 \gamma_{11} \gamma_{12}
    \underbrace{\inner{\hat{\bf U}_1}{\hat{\bf U}_2}}_{\leq n
      \sigma^2} + 2 \gamma_{11} \gamma_{13} \inner{\hat{\bf U}_1}{{\bf
        Y}} + 2 \gamma_{12} \gamma_{13} \inner{\hat{\bf U}_2}{{\bf Y}} \nonumber\\
    & & {} - 2 \gamma_{11} \gamma_{12}
    \underbrace{\inner{{\bf U}_1^{\ast}}{{\bf U}_2^{\ast}}}_{\geq - n
      \sigma^2} - 2 \gamma_{11} \gamma_{13} \inner{{\bf
        U}_1^{\ast}}{{\bf Y}} - 2 \gamma_{12} \gamma_{13}
    \inner{{\bf U}_2^{\ast}}{{\bf Y}} \nonumber\\[3mm]
    & \leq & 4 \gamma_{11} \gamma_{12} n \sigma^2 + 2 \gamma_{11}
    \gamma_{13} \inner{\hat{\bf U}_1 - {\bf U}_1^{\ast}}{{\bf Y}} \nonumber\\
    & & {} + 2 \gamma_{12} \gamma_{13} \inner{\hat{\bf U}_2 - {\bf
        U}_2^{\ast}}{{\bf Y}} \nonumber\\[3mm]
    & \stackrel{a)}{\leq} & 4 \gamma_{11} \gamma_{12} n \sigma^2 + 2 \gamma_{11}
    \gamma_{13} \big( \underbrace{\| \hat{\bf U}_1 - {\bf
          U}_1^{\ast}\|^2}_{\leq 4 n \sigma^2} +
      \|{\bf Y}\|^2 \big) \nonumber\\
    & & {} + 2 \gamma_{12} \gamma_{13} \big( \underbrace{\| \hat{\bf U}_2 -
      {\bf U}_2^{\ast}\|^2}_{\leq 4 n \sigma^2} + \|{\bf Y}\|^2
    \big) \nonumber\\[3mm]
    & \leq & 4 \gamma_{11} \gamma_{12} n \sigma^2 + 2 \gamma_{13}
    (\gamma_{11} + \gamma_{12}) \big( 4 n \sigma^2 + \| {\bf Y} \|^2
    \big) \nonumber\\[3mm]
    & \leq & 4 n \sigma^2 + 4 \gamma_{13} \big( 4 n \sigma^2 + \| {\bf
      Y} \|^2 \big), \label{eq:si-diff-sh-shG1}
  \end{IEEEeqnarray}
  where in $a)$ we have used \eqref{eq:vq-bound inner}, and in the
  last inequality we have used that $0 \leq \gamma_{11}, \gamma_{12} \leq
  1$. We now upper bound the squared norm of ${\bf Y}$ on the RHS of
  \eqref{eq:si-diff-sh-shG1} in terms of ${\bf S}_1$, ${\bf S}_2$,
  ${\bf U}_1^{\ast}$, ${\bf U}_2^{\ast}$ and ${\bf Z}$:
  \begin{IEEEeqnarray}{rCl}
    \| {\bf Y} \|^2 & \leq & \alpha_1^2 \| {\bf S}_1 \|^2 + 2 
     \inner{\alpha_1{\bf S}_1}{\beta_1{\bf U}_1^{\ast}} + 2
    \inner{\alpha_1{\bf S}_1}{\alpha_2{\bf S}_2} + 2 
    \inner{\alpha_1{\bf S}_1}{\beta_2{\bf U}_2^{\ast}} \nonumber\\
    & & {} + 2 \inner{\alpha_1{\bf S}_1}{{\bf Z}} + \beta_1^2 \| {\bf
      U}_1^{\ast} \|^2 + 2 \inner{\beta_1{\bf
        U}_1^{\ast}}{\alpha_2{\bf S}_2} + 2 \inner{\beta_1{\bf
        U}_1^{\ast}}{\beta_2{\bf U}_2^{\ast}} \nonumber\\
    & & {} + 2 \inner{\beta_1{\bf U}_1^{\ast}}{{\bf Z}} + \alpha_2^2
    \| {\bf S}_2 \|^2 + 2 \inner{\alpha_2{\bf S}_2}{\beta_2{\bf
        U}_2^{\ast}} + 2 \inner{\alpha_2{\bf S}_2}{{\bf Z}} \nonumber\\
    & & {} + \beta_2^2 \| {\bf U}_2^{\ast} \|^2 + 2
    \inner{\beta_2{\bf U}_2^{\ast}}{{\bf Z}} + \| {\bf Z} \|^2 \nonumber\\[3mm]
    & \stackrel{a)}{\leq} & 9 \left( \alpha_1^2 \| {\bf S}_1 \|^2 + \alpha_2^2 \| {\bf
        S}_2 \|^2 + \beta_1^2 \| {\bf U}_1^{\ast} \|^2 + \beta_2^2 \|
      {\bf U}_2^{\ast} \|^2 + \| {\bf Z} \|^2 \right) \nonumber\\[3mm]
    & \leq & 9 \left( \alpha_1^2 \| {\bf S}_1 \|^2 + \alpha_2^2 \| {\bf
        S}_2 \|^2 + nP_1 + nP_2 + \| {\bf Z} \|^2 \right). \label{eq:si-norm-Y}
  \end{IEEEeqnarray}
  where $a)$ follows from upper bounding all inner products by
  \eqref{eq:vq-bound inner}. Thus, combining \eqref{eq:si-norm-Y} with
  \eqref{eq:si-diff-sh-shG1} gives
  \begin{IEEEeqnarray}{rCl}
    \| \hat{\bf S}_1 \|^2 - \| \hat{\bf S}_1^{\textnormal{G}} \|^2 &
    \leq & 4 n \sigma^2 + 16 \gamma_{13} n \sigma^2 + 36 n \gamma_{13}
    (P_1 + P_2) \nonumber\\
    & & + 36 \gamma_{13} \left( \alpha_1^2 \| {\bf S}_1 \|^2 + \alpha_2^2 \| {\bf
        S}_2 \|^2 + \| {\bf Z} \|^2 \right). \quad \label{eq:si-diff-sh-shG2}
  \end{IEEEeqnarray}
  And combining \eqref{eq:si-diff-sh-shG2} with
  \eqref{eq:si-diff-sh-shG0} gives
  \begin{IEEEeqnarray}{rCl}
    \frac{1}{n} \E{\|\hat{\bf S}_1\|^2 - \|\hat{\bf
        S}_1^{\textnormal{G}}\|^2} & \leq & 4 \sigma^2
    \Prv{\mathcal{E}_{\hat{\bf U}}} + 16 \gamma_{13} \sigma^2
    \Prv{\mathcal{E}_{\hat{\bf U}}} \nonumber\\ 
    & & {} + 36 \gamma_{13} (P_1 + P_2)
    \Prv{\mathcal{E}_{\hat{\bf U}}} \nonumber\\
    & & {} + 36 \gamma_{13} \Big( \alpha_1^2 \frac{1}{n} \E{\| {\bf S}_1 \|^2 \big|
      \mathcal{E}_{\hat{\bf U}}} \Prv{\mathcal{E}_{\hat{\bf U}}} \nonumber\\
    & & \hspace{16mm} {} + \alpha_2^2 \frac{1}{n} \E{\| {\bf S}_2 \|^2\big|
      \mathcal{E}_{\hat{\bf U}}} \Prv{\mathcal{E}_{\hat{\bf U}}} \nonumber\\
    & & \hspace{16mm} {} + \frac{1}{n} \E{\| {\bf Z} \|^2 \big|
      \mathcal{E}_{\hat{\bf U}}} \Prv{\mathcal{E}_{\hat{\bf U}}} \Big).
    \label{eq:si-diff-sh-shG3}
  \end{IEEEeqnarray}
  It now remains to upper bound the expectations on ${\bf S}_1$, ${\bf
    S}_2$ and ${\bf Z}$ on the RHS of
  \eqref{eq:si-diff-sh-shG3}. Since ${\bf S}_1$, ${\bf S}_2$ and ${\bf
    Z}$ are each Gaussian, their corresponding terms can be bounded in
  similar ways. We show here the derivation for ${\bf S}_1$.
  \begin{IEEEeqnarray}{rCl}
    \frac{1}{n}\E{\| {\bf S}_1 \|^2 \big| \mathcal{E}_{\hat{\bf U}}}
    \Prv{\mathcal{E}_{\hat{\bf U}}} & = & \frac{1}{n}\E{\| {\bf S}_1 \|^2 \big|
      \mathcal{E}_{\hat{\bf U}} \cap \mathcal{E}_{{\bf S}}}
    \Prv{\mathcal{E}_{\hat{\bf U}} \cap \mathcal{E}_{{\bf S}}} \nonumber\\
    & & {} + \frac{1}{n}\E{\| {\bf S}_1 \|^2 \big| \mathcal{E}_{\hat{\bf U}} \cap
      \mathcal{E}_{{\bf S}}^c} \Prv{\mathcal{E}_{\hat{\bf U}} \cap
      \mathcal{E}_{{\bf S}}^c} \nonumber\\[3mm]
    & \leq & \frac{1}{n}\E{\| {\bf S}_1 \|^2 \big| \mathcal{E}_{{\bf S}}}
    \Prv{\mathcal{E}_{{\bf S}}} \nonumber\\
    & & {} + \sigma^2 (1 + \epsilon) \Prv{\mathcal{E}_{\hat{\bf
          U}}} \nonumber\\[3mm]
    & \leq & \sigma^2 (\epsilon + \Prv{\mathcal{E}_{\bf S}}) +
    \sigma^2 (1 + \epsilon) \Prv{\mathcal{E}_{\hat{\bf
          U}}}, \label{eq:si-ub-ES1-EUh}
  \end{IEEEeqnarray}
  where in the last step we have used Lemma \ref{lm:vq-D-E1}. For the
  expectations on ${\bf S}_2$ and ${\bf Z}$, we similarly obtain
  \begin{IEEEeqnarray}{rCl}
    \frac{1}{n}\E{\| {\bf S}_2 \|^2 \big| \mathcal{E}_{\hat{\bf U}}}
    \Prv{\mathcal{E}_{\hat{\bf U}}} & \leq & \sigma^2 (\epsilon +
    \Prv{\mathcal{E}_{\bf S}}) + \sigma^2 (1 + \epsilon)
    \Prv{\mathcal{E}_{\hat{\bf U}}}, \label{eq:si-ub-ES2-EUh}
  \end{IEEEeqnarray}
  and
  \begin{IEEEeqnarray}{rCl}
    \frac{1}{n}\E{\| {\bf Z} \|^2 \big| \mathcal{E}_{\hat{\bf U}}}
    \Prv{\mathcal{E}_{\hat{\bf U}}} & \leq & N (\epsilon +
    \Prv{\mathcal{E}_{\bf Z}}) + N (1 + \epsilon)
    \Prv{\mathcal{E}_{\hat{\bf U}}}. \label{eq:si-ub-EZ-EUh}
  \end{IEEEeqnarray}
  Thus, combining \eqref{eq:si-ub-ES1-EUh} -- \eqref{eq:si-ub-EZ-EUh}
  with \eqref{eq:si-diff-sh-shG3} gives
  \begin{IEEEeqnarray*}{rCl}
    \hspace{1mm} \frac{1}{n} \E{\|\hat{\bf S}_1\|^2 - \|\hat{\bf
        S}_1^{\textnormal{G}}\|^2} & \leq & 4 \sigma^2
    \Prv{\mathcal{E}_{\hat{\bf U}}} + 16 \gamma_{13} \sigma^2
    \Prv{\mathcal{E}_{\hat{\bf U}}} \nonumber\\ 
    & & {} + 36 \gamma_{13} (P_1 + P_2)
    \Prv{\mathcal{E}_{\hat{\bf U}}} \nonumber\\
    & & {} + 36 \gamma_{13} \Big( (P_1+P_2) (\epsilon +
    \Prv{\mathcal{E}_{\bf S}})\\ 
    & & \hspace{17mm} {} + (P_1+P_2+N) (1 + \epsilon)
      \Prv{\mathcal{E}_{\hat{\bf U}}} \nonumber\\
    & & \hspace{17mm} {} + N (\epsilon +
    \Prv{\mathcal{E}_{\bf Z}}) \Big) \nonumber\\[2mm]
    & \leq & 4 \Big( \sigma^2(1+4\gamma_{13}) + 9 \gamma_{13}
      (P_1+P_2+N) (1 + \epsilon) \Big) \Prv{\mathcal{E}_{\hat{\bf
          U}}}\\
    & & {} + 36 \gamma_{13} \Big( (P_1 + P_2)\Prv{\mathcal{E}_{\bf S}}
    + N \Prv{\mathcal{E}_{\bf Z}} + (P_1 + P_2 + N) \epsilon
    \Big). \hspace{1mm} \qedhere
  \end{IEEEeqnarray*}
\end{proof}

Based on Lemma \ref{lm:si-bd-inner-S1-S1h-S1G} and Lemma
\ref{lm:si-bd-S1h-S1G}, the proof of Proposition
\ref{prp:si-D1-eql-genie} now follows easily.


\begin{proof}[Proof of Proposition \ref{prp:si-D1-eql-genie}]
  
  \begin{IEEEeqnarray}{l}
    \frac{1}{n}\E{\| {\bf S}_1 - \hat{\bf S}_1 \|^2} -
    \frac{1}{n}\E{\| {\bf S}_1 - \hat{\bf S}_1^{\textnormal{G}} \|^2} \nonumber\\[2mm]
    \hspace{35mm} = \frac{1}{n}\E{\| {\bf S}_1 - \hat{\bf S}_1 \|^2
      -\| {\bf
        S}_1 - \hat{\bf S}_1^{\textnormal{G}} \|^2} \nonumber\\[2mm]
    \hspace{35mm} = \frac{1}{n} \bigg( \E{\|{\bf S}_1\|^2} - 2
    \E{\inner{{\bf S}_1}{\hat{\bf
          S}_1}} + \E{\|\hat{\bf S}_1\|^2} \nonumber\\
    \hspace{44mm} {} - \E{\|{\bf S}_1\|^2} + 2 \E{\inner{{\bf
          S}_1}{\hat{\bf S}_1^{\textnormal{G}}}} - \E{\|\hat{\bf
        S}_1^{\textnormal{G}}\|^2} \bigg) \nonumber\\[2mm]
    \hspace{35mm} = 2 \frac{1}{n} \E{\inner{{\bf S}_1}{\hat{\bf
          S}_1^{\textnormal{G}} - \hat{\bf S}_1}} + \frac{1}{n}
    \E{\|\hat{\bf S}_1\|^2 - \|\hat{\bf S}_1^{\textnormal{G}}\|^2} \nonumber\\[2mm]
    \hspace{35mm} \stackrel{a)}{\leq} 2 \sigma^2 \left( \epsilon + 17
      \Prv{\mathcal{E}_{\bf S}} + (17 + \epsilon)
      \Prv{\mathcal{E}_{\hat{\bf U}}} \right) \nonumber\\
    \hspace{39mm} {} + 4 \Big( \sigma^2(1+4\gamma_{13}) + 9
    \gamma_{13} (P_1+P_2+N) (1 + \epsilon) \Big)
    \Prv{\mathcal{E}_{\hat{\bf
          U}}} \nonumber\\
    \hspace{39mm} {} + 36 \gamma_{13} \Big( (P_1 +
    P_2)\Prv{\mathcal{E}_{\bf S}} + N \Prv{\mathcal{E}_{\bf Z}} + (P_1
    + P_2 + N) \epsilon \Big)
    \nonumber\\[2mm]
    \hspace{35mm} = \xi_1 \Prv{\mathcal{E}_{\hat{\bf U}}} + \xi_2
    \Prv{\mathcal{E}_{{\bf S}}} + \xi_3 \Prv{\mathcal{E}_{{\bf Z}}} +
    \xi_4 \epsilon, \label{eq:si-diff-genie}
  \end{IEEEeqnarray}
  where in step $a)$ we have used Lemma
  \ref{lm:si-bd-inner-S1-S1h-S1G} and Lemma \ref{lm:si-bd-S1h-S1G},
  and where $\xi_{\ell}$, $\ell \in \{ 1,2,3,4 \}$, depend only on
  $\sigma^2$, $\gamma_{13}$, $P_1$, $P_2$ and $N$. Combining
  \eqref{eq:si-diff-genie} with Lemma \ref{lm:vq-Pr-E4}, Lemma
  \ref{lm:vq-Pr-E1} and Lemma \ref{lm:vq-Pr-E3} gives that for every
  $\delta > 0$ and $0.3 > \epsilon > 0$, there exists an
  $n'(\delta,\epsilon) > 0$ such that for all $(R_1,R_2) \in
  \mathcal{R}'(\epsilon)$ and $n > n'(\delta,\epsilon)$
  \begin{IEEEeqnarray*}{rCl}
    \frac{1}{n}\E{\| {\bf S}_1 - \hat{\bf S}_1 \|^2} -
    \frac{1}{n}\E{\| {\bf S}_1 - \hat{\bf S}_1^{\textnormal{G}} \|^2}
    & < & \xi_1' \delta + \xi_42' \epsilon,
  \end{IEEEeqnarray*}
  where $\xi_1'$ and $\xi_2'$ depend only on $\sigma^2$,
  $\gamma_{13}$, $P_1$, $P_2$ and $N$.
\end{proof}

\subsection{Upper Bound on Expected Distortion}\label{sec:si-ub-D1}

We now derive an upper bound on the achievable distortion for the
proposed vector-quantizer scheme. By Corollary \ref{cor:si-genie}, it
suffices to analyze the genie-aided scheme. Using that $\hat{\bf
  S}_1^{\textnormal{G}} = \gamma_{11} {\bf U}_1^{\ast} + \gamma_{12}
{\bf U}_2^{\ast} + \gamma_{13} {\bf Y}$, we have
\begin{IEEEeqnarray}{rCl}
  \frac{1}{n} \E{\| {\bf S}_1 - \hat{\bf S}_1^{\textnormal{G}} \|^2} &
  = \frac{1}{n} \Big( & \E{\|{\bf S}_1\|^2} - 2 \gamma_{11}
  \E{\inner{{\bf S}_1}{{\bf U}_1^{\ast}}} - 2 \gamma_{12}
  \E{\inner{{\bf S}_1}{{\bf U}_2^{\ast}}} \nonumber\\
  & & {} - 2 \gamma_{13} \E{\inner{{\bf S}_1}{{\bf Y}}} +
  \gamma_{11}^2 \E{\|{\bf U}_1^{\ast}\|^2} + 2 \gamma_{11} \gamma_{12}
  \E{\inner{{\bf U}_1^{\ast}}{{\bf U}_2^{\ast}}} \nonumber\\
  & & {} + 2 \gamma_{11} \gamma_{13}
  \E{\inner{{\bf U}_1^{\ast}}{{\bf Y}}} + \gamma_{12}^2 \E{\|{\bf
      U}_2^{\ast}\|^2} \nonumber\\
  & & {} + 2 \gamma_{12} \gamma_{13} \E{\inner{{\bf U}_2^{\ast}}{{\bf
        Y}}} + \gamma_{13}^2 \E{\|{\bf Y}\|^2} \Big). \label{eq:si-dist1-S1-S1G}
\end{IEEEeqnarray}
Some of the expectation terms are bounded straightforwardly. In
particular, we have $\E{\|{\bf S}_1\|^2} = n \sigma^2$, $\E{\| {\bf
    U}_1^{\ast} \|^2} = n \sigma^2 (1-2^{-2R_1})$, and $\E{\| {\bf
    U}_2^{\ast} \|^2} = n \sigma^2 (1-2^{-2R_2})$. For three further
terms we take over the bounds from the analysis of the
vector-quantizer scheme. That is, by Lemma \ref{lm1:vq-bd-D1genie} we
have that for every $\delta > 0$ and $0 < \epsilon < 0.3$ and every
positive integer $n$
\begin{IEEEeqnarray}{rCl}
  \frac{1}{n}\E{\inner{{\bf S}_1}{{\bf U}_1^{\ast}}} & \geq & \sigma^2
  (1-2^{-2R_1}) - \zeta_1 (\delta,\epsilon) \nonumber\\
  & = & \mat{c}_{11} - \zeta_1 (\delta,\epsilon), \label{eq:si-lb-S1U1}
\end{IEEEeqnarray}
where $\zeta_1(\delta,\epsilon)$ is such that $\lim_{\delta,\epsilon
  \rightarrow 0} \zeta_1(\delta,\epsilon) = 0$. By Lemma
\ref{lm2:vq-bd-D1genie} we have that for every $\delta > 0$ and $0 <
\epsilon < 0.3$ there exists an $n_2'(\delta,\epsilon) \in \Naturals$
such that for all $n > n_2'(\delta,\epsilon)$
\begin{IEEEeqnarray}{rCl}
  \frac{1}{n}\E{\inner{{\bf U}_1^{\ast}}{{\bf U}_2^{\ast}}} & \leq & \sigma^2
  \rho (1-2^{-2R_1}) (1-2^{-2R_2}) + \zeta_2(\delta,\epsilon) \nonumber\\
  & = & \mat{k}_{12} + \zeta_2(\delta,\epsilon), \label{eq:si-ub-U1U2}
\end{IEEEeqnarray}
where $\zeta_2(\delta,\epsilon)$ is such that $\lim_{\delta,\epsilon
  \rightarrow 0} \zeta_2(\delta,\epsilon) = 0$. And by Lemma
\ref{lm3:vq-bd-D1genie} we have that for every $\delta > 0$ and $0 <
\epsilon < 0.3$ there exists an $n'(\delta,\epsilon) \in \Naturals$
such that for all $n > n'(\delta,\epsilon)$
\begin{IEEEeqnarray}{rCl}
  \frac{1}{n}\E{\inner{{\bf S}_1}{{\bf U}_2^{\ast}}} & \geq & \sigma^2 \rho
  (1-2^{-2R_2}) - \zeta_3(\delta,\epsilon) \nonumber\\
  & = & \mat{c}_{12} - \zeta_3(\delta,\epsilon), \label{eq:si-lb-S1U2}
\end{IEEEeqnarray}
where $\zeta_3(\delta,\epsilon)$ is such that $\lim_{\delta,\epsilon
  \rightarrow 0} \zeta_3(\delta,\epsilon) = 0$. Next, recalling that
${\bf Y} = \alpha_1 {\bf S}_1 + \beta_1 {\bf U}_1^{\ast} + \alpha_2
{\bf S}_2 + \beta_2 {\bf U}_2^{\ast} + {\bf Z} $, gives
\begin{IEEEeqnarray}{rCl}
  \frac{1}{n}\E{\inner{{\bf S}_1}{{\bf Y}}} & = \frac{1}{n} \Big( & \alpha_1 \E{\| {\bf S}_1 \|^2}
  + \beta_1 \E{\inner{{\bf S}_1}{{\bf U}_1^{\ast}}} + \alpha_2
  \E{\inner{{\bf S}_1}{{\bf S}_2}} \nonumber\\[2mm]
  & & {} + \beta_2 \E{\inner{{\bf S}_1}{{\bf U}_2^{\ast}}} +
  \underbrace{\E{\inner{{\bf S}_1}{{\bf Z}}}}_{=0} \Big)\nonumber\\[-1mm]
  & \stackrel{a)}{\geq} & \alpha_1 \sigma^2 + \beta_1 \left( \mat{c}_{11} -
    \zeta_1(\delta,\epsilon) \right) + \alpha_2 \rho \sigma^2 +
  \beta_2 \left( \mat{c}_{12} - \zeta_3(\delta,\epsilon) \right) \nonumber\\[2mm]
  & = & \mat{c}_{13} - \zeta_4(\delta,\epsilon), \label{eq:si-lb-S1Y}
\end{IEEEeqnarray}
where in $a)$ we have used \eqref{eq:si-lb-S1U1},
\eqref{eq:si-ub-U1U2} and \eqref{eq:si-lb-S1U2}, and where
$\zeta_4(\delta,\epsilon)$ is such that $\lim_{\delta,\epsilon
  \rightarrow 0} \zeta_4(\delta,\epsilon) = 0$. For the remaining
terms in \eqref{eq:si-dist1-S1-S1G}, it can be shown, similarly as for
\eqref{eq:si-lb-S1U1} and \eqref{eq:si-lb-S1U2}, that for every
$\delta > 0$ and $0 < \epsilon < 0.3$ there exists an
$n''(\delta,\epsilon) \in \Naturals$ such that for all $n >
n''(\delta,\epsilon)$
\begin{IEEEeqnarray}{rCl}
  \frac{1}{n}\E{\inner{{\bf S}_1}{{\bf U}_1^{\ast}}} & \leq & \mat{c}_{11} +
  \zeta_5(\delta,\epsilon) \label{eq:si-ub-S1U1}\\[2mm]
  \frac{1}{n}\E{\inner{{\bf S}_2}{{\bf U}_1^{\ast}}} & \leq & \mat{c}_{21} +
  \zeta_6(\delta,\epsilon) \label{eq:si-ub-S2U1}\\[2mm]
  \frac{1}{n}\E{\inner{{\bf S}_1}{{\bf U}_2^{\ast}}} & \leq & \mat{c}_{12} +
  \zeta_7(\delta,\epsilon) \label{eq:si-ub-S1U2}\\[2mm]
  \frac{1}{n}\E{\inner{{\bf S}_2}{{\bf U}_2^{\ast}}} & \leq & \mat{c}_{22} +
  \zeta_8(\delta,\epsilon), \label{eq:si-ub-S2U2}
\end{IEEEeqnarray}
where $\zeta_j(\delta,\epsilon)$, $j \in \{ 5, \ldots 8 \}$, are such
that $\lim_{\delta,\epsilon \rightarrow 0} \zeta_j(\delta,\epsilon) =
0$. Using \eqref{eq:si-ub-U1U2} and \eqref{eq:si-ub-S1U1} --
\eqref{eq:si-ub-S2U2}, we now get that for every $\delta > 0$ and $0 <
\epsilon < 0.3$ there exists an $\tilde{n}_1(\delta,\epsilon) \in
\Naturals$ such that for all $n > \tilde{n}_1(\delta,\epsilon)$
\begin{IEEEeqnarray}{rCl}
  \frac{1}{n} \E{\inner{{\bf U}_1^{\ast}}{{\bf Y}}} & = & \frac{1}{n}
  \Big( \alpha_1 \E{\inner{{\bf U}_1^{\ast}}{{\bf S}_1}} + \beta_1
  \E{\| {\bf U}_1^{\ast} \|^2} + \alpha_2 \E{\inner{{\bf
        U}_1^{\ast}}{{\bf S}_2}} \nonumber\\[1mm] 
  & & \hspace{5mm} {} + \beta_2 \E{\inner{{\bf U}_1^{\ast}}{{\bf U}_2^{\ast}}} +
  \underbrace{\E{\inner{{\bf U}_1^{\ast}}{{\bf Z}}}}_{=0} \Big)\nonumber\\[1mm]
  & \leq & \alpha_1 \left(\mat{k}_{11} + \zeta_5(\delta,\epsilon)\right)
  + \beta_1 \mat{k}_{11} + \alpha_2 \left( \mat{c}_{21} +
    \zeta_6(\delta,\epsilon) \right) \nonumber\\[1mm]
  & & {} + \beta_2 \left( \mat{k}_{12} + \zeta_2(\delta,\epsilon)
  \right) \nonumber\\[3mm]
  & = & \mat{k}_{13} + \tilde{\zeta}_1(\delta,\epsilon), \label{eq:si-ub-U1Y}
\end{IEEEeqnarray}
where $\tilde{\zeta}_1(\delta,\epsilon)$ is such that
$\lim_{\delta,\epsilon \rightarrow 0} \tilde{\zeta}_1(\delta,\epsilon)
= 0$. Similarly, it can be shown that for every $\delta > 0$ and $0 <
\epsilon < 0.3$ there exists an $\tilde{n}_2(\delta,\epsilon) \in
\Naturals$ such that for all $n > \tilde{n}_2(\delta,\epsilon)$
\begin{IEEEeqnarray}{rCl}
  \frac{1}{n}\E{\inner{{\bf U}_2^{\ast}}{{\bf Y}}} & \leq & \mat{k}_{23} +
  \tilde{\zeta}_2(\delta,\epsilon), \label{eq:si-ub-U2Y}
\end{IEEEeqnarray}
where $\tilde{\zeta}_2(\delta,\epsilon)$ is such that
$\lim_{\delta,\epsilon \rightarrow 0} \tilde{\zeta}_2(\delta,\epsilon)
= 0$. And finally, we have that for every $\delta > 0$ and $0 <
\epsilon < 0.3$ there exists an $\tilde{n}_3(\delta,\epsilon) \in
\Naturals$ such that for all $n > \tilde{n}_3(\delta,\epsilon)$
\begin{IEEEeqnarray}{rCl}
  \frac{1}{n}\E{\|{\bf Y}\|^2} & = & \frac{1}{n} \bigg( \alpha_1^2 \E{\| {\bf S}_1 \|^2} + 2
  \alpha_1 \beta_1 \E{\inner{{\bf S}_1}{{\bf U}_1^{\ast}}} + 2
  \alpha_1 \alpha_2 \E{\inner{{\bf S}_1}{{\bf S}_2}} + 2 \alpha_1
  \beta_2 \E{\inner{{\bf S}_1}{{\bf U}_2^{\ast}}} \nonumber\\
  & & \hspace{5mm} {} + 2 \alpha_1 \underbrace{\E{\inner{{\bf S}_1}{{\bf Z}}}}_{=0}
  + \beta_1^2 \E{\| {\bf U}_1^{\ast} \|^2} + 2 \beta_1 \alpha_2 \E{\inner{{\bf
        U}_1^{\ast}}{{\bf S}_2}} + 2 \beta_1 \beta_2 \E{\inner{{\bf
        U}_1^{\ast}}{{\bf U}_2^{\ast}}} \nonumber\\
  & & \hspace{5mm} {} + 2 \beta_1 \underbrace{\E{\inner{{\bf U}_1^{\ast}}{{\bf Z}}}}_{=0} +
  \alpha_2^2 \E{\| {\bf S}_2 \|^2} + 2 \alpha_2 \beta_2 \E{\inner{{\bf
        S}_2}{{\bf U}_2^{\ast}}} + 2 \alpha_2 \underbrace{\E{\inner{{\bf
          S}_2}{{\bf Z}}}}_{=0} \nonumber\\
  & & \hspace{5mm} {} + \beta_2^2 \E{\| {\bf U}_2^{\ast} \|^2} + 2 \beta_2
  \underbrace{\E{\inner{{\bf U}_2^{\ast}}{{\bf Z}}}}_{=0} + \E{\| {\bf
      Z} \|^2} \bigg) \nonumber\\[2mm]
  & \leq & \alpha_1^2 \sigma^2 + 2 \alpha_1 \beta_1 (\mat{c}_{11}
  + \zeta_5(\delta,\epsilon)) + 2 \alpha_1 \alpha_2 \rho \sigma^2
  + 2 \alpha_1 \beta_2 (\mat{c}_{12} + \zeta_7(\delta,\epsilon)) \nonumber\\
  & & {} + \beta_1^2 \mat{k}_{11} + 2 \beta_1 \alpha_2 (\mat{c}_{21}
  + \zeta_6(\delta,\epsilon)) + 2 \beta_1 \beta_2 (\mat{k}_{12}
  + \zeta_2(\delta,\epsilon)) \nonumber\\
  & & {} + \alpha_2^2 \sigma^2 + 2 \alpha_2 \beta_2 (\mat{c}_{22}
  + \zeta_8(\delta,\epsilon)) \nonumber\\
  & & {} + \beta_2^2 \mat{k}_{22} + N \nonumber\\[3mm]
  & = & \mat{k}_{33} + \tilde{\zeta}_3(\delta,\epsilon)), \label{eq:si-ub-Y2}
\end{IEEEeqnarray}
where $\tilde{\zeta}_3(\delta,\epsilon)$ is such that
$\lim_{\delta,\epsilon \rightarrow 0} \tilde{\zeta}_3(\delta,\epsilon)
= 0$. Thus, combining \eqref{eq:si-lb-S1U1} -- \eqref{eq:si-lb-S1Y}
and \eqref{eq:si-ub-U1Y} -- \eqref{eq:si-ub-Y2} with
\eqref{eq:si-dist1-S1-S1G} gives that for every $\delta > 0$ and $0 <
\epsilon < 0.3$ there exists an $n'(\delta,\epsilon) \in \Naturals$
such that for all $n > n'(\delta,\epsilon)$
\begin{IEEEeqnarray}{rCl}
  \frac{1}{n} \E{\| {\bf S}_1 - \hat{\bf S}_1^{\textnormal{G}} \|^2} &
  \leq & \sigma^2 - 2 \gamma_{11} \mat{c}_{11} - 2 \gamma_{12}
  \mat{c}_{12} - 2\gamma_{13} \mat{c}_{13} + \gamma_{11}^2
  \mat{k}_{11} + 2 \gamma_{11} \gamma_{12} \mat{k}_{12} \nonumber\\
  & & {} + 2 \gamma_{11} \gamma_{13} \mat{k}_{13} + \gamma_{12}^2
  \mat{k}_{22} + \gamma_{12}\gamma_{13} \mat{k}_{23} + \gamma_{13}^2
  \mat{k}_{33} + \zeta'(\delta,\epsilon) \nonumber\\[3mm]
  & = & \sigma^2 - 2 \gamma_{11} \mat{c}_{11} - 2 \gamma_{12}
  \mat{c}_{12} - 2\gamma_{13} \mat{c}_{13} \nonumber\\[-1mm]
  & & \hspace{30mm} {} + \left( \begin{array}{c c c}
      \hspace{-1mm}\gamma_{11} & \hspace{-1mm} \gamma_{12} &
      \hspace{-1mm} \gamma_{13} \hspace{-1mm}\end{array}\right)
  \mat{K} \left( \begin{array}{c} \gamma_{11}\\ \gamma_{12}\\
      \gamma_{13} \end{array}\right) + \zeta'(\delta,\epsilon) \nonumber\\[3mm]
  & \stackrel{a)}{=} & \sigma^2 - 2 \gamma_{11} \mat{c}_{11} - 2
  \gamma_{12}
  \mat{c}_{12} - 2\gamma_{13} \mat{c}_{13} \nonumber\\[-1mm]
  & & \hspace{30mm} {} + \left( \begin{array}{c c c} \hspace{-1mm}
      \gamma_{11} & \hspace{-1mm} \gamma_{12} & \hspace{-1mm}
      \gamma_{13} \hspace{-1mm} \end{array}\right) \mat{K}
  \mat{K}^{-1} \left( \begin{array}{c} \mat{c}_{11}\\ \mat{c}_{12}\\
      \mat{c}_{13} \end{array}\right) +
  \zeta'(\delta,\epsilon) \nonumber\\[3mm]
  & = & \sigma^2 - 2 \gamma_{11} \mat{c}_{11} - 2 \gamma_{12}
  \mat{c}_{12} - 2\gamma_{13} \mat{c}_{13} \nonumber\\[-1mm]
  & & \hspace{30mm} {} + \left( \begin{array}{c c c} \hspace{-1mm}
      \gamma_{11} & \hspace{-1mm} \gamma_{12} & \hspace{-1mm}
      \gamma_{13} \hspace{-1mm} \end{array}\right)
  \left( \begin{array}{c} \mat{c}_{11}\\ \mat{c}_{12}\\
      \mat{c}_{13} \end{array}\right) + \zeta'(\delta,\epsilon) \nonumber\\[3mm]
  & = & \sigma^2 - \gamma_{11} \mat{c}_{11} - \gamma_{12} \mat{c}_{12}
  - \gamma_{13} \mat{c}_{13} +
  \zeta'(\delta,\epsilon), \label{eq:si-dist2-S1-S1G}
\end{IEEEeqnarray}
where we have used the shorthand notation $\mat{K}$ for
$\mat{K}(R_1,R_2)$, and where in $a)$ we have used the definition of
the coefficients $\gamma_{ij}$ in \eqref{eq:src-estmt-coeff}, and
where $\zeta'(\delta,\epsilon)$ is such that $\lim_{\delta,\epsilon
  \rightarrow 0} \zeta'(\delta,\epsilon) = 0$. Now, letting in
\eqref{eq:si-dist2-S1-S1G} first $n \rightarrow \infty$ and then
$\delta, \epsilon \rightarrow 0$, and combining the result with
Corollary \ref{cor:si-genie} gives
\begin{IEEEeqnarray*}{rCl}
  \varlimsup_{n \rightarrow \infty} \E{\| {\bf S}_1 - \hat{\bf
      S}_1 \|^2} & \leq &  \sigma^2 - \gamma_{11}
  \mat{c}_{11} - \gamma_{12} \mat{c}_{12} - \gamma_{13} \mat{c}_{13},
\end{IEEEeqnarray*}
whenever $(R_1,R_2)$ satisfy
\begin{IEEEeqnarray*}{rCl}
  R_1 & < & \frac{1}{2} \log_2 \left( \frac{\beta_1'^2 \| {\bf U}_1 \|^2
      (1-\tilde{\rho}^2) + N'}{N' (1-\tilde{\rho}^2)} \right)\\[2mm]
  R_2 & < & \frac{1}{2} \log_2 \left( \frac{\beta_2'^2 \| {\bf U}_2 \|^2
      (1-\tilde{\rho}^2) + N'}{N' (1-\tilde{\rho}^2)}\right)\\[2mm]
  R_1 + R_2 & < & \frac{1}{2} \log_2 \left( \frac{\beta_1'^2 \| {\bf
        U}_1 \|^2 + \beta_2'^2 \| {\bf U}_2 \|^2 + 2 \tilde{\rho}
        \beta_1' \beta_2' \| {\bf U}_1 \| \| {\bf U}_2 \| +
        N'}{N' (1-\tilde{\rho}^2)}\right).
\end{IEEEeqnarray*}

\end{document}